\documentclass{article}

\usepackage{fullpage}
\usepackage{authblk}
\usepackage{enumerate}
\usepackage[hidelinks]{hyperref}

\usepackage[american]{babel}
\usepackage[utf8]{inputenc}
\usepackage[T1]{fontenc}
\usepackage[english=american]{csquotes}
\usepackage[numbers,sort]{natbib}
\usepackage{graphicx}
\usepackage{epstopdf}
\usepackage{marvosym}
\usepackage{amsmath}
\usepackage{amssymb}
\usepackage{amsthm}
\usepackage{mathtools}
\usepackage{color}
\usepackage{algorithmic}

\usepackage[ruled,vlined]{algorithm2e}
\usepackage[normalem]{ulem}
\usepackage{nicefrac}
\usepackage{soul}
\usepackage{amsfonts}    	
\usepackage{xcolor}
\usepackage{csquotes}
\usepackage{float}
\usepackage{url}
\usepackage[format=plain,labelfont=bf]{caption}

\usepackage{subcaption}
\usepackage{capt-of}
\usepackage{lmodern}
\usepackage{booktabs}
\newcommand{\bigcell}[2]{\begin{tabular}{@{}#1@{}}#2\end{tabular}}
\usepackage{longtable}
\usepackage{multirow}
\usepackage{pifont}
\usepackage{rotating}
\usepackage{tabto}
\usepackage{verbatim}
\usepackage{lscape}
\usepackage{svg}

\newcommand\blfootnote[1]{%
  \begingroup
  \renewcommand\thefootnote{}\footnote{#1}%
  \addtocounter{footnote}{-1}%
  \endgroup
}

\newcommand{\props}[3]{\textbf{#1} #2\textbf{:} #3}
\newcommand*\rot{\rotatebox{270}}
\newcommand*\OK{\ding{51}}

\theoremstyle{definition}
\newtheorem{definition}{Definition}

\newtheorem{remark}{Remark}

\newtheorem{observation}{Observation}

\theoremstyle{plain}
\newtheorem{lemma}{Lemma}
\newtheorem{theorem}{Theorem}
\newtheorem{proposition}{Proposition}
\newtheorem{corollary}{Corollary}

\newcommand{\Tnstar}{\mathcal{T}^\ast_n}
\newcommand{\BTnstar}{\mathcal{BT}^\ast_n}
\newcommand{\Tn}{\mathcal{T}_n}
\newcommand{\BTn}{\mathcal{BT}_n}
\newcommand{\Tfb}{T^{\mathit{fb}}_h}
\newcommand{\Tmb}{T^{\mathit{mb}}_n}
\newcommand{\Tcat}{T^{\mathit{cat}}_n}
\newcommand{\Tstar}{T^{\mathit{star}}_n}
\newcommand{\Tgfb}{T^{\mathit{gfb}}_n}

\title{Tree balance indices: a comprehensive survey \textcolor{red}{Attention: superseded manuscript}}
\date{}

\author[1$^\ast$]{Mareike Fischer}
\author[2]{Lina Herbst}
\author[1]{Sophie Kersting}
\author[1]{Luise Kühn}
\author[3]{Kristina Wicke}

\affil[1]{Institute of Mathematics and Computer Science, University of Greifswald, Greifswald, Germany}
\affil[2]{Transmission, Infection, Diversification \& Evolution Group, Max Planck Institute for the Science of Human History, Jena, Germany}
\affil[3]{Department of Mathematics, The Ohio State University, Columbus OH, USA}

\begin{document}
\maketitle

\begin{abstract} 
\textcolor{red}{\textbf{Attention:} This manuscript has been superseded by the SpringerNature book \enquote{Tree balance indices -- A comprehensive survey}, ISBN 978-3-031-39799-8 (hardcover) and 978-3-031-39800-1 (e-book). Please refer to and cite this book when using our results. Thank you!}

Tree balance plays an important role in phylogenetics and other research areas, which is why several indices to measure tree balance have been introduced over the years. Nevertheless, a formal definition of what a balance index actually is and what makes it a useful measure of balance (or, in other cases, imbalance), has so far not been introduced in the literature. While the established indices all summarize the (im)balance of a tree in a single number, they vary in their definitions and underlying principles. It is the aim of the present manuscript to introduce formal definitions of balance and imbalance indices that classify desirable properties of such indices and to analyze and categorize established indices accordingly. In this regard, we review 19 established (im)balance indices from the literature, summarize their general, statistical and combinatorial properties (where known), prove numerous additional results and indicate directions for future research by making explicit open questions and gaps in the literature. We also prove that a few tree shape statistics that have been used to measure tree balance in the literature do not fulfill our definition of an (im)balance index, which might indicate that their properties are not as useful for practical purposes. Moreover, we show that five additional tree shape statistics from other contexts actually are tree (im)balance indices according to our definition. The manuscript is accompanied by the website \url{treebalance.wordpress.com} containing fact sheets of the discussed indices. Moreover, we introduce the software package \verb|treebalance| implemented in $\mathsf{R}$ that can be used to calculate all indices discussed.
\end{abstract}

\textit{Keywords:}
tree balance index, Sackin index, Colless index, total cophenetic index, Yule model, uniform model, phylogenetics 

\blfootnote{$^\ast$Corresponding author\\ \textit{Email address:} email@mareikefischer.de}

\section{Introduction}\label{Sec_Introduction}
Tree shape statistics, in particular measures of tree (im)balance, play an important role in the analysis of phylogenetic trees. They are used in various ways, among other applications, to test evolutionary models (e.g.,~\citep{Aldous_stochastic_2001, Blum2005, Kirkpatrick1993, Mooers1997}), to assess the impact of fertility inheritance and selection (e.g.,~\citep{Blum2006b, Maia2004, Verboom2020}), or to study tumor evolution (e.g.,~\citep{DiNardo2019,Scott2019}). However, the concept of tree balance is not limited to phylogenetics, but plays an important role in other areas of research, such as computer science, as well (see, for example,~\citep{Andersson1993,Nievergelt1973,Roura2013}).

The balance of a tree is usually summarized in a single number, called a balance or imbalance index, and to the present day at least 19 (im)balance indices have been introduced in the literature. These range from old and widely used indices such as the Sackin and Colless index (\citet{Sackin1972,Colless1982}), over rarely used statistics such as the equal weights Colless ($I_2$) index defined by \citet{Mooers1997}, to fairly recent and new approaches such as the total cophenetic index (\citet{Mir2013}) or the rooted quartet index (\citet{Coronado2019}). 

These indices differ not only in the way they are calculated but also in their behavior and properties. Some indices quantify the balance of a tree, others measure imbalance. Some are solely defined for binary (or bifurcating) trees, others also make sense for arbitrary trees. Finally, the indices differ in their range of values, their resolution power, and the ordering they induce on the set of trees with a given leaf number. While the perception prevails that the so-called rooted caterpillar tree (or comb) on $n$ leaves is the unique most imbalanced tree, and -- provided that $n$ is a power of two -- the so-called fully balanced (or fully symmetric) tree of height $\log_2(n)$ is the unique most balanced binary tree, there is no such consensus when $n$ is not a power of two. Not only do different indices consider different trees as \enquote{most balanced}, but they might also disagree on whether there is a unique most balanced tree, or whether several trees should equally be considered as such. Thus, the notion of \emph{the} most balanced tree might not always be unambiguous and will depend on the index used.

In this manuscript, however, we do not aim at resolving this issue or selecting the \enquote{best} balance index to date. On the contrary, we aim at providing both empiricists and theoreticians with a thorough review of the suite of (im)balance indices available. 

Surprisingly, while the terms \enquote{balance index} and \enquote{imbalance index} are frequently used in the literature, a formal definition of either one seems to be missing. We thus establish certain desirable criteria that make a tree shape statistic a balance index, respectively an imbalance index, and provide formal definitions of the two concepts.\footnote{We remark that similar criteria for tree balance indices on a more general class of trees than the ones considered here have recently and independently been introduced in a preprint by \citet{Lemant2021}.} These criteria allow us to categorize the suite of indices available into imbalance indices, balance indices, as well as concepts that are neither (the popular cherry index \cite{McKenzie2000} being one example). In addition, we identify five less established tree shape statistics that do fulfill our definitions. We then provide a thorough review of 19 established (im)balance indices (and due to its popularity the cherry index), summarizing both general, combinatorial, and statistical properties, as well as indicating relatedness among indices. In addition to summarizing and reviewing the current state of the literature, we establish numerous new mathematical results and solve open problems. For instance, for several of the established indices used in the literature, we prove that they indeed satisfy the criteria of an (im)balance index. Simultaneously, by making explicit open questions and gaps in the literature, we aim at inspiring and stimulating new research in this field. In addition to this manuscript, we have created a website on tree balance (\url{treebalance.wordpress.com}) that will be updated regularly to serve the community as a comprehensive reference for the suite of (im)balance indices available as well as a repository of their properties and open problems.

Finally, by introducing our software package \verb|treebalance| implemented in the free programming language $\mathsf{R}$, we provide a convenient and unifying way to calculate all (im)balance indices discussed in this study without having to use different packages for different indices. In fact, while there currently exist at least 17 different packages that allow for the calculation of some of the indices discussed in this study (cf. Table~\ref{Table_ExistingTools}), we aim at providing a tool that can be used to calculate all of them. 

\paragraph{Using this manuscript as a reference guide.}
In order to allow the reader a quick and easy access to the comprised information, we have organized the remainder of this manuscript as follows: We first review some central definitions and notations which are used throughout this manuscript as well as on the website (Section~\ref{Sec_Notation}). Afterwards, in Section~\ref{Sec_Def_Balance}, we first give precise definitions of (im)balance indices and then introduce 19 established indices of tree (im)balance from the literature. More precisely, we include the average leaf depth \cite{Shao1990, Kirkpatrick1993}, the $B_1$ \cite{Shao1990} and $B_2$ index \cite{Shao1990}, the Colijn-Plazzotta rank \cite{Colijn2018}, the normal \cite{Colless1982, Shao1990}, corrected \cite{Heard1992}, quadratic \cite{Bartoszek2021} and equal weights Colless index \cite{Mooers1997}, the family of Colless-like indices \cite{Mir2018}, the mean $I'$ index \cite{Fusco1995, Purvis2002}, the Rogers $J$ index \cite{Rogers1996}, the Furnas rank \cite{Furnas1984, Kirkpatrick1993}, the rooted quartet index \cite{Coronado2019}, the $\widehat{s}$-shape statistic \cite{blum2006c}\footnote{Note that \citet{blum2006c} called this index $s$; to avoid confusion with other notation used in this manuscript, we refer to it as $\widehat{s}$.}, the Sackin index \cite{Sackin1972, Shao1990}, the symmetry nodes index \cite{Kersting2021}, the total cophenetic index \cite{Mir2013}, and the variance of leaf depths \cite{Sackin1972, Coronado2020b}. 
In Section \ref{Sec_Combinatorial_and_Statistical} we introduce the general, combinatorial, and statistical properties that will subsequently be discussed for all 19 indices.
Table \ref{table_Overview} of the same section allows the reader to quickly assess for each of the established indices which of those properties have already been analyzed in the literature before (marked with \OK), which are addressed in this manuscript (marked with $*$) and which properties are still unknown. The table is thus a summary of Section \ref{Sec_factsheets}, where we provide comprehensive fact sheets for all indices including statements on their properties, references to the original sources of the statements or to the respective proofs in Appendix \ref{app_additionals}, comments on gaps in the literature as well as efforts that have been made in filling them. Note that all gaps in the literature are marked as \emph{open problems}, even if partial results have already been obtained (for instance, there might exist a recursive formula for the minimum value of a certain (im)balance index but not an explicit one). In Section \ref{Sec_tssAreBalInd}, we then discuss several additional tree shape statistics that satisfy our definition of an (im)balance index but have not been thoroughly analyzed in terms of tree balance in the literature yet. Finally, in Section~\ref{Sec_nbi_tss}, we consider further tree shape statistics, some of them used as measures of tree (im)balance in the literature, that do not satisfy our definition of an (im)balance index.
Following this, we briefly discuss approaches for obtaining new (im)balance indices from established ones (Section~\ref{Sec_new}), discuss approaches of normalizing (im)balance indices (Section~\ref{Sec_normalization}), and mention some further concepts related to tree (im)balance (Section~\ref{Sec_related}). In Section \ref{Sec_Software}, we introduce our software package \verb|treebalance|, before giving a brief summary and discussion of our results in Section~\ref{Sec_Discussion1} and indicating some directions for future research in Section~\ref{Sec_FutureResearch}.

Last but not least, the appendix is divided into three parts: In the first part, we fill numerous gaps in the literature concerning the established indices by providing proofs for results that have either been unknown until now or have been mentioned before but to our knowledge have not been formally proven yet. In the second part, we prove that the maximal width~\cite{Colijn_phylogenetic_2014}, the maximal difference in widths~\cite{Colijn_phylogenetic_2014}, and the maximal depth~\cite{Colijn_phylogenetic_2014} fulfill our definition of an (im)balance index, although they have not been strongly linked to tree balance before. Finally, in the third part, we provide some additional figures accompanying Table~\ref{Table_nbi_tss}, as well as results concerning the cherry index~\cite{McKenzie2000} and the so-called clades of size $x$ measure~\cite{Rosenberg_mean_2006}, both of which do not satisfy our definition of an (im)balance index.

\section{Preliminaries} \label{Sec_Notation}

To begin, we need to introduce some definitions, notations, and concepts that will be of relevance throughout this manuscript.

\subsection{General notation and concepts}
First, consider an integer $n \in \mathbb{N}_{\geq 0}$ and its binary expansion $n = \sum\limits_{i=0}^N n_i\cdot 2^i$ with $N \in \mathbb{N}_{\geq 0}$ being the maximal index such that $n_N=1$. Then, the \emph{binary weight} of $n$, denoted by $wt(n)$, is given as $wt(n) = \sum\limits_{i=0}^N n_i $, i.e. the number of 1's in the binary expansion of $n$.

Second, by $H_n$ we will denote the \emph{$n$-th harmonic number} $H_n = \sum\limits_{i=1}^n \frac{1}{i}$, and by $H_n^{(2)}$ we will denote the quantity $H_n^{(2)} = \sum\limits_{i=1}^n \frac{1}{i^2}$.

Third, by $\mathcal{I}(x)$ we denote the indicator function, whose result equals 1 if the expression $x$ is true and $0$ if the expression $x$ is false.

Fourth, whenever we use logarithms in this manuscript, we simply write $\log$ if the logarithm base is irrelevant. If the logarithm base is relevant, we indicate this by writing $\log_b$, where $b$ is the base (as for example in $\log_2$) or $\ln = \log_e$ in case of the natural logarithm.

Finally, we are using the conventions that a sum with an empty index set evaluates to zero and a product with an empty index set evaluates to one. We are also using $\frac{0}{0}=0$. Although this is not a general mathematical convention, this assumption appears quite regularly in the literature (explicitly in \citep[Footnote 1]{Matsen2007} and implicitly in \citep{Lima2020, Heard1992, Hitchin1997, Heard2007, Mooers1997, Kirkpatrick1993}) and it is helpful when dealing with initial values of some indices (e.g. corrected Colless index). 

\subsection{Rooted trees and related concepts}
\paragraph{Rooted (binary) trees.}
Throughout this paper, by a \emph{tree} we mean a non-empty \emph{rooted tree} without vertices of in-degree and out-degree 1, that is a directed graph $T=(V(T),E(T))$ with vertex set $V(T)$ and edge set $E(T)$, containing precisely one vertex of in-degree 0, the \emph{root} (denoted by $\rho$), such that for every $v \in V(T)$ there exists a unique path from $\rho$ to $v$, and such that there are no vertices with out-degree (denoted $deg^+$) 1. In particular, the edges are directed away from the root. We use $V_L(T) \subseteq V(T)$ to refer to the leaf set of $T$ (i.e. $V_L(T) = \{ v \in V(T): deg^+(v)=0\}$), and we use $\mathring{V}(T)$ to denote the set of inner vertices of $T$, i.e. $\mathring{V}(T) = V(T) \setminus V_L(T)$. We generally use $n$ to denote the number of leaves of $T$, i.e. $n=|V_L(T)|$. Note that $\rho \in \mathring{V}(T)$ if $n\geq 2$. If $n=1$, the tree consists of only one vertex, which is at the same time the root and its only leaf. We consider two trees as equal when they are isomorphic.

A rooted tree is called \emph{binary} if all inner vertices have out-degree 2, and for every $n \in \mathbb{N}_{\geq 1}$ we denote by $\BTnstar$ the set of (isomorphism classes of) rooted binary trees with $n$ leaves and by $\Tnstar$ the set of (isomorphism classes of) all rooted trees with $n$ leaves. Moreover, we often refer to a vertex with out-degree 2 as a \emph{binary vertex} or \emph{binary node}.

\paragraph{Ancestors, descendants, and cherries.}
Whenever there exists a path from a vertex $u$ to a vertex $v$ in $T$, we say that $u$ is an \emph{ancestor} of $v$ and $v$ is a \emph{descendant} of $u$. If $u$ and $v$ are connected by an edge, i.e. if $(u,v) \in E(T)$, we also say that $u$ is the \emph{parent} of $v$ and $v$ is a \emph{child} of $u$. The set of ancestors of a vertex $v$ excluding $v$ itself will be denoted $anc(v)$ and the set of children of $v$ will be denoted by $child(v)$.
Two leaves, say $x,y \in V_L(T)$ are said to form a \emph{cherry}, denoted by $[x,y]$, if they have the same parent. Note that every rooted (not necessarily binary) tree with at least two leaves has at least one cherry. We denote by $c(T)$ the number of cherries of $T$. 
The \emph{lowest common ancestor} $LCA_T(u,v)$ of two vertices $u, v \in V(T)$ is the unique common ancestor of $u$ and $v$ that is a descendant of every other common ancestor of them.

Given a leaf $x\in V_L(T)$, we denote by $p_x(T)$ the probability of reaching $x$ when starting at the root and assuming equiprobable branching at each inner vertex, more precisely $p_x(T)=\prod\limits_{v\in anc(x)} \frac{1}{|child(v)|}$. Note that in a binary tree we have $p_x(T)=(1/2)^{|anc(x)|}=(1/2)^{\delta_T(x)}$, because each inner vertex has exactly two children. Whenever there is no ambiguity, we use the shorthand $p_x$ for $p_x(T)$.

\paragraph{Depth, height, and width.}
The \emph{depth} $\delta_T(v)$ of a vertex $v \in V(T)$ is the number of edges on the path from $\rho$ to $v$, and the \emph{height} $h(T)$ of $T$ is the maximum depth of any leaf of $T$, i.e. $h(T) = \max\limits_{x \in V_L(T)} \delta_T(x)$. The \emph{width} $w_T(i)$ of $T$ at depth $i$ is the number of vertices $v \in V(T)$ that have $\delta_T(v)=i$.

\paragraph{Cophenetic value and nodal distance.}
Given two leaves $x,y \in V_L(T)$, their \emph{cophenetic value} $\varphi_T(x,y)$ (\citet{Sokal1962}) is the depth of their lowest common ancestor, i.e. $\varphi_T(x,y) = \delta_T(LCA_T(x,y))$. Moreover, the \emph{nodal distance} $d_T(x,y)$ between $x$ and $y$ is the number of edges on the unique undirected shortest path connecting them.

\paragraph{Pending subtrees and decomposition of rooted trees.}
Given a tree $T$ and an arbitrary vertex $v \in V(T)$, we denote by $T_v$ the pending subtree of $T$ rooted at $v$ and we use $n_v$ (or in case of ambiguity $n_T(v)$) to denote the number of leaves in $T_v$. Note that we sometimes call the descendant leaves of $v$ a \emph{clade}, and refer to $n_v$ as the \emph{clade size}. Also, recall that any rooted tree $T$ with $n\geq 2$ leaves can be decomposed into its maximal pending subtrees rooted at the children of $\rho$. We refer to this decomposition as $T=(T_1, \ldots, T_k)$, where $k$ refers to the number of children of $\rho$. If not stated otherwise we assume without loss of generality that the maximal pending subtrees of each subtree of $T$ are ordered decreasingly according to their number of leaves and thus, in particular, $n_1 \geq \ldots \geq n_k$.

\paragraph{Rooted quartets.}
Given a tree $T$ and a subset $Y \subseteq V_L(T)$, the \emph{restriction} $T_{|Y}$ of $T$ to $Y$ is the tree obtained from the minimal subtree of $T$ connecting the elements in $Y$ by suppressing all non-root degree-2 vertices. If $|Y|=4$, we call $T_{|Y}$ a \emph{rooted quartet displayed by $T$} and we use $\mathcal{Q}(T)$ to denote the multiset of all rooted quartets displayed by $T$, i.e. $\mathcal{Q}(T) = \{ T_{|Y}: \, Y \subseteq V_L(T), |Y|=4\}$. Note that there are five elements in $\mathcal{T}^\ast_4$, and thus five rooted quartet trees $Q_0^\ast, \ldots, Q_4^\ast$ (cf. Figure \ref{Fig_Quartets}). These are ordered according to their symmetry (measured in terms of their numbers of automorphisms, see \citep{Coronado2019} for more details), and the \emph{$rQI$-value} $q_i$ (\citet{Coronado2019}) that is associated with quartet $i$ is a value that strictly increases with the quartet tree's symmetry, i.e. $0=q_0<q_1<q_2<q_3<q_4$. As stated in \citep{Coronado2019}, the specific numerical values can be chosen in order to magnify the differences in symmetry between specific pairs of trees. For instance, \citet{Coronado2019} suggest to take $q_i = i$, or $q_i = 2^i$.

\begin{figure}[htbp]
    \centering
    \includegraphics[scale=0.165]{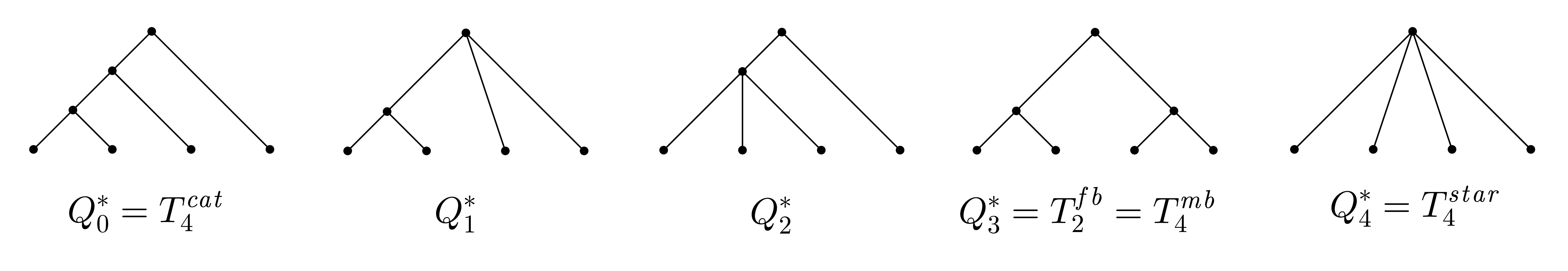}
    \caption{The five elements in $\mathcal{T}^\ast_4$ (Figure adapted from \citet{Coronado2019}).}
    \label{Fig_Quartets}
\end{figure}

\paragraph{Symmetry vertices, balance values, and $I_v$ values.}
Now, let $T$ be a rooted binary tree and let $v \in \mathring{V}(T)$ be an inner vertex of $T$ with children $v_1$ and $v_2$. The \emph{balance value} $bal_T(v)$ of $v$ is defined as $bal_T(v)=|n_{v_1}-n_{v_2}|$ and $v$ is called \emph{balanced} if it fulfills $bal_T(v)\leq 1$ and \emph{perfectly balanced} if $bal_T(v)=0$. Moreover, we call $v$ a \emph{symmetry vertex} if $T_{v_1}$ and $T_{v_2}$ are isomorphic. The number of symmetry vertices of $T$ is denoted $s(T)$.

Referring to \citet{Mir2018}, let $T\in\Tnstar$ be a rooted tree and let $f: \mathbb{N}_{\geq 0} \rightarrow \mathbb{R}_{\geq 0}$ be a function that maps any natural number to a non-negative real number. Then, the \emph{$f$-size} of $T$ is defined as $\Delta_f(T) = \sum\limits_{v \in V(T)} f(deg^+(v))$, which is a weighted sum where the out-degree of each vertex is weighted by means of the function $f$. Also, let $\mathbb{R}^+ = \{(x_1, \ldots, x_k) \, | \, k \geq 1, x_1, \ldots, x_k \in \mathbb{R}\}$ be the set of all non-empty finite-length sequences of real numbers. A \emph{dissimilarity} on $\mathbb{R}^+$ is any mapping $D: \mathbb{R}^+ \rightarrow \mathbb{R}_{\geq 0}$ satisfying the following two conditions: For every $(x_1, \ldots, x_k) \in \mathbb{R}^+$ we have $D(x_1, \ldots, x_k) = D(x_{\sigma(1)}, \ldots, x_{\sigma(k)})$ for every permutation $\sigma$ of $\{1, \ldots, k\}$ and we have $D(x_1, \ldots, x_k)=0$ if and only if $x_1 = \ldots = x_k$. Examples of dissimilarities in this sense include the (sample) variance, (sample) standard deviation, and mean deviation from the median. Now, given the tuple $(D,f)$ of dissimilarity $D$ and function $f$ the \emph{$(D,f)$-balance value} $bal_{D,f}(v)$ of a vertex $v$ in a tree $T\in\Tnstar$ is defined as $bal_{D,f}(v) = D(\Delta_f(T_{v_1}), \ldots, \Delta_f(T_{v_k}))$.

Now, let $v\in\mathring{V}(T)$ be an inner vertex with $n_v\geq 4$ and $|child(v)|=2$ (denoted $v \in \mathring{V}_{bin,\geq 4}$). Then, its \emph{$I_v$ value}\citep{Fusco1995} is defined as the ratio between the observed deviation of the leaf number of the larger maximal pending subtree of $T_v$ from the minimal possible value and the maximal possible deviation, more precisely $I_v = \frac{n_{v_1}-\lceil\frac{n_v}{2}\rceil}{n_v-1-\lceil\frac{n_v}{2}\rceil}$, where $v_1, v_2$ denote the children of $v$ and $n_{v_1}\geq n_{v_2}$.\footnote{Note that \citet{Fusco1995} introduced the $I_v$ value in a slightly more general way than the one considered here. More precisely, \citep{Fusco1995} allow each leaf of a tree to represent several species and then consider the number of descending terminal species instead of the number of descending leaves for each binary vertex $v$ when calculating the $I_v$ value.} As the expected value of $I_v$ under the Yule model depends on $n_v$, \citet{Purvis2002} introduced two modified versions of it, namely the \emph{$I_v'$ value} and the \emph{$I^w_v$ value}. The $I_v'$ value is defined as $I_v'= \begin{cases} I_v & \text{if } n_v \text{ is odd} \\ \frac{n_v-1}{n_v} \cdot I_v &\text{else} \end{cases}$. The $I^w_v$ value of a vertex $v$ is another weighted version of the $I_v$ value and is defined as  
\[I^w_v = \frac{w(I_v) \cdot I_v}{\text{mean}_{v \in \mathring{V}_{bin,\geq 4}} w(I_v)} \qquad \text{with weights} \qquad w(I_v) = \begin{cases} 
    1 & \text{if } n_v \text{ is odd} \\ 
    \frac{n_v-1}{n_v} & \text{if } n_v \text{ is even and } I_v>0 \\
    \frac{2\cdot(n_v-1)}{n_v} & \text{if } n_v \text{ is even and } I_v=0. 
\end{cases}\]

\paragraph{Special trees.}
Finally, we need to introduce several specific families of trees that will be important in what follows. First, a \emph{rooted star tree}, denoted by $\Tstar$, is a rooted tree with $n$ leaves that fulfills $n=1$, or $n \geq 2$ and has a single inner vertex (the root), which is adjacent to all leaves. Second, a \emph{rooted caterpillar tree}, denoted by $\Tcat$, is a rooted binary tree with $n$ leaves that fulfills $n=1$, or $n \geq 2$ and has exactly one cherry. Third, a \emph{fully balanced tree of height $h$}, denoted by $\Tfb$, is a rooted binary tree with $n=2^h$ leaves with $h \in \mathbb{N}_{\geq 0}$ in which all leaves have depth exactly $h$. Note that for $n \geq 2$ both maximal pending subtrees of a fully balanced tree are again fully balanced trees, and we have $\Tfb=(T^{\mathit{fb}}_{h-1}, T^{\mathit{fb}}_{h-1})$.

Next to the fully balanced tree, following \citet[Definition 3.4]{Kersting2021}, we introduce a special subset of rooted binary trees involving fully balanced subtrees.
More precisely, for each $n \in \mathbb{N}_{\geq 1}$, the set of \emph{rooted binary weight trees} $\widetilde{T}_n$ contains all trees that can be constructed as follows: Let $wt(n)$ denote the binary weight of $n$. Then, a tree $ \widetilde{T} \in \widetilde{T}_n$ consists of an arbitrary tree $\widetilde{T}_{top}$ with $wt(n)$ leaves, called the \enquote{top tree}, and a collection of distinct fully balanced subtrees $(T_i^\mathit{fb})_{i \in I_n}$ with $I_n$ being the set of indices where the binary expansion of $n$ is 1. These subtrees are attached to the top tree by identifying each leaf of the top tree with the root of one of the fully balanced subtrees (for further details see \citet{Kersting2021}). Note that when $n$ is a power of two, i.e. $n=2^h$ for some $h \in \mathbb{N}_{\geq 0}$, $wt(n)=1$ and there is precisely one rooted binary weight tree, namely $\Tfb$ \citep[Theorem 3.5]{Kersting2021}.

Moreover, a \emph{maximally balanced tree}, denoted by $\Tmb$, is a rooted binary tree with $n$ leaves in which all inner vertices are balanced. Recursively, a rooted binary tree with $n \geq 2$ leaves is maximally balanced if its root is balanced and its two maximal pending subtrees are maximally balanced, i.e. $\Tmb = (T_{\lceil n/2 \rceil}^{\mathit{mb}}, T_{\lfloor n/2 \rfloor}^{\mathit{mb}})$. Note also that $\Tfb = T_{2^h}^{\mathit{mb}}$, because in the special case when $n=2^h$, $\Tfb$ is the unique tree all of whose inner vertices have a balance value of zero.

Finally, a \emph{greedy from the bottom (GFB) tree}, denoted by $\Tgfb$, is a rooted binary tree with $n$ leaves that results from greedily clustering trees of minimal numbers of leaves starting with single vertices and proceeding until only one tree is left as described by Algorithm 2 in \citet{Coronado2020a}.
Alternatively, $\Tgfb$ can be built via a top-down approach by recursively partitioning the number of leaves into two parts of suitable sizes. More precisely, for every $T \in \BTnstar$, we have that $T = \Tgfb$ if and only if for every $v \in \mathring{V}(T)$, if we write $n_v = 2^k+s$ with $k = \lfloor \log_2(n_v) \rfloor$ and $0 \leq s < 2^k$, then the numbers of descendant leaves of the children of $v$ are
respectively $2^{k-1}+s$ and $2^{k-1}$ (if $0 \leq s \leq 2^{k-1}$) or $2^k$ and $s$ (if $2^{k-1} \leq s < 2^k$) (cf. \citet[Corollary 8]{Coronado2020a}).

\subsection{Ordering and enumerating rooted (binary) trees}
\paragraph{Furnas ranking scheme.}
\citet{Furnas1984} introduced the following ordering for rooted binary trees.
\begin{definition}[Left-light rooted ordering; adapted from \cite{Furnas1984}] \label{def_LLR_ordering}
The \emph{left-light rooted ordering} (LLR ordering) $\prec$ is recursively defined as follows: For two rooted binary trees $T'$ and $T$ we have $T'\prec T$ if and only if
\begin{enumerate}
    \item $|V_L(T')|<|V_L(T)|$, or
    \item $|V_L(T')|=|V_L(T)|$ and $T_L'\prec T_L$, or
    \item $|V_L(T')|=|V_L(T)|$ and $T_L'=T_L$ and $T_R'\prec T_R$,
\end{enumerate}
where, provided that $T$ has at least two leaves, $T_L$ and $T_R$ denote the two maximal pending subtrees of $T$ with $T_L\preceq T_R$, i.e. $T_L=T_R$ or $T_L\prec T_R$, and, provided that $T'$ has at least two leaves, $T_L'$ and $T_R'$ denote the two maximal pending subtrees of $T'$ with $T_L'\preceq T_R'$.
\end{definition}

For a fixed value of $n$, the \emph{rank} $r_n(T)$ of a tree $T$ in this ordering of all binary trees with $n$ leaves is one more than the number of other binary trees, $T'$, before $T$ in the ordering. This rank was denoted $RRANK_n(T)$ by \citet{Furnas1984}. However, for brevity, we use $r_n(T)$.

\paragraph{Number of rooted (binary) trees.}
The number of rooted binary trees $T=(T_1,T_2)$ with $n \geq 2$ leaves can be calculated recursively by counting the number of pairings of a subtree $T_2$ of size $n_2$ with a subtree $T_1$ of size $n_1=n-n_2$ for all $n_2=1, \ldots, \lfloor \frac{n}{2} \rfloor$. Formally, the number of rooted binary trees with $n$ leaves is given by the \emph{Wedderburn-Etherington number} $we(n)$ (sequence A001190 in the OEIS\footnote{\url{https://oeis.org/A001190}}), where
\begin{equation}
    we(n) = \sum\limits_{i=1}^{\lfloor \frac{n-1}{2} \rfloor} we(i)\cdot we(n-i) + \frac{1}{2}\cdot we\left(\frac{n}{2}\right)\cdot\left(we\left(\frac{n}{2}\right)+1\right)
\end{equation}
with the initial cases $we(1)=we(2)=1$ and $we(x)=0$ if $x\notin\mathbb{N}_{\geq 1}$.

For the number $|\Tnstar|$ of arbitrary rooted trees with $n$ leaves no explicit formula is known. However, it can be computed recursively (e.g., \citet[Algorithm 1]{Xiang2009}). Note that the numbers $(|\Tnstar|)_{n}$ form sequence A000669 in the OEIS\footnote{\url{https://oeis.org/A000669}}, where more information about it, e.g. a generating function, can be found.
 
\subsection{Probabilistic models of phylogenetic trees}
\paragraph{Phylogenetic trees.}
A \emph{rooted phylogenetic $X$-tree} (or simply \emph{phylogenetic $X$-tree} as this manuscript only considers rooted trees) $\mathcal{T}$ is a tuple $\mathcal{T}=(T, \phi)$, where $T$ is a rooted tree and $\phi$ is a bijection from $V_L(T)$ to $X$. $T$ is often referred to as the \emph{topology} or \emph{tree shape} of $\mathcal{T}$ and $X$ is called the \emph{taxon set} of $\mathcal{T}$. Two phylogenetic $X$-trees are called \emph{isomorphic} if there is an isomorphism between their topologies that preserves leaf labels.
Moreover, a phylogentic $X$-tree $\mathcal{T}$ is called \emph{binary} when $T$ is binary. In the following, we will always assume that $X=\{1, \ldots, n\}$. We use $\BTn$ to denote the space (of isomorphism classes) of binary phylogenetic $X$-trees with $|X|=n$, and similarly we use $\Tn$ to denote the space of all phylogenetic $X$-trees with $n$ leaves.  Note that $|\BTn|=1$ for $n = 1$ and $|\BTn|=(2n-3)!!$ for $n \geq 2$ (\citet[Corollary 2.2.4]{Semple2003}).

\paragraph{Yule and uniform model of binary phylogenetic trees.}
A \emph{probabilistic model of phylogenetic trees $P_n$}, with $n \geq 1$, is a family of probability mappings $P_n: \Tn \rightarrow [0,1]$, associating a phylogenetic tree in $\Tn$ to its probability under the model. Two very important probabilistic models for binary phylogenetic trees are the \emph{Yule-Harding model} (\citet{Yule1925,Harding1971}), and the \emph{uniform model}.

Under the Yule model (also known as equal-rates-Markov-model (ERM), random branching model, Markovian dichotomous branching model, Yule-Harding-Kingman model (YHK) or simply Markovian model), binary phylogenetic trees on $n$ leaves are generated through the following stochastic process: starting with a single vertex, at each step, a leaf is chosen uniformly at random and is replaced by a cherry; when the desired number $n$ of leaves is reached, leaf labels are assigned uniformly at random to the leaves. The probability $P_{Y,n}(\mathcal{T})$ of generating a phylogenetic $X$-tree $\mathcal{T}=(T,\phi)$ under the Yule model is then given by (see, for instance, \citet[Proposition 3.2]{Steel2016})
\begin{align}\label{Prob_Yule}
     P_{Y,n}(\mathcal{T}) &= \frac{2^{n-1}}{n!} \cdot \prod\limits_{v \in \mathring{V}(T)} \frac{1}{n_v-1}.
\end{align}

The uniform model (also known as proportional-to-distinguishable-arrangements (or -types) model (PDA)) on the other hand, simply selects a phylogenetic $X$-tree uniformly at random from $\BTn$ (\citet{Rosen1978}). 
As $|\mathcal{BT}_n|=(2n-3)!!=(2n-3)(2n-5)\cdots1$ for every $n \geq 1$ (with the convention that $(-1)!!=1$; see, for instance, \citet[Corollary 2.2.4]{Semple2003}), the probability $P_{U,n}(\mathcal{T})$ of generating a phylogenetic $X$-tree $\mathcal{T}$ under the uniform model is thus given by 
\begin{align}\label{Prob_Uniform}
     P_{U,n}(\mathcal{T}) &= \frac{1}{(2n-3)!!}.
\end{align}
For more information on these models, see, for instance, Chapter 3 in \citet{Steel2016}. If not stated otherwise we use $E_Y(t(T_n))$ and $V_Y(t(T_n))$ as well as $E_U(t(T_n))$ and $V_U(t(T_n))$ to denote the expected value and the variance of some (im)balance index $t$ computed for a rooted binary tree $T_n$ with $n$ leaves randomly drawn under the Yule or uniform model, respectively.

\section{Summary of tree balance indices} \label{Sec_Summary}

In this section, we summarize the current state of the literature (including results obtained in this manuscript) on balance and imbalance indices. In doing so, we first give precise definitions for both. We then list definitions and short descriptions of 15 established indices of tree imbalance (Table \ref{Table_Imbalance}) and 4 of tree balance (Table \ref{Table_Balance}) and exemplarily calculate their values for all rooted binary trees with 6 leaves (Table \ref{Table_Examples}). Table \ref{table_Overview} summarizes the current state of knowledge on the topic by providing a quick overview of some general, combinatorial and statistical properties that we consider relevant for empiricists and theoreticians. Open questions and gaps in the literature are explicitly pointed out. A comprehensive fact sheet for each index containing the properties in Table \ref{table_Overview} is given in Section \ref{Sec_factsheets} and on the website \url{treebalance.wordpress.com}. In addition, we give a list of tree shape statistics that comply with our definition of an (im)balance index but did not receive a fact sheet because they have not been used in balance analyses yet nor have they been extensively analyzed so far. Finally, we briefly discuss 13 additional tree shape statistics that do not satisfy our criteria for an (im)balance index, even though some of them, e.g. the cherry index~\cite{McKenzie2000}, are used as measures of tree (im)balance in the literature.

\subsection{Definition of tree (im)balance} \label{Sec_Def_Balance}

Although the terms \enquote{balance index} and \enquote{imbalance index} are used quite regularly in the literature, a precise definition for either one seems to be missing.\footnote{As indicated in Section \ref{Sec_Introduction}, similar criteria for a more general class of trees than the ones considered here, were recently established by \citet{Lemant2021}.}  We have thus gathered some requirements that the mathematical community seems to agree on and combined them in Definition \ref{def_balance} and \ref{def_imbalance}. To begin with, a function $t:\mathcal{T}^\ast(\mathcal{BT}^\ast)\rightarrow\mathbb{R}_{\geq 0}$ is called a \textit{(rooted binary) tree shape statistic (TSS)} if $t(T)$ depends only on the shape of $T$ and not on the labeling of vertices or the lengths of edges. Now, we can define tree (im)balance indices as follows:

\begin{definition}[Balance index]\label{def_balance}
A (binary) tree shape statistic $t$ is called a \emph{balance index} if and only if 
\begin{enumerate}[i)]
    \item the caterpillar tree $\Tcat$ is the unique tree minimizing $t$ on $\Tnstar(\BTnstar)$ for all $n\geq 1$,
    \item the fully balanced tree $\Tfb$ is the unique tree maximizing $t$ on $\BTnstar$ for all $n=2^h$ with $h\in\mathbb{N}_{\geq 0}$.
\end{enumerate}
\end{definition}

\begin{definition}[Imbalance index] \label{def_imbalance}
A (binary) tree shape statistic $t$ is called an \emph{imbalance index} if and only if
\begin{enumerate}[i)]
    \item the caterpillar tree $\Tcat$ is the unique tree maximizing $t$ on $\Tnstar(\BTnstar)$ for all $n\geq 1$,
    \item the fully balanced tree $\Tfb$ is the unique tree minimizing $t$ on $\BTnstar$ for all $n=2^h$ with $h\in\mathbb{N}_{\geq 0}$.
\end{enumerate}
\end{definition}

Note that the second condition is restricted to $\BTnstar$, while the first condition must hold for $\Tnstar$ or $\BTnstar$ depending on the domain of $t$. If a balance (imbalance) index $t$ is defined for arbitrary (i.e. not just binary) trees, it is also desirable that the rooted star tree $\Tstar$ maximizes (minimizes) $t$ on $\Tnstar$ for all $n\geq 1$. But since this criterion appears only sporadically in the literature, we have excluded it from the definition.

Also, note that we define all indices for unlabeled trees even though some of them were originally defined for phylogenetic (i.e. labeled) trees (e.g. the rooted quartet index), because it is commonly accepted that the degree of balance should not depend on the vertex labels. The (im)balance value of a phylogenetic tree is then simply the (im)balance value of its topology. Moreover, we assume that all leaves correspond to taxonomic units of equal rank (e.g., leaves might represent species or genera or families, but the taxonomic ranks do not vary across the tree) and we thus assume a one-to-one correspondence between the number of leaves of a tree and the number of taxonomic units it represents. We remark, however, that it has been argued that a balance index should not rely on this assumption and should thus not depend on the number of leaves of a tree but on the total number of species it represents (cf. \citet{Purvis2002}).

Note that while we require balance indices to agree on the most symmetric and asymmetric trees, their assessment of other trees can be (and is in many cases) contradicting. 
In general, if $t$ is a balance (imbalance) index and $t(T)>t(T')$ for two trees $T,T'\in\Tnstar(\BTnstar)$ we say that $T$ is \emph{more (less) balanced} than $T'$ regarding $t$. This means that the value of a balance index increases with increasing balance, while the value of an imbalance index decreases with increasing balance. Please note that this statement is usually restricted to trees that have the same number of leaves. This should be kept in mind as some of the following indices are normalized -- enabling a comparison of trees with different leaf numbers -- while others are not.

Now, consider Table \ref{Table_Imbalance} and \ref{Table_Balance} for the definitions and explanatory descriptions of 15 imbalance and 4 balance indices.

\begin{longtable}{lll}
\caption{List of established measures that comply with our definition of an imbalance index (Definition \ref{def_imbalance}). Indices marked with (a) make sense for arbitrary trees, whereas indices marked with (a*) are restricted to a subset, generally trees that have only a small fraction of non-binary vertices, and indices marked with (b) are only defined for binary trees.
Note that although the $\widehat{s}$-shape statistic is defined for arbitrary logarithm base, it is common to use $\log_2$ when working with binary trees. Also note that despite being defined for arbitrary trees, the $\widehat{s}$-shape statistic only satisfies the criteria of a balance index on the space of rooted binary trees $\BTnstar$, while it violates them on $\Tnstar$ (for details see Section \ref{factsheet_s-shape}).
For additional details on an index we refer the interested reader to the respective fact sheet in the indicated section or to the website \url{treebalance.wordpress.com}.} \\
\toprule
\label{Table_Imbalance}
\centering
\textbf{Imbalance indices} & & \\
Index & Description &  Definition\\
\midrule
\endfirsthead
\multicolumn{3}{c}%
{\textbf{Table \ref{Table_Imbalance}}: \textit{Continued from previous page}} \\
\toprule
\textbf{Imbalance indices} & & \\
Index & Description &  Definition\\
\midrule
\endhead
\bigcell{l}{Average leaf depth (a) \\ \citep{Sackin1972, Shao1990} \\ $\rightarrow$ see Section \ref{factsheet_ALD}} & \bigcell{l}{mean depth of the \\ leaves of $T$} & $\overline{N}(T) \coloneqq \frac{1}{n}\cdot\sum\limits_{x \in V_L(T)} \delta_T(x)$ \\ \addlinespace
\bigcell{l}{Colijn-Plazotta rank (b) \\ \citep{Colijn2018,Rosenberg2020} \\ $\rightarrow$ see Section \ref{factsheet_CP}} & \bigcell{l}{recursively defined \\ bijective mapping \\ between rooted \\ binary trees and \\ the positive integers} & \bigcell{l}{$CP(T)\coloneqq 1$ for $T \in \mathcal{BT}_1^\ast$ and for \\ $T=(T_1,T_2) \in \BTnstar$ with $n \geq 2$ \\ $CP(T) \coloneqq \frac{1}{2} CP(T_1)  (CP(T_1)-1)$ \\\phantom{$CP(T) \coloneqq$} $+ CP(T_2) + 1$ \\ (with $CP(T_1) \geq CP(T_2)$)} \\ \addlinespace
\bigcell{l}{Colless index (b) \\ \citep{Colless1982} \\ $\rightarrow$ see Section \ref{factsheet_Colless}} & \bigcell{l}{sum of the balance \\ values of the inner \\ vertices of $T$} & $C(T) \coloneqq \sum\limits_{v \in \mathring{V}(T)} bal_T(v)$ \\ \addlinespace
\bigcell{l}{Colless-like indices (a) \\ \citep{Mir2018} \\ $\rightarrow$ see Section \ref{factsheet_Colless-like}} & \bigcell{l}{sum of the $(D,f)$-balance \\ values of the inner \\ vertices of $T$}& \bigcell{l}{$\mathfrak{C}_{D,f}(T) \coloneqq \sum\limits_{v \in \mathring{V}(T)} bal_{D,f}(v) $} \\ \addlinespace
\bigcell{l}{Corrected Colless index (b) \\ \citep{Heard1992} \\ $\rightarrow$ see Section \ref{factsheet_cColless}} & \bigcell{l}{$C(T)$ divided by its\\ maximum possible value} & $I_C(T) \coloneqq \frac{2\cdot C(T)}{(n-1)(n-2)}$ \\ \addlinespace
\bigcell{l}{Equal weights Colless index\\ ($I_2$ index) (b) \\ \citep{Mooers1997} \\ $\rightarrow$ see Section \ref{factsheet_ewColless}} & \bigcell{l}{modification of $C(T)$ that \\ weighs the imbalance \\ values of all inner vertices \\ of $T$ equally} & $I_2(T) \coloneqq \frac{1}{n-2}\cdot\sum\limits_{\substack{v \in \mathring{V}(T) \\ n_v > 2}} \frac{bal_T(v)}{n_v-2}$\\ \addlinespace
\bigcell{l}{Total $I$ ($I'$) index (a$^\ast$) \\ \citep{Fusco1995, Purvis2002} \\ $\rightarrow$ see Section \ref{factsheet_Ibased}} & \bigcell{l}{sum/total of the $I_v$ ($I_v'$)\\ values over all inner binary\\ vertices $v$ with $n_v \geq 4$ \\ (denoted $\mathring{V}_{bin,\geq 4}$)} & \bigcell{l}{$\Sigma I(T) \coloneqq \sum\limits_{v \in \mathring{V}_{bin,\geq 4}(T)} I_v$ \\ $\Sigma I'(T) \coloneqq \sum\limits_{v \in \mathring{V}_{bin,\geq 4}(T)} I'_v$}  \\ \addlinespace
\bigcell{l}{Mean $I$ ($I'$) index (a$^\ast$) \\ \citep{Fusco1995, Purvis2002} \\ $\rightarrow$ see Section \ref{factsheet_Ibased}} & \bigcell{l}{mean of the $I_v$ ($I_v'$)\\ values over all inner binary\\ vertices $v$ with $n_v \geq 4$ \\ (denoted $\mathring{V}_{bin,\geq 4}$)} & \bigcell{l}{$\overline{I}(T) \coloneqq \frac{1}{|\mathring{V}_{bin,\geq 4}(T)|} \cdot  \sum\limits_{v \in \mathring{V}_{bin,\geq 4}(T)} I_v$ \\ $\overline{I'}(T) \coloneqq \frac{1}{|\mathring{V}_{bin,\geq 4}(T)|} \cdot  \sum\limits_{v \in \mathring{V}_{bin,\geq 4}(T)} I'_v$}  \\ \addlinespace
\bigcell{l}{Quadratic Colless index (b) \\ \citep{Bartoszek2021} \\ $\rightarrow$ see Section \ref{factsheet_qColless}} & \bigcell{l}{sum of the squared \\ balance values of the \\ inner vertices of $T$} &  $QC(T) \coloneqq \sum\limits_{v \in \mathring{V}(T)} bal_T(v)^2$\\ \addlinespace
\bigcell{l}{Rogers $J$ index (b) \\ \citep{Rogers1996} \\ $\rightarrow$ see Section \ref{factsheet_RogersU}} & \bigcell{l}{number of inner vertices \\of $T$ which are not \\perfectly balanced} & \bigcell{l}{$J(T) \coloneqq \sum\limits_{v \in \mathring{V}(T)} (1-\mathcal{I}(bal_T(v)=0))$} \\ \addlinespace
\bigcell{l}{$\widehat{s}$-shape statistic (a) \\ \citep{blum2006c} \\ $\rightarrow$ see Section \ref{factsheet_s-shape}} & \bigcell{l}{sum of $\log(n_v-1)$ over \\ all inner vertices of $T$} &\bigcell{l}{$\widehat{s}(T) \coloneqq \sum\limits_{v \in \mathring{V}(T)} \log(n_v-1)$} \\ \addlinespace
\bigcell{l}{Sackin index (a) \\ \citep{Sackin1972, Shao1990} \\ $\rightarrow$ see Section \ref{factsheet_Sackin}} & \bigcell{l}{sum of the depths of \\ the leaves of $T$} & $S(T) \coloneqq \sum\limits_{x \in V_L(T)} \delta_T(x)$ \\ \addlinespace
\bigcell{l}{Symmetry nodes index (b) \\ \citep{Kersting2021} \\ $\rightarrow$ see Section \ref{factsheet_SNI}} & \bigcell{l}{number of inner vertices \\ of $T$ that are not \\symmetry nodes}  & $SNI(T) \coloneqq (n-1)-s(T)$\\ \addlinespace
\bigcell{l}{Total cophenetic index (a) \\ \citep{Mir2013} \\ $\rightarrow$ see Section \ref{factsheet_TCI}} & \bigcell{l}{sum of the cophenetic \\ values of all different \\ pairs of leaves of $T$}& $\Phi(T) \coloneqq \sum\limits_{\substack{\{x,y\} \in V_L(T)^2 \\ x \neq y}} \varphi_T(x,y)$ \\ \addlinespace
\bigcell{l}{Variance of leaf depths (a) \\ \citep{Coronado2020b, Sackin1972, Shao1990} \\ $\rightarrow$ see Section \ref{factsheet_VLD}} & \bigcell{l}{variance of the depths \\ of the leaves of $T$} & $\sigma_N^2(T) \coloneqq \frac{1}n\cdot\sum\limits_{x \in V_L(T)} \left( \delta_T(x) - \overline{N}(T) \right)^2$ \\ \addlinespace
\bottomrule
\end{longtable}

\begin{table}[htbp]
\caption{List of established measures that comply with our definition of a balance index (Definition \ref{def_balance}). Indices marked with (a) make sense for arbitrary trees, whereas indices marked with (b) are only defined for binary trees. 
Note that although the $B_2$ index is defined for arbitrary logarithm base it is common to use $\log_2$ when working with binary trees.
Also note that despite being defined for arbitrary trees, the $B_1$ index only satisfies the criteria of a balance index on the space of rooted binary trees $\BTnstar$, while it violates them on $\Tnstar$ (for details see Section \ref{factsheet_B1}). 
For additional details on an index we refer the interested reader to the respective fact sheet in the indicated section or to the website \url{treebalance.wordpress.com}.}
\label{Table_Balance}
\centering
\setlength{\tabcolsep}{5mm}
\begin{tabular}{lll}
\toprule
\textbf{Balance indices} & &\\
Index & Description & Definition \\
\midrule
\bigcell{l}{$B_1$ index (a) \\ \citep{Shao1990} \\ $\rightarrow$ see Section \ref{factsheet_B1}} & \bigcell{l}{sum of the reciprocal \\ of the heights of the \\ subtrees of $T$ rooted at \\ inner vertices of $T$ \\ (except for $\rho$) \\ }& $B_1(T) \coloneqq \sum\limits_{v \in \mathring{V}(T) \setminus \{\rho\}} h(T_v)^{-1}$\\ \addlinespace
\bigcell{l}{$B_2$ index (a) \\ \citep{Agapow2002,Hayati2019,Kirkpatrick1993,Shao1990} \\ $\rightarrow$ see Section \ref{factsheet_B2}} & \bigcell{l}{Shannon-Wiener information \\ function (measures the \\ equitability of arriving at \\ the leaves of $T$ when starting \\ at the root and assuming \\ equiprobale branching at each \\ inner vertex)} & \bigcell{l}{$B_2(T) \coloneqq - \sum\limits_{x \in V_L(T)} p_x \cdot \log(p_x)$}\\ \addlinespace
\bigcell{l}{Furnas rank (b) \\ \citep{Furnas1984,Kirkpatrick1993} \\ $\rightarrow$ see Section \ref{factsheet_Furnas}} & \bigcell{l}{rank of $T$ according to Furnas' \\ \enquote{left-light rooted ranking}\\ on $\BTnstar$} & $F(T) \coloneqq r_n(T)$ \\ \addlinespace
\bigcell{l}{Rooted quartet index (a) \\ \citep{Coronado2019} \\ $\rightarrow$ see Section \ref{factsheet_rQuartet}} & \bigcell{l}{family of balance indices \\ based on the symmetry of \\ the quartets displayed by $T$} & $rQI(T) \coloneqq \sum\limits_{Q \in \mathcal{Q}(T)} rQI(Q)$\\
\bottomrule
\end{tabular}
\end{table}

\begin{sidewaystable}[htbp]
\caption{Example index values of all established (im)balance indices considered in this study (except the class of Colless-like indices) for all possible tree shapes with $n=6$ leaves. The maximum value of each index is marked with $\star$ and the minimum value is marked with $\pm$. Note that all indices consider the tree in the first column as least balanced, because it is a caterpillar tree which is by Definition \ref{def_balance} and \ref{def_imbalance} the unique least balanced tree. Opposed to this, the tree with maximal balance is not unambiguous as $n=6$ is not a power of two. Some indices (e.g. the $I_2$ and $B_1$ index) consider solely the greedy from the bottom tree (fifth column) and others (e.g. the quadratic Colless and rooted quartet index) solely the maximally balanced tree (sixth column) as most balanced while some indices (e.g. the Colless and Sackin index) do not discriminate between them.}
\label{Table_Examples}
\centering
	\def\arraystretch{1.5}
	\scalebox{0.8}{
	\begin{tabular}{lllllll}
	\toprule
	Tree & \includegraphics[scale=0.145]{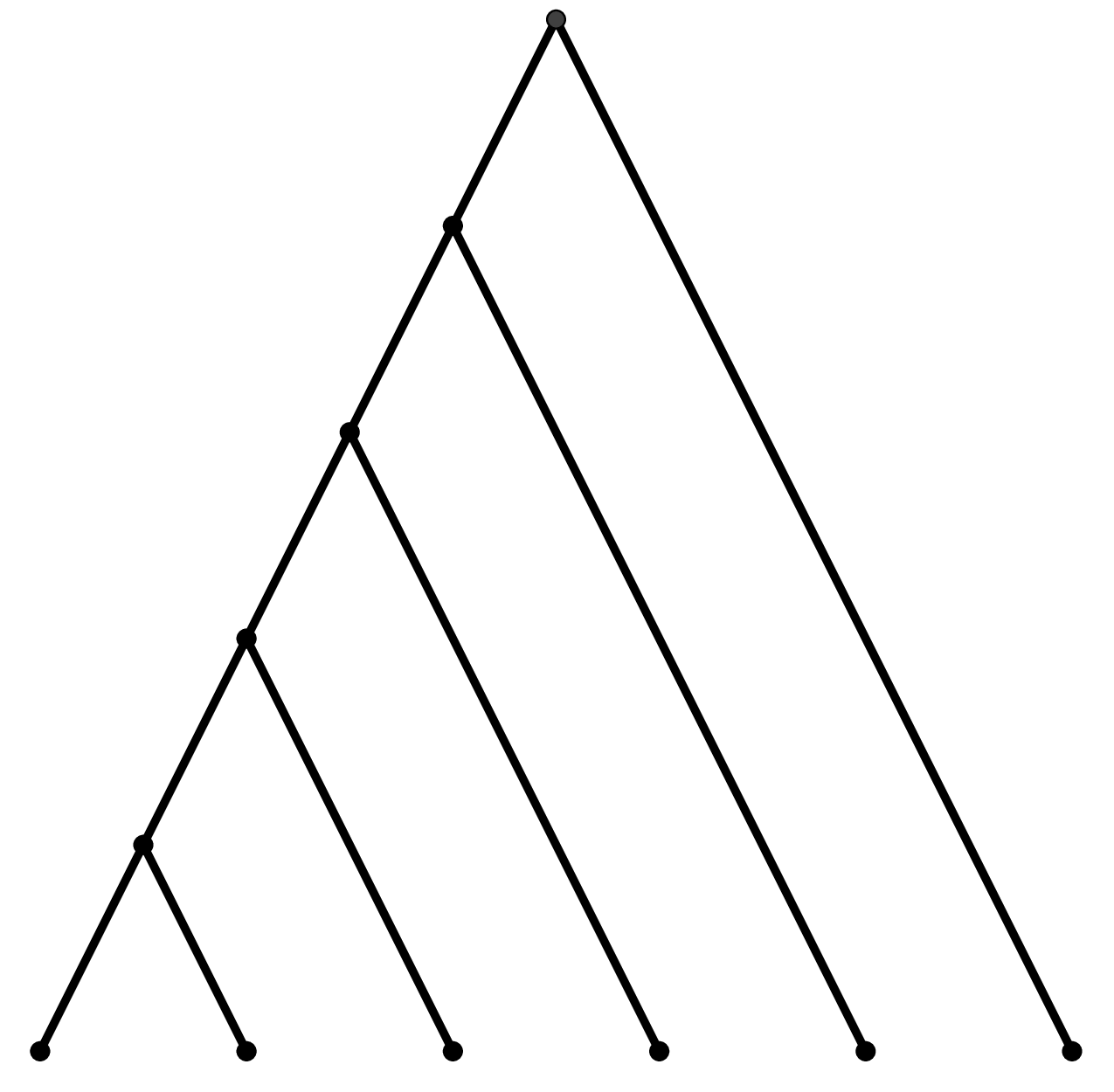} & \includegraphics[scale=0.145]{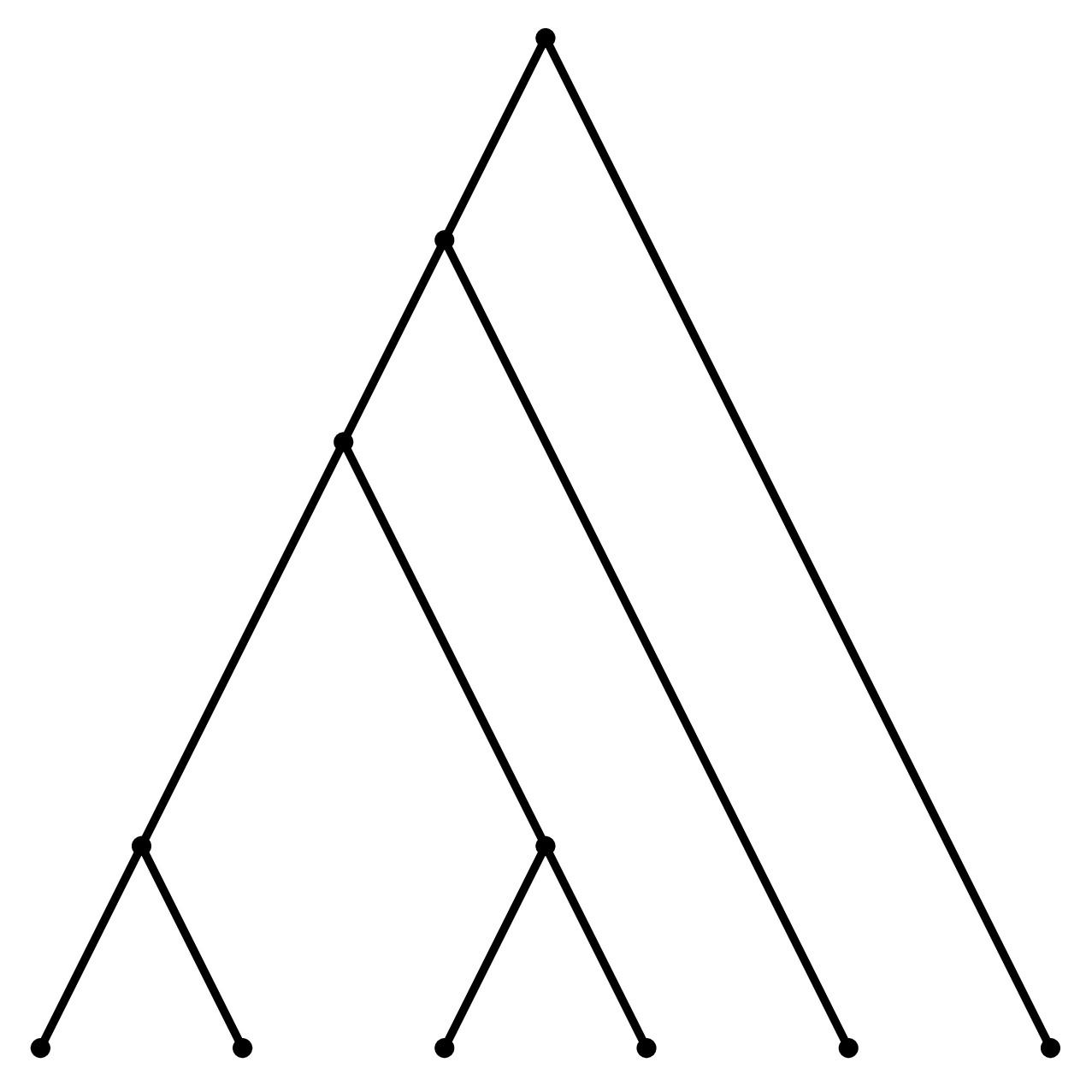} & \includegraphics[scale=0.145]{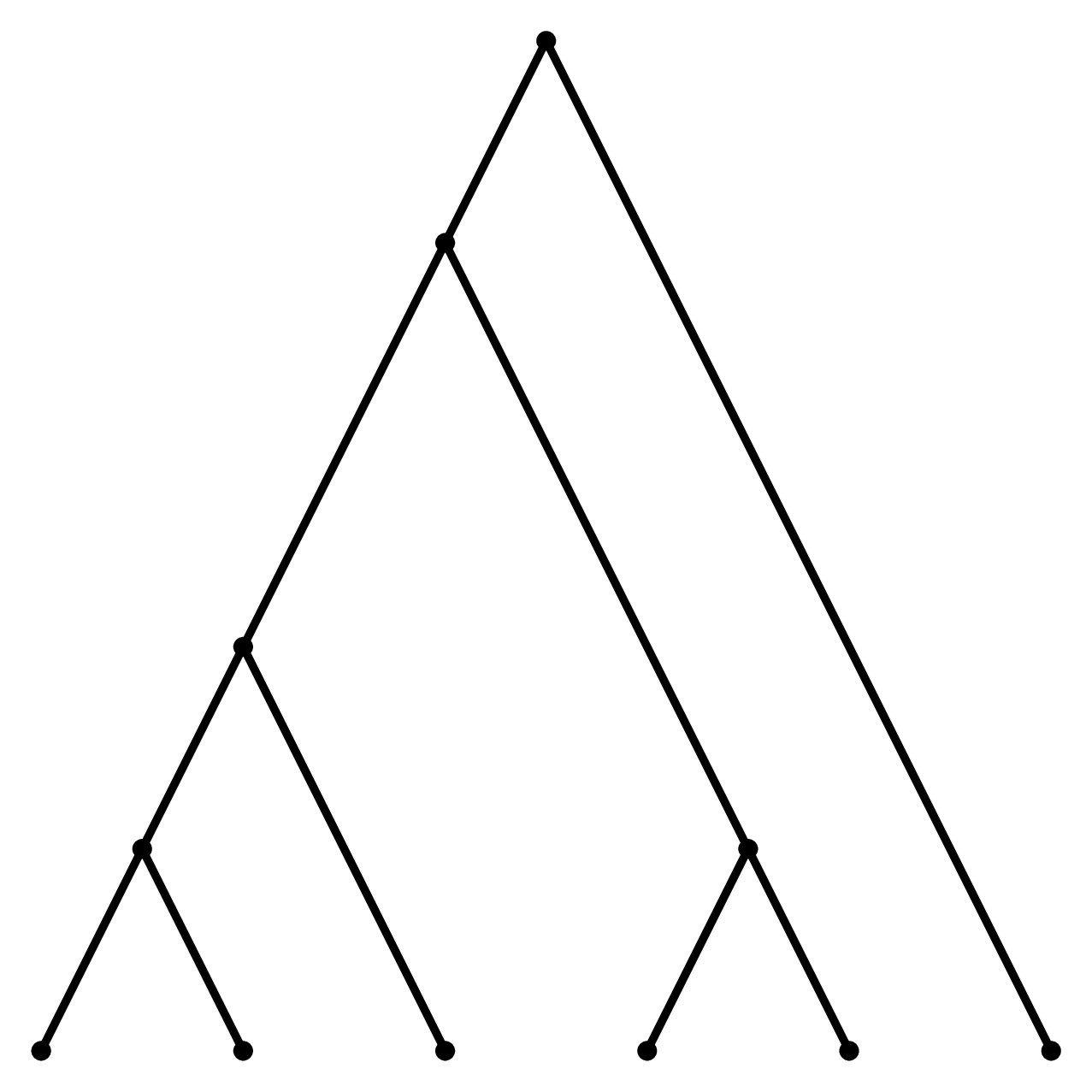} & \includegraphics[scale=0.145]{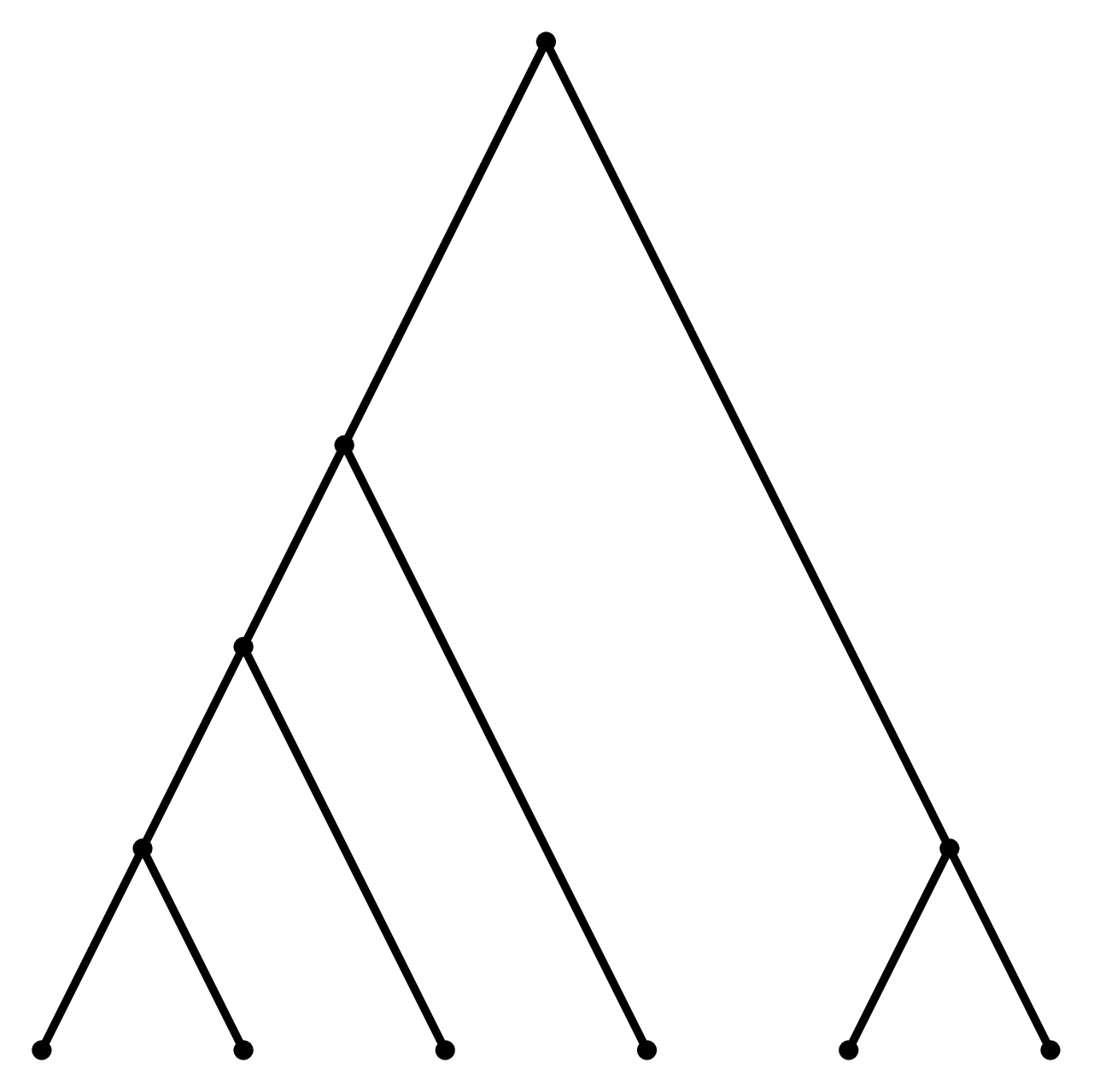} & \includegraphics[scale=0.145]{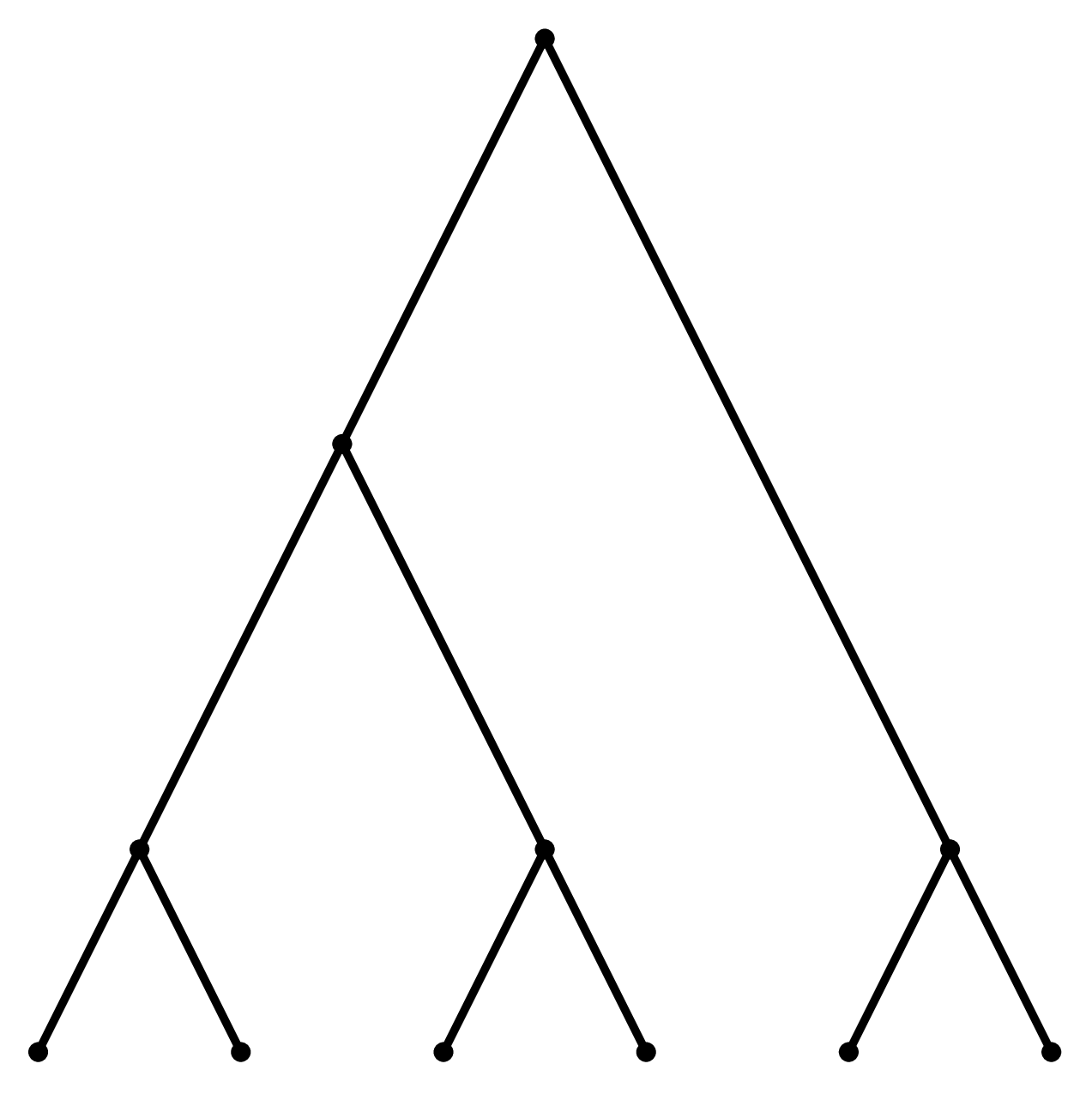} & \includegraphics[scale=0.145]{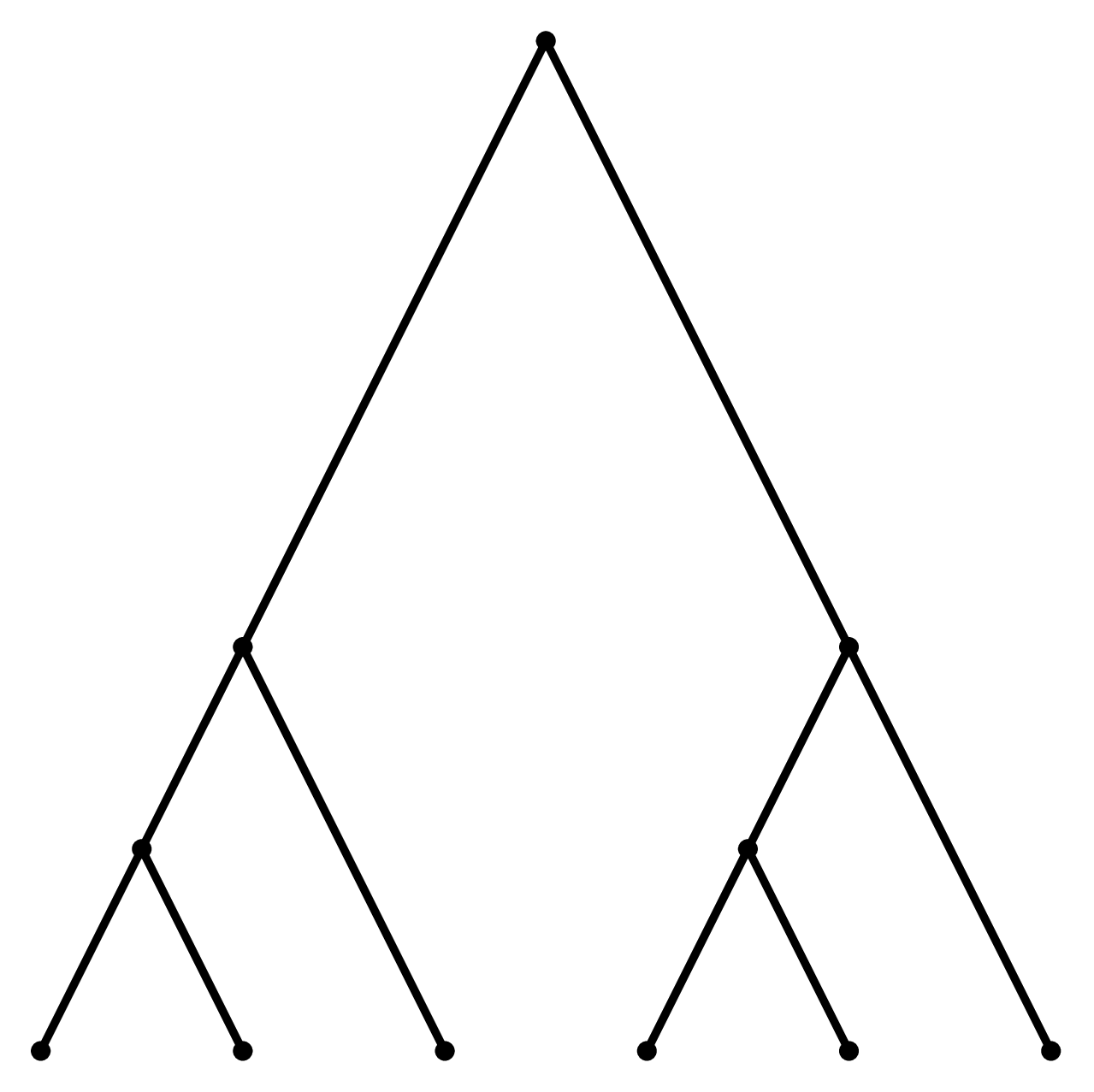}  \\
	 \midrule
	 {Average leaf depth} & $\frac{10}{3}=3.\overline{3}$ $\star$ & $\frac{19}{6}=3.1\overline{6}$ & 3 & $\frac{17}{6}=2.8\overline{3}$ & $\frac{16}{6}=2.\overline{6}$ $\pm$ & $\frac{16}{6}=2.\overline{6}$ $\pm$ \\
	 {Colijn-Plazotta rank} & 68 $\star$ & 30 & 17 & 13 & 9 & 7 $\pm$ \\
	 {Colless index} & 10 $\star$ & 7 & 6 & 5 & 2 $\pm$ & 2 $\pm$ \\
	 {Corrected Colless index} & 1 $\star$ & 0.7 & 0.6 & 0.5 & 0.2 $\pm$ & 0.2 $\pm$ \\
	 {$I_2$ index} & 1 $\star$ & 0.5 & 0.58333 & 0.625 & 0.125 $\pm$ & 0.5 \\
	 {Mean $I$ ($I'$) $[I^w]$ index} & 1 $\star$ ($\frac{31}{36}$ $\star$) $[1 \ \star]$ & $\frac{2}{3}$ ($\frac{11}{18}$) $[0.55]$ & $\frac{1}{2}$ ($\frac{5}{12}$) $[\frac{5}{11}]$ & $\frac{3}{4}$ ($\frac{7}{12}$) $[\frac{14}{19}]$ & $\frac{1}{4}$ ($\frac{5}{24}$) $[\frac{5}{28}]$ & 0 $\pm$ (0 $\pm$) $[0 \ \pm]$ \\
	 {Quadratic Colless index} & 30 $\star$ & 25 & 18 & 9 & 4 & 2 $\pm$ \\
	 {Rogers $J$ index} & 4 $\star$ & 2 & 3 & 3 & 1 $\pm$ & 2 \\
	 {$\widehat{s}$-shape statistic (using $\log_2$)} & $\log_2(120) \approx 6.91$ $\star$ & $\log_2(60) \approx 5.91$ & $\log_2(40) \approx 5.32$ & $\log_2(30) \approx 4.91$ & $\log_2(15) \approx 3.91$ $\pm$ & $\log_2(20) \approx 4.32$ \\
	 {Sackin index} & 20 $\star$ & 19 & 18 & 17 & 16 $\pm$ & 16 $\pm$ \\
	 {Symmetry nodes index} & 4 $\star$ & 2 & 3 & 3 & 1 $\pm$ & 2 \\
	 {Total cophenetic index} & 20 $\star$ & 18 & 15 & 11 & 9 & 8 $\pm$ \\
	 {Total $I$ ($I'$) $[I^w]$ index} & 3 $\star$ ($\frac{31}{12}$ $\star$) $[3 \ \star]$ & 2 ($\frac{11}{6}$) $[1.65]$ & 1 ($\frac{5}{6}$) $[\frac{10}{11}]$ & $\frac{3}{2}$ ($\frac{7}{6}$) $[\frac{28}{19}]$ & $\frac{1}{2}$ ($\frac{5}{12}$) $[\frac{5}{14}]$ & 0 $\pm$ (0 $\pm$) $[0 \ \pm]$ \\
	 {Variance of leaf depths} & $2.\overline{2}$ $\star$ & $1.47\overline{2}$ & 1 & $0.80\overline{5}$ & $0.\overline{2}$ $\pm$ & $0.\overline{2}$ $\pm$ \\
	 \midrule
	 \midrule
	 {$B_1$ index} & $\frac{25}{12}=2.08\overline{3}$ $\pm$ & $\frac{17}{6}=2.8\overline{3}$ & $\frac{17}{6}=2.8\overline{3}$ & $\frac{17}{6}=2.8\overline{3}$ & $\frac{7}{2}=3.5$ $\star$ & 3 \\
	 {$B_2$ index (using $\log_2$)} & $\frac{31}{16}=1.9375$ $\pm$ & $2$ & $\frac{17}{8}=2.125$ & $\frac{19}{8}=2.375$ & $\frac{5}{2}=2.5$ $\star$ & $\frac{5}{2}=2.5$ $\star$ \\
	 {Furnas rank} & 1 $\pm$ & 2 & 3 & 4 & 5 & 6 $\star$ \\
	 {Rooted quartet index ($q_i=i$)} & 0 $\pm$ & 3 & 9 & 18 & 21 & 27 $\star$ \\
	 \bottomrule
	\end{tabular}}
\end{sidewaystable}

\clearpage

\newpage
\subsection{General, combinatorial and statistical properties} \label{Sec_Combinatorial_and_Statistical}

In this section, we aim at providing the reader with a quick overview of the current state of knowledge on tree balance. In order to do this, Table \ref{table_Overview} summarizes some general, combinatorial and statistical properties which we consider relevant for empiricists and theoreticians. It especially allows to quickly assess which of those properties have been addressed before and which questions are still unsolved. Comprehensive fact sheets for each index\footnote{Additionally, we include a fact sheet for the cherry index (which is by our definition neither a balance nor imbalance index, see Table~\ref{Table_nbi_tss}) because it is a well-known tree shape statistic that has been used to measure tree balance before.} containing these properties are given in Section \ref{Sec_factsheets} in alphabetical order and on the website \url{treebalance.wordpress.com}. They include known properties, references to the original sources of the contained statements or -- in case the assertion has not been formally shown until now -- refer the reader to the respective proof in Appendix \ref{app_additionals}. We also point out gaps in the literature and list -- where available -- efforts that have been made in filling those gaps.

The properties that we consider can be divided into general, combinatorial and statistical properties. The first general property is the \textit{computation time}, i.e. the asymptotic time that is needed to compute the (im)balance value of a given tree. Note that it is a theoretical one assuming that the children and the parent of a vertex $v$ can be found in constant time. Whether this is possible depends on the data structure used for the tree. Most of the herein considered indices can be computed in time $O(n)$ by using either post-order tree traversal, in which the vertices are considered from the leaves towards the root (left, right, root), or pre-order tree traversal, in which the vertices are considered from the root towards the leaves (root, left, right).

The second general property is the \textit{recursiveness}. A \emph{recursive tree shape statistic}\footnote{We adapted this definition from \citet{Matsen2007}, who considered only binary trees. We also slightly changed the notation and use $r$ and $x$ instead of $\rho$ and $n$ in order to avoid confusion with the other notation in this manuscript} of length $x$ is an ordered pair $(\lambda,r)$, where $\lambda\in\mathbb{R}^x$ and $r$ is an $x$-vector of $Symm^k(\mathbb{R}^x)\rightarrow\mathbb{R}$ maps, in which $k$ denotes the maximal number of children of any vertex in the tree. In this definition, $x$ is the number of recursions that are used to calculate the index, the vector $\lambda$ contains the start value for each recursion, i.e. the value of $T\in\Tnstar$ if $n=1$, and the vector $r$ contains the recursions themselves, i.e. it indicates how the value of $T=(T_1,\ldots,T_k)$ can be calculated from the values of $T_1,\ldots,T_k$. In particular, each recursion in $r$ must be independent of the order of $T_1,\ldots,T_k$. If $k=2$ is fixed, we use the term \emph{binary recursive tree shape statistic}. Note that the fact sheets contain only the recursion for the (im)balance index itself, while the complete pair $(\lambda,r)$ can be found in the respective proof.

The third general property that we consider is \textit{locality}. An (im)balance index is considered to be local when it fulfills the following criterion: If two trees differ only in one rooted subtree, the difference between their indices is equal to the difference between the indices of these differing subtrees \citep{Mir2013}. Formally, we have: Let $T$ be a tree and let $v$ be one of its vertices. We obtain the tree $T'$ by exchanging the subtree $T_v$, rooted at $v$, by a subtree $T_v'$ with the same leaf number and also rooted at $v$. Then, the (im)balance index $t$ is called \emph{local} if it fulfills $t(T)-t(T')=t(T_v)-t(T_v')$ for all $v\in V(T)$ \citep{Mir2013}.

The combinatorial properties comprise the maximal and minimal value that a tree with $n$ leaves can have, together with a characterization and the number of trees that achieve these extremal values. We consider these combinatorial properties both for arbitrary (where applicable) and binary trees.

Finally, as statistical properties, we consider the expected value and variance under the Yule model and the uniform model.

\begin{sidewaystable}
\caption{Overview of general, combinatorial and statistical properties. Properties marked with (a) refer to the set of arbitrary trees and properties marked with (b) refer to the set of binary trees. Do note that the cherry index does not fulfill our definition of an (im)balance index.}
\label{table_Overview}
\centering
    \begin{tabular}{l*{20}c}
        & \rot{computation time} & \rot{recursiveness} & \rot{locality} & \rot{maximal value (a)} & \rot{maximal value (b)} & \rot{\shortstack[l]{trees with\\maximal value (a)}} & \rot{\shortstack[l]{trees with\\maximal value (b)}} & \rot{\shortstack[l]{number of trees with\\maximal value (a)}} & \rot{\shortstack[l]{number of trees with\\maximal value (b)}} & \rot{minimal value (a)} & \rot{minimal value (b)} & \rot{\shortstack[l]{trees with\\minimal value (a)}} & \rot{\shortstack[l]{trees with\\minimal value (b)}} & \rot{\shortstack[l]{number of trees with\\minimal value (a)}} & \rot{\shortstack[l]{number of trees with\\minimal value (b)}} & \rot{\shortstack[l]{expected value under\\the Yule model}} & \rot{\shortstack[l]{variance under\\the Yule model}} & \rot{\shortstack[l]{expected value under\\the uniform model}} & \rot{\shortstack[l]{variance under\\the uniform model}}\\
        \cmidrule{1-20}
        average leaf depth  & * & * & * & * & * & \OK* & \OK* & \OK* & \OK* & * & * & \OK* & \OK* & \OK* & (\OK*) & \OK & * & \OK & * \\
        $B_1$ index  & * & * & * &  & (*) &  & (*) &  & (*) & \OK* & \OK* & * & * & * & * &  &  &  &  \\
        $B_2$ index  & * & \OK & * & \OK & \OK & \OK & \OK &  & (*) & *  & \OK* & *  & \OK* &  & \OK* & \OK & (\OK) & \OK & (\OK) \\
        Colijn-Plazzotta rank & (\OK) & (\OK) & * &  & (\OK) &  & \OK &  & \OK &  & (\OK) &  & \OK &  & \OK &  &  &  &  \\
        Colless index & * & \OK & * & X & \OK & X & \OK & X & \OK & X & \OK & X & \OK & X & (\OK) & \OK & \OK & (\OK) & (\OK)\\
        Colless-like indices & \OK & \OK &  & [\OK] & [\OK] & [\OK] & [\OK] & [\OK] & [\OK] & [\OK] & [\OK] & [\OK] & [\OK] & [\OK] & [\OK] &  &  &  & \\
        corrected Colless index & * & * & * & X & * & X & * & X & * & X & * & X & * & X & (*) & \OK * & * & (*) & (*)\\
        $I_2$ index & * & * & * & X & * & X & * & X & * & X & (*) & X & (*) & X & (*) &  &  &  &  \\
        Furnas rank & * & * & * & X & * & X & * & X & * & X & * & X & * & X & * &  &  &  & \\
        $I$-based ($\overline{I}$, $\overline{I'}$, $\Sigma I$, $\Sigma I'$) & [*] & [*] & [*] & [*] & [*] & [*] & [*] & [*] & [*] & [*] & [*] & [*] & [*] & [*] & [*] & [\OK*] & [*] &  &  \\
        quadratic Colless index & * & \OK* & * & X & \OK & X & \OK & X & \OK & X & \OK & X & \OK & X & \OK & \OK & \OK & \OK & \OK\\
        Rogers $J$ index & * & * & * & X & \OK & X & \OK & X & \OK & X & \OK & X & \OK & X & \OK &  &  &  & \\
        rooted quartet index & \OK & \OK & * & \OK & (\OK) & \OK & \OK & \OK & \OK & \OK & \OK & \OK & \OK & \OK & \OK & \OK & \OK & \OK & \OK\\
        $\widehat{s}$-shape statistic & * & * & * & * & * & * & * & * & * & * & (*) & * & (*) & * & (*) &  &  &  & \\
        Sackin index & * & * & * & \OK* & \OK & \OK* & \OK & \OK* & \OK & \OK* & \OK & \OK* & \OK & \OK* & (\OK) & \OK & \OK & \OK & \OK\\
        symmetry nodes index & \OK & * & * & X & \OK & X & \OK & X & \OK & X & \OK & X & \OK & X & \OK &  &  & (*) & (*)\\
        total cophenetic index & \OK & \OK & \OK & \OK & \OK & \OK & \OK & \OK & \OK & \OK & \OK & \OK & \OK & \OK & \OK & \OK & \OK & \OK & \OK\\
        variance of leaf depths & * & * & * & \OK & \OK & \OK & \OK & \OK & \OK & \OK & (\OK) & \OK & (\OK) &  & (\OK) & \OK &  & \OK &  \\
        \cmidrule{1-20}
        cherry index\footnote{The cherry index does not fulfill the definition of an (im)balance index, it was only included because of its popularity.} & * & * & * & * & \OK & * & \OK & * & \OK & * & \OK & * & \OK & * & \OK & \OK & \OK & \OK & \OK \\
        \cmidrule{1-20}
        \cmidrule[1pt]{1-20}
        \multicolumn{20}{l}{\shortstack[l]{\quad Empty fields indicate that the question has not been answered yet.\\
        * These results are proven in this manuscript. Some of them have been stated in the literature before, but we could not find a formal proof.\\
        \OK These results have already been proven in the literature.\\
        X This information is not available, because the index is not defined for this situation.\\
        (\OK/*) The question has not entirely been answered yet, but some partial results have already been obtained.\\  
        $[$\OK/*$]$ The question has been answered for at least one index of the index class, but not for all of them.}}
    \end{tabular}
\end{sidewaystable}

\subsection{Tree shape statistics that are (im)balance indices} \label{Sec_tssAreBalInd}

In addition to the already mentioned (im)balance indices, there is a vast number of tree shape statistics that also map trees to real numbers, but do not have or have not yet been shown to have the defining properties of (im)balance indices, i.e. the properties that the opposite extreme values are uniquely assigned to the caterpillar and (provided that $n$ is a power of two) the fully balanced tree -- until now. In Table \ref{Table_bi_tss} we list several tree shape statistics that indeed satisfy the definitions, and name their values for the caterpillar and the fully balanced tree. Since it would be beyond the scope of this manuscript, we leave the analysis of their remaining combinatorial and statistical properties up to future research. Of course, there are also several tree shape statistics that do not comply with the definition of an (im)balance index, a selection of which, are briefly discussed in the next section.

\begin{sidewaystable}[htbp]
\caption{Further tree shape statistics that are (im)balance indices. $t(\Tcat)$ and $t(\Tfb)$ denote the respective values of the tree shape statistics for the caterpillar and the fully balanced tree. The references to the proofs for these extremal properties and that they are in fact (im)balance indices can be found in the last column. All of these have been implemented in our \textsf{R} package \texttt{treebalance} and are accessible via the functions \texttt{maxWidth, maxDelW, maxDepth, stairs1} and \texttt{stairs2}, respectively.}
\label{Table_bi_tss}
\centering
	\def\arraystretch{1}
	\scalebox{0.8}{
	\begin{tabular}{llp{8cm}lll}
	 \toprule
     Index & Description & Definition & $t(\Tcat)$ & $t(\Tfb)$ & Proof \\
	 \midrule
	 \multicolumn{2}{l}{\textbf{TSS that are balance indices}}& & & &\\ \addlinespace
	 maximum width \citep{Colijn_phylogenetic_2014} &  \bigcell{l}{maximal width at \\ any depth} & $\begin{array}{rcl}
	  \displaystyle mW(T) & \coloneqq & \max\limits_{i=0,\ldots,h(T)} w_T(i)
	 \end{array}$ & $mW(\Tcat)=\begin{cases} 1 & \text{if } n=1 \\ 2 & \text{if } n\geq 2\end{cases}$ & $mW(\Tfb)=2^h$ &  \bigcell{l}{Theorem \ref{prop_mW_max} \\and \ref{prop_mW_min}}\\ \addlinespace
	 \bigcell{l}{maximum difference \\in widths \citep{Colijn_phylogenetic_2014}} &  \bigcell{l}{maximal difference of \\widths at any depth} & $\begin{array}{rcl} delW(T) & \coloneqq & \max\limits_{i=0,\ldots,h(T)-1} {w_T(i+1)-w_T(i)} \end{array}$ & $delW(\Tcat)=1$ & $delW(\Tfb)=2^{h-1}$ &   \bigcell{l}{Theorem \ref{prop_delW_max} \\and \ref{prop_delW_min}}\\ \addlinespace
	 \midrule
	 \multicolumn{2}{l}{\textbf{TSS that are imbalance indices}} & & & &\\ \addlinespace
	 maximal depth \citep{Colijn_phylogenetic_2014} &  \bigcell{l}{height of the tree} &
              {$\begin{array}{rcl}
        mD(T) & \coloneqq & \max\limits_{x\in V_L(T)} \delta_T(x) \\
        & = & \max\limits_{v\in V(T)} \delta_T(v) \\
              & = & \max\limits_{v\in V(T)} |anc(v)| = h(T) \end{array}$} 
      &  $ mD(\Tcat)=n-1$ & $mD(\Tfb)=h$ &  \bigcell{l}{Theorem \ref{prop_mD_max} \\and \ref{prop_mD_min}}\\ \addlinespace
      \bigcell{l}{stairs or stairs1 \citep{Norstrom_phylotempo_2012}} &  \bigcell{l}{for $n\geq 2$, the proportion \\ of inner vertices that are \\ not perfectly balanced \\ among all inner vertices, \\modified version of the \\Rogers $J$ index} & {$\begin{array}{rcl}
      st(T) & = & st1(T) \\
      & \coloneqq & \frac{1}{n-1} \cdot \sum\limits_{v\in\mathring{V}(T)} (1-\mathcal{I}(bal_T(v)=0)) \\
      & = & \frac{J(T)}{n-1} 
      \end{array}$}
      &  $st(\Tcat)=\frac{n-2}{n-1}$ & $st(\Tfb)=0$ &  \bigcell{l}{see Rogers $J$ \\ $\rightarrow$ see \ref{factsheet_RogersU}}\\ \addlinespace
      \bigcell{l}{stairs2 \citep{Colijn_phylogenetic_2014}} &  \bigcell{l}{the mean ratio between \\the leaf numbers of the \\smaller and larger \\pending subtree over all \\inner vertices} & {$\begin{array}{rcl}
      st2(T) &\coloneqq \frac{1}{n-1}\cdot\sum\limits_{v\in\mathring{V}(T)} \frac{\min\{n_{v_1},n_{v_2}\}}{\max\{n_{v_1},n_{v_2}\}} \\
      &= \frac{1}{n-1}\cdot\sum\limits_{v\in\mathring{V}(T)} \frac{n_v-n_{v_1}}{n_{v_1}} \end{array}$} 
      &  $st2(\Tcat)=\begin{cases} 0 & \text{if } n=1 \\ \frac{H_{n-1}}{n-1} & \text{if } n\geq 2 \end{cases}$ & $st2(\Tfb)=1$ &  \bigcell{l}{analogous to \\Theorem \ref{prop_meanI_max_b} \\and \ref{prop_meanI_min_b}\\ of the Mean $I$ \\index}\\ \addlinespace
	 \bottomrule
	\end{tabular}}
\end{sidewaystable}

\subsection{Tree shape statistics that are not (im)balance indices} \label{Sec_nbi_tss}
In this section, we briefly summarize some tree shape statistics which have been introduced or suggested in the literature as measures for tree (im)balance but which do not meet our criteria for (im)balance indices and which therefore are not investigated further in the present manuscript.

We review 13 such tree shape statistics in Table~\ref{Table_nbi_tss}, where give a brief description and formal definition of each of them. More importantly, though, we highlight why they are not (im)balance indices according to our definition. As an example, the well-known cherry index~\cite{McKenzie2000} is not a balance index because its maximum value on $\BTnstar$ is not uniquely achieved by the fully balanced tree $\Tfb$ when $n$ is a power of two (see also Figure~\ref{Fig_cherry_max}). Note that this list of measures is restricted to only the most established tree shape statistics or modification and generalizations of known and established balance indices or tree shape statistics. Two of these, the Area Per Pair index and the $D$ index (also referred to as the weighted $\ell_1$ distance) have also been implemented in our \textsf{R} package \texttt{treebalance} and are accessible via the functions \texttt{areaPerPairI} and \texttt{weighL1dist}, respectively.

\begin{sidewaystable}[htbp]
\caption{Further tree shape statistics that are \emph{not} (im)balance indices as postulated in Definition \ref{def_balance} and \ref{def_imbalance}. Measures marked with (a) make sense for arbitrary trees, while measures marked with (a*) are restricted to a subset, generally trees that have only a small fraction of non-binary nodes or trees that have a binary root, and measures marked with (b) are only defined for binary trees. At least one reason why these measures do not meet the requirements of an (im)balance index can be found in the last column in the form of counterexamples or references to proofs. Figures~\ref{Fig_APP_Extremal} -- \ref{Fig_d1Trees} referenced in this table can be found in Appendix~\ref{Appendix_additionalfigures}.}
\label{Table_nbi_tss}
\centering
	\def\arraystretch{1}
	\scalebox{0.7}{
	\begin{tabular}{llp{7cm}l}
	 \toprule
     Index and source & Description & Definition & Reason why no (im)balance index \\
	 \midrule
	 \bigcell{l}{Area Per Pair index (a) \\ \citep{Lima2020}} & \bigcell{l}{average distance between \\ all pairs of leaves of $T$} & $\bar{d}_n(T) \coloneqq  \frac{2}{n(n-1)}\sum\limits_{\substack{\{i,j\} \in V_L(T)^2 \\ i \neq j}} d_T(i,j)$  & Neither $T_n^{cat}$ nor $T_h^{fb}$ are generally extremal, e.g. for $n=8=2^3$, cf. Figure \ref{Fig_APP_Extremal}. \\\addlinespace
    \bigcell{l}{Wiener index (a) \\ \cite{Mohar_how_1988, Chindelevitch2019} } & \bigcell{l}{total distance between \\ all pairs of leaves of $T$} & $WI(T)\coloneqq \sum\limits_{\substack{\{i,j\} \in V_L(T)^2 \\ i \neq j}} d_T(i,j)=\frac{n(n-1)}{2}\bar{d}_n(T)$ & \bigcell{l}{See Area Per Pair index.}\\ \addlinespace
    \bigcell{l}{Degree of root imbalance \\(a*) \citep{Guyer1993}} & \bigcell{l}{ratio between the size \\ $n_{a}$ of the larger \\ maximal pending subtree \\and $n$}& \bigcell{l}{$DU_\rho(T) \coloneqq \frac{n_{a}}{n}$}  & \bigcell{l}{Caterpillar not unique maximally imbalanced tree, \\see  e.g. first, second and third tree in Table \ref{Table_Examples}.}\\ \addlinespace
    \bigcell{l}{$I$ value  (a*) \citep{Fusco1995}\\ (with correction $I_\rho'$ \citep{Purvis2002})\\ $\rightarrow$ see \ref{factsheet_Ibased}} & \bigcell{l}{ratio between the observed \\ deviation of the larger pending \\subtree from  the minimum \\value possible and the maximum \\ deviation possible }& \bigcell{l}{$I_\rho(T) \coloneqq I_\rho = \frac{n_{1}-\lceil \frac{n}{2} \rceil}{n-1- \lceil \frac{n}{2} \rceil}$ \\$I_\rho' \coloneqq \begin{cases} I_\rho &\text{if } n \text{ is odd} \\ \frac{n-1}{n} \cdot I_\rho &\text{else} \end{cases}$}  & \bigcell{l}{Caterpillar not unique maximally imbalanced tree, \\see  e.g. first, second and third tree in Table \ref{Table_Examples}.}\\ \addlinespace
    \bigcell{l}{Mean $I_{10}'$ index (b) \\ \citep{Agapow2002}\\ $\rightarrow$ see \ref{factsheet_Ibased}} & \bigcell{l}{$\overline{I}_{10}'(T)$ is the mean of the $I_v'$ \\ values over the ten oldest \\ inner vertices of $T$} & \bigcell{l}{Very similar to the  mean $I'$ index \\(cf. Table \ref{Table_Imbalance}), but  assumes \\ranked trees (in order to  determine \\the \enquote{ten oldest inner vertices}).}  & \bigcell{l}{Only defined for ranked trees; value does not solely depend on tree shape \\ but also on given ranking, cf. Figure \ref{fig_ex_Ivprime}. Was not really intended to be used \\ as an (im)balance index \citep{Agapow2002}.}
    \\ \addlinespace
    \bigcell{l}{$D$ index (also referred to as \\ weighted $\ell_1$  distance) (b)  \citep{Blum2005}} & \bigcell{l}{weighted $\ell_1$ distance between \\ the observed distribution $f_n(z)$ \\ of the number of subtrees of size \\ $z$ and the expected distribution \\ $p_n(z)$ under the Yule model} & \bigcell{l}{$D(T) \coloneqq \sum\limits_{z=2}^n z | f_n(z)- p_n(z)|$, where \\ $p_n(z) = \begin{cases} \frac{n}{n-1} \frac{2}{z(z+1)} &z= 2, \ldots, n-1 \\ \frac{1}{n-1} &z=n. \end{cases}$}  & \bigcell{l}{Used in \cite{Blum2005} to quantify \enquote{departure from the Yule model}, not to measure \\ tree balance; fully balanced tree and caterpillar are not extremal; cf. Figure \ref{Fig_d1Trees}.} \\ \addlinespace
    \bigcell{l}{Cherry index (a) \\ \citep{McKenzie2000} \\ $\rightarrow$ see \ref{factsheet_cherry}\\ Modified cherry index (b) \\ \citep{Kersting2021}} & \bigcell{l}{number of cherries in $T$ and \\number of leaves not in a cherry} & $ChI(T) \coloneqq c(T)$, $mCI(T) \coloneqq n-2\cdot c(T)$ & \bigcell{l}{Maximally balanced trees not unique for $n=2^h$, \\$h\in\mathbb{N}_{\geq 3}$, see e.g. Figure \ref{Fig_cherry_max} on the left.}\\ \addlinespace
    \bigcell{l}{Number of pitchforks (a) \\ \citep{Choi2020}} & \bigcell{l}{number of pitchforks in $T$, \\i.e. the number of times $T^{cat}_3$ \\appears in $T$ as a maximal \\pending subtree} & $pf(T)$ & \bigcell{l}{The caterpillar is not always the most imbalanced tree, \\see e.g. first, fifth and sixth tree in Table \ref{Table_Examples}.}\\ \addlinespace
    \bigcell{l}{Number of double cherries \\or 4-caterpillars (a) \\ \cite{Chindelevitch2019, Rosenberg_mean_2006}} & \bigcell{l}{number of $T^{bal}_2$ or $T^{cat}_4$ in $T$} &  & \bigcell{l}{Maximally balanced trees not unique for $n=2^h$, \\$h\in\mathbb{N}_{\geq 4}$, see e.g. Figure \ref{Fig_cherry_max} on the right.}\\\addlinespace
    \bigcell{l}{Number of clades of \\size $x$ (a) \\ \cite{Rosenberg_mean_2006}} & \bigcell{l}{counts subtrees of a certain size $x$, \\for $x=2$ and $x=3$ it matches \\$CI(T)$ and $pf(T)$ and for $x=4$ the \\sum of $dc(T)$ and $cat_4(T)$} & $num_x(T)$ & \bigcell{l}{For $x\leq 3$ see $CI(T)$ and $pf(T)$;  for $x\geq 4$ (for $n>x$) $T^{cat}_n$  is not the  \\unique maximally imbalanced tree (Proposition \ref{prop_Numx_nbi} i.t.m.). }\\\addlinespace
    \bigcell{l}{diameter (a) \\ \cite{Chindelevitch2019}} & \bigcell{l}{maximal distance of two vertices} & $diam(T) \coloneqq \max\limits_{i,j \in V(T)}{d_{i,j}}$ & \bigcell{l}{Caterpillar not unique maximally imbalanced tree, as all trees $T=(T^{cat}_{n-k},T^{cat}_{k})$ \\ for a $k \in \{1,\ldots,n-1\}$ (including $\Tcat$) are maximal with $diam(T)=n$.}\\\addlinespace
    \bigcell{l}{ILnumber (a) \\ \cite{Colijn_phylogenetic_2014}} & \bigcell{l}{number of inner vertices with a \\single leaf child (\enquote{IL} nodes), \\equivalent to $mCI(T)$ for rooted \\binary trees} &  & \bigcell{l}{Maximally balanced trees not unique for $n=2^h$, \\$h\in\mathbb{N}_{\geq 3}$, see e.g. Figure \ref{Fig_cherry_max} on the left.}\\\addlinespace
    \bigcell{l}{ladder length (a) \\ \cite{Colijn_phylogenetic_2014}} & \bigcell{l}{the maximal length of a path \\consisting only of \enquote{IL} nodes} &  & \bigcell{l}{Maximally balanced trees not unique for $n=2^h$, \\$h\in\mathbb{N}_{\geq 3}$, see e.g. Figure \ref{Fig_cherry_max} on the left.}\\\addlinespace
	 \bottomrule
	\end{tabular}}
\end{sidewaystable}

\section{Obtaining new balance indices from established indices} \label{Sec_new}
Even though there already exists a multitude of (im)balance indices in the literature, several approaches for the construction of new balance indices from established ones have been investigated in recent years. \citet{Matsen2006}, for example, developed binary recursive tree shape statistics (BRTSS) as a framework to create new indices. Using a genetic algorithm (i.e. simulating evolution by randomly changing a population of indices over several generations, each time using well-performing ones to create the next generation) an initial set of BRTSS can be modified and evolved in order to optimize any objective function. In his study, he used a resolution function based on the so-called NNI distance matrix (see \cite{Matsen2006} for more information) reasoning that a good balance index should take on similar values for similar trees, but different values for significantly different trees.

Another approach is to combine already existing indices in order to create new ones. \citet{Cardona2012}, for example, analyzed the stochastic properties of the sum of the Sackin index and the total cophenetic index.

Similarly, but on a larger scale, \citet{Hayati2019} analyzed several linear combinations of pairs of balance indices. They evaluated and compared the performance of the resulting new indices with the established ones using a new resolution function based on so-called Laplacian matrices. One linear combination showed to be promising: The so-called Saless index -- a linear combination of the average leaf depth $\overline{N}$ and the corrected Colless index $I_C$ -- which is defined as $Saless(T) \coloneqq \lambda\cdot\overline{N}(T) + I_C(T)$. Note that the authors use the terms Sackin and Colless index for $\overline{N}$ and $I_C$. $\lambda$ is chosen to maximize the resolution (discriminatory power of a tree shape statistic measured by the newly developed Laplacian resolution function; see \cite{Hayati2019} for more details), and therefore $\lambda$ may be different for trees with varying numbers of leaves. The experiments of \citet{Hayati2019} suggest (but they have not formally proven it) that $\lambda$ converges to a limiting value as the number of leaves goes to infinity. They also show that  when using Matsen's distance resolution function for $n=7,\ldots,17$ the Saless index has in fact a higher resolution than the corrected Colless index, average leaf depth, variance of leaf depth, $I_2$ index, $B_1$ index, and $B_2$ index.

\section{Normalizing balance indices} \label{Sec_normalization}
Next to obtaining new balance indices from established indices, it is also often of interest to normalize existing indices in order to use them in the comparison of trees of different sizes. 
Indeed, when ordering trees from \enquote{unbalanced} to \enquote{balanced} (or vice versa) most indices presented in this manuscript only provide a meaningful ordering for trees with the same number of leaves, because the value of an index often depends on the size of the tree. When closed expressions for the minimum and maximum value on $\Tnstar$ (or $\BTnstar$) of some tree shape statistic $t$ are known and differ from each other, $t$ can be normalized such that its range becomes the unit interval $[0,1]$ by means of the usual affine transformation 
\[ \widetilde{t}(T) = \frac{t(T) - \min\limits_{T' \in \Tnstar \, (\BTnstar)} t(T')}{\max\limits_{T' \in \Tnstar  \, (\BTnstar)} t(T') - \min\limits_{T' \in \Tnstar  \, (\BTnstar)} t(T')}. \]
This normalized index then allows for the comparison of the balance of trees with different numbers of leaves. However, for certain indices such as the Colless index, this normalization tends to 0 in $L_2$ and probability when $n \rightarrow \infty$ both under the Yule and the uniform model for phylogenetic trees (see Discussion in \citet{Coronado2020a}). Moreover, normalization via affine transformation is only possible when both the minimum and maximum value of a tree shape statistic $t$ are known for all numbers of leaves which is not yet the case for all indices currently used in the literature. In addition to these technical problems, the affine transformation is also disputed for its meaningfulness. \citet{Heard1992}, for example, stated that normalizing an index to 0 on $\BTnstar$ for any $n\in\mathbb{N}_{\geq 1}$ is unsatisfactory because a tree whose leaf number is not a power of two can never be \enquote{perfectly balanced} even though it might be the most balanced tree of the considered space. \citet{Shao1990} mentioned the alternative of simply dividing the index of the tree by the maximal possible value for the same leaf number. Although they do not recommend this method, it has the advantage that only the maximal value on $\Tnstar$ ($\BTnstar$) must be known.

Another popular normalization approach (see e.g. \citet{Kirkpatrick1993}) for a tree shape statistic $t$, relative to a probabilistic model of phylogenetic trees, is to consider the standardization of $t$ under the respective model. Let $T_n$ be a phylogenetic tree with $n$ leaves sampled under some probabilistic model $\ast$ (e.g., the Yule model or the uniform model). Then, a tree shape statistic $t$ can be standardized by subtracting the expected value of $t$ under $\ast$ and dividing by its standard deviation under $\ast$:
$$ \widehat{t}(T) = \frac{t(T)-E_{\ast}(t(T_n))}{\sqrt{Var_{\ast}(t(T_n))}}.$$
A related approach was used by \citet{Blum2006a}, who normalized the Sackin index of a tree $T$ by subtracting its expected value and dividing by the number of leaves. Note that normalization via standardization does not require knowledge of the extremal values of a tree shape statistic $t$, but it requires knowledge of the expected value and variance of $t$ under the probabilistic model of interest. Again, these quantities are not known for all tree (im)balance indices currently discussed in the literature.

\section{Related concepts} \label{Sec_related} 

As we have shown in this manuscript, there exists a multitude of indices to measure how balanced or imbalanced a rooted tree is. Before we introduce our software package to compute these indices, we want to mention some additional concepts which are related to tree balance but which are beyond the scope of our manuscript.

While balance indices describe the degree of tree balance in a single number, several concepts that try to encapsulate more information have been proposed in the literature over the years. One of the best-known ideas was proposed by Aldous in 2001 \citep{Aldous_stochastic_2001}. Aldous suggested constructing a scatterplot of the inner vertices $v\in \mathring{V}(T)$: The $x$-coordinate is the size of the pending subtree rooted at $v$ (size of parent clade) and the $y$-coordinate is the size of the smaller pending subtree rooted at a child of $v$ (size of smaller daughter clade). Aldous proposed to perform nonlinear median regression on the log-log version of this scatter plot -- in the literature known as \enquote{Aldous scatterplot} -- to estimate the median size of the smaller daughter clade as a function of the size of the parent clade. The fitted function can then serve as a descriptor of the tree shape.

Similar to the Aldous scatterplot approach, there is another balance measurement that uses a scatterplot and statistical means to describe the degree of symmetry in trees. This measurement described by Stich and Manrubia \citep{Stich2009} has already been used for phylogenetic trees (see e.g. \citet{Herrada2008}), but also in other fields like transportation networks and food webs. For each vertex $v\in V(T)$, two values are calculated: the subtree size $A_v$, here defined as the number of all vertices (including $v$) in the pending subtree rooted at $v$, as well as the cumulative branch size $C_v$, defined as the sum of $A_w$ over all vertices $w$ in the pending subtree rooted at $v$ (i.e. $v$ and all descendants of $v$). Again, as in the Aldous scatterplot, the goal is to estimate the average value of $C$ as a function $f$ of $A$, i.e. this function tells us for any node $v$ in the tree with a subtree size of $A_v$ which value $C_v\approx f(A_v)$ we should be expecting for the cumulative branch size. Under various tree models like the Yule model as well as for real data, $A$ and $C$ approximately follow a power law: $C \sim A^{\eta}$ with $\eta=(1-\alpha)/(1-\gamma)$ and these parameters of the exponent can describe the degree of balance, e.g.  $(\alpha=2,\ \gamma=2,\ \eta=1)$ for fully balanced trees and $(\alpha=0,\ \gamma=1/2,\ \eta=2)$ for caterpillars (with correction). The exponent $\eta$ can be estimated for a tree using a least-squares regression after weighing the data to avoid over-representation of vertices close to the leaves.

Another approach, which was also intended to summarize balance in more than one number and which should therefore be mentioned here, are the $I$ values by \citet{Fusco1995}. In fact, the main idea had been to compare the symmetry of trees by comparing the frequency distribution of the $I$ values. Using statistics like the sum, mean, median or quartile distance which can be applied to the $I$ data and which have been used as balance indices several times in literature (see e.g. \cite{Agapow2002, Blum2006b}) have not been intended primarily, but as a second step in the analysis.

\section{Software} \label{Sec_Software}

In the previous sections, we have seen that there is a large variety of indices of tree (im)balance proposed and used in the literature. Similarly, there exists a multitude of software packages for calculating those measures. In fact, there currently exist at least 17 different packages that allow for the calculation of some of the (im)balance indices and tree shape statistics discussed above. We summarize them in Table~\ref{Table_ExistingTools}. While some popular indices, e.g. the Colless index or Sackin index, are implemented in several tools, other indices like the rooted quartet index are only available in one particular package. In order to compute different indices of tree (im)balance, it is thus currently necessary to refer to different software tools (and different programming languages). Moreover, some of the indices discussed above have not been implemented at all or are only available from the authors of the corresponding paper on request.

Our software package \verb|treebalance|, which accompanies this manuscript, will hopefully provide a remedy for this problem. \verb|treebalance| is a package of the statistical programming environment $\mathsf{R}$ \citep{RCoreTeam2019}, which can be downloaded from the comprehensive $\mathsf{R}$ archive network (CRAN, https://cran.r-project.org/package=treebalance) for all computing platforms. Its aim is to provide a system-independent, costless, fast and user-friendly tool to calculate a large variety of (im)balance indices for rooted trees. As such, it merges and unifies source code from different programming languages and different sources and expands it by indices that, so far, are not openly available. As of September 2021, it includes functions for all (im)balance indices (and tree shape statistics that comply with the definitions of (im)balance indices) that are mentioned in this manuscript, and it is planned to keep it up to date with the current state of research. The package, including a manual, can be obtained by downloading it from CRAN or by entering the following commands into $\mathsf{R}$’s console:

\indent\verb|> install.packages("treebalance")|\\
\indent\verb|> library("treebalance")|

\begin{longtable}{lll}
\caption[Existing software packages]{Existing software packages for calculating indices of tree balance. We remark that this list might not be exhaustive and might miss existing software packages. For example, \citet{Heard1992} mentions a computer program written in BASIC, but this does not seem to be publicly available. Moreover, \citet{Rogers1994} mentions a computer program written in FORTRAN that is only available from the author. Also, some indices have been implemented by \citet{Agapow2002} in MeSA and source code is available from \url{https://github.com/agapow/mesa-revenant}; however, MeSA does not seem to be maintained anymore (note that we downloaded the source code and tried to compile it on an Ubuntu 20.04.2 as well as a macOS Catalina 10.15.7 operating system, resulting in errors both times). }\\
\toprule
\label{Table_ExistingTools}
\centering
\setlength{\tabcolsep}{5mm}
Tool & Language & Indices implemented \\
\midrule
\endfirsthead
\multicolumn{3}{c}%
{\textbf{Table \ref{Table_ExistingTools}}: \textit{Continued from previous page}} \\
\toprule
Tool & Language & Indices implemented \\
\midrule
\endhead
\texttt{ape} \citep{Paradis2018} &  $\mathsf{R}$ &  \bigcell{l}{cherry index \citep{McKenzie2000} \\ \textit{Note that} \texttt{ape} \textit{offers a variety of tools to} \\ \textit{manipulate and analyze phylogenetic trees} \\ \textit{beyond the cherry index.}} \\ \addlinespace
\texttt{apTreeshape} \citep{Bortolussi2005} & $\mathsf{R}$ & \bigcell{l}{Colless index \citep{Colless1982} \\ Sackin index \citep{Sackin1972, Shao1990} \\ $\widehat{s}$-shape statistic \citep{blum2006c}} \\ \addlinespace
\texttt{castor} \citep{Louca2017} & $\mathsf{R}$ &  \bigcell{l}{Colless index \citep{Colless1982}\\ corrected Colless index \citep{Heard1992} \\ Sackin index \citep{Sackin1972, Shao1990} \\ \textit{Note that} \texttt{castor} \textit{offers a variety of tools to} \\ \textit{manipulate and analyze phylogenetic trees} \\ \textit{beyond balance.}} \\ \addlinespace
\texttt{CollessLike} \citep{Mir2018} & $\mathsf{R}$ & \bigcell{l}{Colless-like indices \citep{Mir2018}  \\ Sackin index \citep{Sackin1972, Shao1990} \\ total cophenetic index \citep{Mir2013}} \\\addlinespace
\texttt{phyloTop} \citep{Kendall2018} & $\mathsf{R}$ & \bigcell{l}{cherry index \citep{McKenzie2000} \\ Colless index \citep{Colless1982} \\ Sackin index \citep{Sackin1972, Shao1990} \\ \textit{Note that this tool also allows for the} \\ \textit{calculation of other toplogical properties of} \\ \textit{phylogenetic trees.} } \\ \addlinespace
\texttt{symmeTree} \citep{Kersting2021} & $\mathsf{R}$ & \bigcell{l}{symmetry nodes index \citep{Kersting2021} \\ Rogers $J$ index \citep{Rogers1996}}\\ \addlinespace
\texttt{TotalCopheneticIndex} \citep{Mir2013} & $\mathsf{R}$ & total cophenetic index \citep{Mir2013}\\ \addlinespace
\texttt{TreeShapeStats} \citep{Hayati2019} & $\mathsf{R}$ & \bigcell{l}{Saless index \citep{Hayati2019} and \\ other linear combinations \\ of estbalished indices} \\ \addlinespace
\texttt{TreeTools} \citep{TreeTools} & $\mathsf{R}$ & \bigcell{l}{total cophenetic index \citep{Mir2013} \\ \textit{Note that this tool implements various} \\ \textit{functions for the creation, modification,} \\ \textit{and analysis of phylogenetic trees.}} \\ \addlinespace
\texttt{treetop} \citep{Colijn2018} & $\mathsf{R}$ & \bigcell{l}{Given an integer $z$, this tool \\ returns the tree $T$ whose \\ Colijn-Plazotta rank is $z$. \\ \textit{Additionally this tool implements certain} \\ \textit{metrics for phylogenetic trees.}}\\ \addlinespace
\texttt{simmons} \citep{Matsen2006} & \bigcell{l}{command-line \\ program  written \\ in ocaml}& \bigcell{l}{average leaf depth \citep{Sackin1972, Shao1990} \\ variance of leaf depths \citep{Coronado2020b, Sackin1972, Shao1990} \\ corrected Colless index \citep{Heard1992} \\ $I_2$ index \citep{Mooers1997} \\ $B_1$ index \citep{Shao1990} \\ $B_2$ index for binary trees \citep{Kirkpatrick1993} \\ cherry index \citep{McKenzie2000}}\\ \addlinespace
\texttt{SkewMatic 2.01} \citep{Heard2007} & \bigcell{l}{Windows \\ application} & \bigcell{l}{corrected Colless index \citep{Heard1992} \\ \textit{Note that this tool also computes so-called} \\ \textit{phylogenetic diversity}.} \\ \addlinespace
\texttt{Bio::Phylo::Forest::TreeRole} \citep{Vos2011} & BioPerl & \bigcell{l}{cherry index \citep{McKenzie2000} \\ Colless index \citep{Colless1982} \\ $I_2$ index \citep{Mooers1997}} \\ \addlinespace
\texttt{Quartet\_Index} \citep{Coronado2019} & Python & rooted quartet index \citep{Coronado2019}\\ \addlinespace
\texttt{var\_depths} \citep{Coronado2020b} & Python & variance of leaf depths \citep{Coronado2020b, Sackin1972, Shao1990} \\ \addlinespace
\texttt{colless} \citep{Coronado2020a} & Python &  \bigcell{l}{Colless index \citep{Colless1982} \\ \textit{Additionally this tool computes all minimal} \\ \textit{Colless trees and related quantities.}} \\ \addlinespace
\texttt{variances} \citep{Cardona2012} & Python & \bigcell{l}{Sackin index \citep{Sackin1972, Shao1990} \\ Colless index \citep{Colless1982} \\ total cophenetic index  \citep{Mir2013} \\ sum of Sackin index and total cophenetic index \\ \quad  \citep{Cardona2012, Mir2013} \\ \textit{Note that this tool seems to have been} \\ \textit{established primarily for computing the} \\ \textit{expected values and variances of the } \\ \textit{indices as well as some covariances} \\ \textit{under the Yule model.}}\\ \addlinespace
\bottomrule
\end{longtable}

\section{Discussion and outlook} \label{Sec_Discussion}

\subsection{Summary and discussion} \label{Sec_Discussion1}
The aim of the present manuscript was to provide a thorough and comprehensive review of the multitude of measures of tree (im)balance and related concepts currently discussed in the literature. In fact, while the terms \enquote{balance index} and \enquote{imbalance index} are frequently used in the literature and there is some basic agreement on what these concepts mean, to our knowledge they had not been formally defined before. In this manuscript, we established precise criteria for a tree shape statistic to be a balance index (Definition \ref{def_balance}) or an imbalance index (Definition \ref{def_imbalance}), and classified the existing 19 measures of tree balance into 4 balance indices and 15 imbalance indices. We also identified 5 tree shape statistics that are (im)balance indices but have not been considered as such as well as 13 tree shape statistics that are used as (im)balance indices in the literature (e.g., the cherry index) but do not comply with our definition of an (im)balance index. The fact that established indices like the cherry index do not satisfy our criteria of an (im)balance index, could of course be seen as a weakness of our definition. In other words, our definitions could be criticized for being too strict. We thus remark that the indices we discarded might still be useful for assessing certain aspects of tree shapes; they simply do not fully comply with the general agreement in the literature that the caterpillar tree on $n$ leaves should be the unique most imbalanced tree for all positive integers $n$, whereas the fully balanced tree on $n$ leaves should be the unique most balanced binary tree for all positive integers $n$ that are powers of two. 

After categorizing the established indices, we then provided comprehensive fact sheets reviewing their general, combinatorial, and statistical properties. In order to do so, we summarized the current state of the literature on tree balance, but also provided numerous new mathematical results. For instance, for many of the \enquote{established} indices of tree (im)balance, we proved that they are indeed (im)balance indices. In addition, we established a web repository on tree balance (\url{treebalance.wordpress.com}) and introduced the new software package \texttt{treebalance} to compute all indices of tree (im)balance and additional tree shape statistics considered in this manuscript.

While we considered 19 different indices of tree (im)balance in this manuscript, our study also made explicit that many indices are closely related and sometimes there are only subtle differences between them. In some cases this leads to redundancy in the sense that two indices induce the same ordering on a set of trees from \enquote{unbalanced} to \enquote{balanced} (or vice versa) or can directly be obtained from one another; in other cases, even subtle differences in the definition of two indices lead to different assessments of tree balance. For instance, there is redundancy between the Colless index and the corrected Colless index since the latter is simply a normalization of the former; on the other hand, while the quadratic Colless index is also closely related to the Colless index (the Colless index is the sum of the balance values of all inner vertices, while the quadratic Colless index is the sum of the squared balance values of all inner vertices), there are striking differences between them. In particular, while there is precisely one tree in $\BTnstar$ minimizing the quadratic Colless index, namely the maximally balanced tree $\Tmb$, in general, there are several trees minimizing the Colless index with $\Tmb$ just being one of them. 

On the one hand, redundancy among different indices raises the question of whether 19 different indices (and potentially more being currently developed) are really needed or whether the set of indices used in practice should be reduced. On the other hand, differences between indices lead to the problem of determining the \enquote{best} (im)balance index to date. While there might not be a unique answer to the second question (since the usefulness of an index will likely depend on the application it is used for), some attempts have been made to assess the performance of different indices for different purposes (see, e.g.~\cite{Agapow2002, Blum2005, Hayati2019, Kirkpatrick1993, Matsen2006}). Developing rigorous criteria and guidelines for determining the best index for a particular purpose, might nevertheless be beneficial for both empiricists and theoreticians. Note, however, that it might even be questionable whether an (im)balance index in form of a single number is completely adequate to measure the balance of a tree or whether more encompassing approaches (e.g., the ones discussed in Section~\ref{Sec_related}) are needed and should be used in practice. Again, the answer to this question will likely depend on the situation and context, but it stresses the fact that developing guidelines on which measure to use in which situation, would be beneficial for future research in the field.

\subsection{Directions for future research} \label{Sec_FutureResearch}

As one of the aims of this study was to make explicit open questions concerning (im)balance indices and stimulate further research, we end this manuscript by discussing several directions for future research. A quick look at Table \ref{table_Overview} and the fact sheets in Section \ref{Sec_factsheets} reveals that there are still several open questions related to the established indices of tree (im)balance discussed in this manuscript. Filling these gaps in the literature is an immediate direction for future research. 

While for the established indices many properties are often already known and only certain aspects are still open (e.g., the combinatorial properties are known but not the statistical ones), there are several more open questions regarding the additional tree shape statistics discussed in Table~\ref{Table_bi_tss}. In particular, while we showed that all of them satisfy the criteria of an (im)balance index, we did not explore their other general, combinatorial and statistical properties any further. This again leads to several directions for future research. 

Moreover, in addition to the tree shape statistics mentioned in Table~\ref{Table_bi_tss}, there is a number of tree shape statistics inspired by network science like the \emph{maximal betweenness centrality} or the \emph{closeness centrality} with or without weighting that could potentially spawn additional (im)balance indices. Some of their properties have already been explored in \citep{Chindelevitch2019} although partially without explicit proof or only for certain choices of $n$. The authors only analyzed the maximum of each of these vertex measures per tree and they mentioned that other statistics like the minimum, mean, median and variance could also be considered, resulting in a wide range of candidates. The question if these fulfill the (im)balance index definition might be interesting for future research as it has -- to our knowledge -- not thoroughly been answered yet for any of them.

It might also be interesting to explore the approaches of other fields of research to describe tree topologies and their properties. While in phylogenetics a tree is often considered to be directed from the root to the leaves, giving rise to definitions of e.g. depth, height or the number of descendants, in hydrology for example the tree is analyzed from the leaves, the water sources, to the root where all streams have joined. This perspective led to the introduction of different vertex or edge orders (slightly similar to the idea of an inverse depth) and various tree measurements like the \emph{bifurcation ratio} and the \emph{stream length ratio} \cite[p.~307-313]{Smart_Channel_1972} which could, in turn, be useful in phylogenetics. There are also interesting tree shape statistics explored in the field of cell morphology, even though sometimes different terms for the graph-theoretical objects are used therein which makes it slightly more difficult to compare the measures and results. An example is the \emph{mean centrifugal order}, which in our notation could be described as the average vertex depth, whose stochastic properties have been analyzed as well \cite{VanPelt_centrifugal_1989}.

Another pathway to obtain new tree shape statistics could be to use tree metrics $d$ and measure the distance of a tree $T$ to the caterpillar $d(T,\Tcat )$. Example tree metrics could be the partition metric, which is based on splits, or the quartet metric, which uses a similar idea as the quartet index, \cite{Steel_distributions_1993} or even the $CP$-rank in its original intention \cite{Colijn2018}. It would be interesting to explore if these outcomes do in fact fulfill the definition of an (im)balance index and what their other properties are.
 
Finally, we want to mention that of course tree balance is not limited to rooted trees. However, while several authors have considered balance in the unrooted setting (\citet{Fischer2015, Wang2019}), tree balance for unrooted trees has been less extensively studied in the literature than for rooted trees, and there are no standard established indices (such as Sackin, Colless or the total cophenetic index for rooted trees). This again leaves room for future research.

\section{Fact sheets} \label{Sec_factsheets}

In this section, we provide a fact sheet containing statements on its general, combinatorial and statistical properties (see Section \ref{Sec_Combinatorial_and_Statistical} for more information) for each of the 19 established balance and imbalance indices (and additionally the cherry index). Intended as a reference guide for empiricists and theoreticians, it allows the reader to quickly access information on the desired index without having to perform a thorough literature research, while still being pointed to relevant publications.

The listed general properties of an index are its computation time, its recursiveness and its locality. The combinatorial properties are its maximal and minimal value together with a characterization and the number of trees achieving these extremal values. These are considered both on the space of arbitrary trees $\Tnstar$ and on the space of binary trees $\BTnstar$. As statistical properties, we list the expected value and variance under the Yule and uniform model. For each statement, we either provide a reference or -- if the statement has not been proven until now -- refer the reader to the respective proof in Appendix \ref{app_additionals}. Additionally, we mark properties that have not yet been (fully) analyzed as \textit{open problems.}

\subsection{Average leaf depth} \label{factsheet_ALD}

The average leaf depth is a normalized version of the Sackin index (Section \ref{factsheet_Sackin}). As such, it is defined for arbitrary trees and it is an imbalance index, i.e. for a fixed $n\in\mathbb{N}_{\geq 1}$ it increases with decreasing balance of the tree. It can be calculated using the function \texttt{avgLeafDepI} from our \textsf{R} package \texttt{treebalance}.

\begin{description}
    \item \props{Definition}{(\citet{Shao1990, Kirkpatrick1993})}{The average leaf depth $\overline{N}(T)$ of a tree $T\in\Tnstar$ is defined as \[ \overline{N}(T) \coloneqq \frac{1}{n}\cdot\sum\limits_{v\in V_L(T)} \delta_T(v) = \frac{1}{n}\cdot S(T). \]}
    \item \props{Computation time}{(this manuscript, see Proposition \ref{runtime_ALD})}{For every tree $T\in\Tnstar$, the average leaf depth $\overline{N}(T)$ can be computed in time $O(n)$.}
    \item \props{Recursiveness}{(this manuscript, see Proposition \ref{recursiveness_ALD})}{The average leaf depth is a recursive tree shape statistic. We have $\overline{N}(T)=0$ for $T\in \mathcal{T}_1^\ast$, and for every tree $T\in\Tnstar$ with $n \geq 2$ and standard decomposition $T=(T_1,\ldots,T_k)$ we have \[ \overline{N}(T)=1+\left(\sum\limits_{i=1}^k n_i\cdot\overline{N}(T_i)\right)\cdot\left(\sum\limits_{i=1}^k n_i\right)^{-1}. \]}
    \item \props{Locality}{(this manuscript, see Proposition \ref{locality_ALD})}{The average leaf depth is not local.}
    \item \props{Maximal value on $\Tnstar$ and $\BTnstar$}{(this manuscript, see Corollary \ref{max_ALD}; \citet[Theorem 4, p. 522]{Fischer2018})}{For every tree $T\in\Tnstar$ with $n$ leaves and $m$ inner vertices, the average leaf depth fulfills $\overline{N}(T)\leq m-\frac{(m-1)\cdot m}{2n}$. This bound is tight for $n=1$ and $m=0$ and for all $n\in\mathbb{N}_{\geq 2}$ and $m\in\mathbb{N}_{\geq 1}$ with $m\leq n-1$. Moreover, we have $\overline{N}(T)\leq\frac{n+1}{2}-\frac{1}{n}=\overline{N}(\Tcat)$.}
    \item \props{(Number of) trees with maximal value on $\Tnstar$ and $\BTnstar$}{(this manuscript, see Corollary \ref{max_ALD}; \citet[Theorem 1, Theorem 4, p.522]{Fischer2021})}{For the combination of $n=1$ and $m=0$ as well as for any given $n\in\mathbb{N}_{\geq 2}$ and $m\in\mathbb{N}_{\geq 1}$ with $m\leq n-1$, there is exactly one tree $T\in\Tnstar$ with $n$ leaves, $m$ inner vertices and maximal average leaf depth, namely the caterpillar tree on $m+1$ leaves that has $n-m-1$ additional leaves attached to the inner vertex with the largest depth. Moreover, we have $\overline{N}(T)<\overline{N}(\Tcat)$ if $T\neq \Tcat$, i.e. for $n\geq 1$ the \emph{binary} caterpillar tree $\Tcat$ is the unique maximal tree on $\BTnstar$ and on $\Tnstar$ if $m$ is not fixed.}
    \item \props{Minimal value on $\Tnstar$ and $\BTnstar$}{(this manuscript, see Corollary \ref{mintrees_ALD}; \citet[Theorem 5, p. 522]{Fischer2021})}{
    For every tree $T\in\Tnstar$ with $n$ leaves and $m$ inner vertices the average leaf depth fulfills \[ \overline{N}(T) \geq \begin{cases} 0 & \text{if } n=1 \text{ and } m=0, \\ \left\lfloor\log_2\left(\frac{n}{k}\right)\right\rfloor+3-\frac{k}{n}\cdot 2^{\lfloor\log_2(\frac{n}{k})\rfloor+1} & \text{else.} \end{cases} \] with $k=n-m+1$. This bound is tight for $n=1$ and $m=0$ as well as for all $n\in\mathbb{N}_{\geq 2}$ and all $m\in\mathbb{N}_{\geq 1}$ with $m\leq n-1$. Moreover, we have $\overline{N}(T)>1=\overline{N}(\Tstar)$ if $T\neq\Tstar$.} 
    \item \props{Trees with minimal value on $\Tnstar$ and $\BTnstar$}{(this manuscript, see Corollary \ref{mintrees_ALD}; \citet[Theorem 5, p. 522]{Fischer2021})}{If $n=1$ and $m=0$, there is precisely one tree with minimal average leaf depth, namely $T^\mathit{star}_1$. Moreover, for any given $n\in\mathbb{N}_{\geq 2}$ and $m\in\mathbb{N}_{\geq 1}$ with $m\leq n-1$, a tree $T$ with $n$ leaves and $m$ inner vertices has minimal average leaf depth if and only if it has $k=n-m+1$ maximal pending subtrees rooted in the children of the root $\rho$ and fulfills $|\delta_T(x)-\delta_T(y)|\leq 1$ for all $x,y\in V_L(T)$. If $m$ is not fixed the star tree $\Tstar$ is the unique tree minimizing the average leaf depth on $\Tnstar$.}
    \item \props{Number of trees with minimal value on $\Tnstar$}{(this manuscript, see Corollary \ref{nummintrees_ALD}; \citet[Theorem 5, p. 522]{Fischer2021})}{Let $n \in \mathbb{N}_{\geq1}$ and $m \in \mathbb{N}_{\geq0}$, and let $avl(n,m)$ denote the number of trees with $n$ leaves, $m$ inner vertices and minimal average leaf depth, and let $k=n-m+1$. Also, denote by $\mathcal{P}_k(n)$ the set of all sets of pairs $\{(a_1,\widetilde{n}_1),\ldots,(a_l,\widetilde{n}_l)\}$ where $a_i, \widetilde{n}_i\in\mathbb{N}_{\geq 1}$ are integers such that $\widetilde{n}_i\neq \widetilde{n}_j$ if $i\neq j$ and $2^{\delta-1}\leq \widetilde{n}_i\leq 2^{\delta}$ for $\delta=\left\lfloor\log_2\left(\frac{n}{k}\right)\right\rfloor+1$ and $a_1+\ldots+a_l=k$ and $a_1\cdot\widetilde{n}_1+\ldots+a_l\cdot\widetilde{n}_l=n$, i.e. each element in $\mathcal{P}_k(n)$ represents a specific unique integer partition of $n$. Then, we have $avl(n,m)=0$ if $m>n-1$ or if $m=0$ and $n>1$, and $avl(1,0)=1$ and otherwise: \[ avl(n,m)=\sum\limits_{ \substack{\{(a_1,\widetilde{n}_1),\ldots,(a_l,\widetilde{n}_l)\} \\ \in \mathcal{P}_k(n)} } \ \prod\limits_{i=1}^l \binom{avl(\widetilde{n}_i,\widetilde{n}_i-1)+a_i-1}{a_i}, \] 
    where $avl(n,n-1)$ corresponds to the number of rooted binary trees with $n$ leaves and minimal average leaf depth, which can be calculated by the formula presented in \cite[Theorem 3]{Fischer2021} (see also Online Encyclopedia of Integer Sequences \cite[Sequence A299037]{OEIS}, and next paragraph about the number of minimal binary trees in this fact sheet).\\
    Moreover, if the number of inner vertices $m$ is not fixed, the star tree $\Tstar$ is the unique tree minimizing the average leaf depth.}
    \item \props{Number of trees with minimal value on $\BTnstar$}{(this manuscript, see Corollary \ref{nummintrees_ALD}; \citet[p. 522, Theorem 3, Corollary 1]{Fischer2021})}{\textit{Open problem.} Let $n\in\mathbb{N}_{\geq 1}$ and let $\widehat{avl}(n)$ denote the number of binary trees with $n$ leaves that have minimal average leaf depth. Let $A(n)$ denote the set of pairs $A(n)=\{(n_a,n_b)|n_a,n_b\in\mathbb{N}_{\geq 1},n_a+n_b=n, \frac{n}{2}<n_a\leq 2^{\lceil\log_2(n)\rceil-1},n_b\geq2^{\lceil\log_2(n)\rceil-2}\}$, let \[ f(n)=\begin{cases} 0 & \text{if } n \text{ is odd} \\ \binom{\widehat{avl}\left(\frac{n}{2}\right)+1}{2} & \text{if } n \text{ is even.} \end{cases} \] Then, $\widehat{avl}(n)$ fulfills the recursion $\widehat{avl}(1)=1$ and for $n\geq 2$ \[ \widehat{avl}(n)=\sum\limits_{(n_a,n_b)\in A(n)} \widehat{avl}(n_a)\cdot\widehat{avl}(n_b)+f(n). \]
    A closed formula is not known yet. If $n\in\{2^m-1,2^m,2^m+1\}$ for some $m\in\mathbb{N}_{\geq 1}$, there is exactly one tree in $\BTnstar$ with minimal average leaf depth (in particular, if $n$ is a power of two, the fully balanced tree $\Tfb$ is the unique minimal tree). For all other $n$, there exist at least two trees in $\BTnstar$ with minimal average leaf depth.}
    \item \props{Expected value under the Yule model}{(\citet[Remark 1]{Coronado2020b}, \citet[Appendix]{Kirkpatrick1993})}{Let $T_n$ be a phylogenetic tree with $n$ leaves sampled under the Yule model. Then, the expected value of $\overline{N}$ of $T_n$ is \[ E_Y(\overline{N}(T_n)) = 2\cdot H_n-2. \]
    Moreover, in the limit
    \[ E_Y(\overline{N}(T_n)) \sim 2\cdot\ln(n). \]}
    \item \props{Variance under the Yule model}{(this manuscript, see Proposition \ref{varY_ALD})}{Let $T_n$ be a phylogenetic tree with $n$ leaves sampled under the Yule model. Then, the variance of $\overline{N}$ of $T_n$ is \[ V_Y(\overline{N}(T_n))=7-4\cdot H_n^{(2)}-\frac{2}{n}\cdot H_n-\frac{1}{n}. \] Moreover, in the limit \[ V_Y(\overline{N}(T_n)) \sim 7-\frac{2\pi^2}{3}. \]}
    \item \props{Expected value under the uniform model}{(\citet[Remark 2]{Coronado2020b})}{Let $T_n$ be a phylogenetic tree with $n$ leaves sampled under the uniform model. Then, the expected value of $\overline{N}$ of $T_n$ is \[ E_U(\overline{N}(T_n)) = \frac{(2n-2)!!}{(2n-3)!!}-1. \]
    Moreover, in the limit \[ E_U(\overline{N}(T_n)) \sim \sqrt{\pi n}. \]}
    \item \props{Variance under the uniform model}{(this manuscript, see Proposition \ref{varU_ALD})}{Let $T_n$ be a phylogenetic tree with $n$ leaves sampled under the uniform model. Then, the variance of $\overline{N}$ of $T_n$ is \[ V_U(\overline{N}(T_n))=\frac{10n^2-3n-1}{3n}-\frac{n+1}{2n}\cdot\frac{(2n-2)!!}{(2n-3)!!}-\left(\frac{(2n-2)!!}{(2n-3)!!}\right)^2. \] Moreover, in the limit \[ V_U(\overline{N}(T_n)) \sim \left(\frac{10}{3}-\pi\right)\cdot n. \]}
\end{description}

\subsection{\texorpdfstring{$B_1$}{B\_1} index} \label{factsheet_B1}

The $B_1$ index was the first balance index that \citet{Shao1990} introduced in their paper, hence the name $B_1$. While it is defined on $\Tnstar$ (and this profile sheet includes properties for arbitrary trees), the $B_1$ index does only fulfill our definition of a balance index when it is restricted to binary trees (because otherwise the caterpillar tree is not a unique extremal tree, see Remark \ref{remark_B1_restriction}). On $\BTnstar$ it is a balance index, i.e. for a fixed $n\in\mathbb{N}_{\geq 1}$ it increases with increasing balance of the tree. It can be calculated using the function \texttt{B1I} from our \textsf{R} package \texttt{treebalance}.

\begin{description}
    \item \props{Definition}{(\citet{Shao1990})}{The $B_1$ index $B_1(T)$ of a tree $T\in\Tnstar$ is defined as \[ B_1(T)\coloneqq\sum\limits_{v\in\mathring{V} (T)\setminus\{\rho\}} \frac{1}{h(T_v)}.\]}
    \item \props{Computation time}{(this manuscript, see Proposition \ref{runtime_B1})}{For every tree $T\in\Tnstar$, the $B_1$ index $B_1(T)$ can be computed in time $O(n)$.}
    \item \props{Recursiveness}{(this manuscript, see Proposition \ref{recursiveness_B1})}{The $B_1$ index is a recursive tree shape statistic. We have $B_1(T)=0$ for $T\in\mathcal{T}_1^\ast$, and for every tree $T\in\Tnstar$ with $n \geq 2$ and standard decomposition $T=(T_1,\ldots,T_k)$ we have \[ B_1(T)=\sum\limits_{i=1}^k B_1(T_i)+\sum\limits_{i=1}^k \frac{1-\mathcal{I}(h(T_i)=0)}{h(T_i)}. \]}
    \item \props{Locality}{(this manuscript, see Proposition \ref{locality_B1})}{The $B_1$ index is not local.}
    \item \textbf{Maximal value on $\Tnstar$:} \textit{Open problem.}
    \item \props{Maximal value on $\BTnstar$}{(this manuscript, see Theorem \ref{Prop_B1_fb})}{\textit{Open problem.} Let $n=2^h$ for some $h\in\mathbb{N}_{\geq 0}$. Then, for every binary tree $T\in\BTnstar$, the $B_1$ index fulfills \[ B_1(T)\leq \sum\limits_{i=1}^{h-1}\frac{2^i}{h-i}. \] This bound is tight. An upper bound if $n$ is not a power of two is not known yet\footnote{Note that \citet{Shao1990} have suggested an upper bound for all $n$, but this bound is unfortunately erroneous (see Remark \ref{Rem_B1_Shao}).}.}
    \item \textbf{Trees with maximal value on $\Tnstar$:} \textit{Open problem.}
    \item \textbf{Number of trees with maximal value on $\Tnstar$:} \textit{Open problem.}
    \item \props{(Number of) trees with maximal value on $\BTnstar$}{(this manuscript, see Theorem \ref{Prop_B1_fb})}{\textit{Open problem.} Let $n=2^h$ for some $h\in\mathbb{N}_{\geq 0}$. Then, the fully balanced tree $\Tfb$ is the unique binary tree achieving the maximal $B_1$ value, i.e. $B_1\left(\Tfb\right)=\sum\limits_{i=1}^{h-1}\frac{2^i}{h-i}$. A full characterization of the trees with maximal $B_1$ index if $n$ is not a power of two (and their number) is not known yet.}
    \item \props{Minimal value on $\Tnstar$ and $\BTnstar$}{(\citet{Shao1990}; this manuscript, see Theorem \ref{min_B1_a} and Corollary \ref{nummintrees_B1_b})}{For every tree $T\in\Tnstar$ with $m$ inner vertices, the $B_1$ index fulfills $B_1(T)\geq H_{m-1}$. This bound is tight for $n=1$ and $m=0$ and for all $n\in\mathbb{N}_{\geq 2}$ and $m\in\mathbb{N}_{\geq 1}$ with $m\leq n-1$ (note that $H_{-1}=H_0=0$ as an empty sum evaluates to zero). In particular, we have $B_1(T)\geq 0$ for all $T\in\Tnstar$ and $B_1(T)\geq H_{n-2}$ for all $T\in\BTnstar$.}
    \item \props{(Number of) trees with minimal value on $\Tnstar$ and $\BTnstar$}{(this manuscript, see Proposition \ref{nummintrees_B1}, Corollary \ref{nummintrees_B1_b}, and Remark \ref{remark_B1_restriction})}{Let $n\in\mathbb{N}_{\geq 1}$ and $m\in\mathbb{N}_{\geq 0}$, let $b(n,m)$ denote the number of trees in $\Tnstar$ that have $m$ inner vertices and minimal $B_1$ index. Then, $b(n,m)=0$ if $m>n-1$ or if $m=0$ and $n>1$, $b(1,0)=1$, and otherwise $b(n,m)=\binom{n-2}{n-m-1}$. Moreover, the rooted star tree $\Tstar$ is the unique minimal tree on $\Tnstar$ if $m$ is not fixed, and the caterpillar tree $\Tcat$ is the unique minimal tree on $\BTnstar$.}
    \item \textbf{Expected value under the Yule model:} \textit{Open problem.}
    \item \textbf{Variance under the Yule model:} \textit{Open problem.}
    \item \textbf{Expected value under the uniform model:} \textit{Open problem.}
    \item \textbf{Variance under the uniform model:} \textit{Open problem.}
\end{description}

\subsection{\texorpdfstring{$B_2$}{B\_2} index} \label{factsheet_B2}

While \citet{Shao1990} originally defined the $B_2$ index for trees, \citet{Bienvenu2020} found that, due to its probabilistic interpretation, it is applicable to networks as well. In its definition, the value $p_x$ is the probability of reaching leaf $x$ when starting at the root and assuming equiprobable branching at each inner vertex. The $B_2$ index then measures the equitability of the probabilities of arriving at the leaves. More precisely, the $B_2$ index is the Shannon entropy of the probability distribution $(p_x)_{x \in V_L(T)}$. \citet{Bienvenu2020} described this vividly as letting water drip down from the root along the edges and then measure how evenly it is distributed among the leaves. \citet{Shao1990} did not state what logarithm base to use (possibly $\log_{10}$ considering their example), but \citet{Bienvenu2020} pointed out that base 2 is convenient when working with binary trees. It can be calculated using the function \texttt{B2I} from our \textsf{R} package \texttt{treebalance}.

Note that since we assume that the root has no incoming edge, our definition of $p_x$ differs slightly from the one given by \citet{Shao1990}. 
    
\begin{description}
    \item \props{Definition}{(\citet{Shao1990})}{The $B_2$ index $B_2(T)$ of a tree $T\in\Tnstar$ is defined as \[ B_2(T)\coloneqq -\sum\limits_{x\in V_L(T)} p_x\cdot\log(p_x) \text{\quad\quad with \quad\quad} p_x\coloneqq\prod\limits_{v\in anc(x)} \frac{1}{|child(v)|}.\] Note that in a binary tree, we have $p_x=(1/2)^{|anc(x)|}=(1/2)^{\delta_T(x)}$, because each inner vertex has exactly two children.}
    \item \props{Computation time}{(this manuscript, see Proposition \ref{runtime_B2})}{For every tree $T\in\Tnstar$, the $B_2$ index $B_2(T)$ can be computed in time $O(n)$ (regardless of the logarithm base).}
    \item \props{Recursiveness}{(\citet[Corollary 1.12]{Bienvenu2020}; this manuscript, see Remark \ref{remark_recursiveness_B2})}{The $B_2$ index is a binary recursive tree shape statistic (regardless of the logarithm base). We have $B_2(T)=0$ for $T\in\mathcal{BT}_1^\ast$, and for every binary tree $T\in\BTnstar$ with $n \geq 2$ and standard decomposition $T=(T_1,T_2)$ we have \[ B_2(T)=\frac{1}{2}\cdot(B_2(T_1)+B_2(T_2))+1. \]}
    \item \props{Locality}{(this manuscript, see Proposition \ref{locality_B2})}{The $B_2$ index is not local (regardless of the logarithm base).}
    \item \props{Maximal value on $\Tnstar$}{(\citet[Example 1.5, Proposition 1.7, page 8]{Bienvenu2020})}{Assume that the logarithm base in the definition is 2. Then, for every tree $T\in\Tnstar$, the $B_2$ index fulfills $B_2(T)\leq \log_2(n)$. This bound is tight for all $n\in\mathbb{N}_{\geq 1}$. More generally, as the $B_2$ index is the Shannon entropy of the probability distribution $(p_x)_{x \in V_L(T)}$, we have $B_2(T) \leq \log(|V_L(T)|) = \log(n)$ with equality if and only if $(p_x)_{x \in V_L(T)}$ is uniform (see, e.g. \citet[Chapter 2.5]{Mackay2003}).}
    \item \props{Maximal value on $\BTnstar$}{(\citet[Theorem 2.3]{Bienvenu2020})}{Assume that the logarithm base in the definition is 2. Then, for every binary tree $T\in\BTnstar$, the $B_2$ index fulfills \[ B_2(T)\leq \lfloor\log_2(n)\rfloor+\frac{n-2^{\lfloor\log_2(n)\rfloor}}{2^{\lfloor\log_2(n)\rfloor}}. \] This bound is tight for all $n\in\mathbb{N}_{\geq 1}$.}
    \item \props{Trees with maximal value on $\Tnstar$}{(\citet[Example 1.5, page 8]{Bienvenu2020})}{Assume that the logarithm base in the definition is 2. Then, the star tree $\Tstar$ and -- given that $n=2^h$ for some $h\in\mathbb{N}_{\geq 0}$ -- the fully balanced tree $\Tfb$ are among the arbitrary trees with maximal $B_2$ value. More generally, as the $B_2$ index is a Shannon entropy, it is maximized if and only if $p_x = \frac{1}{n}$ for all $x \in V_L(T)$, i.e. if and only if the probability distribution $(p_x)_{x \in V_L(T)}$ is uniform (e.g. \citet[Chapter 2.5]{Mackay2003}). As for all $n \in \mathbb{N}_{\geq 1}$, there exists at least on rooted tree $T \in \Tnstar$ such that $(p_x)_{x \in V_L(T)}$ is uniform, namely the rooted star tree $\Tstar$, the maximum value of $\log(n)$ is assumed for all $n$, and thus the property $p_x = \frac{1}{n}$ for all $x \in V_L(T)$ characterizes the class of arbitrary rooted trees with maximum $B_2$ index.}
    \item \props{Trees with maximal value on $\BTnstar$}{(\citet[Theorem 2.3]{Bienvenu2020})}
    {Assume that the logarithm base in the definition is 2. Then, a binary tree $T\in\BTnstar$ has maximal $B_2$ index, i.e. $B_2(T)=\lfloor\log_2(n)\rfloor+\frac{n-2^{\lfloor\log_2(n)\rfloor}}{2^{\lfloor\log_2(n)\rfloor}}$, if and only if it fulfills $\max\limits_{x,y\in V_L(T)} |\delta_T(x)-\delta_T(y)|\leq 1$.}
    \item \textbf{Number of trees with maximal value on $\Tnstar$:} \textit{Open problem.}
    \item \props{Number of trees with maximal value on $\BTnstar$}{(this manuscript, see Proposition \ref{nummaxtrees_B2})}{\textit{Open problem.} Assume that the logarithm base is 2. Let $n\in\mathbb{N}_{\geq 1}$ and let $g(n)$ denote the number of binary trees with $n$ leaves that have maximal $B_2$ index. Let $A(n)$ denote the set of pairs $A(n)=\{(n_a,n_b)|n_a,n_b\in\mathbb{N}_{\geq 1},n_a+n_b=n, \frac{n}{2}<n_a\leq 2^{\lceil\log_2(n)\rceil-1},n_b\geq2^{\lceil\log_2(n)\rceil-2}\}$ and let \[ f(n)=\begin{cases} 0 & \text{if } n \text{ is odd} \\ \binom{g\left(\frac{n}{2}\right)+1}{2} & \text{if } n \text{ is even} \end{cases}. \] Then, $g(n)$ fulfills the recursion $g(1)=1$ and for $n\geq 2$ \[ g(n)=\sum\limits_{(n_a,n_b)\in A(n)} g(n_a)\cdot g(n_b)+f(n). \] A closed formula is not known yet. If $n\in\{2^m-1,2^m,2^m+1\}$ for some $m\in\mathbb{N}_{\geq 1}$, there is exactly one tree in $\BTnstar$ with maximal $B_2$ index (in particular, if $n$ is a power of two, the fully balanced tree $\Tfb$ is the unique maximal tree)}. For all other $n$, there exist at least two trees in $\BTnstar$ with maximal $B_2$ index.
    \item \props{Minimal value on $\Tnstar$ and $\BTnstar$}{(\citet[Theorem 2.3]{Bienvenu2020}, this manuscript, see Theorem \ref{thm_B2_min} and \ref{prop_B2cat})} {For every tree $T \in \Tnstar$, the $B_2$ index fulfills $B_2(T) \geq \log(2) \cdot (2-2^{-n+2})$ (regardless of the logarithm base). This bound is tight for all $n \in \mathbb{N}_{\geq 1}$ in both the binary and arbitrary case.}
    \item \props{(Number of) trees with minimal value on $\Tnstar$ and $\BTnstar$} {(\citet[Theorem 2.3]{Bienvenu2020}, this manuscript, see Theorem \ref{thm_B2_min} and \ref{prop_B2cat})}{For any given $n \in \mathbb{N}_{\geq 1}$, there is exactly one tree $T \in \Tnstar$ with minimal $B_2$ index, namely the caterpillar tree $\Tcat$. Since $\Tcat$ is binary, it is also the unique tree with minimal $B_2$ index in $\BTnstar$.}
    \item \props{Expected value under the Yule model}{(\citet[Theorem 3.6, page 7]{Bienvenu2020})}{Assume that the logarithm base is 2. Let $T_n$ be a phylogenetic tree with $n$ leaves sampled under the Yule model. Then, the expected value of $B_2$ of $T_n$ is \[ E_Y(B_2(T_n))=\sum\limits_{i=1}^{n-1} \frac{1}{i} =H_{n-1}. \] Moreover, in the limit $E_Y(B_2(T_n)) \sim \ln(n)$.}
    \item \props{Variance under the Yule model}{(\citet[Theorem 3.6]{Bienvenu2020})}{\textit{Open problem.} Assume that the logarithm base is 2. Let $T_n$ be a phylogenetic tree with $n$ leaves sampled under the Yule model. Then, the variance of $B_2$ of $T_n$ fulfills in the limit \[ V_Y(B_2(T_n))\sim 2-\frac{\pi^2}{6}. \] An exact formula for the variance is not known yet.}
    \item \props{Expected value under the uniform model}{(\citet[Theorem 3.8, page 7]{Bienvenu2020})}{Assume that the logarithm base is 2. Let $T_n$ be a phylogenetic tree with $n$ leaves sampled under the uniform model. Then, the expected value of $B_2$ of $T_n$ is \[ E_U(B_2(T_n))=\frac{3(n-1)}{n+1}. \] Moreover, in the limit $E_U(B_2(T_n)) \sim 3$.}
    \item \props{Variance under the uniform model}{(\citet[Theorem 3.8]{Bienvenu2020})}{\textit{Open problem.} Assume that the logarithm base is 2. Let $T_n$ be a phylogenetic tree with $n$ leaves sampled under the uniform model. Then, the variance of $B_2$ of $T_n$ fulfills in the limit \[ V_U(B_2(T_n))\sim\frac{4}{9}. \] An exact formula for the variance is not known yet.}
    \item \textbf{Comments:} \citet{Shao1990} considered the $B_2$ index as \enquote{deficient}, because the values themselves and the range of values decrease with increasing $n$. \citet[p. 6, Proposition A.2.2]{Bienvenu2020} later proved that the number of distinct $B_2$ values on $\BTnstar$ and $\BTn$ is at least $2^{\lfloor n/2\rfloor-1}$. For $n\geq 20$, this exceeds the number of distinct values of the Sackin and Colless index on $\BTnstar$ and $\BTn$, which is at most $O(n^2)$ as both are restricted to integers, implying that $B_2$ is more suitable to discriminate between binary trees \citep{Bienvenu2020}.
    For additional results on the $B_2$ index of rooted phylogenies (especially the subcase of tree-child networks), see \citep{Bienvenu2020}.
\end{description}

\subsection{Cherry index} \label{factsheet_cherry}

The cherry index is a tree shape statistic that uses the number of cherries to measure tree symmetry. However, it is important to note that while it is often considered an established balance index, the cherry index does \emph{not} fulfill our definition of a balance or imbalance index since the fully balanced tree is not a unique extremum when $n$ is a power of two; see Figure \ref{Fig_cherry_max}.
The cherry index is normally used for rooted trees, but can be applied to unrooted trees as well. It can be calculated using the function \texttt{cherryI} from our \textsf{R} package \texttt{treebalance}.

\begin{description}
    \item \props{Definition}{(\citet{McKenzie2000})}{The cherry index $ChI(T)$ of a tree $T \in \Tnstar$ is defined as the number of cherries $c(T)$: \[ ChI(T) \coloneqq c(T). \]}
    \item \props{Computation time}{(this manuscript, see Proposition \ref{runtime_cherry})}{For every tree $T\in\Tnstar$, the cherry index $ChI(T)$ can be computed in time $O(n)$.}
    \item \props{Recursiveness}{(this manuscript, see Proposition \ref{recursiveness_cherry})}{The cherry index is a recursive tree shape statistic. We have $ChI(T)=0$ for $T\in\mathcal{T}_1^\ast$, and for every tree $T\in\Tnstar$ with $n \geq 2$ and standard decomposition $T=(T_1,\ldots,T_k)$ we have \[ ChI(T)=\sum\limits_{i=1}^k ChI(T_i)+\binom{\sum\limits_{i=1}^k \mathcal{I}(ChI(T_i)=0)}{2}. \]}
    \item \props{Locality}{(this manuscript, see Proposition \ref{locality_cherry})}{The cherry index is local.}
    \item \props{Maximal value on $\Tnstar$}{(this manuscript, see Theorem \ref{prop_cherry_max_a})}{For every tree $T\in\Tnstar$, the cherry index fulfills $ChI(T)\leq \binom{n}{2}$. This bound is tight for all $n\in\mathbb{N}_{\geq 1}$.}
    \item \props{Maximal value on $\BTnstar$}{(\citet[Theorem 4.3]{Kersting2021})}{For every binary tree $T\in\BTnstar$ with $n \geq 2$ leaves, the cherry index fulfills $ChI(T)\leq \lfloor \frac{n}{2}\rfloor$. This bound is tight for all $n\in\mathbb{N}_{\geq 2}$.}
    \item \props{(Number of) trees with maximal value on $\Tnstar$}{(this manuscript, see Theorem \ref{prop_cherry_max_a})}{For any given $n\in\mathbb{N}_{\geq 1}$, there is exactly one tree $T\in\Tnstar$ with maximal cherry index, i.e. $ChI(T)=\binom{n}{2}$, namely the rooted star tree $\Tstar$.}
    \item \props{Trees with maximal value on $\BTnstar$}{(\citet[Theorem 4.3]{Kersting2021})}{ For any given $n\in\mathbb{N}_{\geq 2}$, the maximal binary trees regarding the cherry index consist of a so-called top tree with $\lceil\frac{n}{2}\rceil$ leaves. Cherries are attached to $\lfloor\frac{n}{2}\rfloor$ of these leaves by merging the leaves with the parent vertex of a cherry. Thus, in case of $n$ being odd there will be one single leaf remaining.}
    \item \props{Number of trees with maximal value on $\BTnstar$}{(\citet[Theorem 4.4]{Kersting2021})}{The number of trees in $\BTnstar$ achieving the maximal cherry index is $we(\frac{n}{2})$ if $n$ is even and $A(\lfloor \frac{n }{2} \rfloor )$ if $n$ is odd. Here, $A(m)$ denotes the number of different ways to insert parentheses in the term $x^m \cdot c$ with multiplication being commutative, but not associative (sequence A085748 in Sloane's \emph{On-Line Encyclopedia of Integer Sequences} \citep{OEIS}).}
    \item \props{Minimal value on $\Tnstar$ and $\BTnstar$}{(this manuscript, see Theorem \ref{prop_cherry_min_a}; \citet[Theorem 4.1, p. 4]{Kersting2021})}{For $T\in\mathcal{T}_1^\ast=\mathcal{BT}_1^\ast$ we have $ChI(T)=0$. For every tree $T\in\Tnstar$ with $n \geq 2$ leaves, the cherry index fulfills $ChI(T)\geq 1$. This bound is tight for all $n\in\mathbb{N}_{\geq 2}$ in both the binary and the arbitrary case.}
    \item \props{(Number of) trees with minimal value on $\Tnstar$ and $\BTnstar$}{(this manuscript, see Theorem \ref{prop_cherry_min_a}; \citet[Theorem 4.1, p.4]{Kersting2021})}{For any given $n\in\mathbb{N}_{\geq 1}$, there is exactly one tree $T\in\Tnstar$ with minimal cherry index, i.e. $ChI(T)=0$ if $n=1$ and $ChI(T)=1$ else, namely the caterpillar tree $\Tcat$. Since $\Tcat$ is binary, it is also the unique tree with minimal value in $\BTnstar$.}
    \item \props{Expected value and variance under the Yule model}{(\citet[Theorem 2]{McKenzie2000})}{Let $T_n$ be a phylogenetic tree with $n$ leaves sampled under the Yule model. Then, the expected value of $ChI$ of $T_n$ is $E_Y(ChI(T_n))=0$ if $n=1$, $E_Y(ChI(T_n))=1$ if $n=2$, and  $E_Y(ChI(T_n))=\frac{n}{3}$ for $n \geq 3$; the variance is $V_Y(ChI(T_n))=0$ if $n\in\{1,2,3\}$, $V_Y(ChI(T_n))=\frac{2}{9}$ if $n=4$, and $V_Y(ChI(T_n))=\frac{2n}{45}$ for $n\geq 5$. Furthermore, we have $\frac{ChI(T_n)- n/3}{\sqrt{2n/45}} \rightarrow \mathcal{N}(0,1)$.}
    \item \props{Expected value under the uniform model}{(\citet[Corollary 4]{Wu2015})}{Let $T_n$ be a phylogenetic tree with $n$ leaves sampled under the uniform model. Then the expected value of $ChI$ of $T_n$ is $E_U(ChI(T_n)) = \frac{n(n-1)}{2(2n-3)}$ for $n\geq 1$. Moreover, in the limit $E_U(ChI(T_n)) \sim \frac{n}{4}$.} 
    \item \props{Variance under the uniform model}{(\citet[Corollary 4]{Wu2015})}{Let $T_n$ be a phylogenetic tree with $n$ leaves sampled under the uniform model. Then the variance of $ChI$ of $T_n$ is $V_U(ChI(T_n)) = \frac{n(n-1)(n-2)(n-3)}{2(2n-3)^2(2n-5)}$ for $n \geq 1$. Moreover, in the limits $V_U(ChI(T_n)) \sim \frac{n}{16}$.}
    \item \textbf{Comments:} All combinatorial results for binary trees have also been shown for a modified version of the cherry index defined as $mChI(T)=n-2\cdot c(T)$, counting the leaves that are not in a cherry (\citet[Definition 2.2]{Kersting2021}). This modified version is therefore measuring imbalance as it assigns higher values to trees with a higher degree of asymmetry (fewer cherries) in trees. All results for binary trees can be easily converted to the two different index versions. However, note that the results for arbitrary trees mentioned above are not applicable to this modified version of the cherry index.
\end{description}

\subsection{Colijn-Plazzotta rank} \label{factsheet_CP}

The $CP$-ranking or $CP$-labeling induces a bijective map between the rooted binary trees and the positive integers \citep{Colijn2018, Rosenberg2020}. While \citet{Colijn2018} originally developed it in order to define metrics on tree shapes, \citet{Rosenberg2020} found that it could also be used to measure the balance of a binary tree. As such it is an imbalance index, i.e. for a fixed $n\in\mathbb{N}_{\geq 1}$ it increases with decreasing balance of the tree. 

The $CP$-labeling can be extended to arbitrary trees with a fixed maximal number of children of any vertex (see supplementary material of \citep{Colijn2018}). However, this extension has not been extensively studied yet. In fact, it is even unknown if the $CP$-rank fulfils our (im)balance index criteria when arbitrary trees are considered. As indicated in this fact sheet, the $CP$-rank for arbitrary trees thus induces several open problems.

Since the $CP$-ranking (given the maximal size of any multifurcation) is a bijective map, no two trees have the same value, and it thus has (like the Furnas rank, see Section \ref{factsheet_Furnas}) the highest resolution that any index can have. Note that due to the exponentially increasing number of tree shapes in $\Tnstar$ and the fact that trees with \enquote{adjacent} $CP$-ranks can have quite different leaf numbers, the $CP$-labels can become extremely large for even relatively small $n$ \citep{Colijn2018}.
The function \texttt{colPlaLab} for the calculation of the $CP$-rank as well as \texttt{colPlaLab\_inv}, its inverse function, can be found in our \textsf{R} package \texttt{treebalance}.

\begin{description}
    \item \props{Definition}{(\citet{Colijn2018})}{The Colijn-Plazzotta rank $CP(T)$ of a binary tree $T\in\BTnstar$ is recursively defined as $CP(T)=1$ if $T$ consists of only one leaf and \[ CP(T)\coloneqq\frac{1}{2}\cdot CP(T_1)\cdot (CP(T_1)-1)+CP(T_2)+1 \] (with $CP(T_1)\geq CP(T_2)$) if $T$ has at least two leaves and the standard decomposition $T=(T_1,T_2)$.\\
    Given $\ell$ as the maximal number of children of any vertex, the Colijn-Plazzotta rank $CP(T)$ of an arbitrary tree $T\in\Tnstar$ is recursively defined as $CP(T)=0$ if $T$ is the empty tree (with no vertices), $CP(T)=1$ if $T$ consists of only one leaf and \[CP(T)\coloneqq\sum\limits_{i=1}^\ell \binom{CP(T_i)+i-1}{i}\] (with $CP(T_\ell)\geq CP(T_{\ell-1})\geq\ldots\geq CP(T_1)$) if $T$ has at least two leaves and the standard decomposition $T=(T_1,\ldots,T_k)$ with $k\leq\ell$. Note that if $k < l$, trees $T_{k+1}, \ldots, T_\ell$ are empty and thus $CP(T_{k+1})= \ldots= CP(T_\ell)=0$ in the above sum.}
    \item \props{Computation time}{(\citet[p. 115]{Colijn2018})}{\textit{Open problem.} For every binary tree $T\in\BTnstar$, the Colijn-Plazzotta rank $CP(T)$ can be computed in time $O(n)$. For arbitrary trees, the computation time is to our knowledge not known yet.}
    \item \props{Recursiveness}{(\citet[p. 114]{Colijn2018})}{\textit{Open problem.} The $CP$-rank is a binary recursive tree shape statistic. We have $CP(T)=1$ for $T\in\mathcal{BT}_1^\ast$, and for every binary tree $T\in\BTnstar$ with $n\geq 2$ and standard decomposition $T=(T_1,T_2)$ we have \[ CP(T)=\frac{1}{2}\cdot \max\{CP(T_1),CP(T_2)\}\cdot (\max\{CP(T_1),CP(T_2)\}-1)+\min\{CP(T_1),CP(T_2)\}+1. \] For arbitrary trees, the recursiveness is to our knowledge not known yet.}
    \item \props{Locality}{(this manuscript, see Proposition \ref{locality_CP})}{The Colijn-Plazzotta rank is not local.}
    \item \textbf{Maximal value on $\Tnstar$:} \textit{Open problem.}
    \item \props{Maximal value on $\BTnstar$}{(\citet[Theorem 9, Corollary 14]{Rosenberg2020})}{\textit{Open problem.} For every binary tree $T\in\BTnstar$, the Colijn-Plazzotta ranking fulfills $CP(T)\leq b(n)$ with $b(n)$ following the recursion $b(1)=1$ and \[ b(n)=b(n-1)\cdot \frac{b(n-1)-1}{2}+2 \] for $n\geq 2$. A closed formula is not known yet. This bound is tight for all $n\in\mathbb{N}_{\geq 1}$. The asymptotic behavior of $b(n)$ is $b(n)\sim 2\cdot\beta^{(2^n)}$ for a constant $\beta\approx 1.05653$.}
    \item \textbf{(Number of) trees with maximal value on $\Tnstar$:} \textit{Open problem.}
    \item \props{(Number of) trees with maximal value on $\BTnstar$}{(\citet[Corollary 10]{Rosenberg2020})}{For any given $n\in\mathbb{N}_{\geq 1}$, there is exactly one binary tree $T\in\BTnstar$ with maximal Colijn-Plazzotta rank, i.e. $CP(T)=b(n)$, namely the caterpillar tree $\Tcat$.}
    \item \textbf{Minimal value on $\Tnstar$:} \textit{Open problem.}
    \item \props{Minimal value on $\BTnstar$}{(\citet[Theorem 6, Proposition 15]{Rosenberg2020})}{\textit{Open problem.} For every binary tree $T\in\BTnstar$, the Colijn-Plazzotta rank fulfills $CP(T)\geq a(n)$ with $a(n)$ following the recursion $a(1)=1$ and \[ a(n)=\frac{1}{2}\cdot a\left(\left\lceil\frac{n}{2}\right\rceil\right)\cdot\left(a\left(\left\lceil\frac{n}{2}\right\rceil\right)-1\right)+1+a\left(\left\lfloor\frac{n}{2}\right\rfloor\right) \] for $n\geq 2$. A closed formula is not known yet. This bound is tight for all $n\in\mathbb{N}_{\geq 1}$. The value $a(n)$ is bounded by $a(n)<(\frac{3}{2})^n$ for $n\geq 1$, but \citet{Rosenberg2020} notes that this might be a relatively loose upper bound.}
    \item \textbf{(Number of) trees with minimal value on $\Tnstar$:} \textit{Open problem.}
    \item \props{(Number of) trees with minimal value on $\BTnstar$}{(\citet[Corollary 7]{Rosenberg2020})}{For any given $n\in\mathbb{N}_{\geq 1}$, there is exactly one binary tree $T\in\BTnstar$ with minimal Colijn-Plazzotta rank, i.e. $CP(T)=a(n)$, namely the maximally balanced tree $\Tmb$. In particular, if $n$ is a power of two, $\Tfb$ is the unique minimal tree.}
    \item \textbf{Expected value under the Yule model:} \textit{Open problem.}
    \item \textbf{Variance under the Yule model:} \textit{Open problem.}
    \item \textbf{Expected value under the uniform model:} \textit{Open problem.}
    \item \textbf{Variance under the uniform model:} \textit{Open problem.}
    \item \textbf{Comments:} The inverse function of $CP$ for binary trees is as follows: $CP^{-1}(1)$ is the tree that consists of only one leaf. For $x\geq 2$ we have $CP^{-1}(x)$ is the tree whose left pending subtree has the $CP$-rank $CP(T_1)=\left\lceil\frac{1+\sqrt{8x-7}}{2}\right\rceil-1$ and whose right pending subtree has the $CP$-rank $CP(T_2)=x-\frac{CP(T_1)\cdot (CP(T_1)-1)}{2}-1$ \citep[Corollary 3]{Rosenberg2020}. \\
    Both Furnas' LLR-ordering (Definition \ref{def_LLR_ordering}) and the Colijn-Plazzotta labeling for binary trees are bijective maps from the set of binary trees to the positive integers. But while the former lists all trees in $\BTnstar$ before moving on to $\mathcal{BT}^\ast_{n+1}$, the latter can assign trees with the same leaf number quite distant ranks or assign trees with quite different leaf numbers adjacent ranks \citep{Rosenberg2020}. Also, the left-right notion of a tree generally differs in both approaches \citep{Rosenberg2020}.\\
    In addition to the herein presented $CP$-rank, \citet{Rosenberg2020} has suggested two more methods of how the Colijn-Plazzotta ranking could be used to measure the balance of a binary tree. These are basically normalized versions of the $CP$-rank and can be found in the discussion section of his paper \citep{Rosenberg2020}. For information on the Colijn-Plazzotta ranking for trees with inner vertices of out-degree 1, see \citep{Colijn2018}.
\end{description}

\subsection{Colless index} \label{factsheet_Colless}

This index is only defined for binary rooted trees. It is an imbalance index, i.e. for a fixed $n\in\mathbb{N}_{\geq 1}$ it increases with decreasing balance of the tree. As the maximal possible balance value of a vertex increases with the number of its descendant leaves, the Colless index gives more weight to vertices close to the root than to those closer to the leaves \citep{Kirkpatrick1993}. And since the possible range of values grows with $n$, it is only meaningful to compare the Colless index of two trees if they have the same number of leaves \citep{Coronado2020a}. It can be calculated using the function \texttt{collessI} from our \textsf{R} package \texttt{treebalance} specifying \texttt{\enquote{original}} as the desired method.

Note that despite the name of this index, \citet{Colless1982} originally suggested a normalized version of it (cf. corrected Colless index (Section \ref{factsheet_cColless})). This normalization was later omitted (see e.g. \citet{Shao1990} and \citet{Rogers1993}).

\begin{description}
    \item \props{Definition}{(\citet{Shao1990})}{The Colless index $C(T)$ of a binary tree $T\in\BTnstar$ is defined as \[ C(T) \coloneqq \sum\limits_{v\in\mathring{V}(T)} bal_T(v)=\sum\limits_{v\in\mathring{V}(T)} |n_{v_1}-n_{v_2}| \] with $v_1$ and $v_2$ denoting the children of $v$}.
    \item \props{Computation time}{(this manuscript, see Proposition \ref{runtime_Colless})}{For every binary tree $T\in\BTnstar$, the Colless index $C(T)$ can be computed in time $O(n)$.}
    \item \props{Recursiveness}{(\citet[p. 507]{Matsen2007}, \citet[p. 102]{Rogers1993})}{The Colless index is a binary recursive tree shape statistic. We have $C(T)=0$ if $T\in\mathcal{BT}_1^\ast$, and for every binary tree $T\in\BTnstar$ with $n\geq 2$ and standard decomposition $T=(T_1,T_2)$ we have \[ C(T)=C(T_1)+C(T_2)+|n_1-n_2|. \]}
    \item \props{Locality}{(this manuscript, see Proposition \ref{locality_Colless})}{The Colless index is local.}
    \item \props{Maximal value on $\BTnstar$}{(\citet[p. 1819]{Heard1992}; \citet[Lemma 1]{Mir2018})}{For every binary tree $T\in\BTnstar$, the Colless index fulfills \[ C(T)\leq \frac{(n-1)(n-2)}{2}. \] This bound is tight for all $n\in\mathbb{N}_{\geq 1}$.}
    \item \props{(Number of) trees with maximal value on $\BTnstar$}{(\citet[Lemma 1]{Mir2018})}{For any given $n\in\mathbb{N}_{\geq 1}$, there is exactly one binary tree $T\in\BTnstar$ with maximal Colless index, i.e. $C(T)=\frac{(n-1)(n-2)}{2}$, namely the caterpillar tree $\Tcat$.}
    \item \props{Minimal value on $\BTnstar$}{(\citet[Theorem 4]{Hamoudi2017}; \citet[Theorem 2 and 3]{Coronado2020a})}{Let $T\in\BTnstar$ be a binary tree. First, let $b_a b_{a-1} \ldots b_0$ denote the binary representation of $n$. Then, \[ C(T) \geq 2 \cdot (n \text{ mod } 2^a) + \sum\limits_{j=0}^{a-1} (-1)^{b_j} \cdot (n \text{ mod } 2^{j+1}). \]
    Second, consider the binary expansion of $n$, i.e. write $n=\sum\limits_{j=1}^\ell 2^{d_j}$ with $\ell\geq 1$ and $d_1,\ldots,d_\ell\in\mathbb{N}_{\geq 0}$ such that $d_1>\ldots>d_\ell$. Then, \[ C(T)\geq \sum\limits_{j=2}^\ell 2^{d_j}\cdot(d_1-d_j-2\cdot(j-2)). \]
    Third, let $s(x)$ denote the triangle wave, i.e. the distance from $x\in\mathbb{R}$ to its nearest integer. Then,  \[ C(T)\geq \sum\limits_{j=1}^{\lceil\log_2(n)\rceil-1} 2^j\cdot s(2^{-j}\cdot n). \] These three bounds are equivalent\footnote{The first bound was derived in \citep{Hamoudi2017} and the other two were independently derived in \cite{Coronado2020a}.} and tight for all $n\in\mathbb{N}_{\geq 1}$.}
    \item \props{Trees with minimal value on $\BTnstar$}{(\citet[Algorithm 1, Proposition 1, Proposition 3, Theorem 1, Proposition 6]{Coronado2020a})}{Proposition 1 and 3 in \citep{Coronado2020a} provide a full characterization of trees with minimum Colless index for any given $n \in \mathbb{N}_{\geq 1}$, and Algorithm 1 in \citep{Coronado2020a} computes precisely those trees. In particular, each maximally balanced tree $\Tmb$ and each greedy from the bottom tree $\Tgfb$ has minimal Colless index. Note that the last fact has independently been shown by \citet[Theorem 4]{Hamoudi2017}.}
    \item \props{Number of trees with minimal value on $\BTnstar$}{(\citet[Proposition 4, Corollary 7]{Coronado2020a})}{\textit{Open problem.} Let $c_n$ denote the minimal Colless index for a given $n\in\mathbb{N}_{\geq 1}$, let $A(n)$ denote the set of pairs $A(n)=\{(n_a,n_b)|n_a,n_b\in\mathbb{N}, n_a>n_b\geq 1, n_a+n_b=n, c_{n_a}+c_{n_b}+n_a-n_b=c_n\}$, and let $\Tilde{c}(n)$ denote the number of binary trees with $n$ leaves that have minimal Colless index. Then, $\Tilde{c}(n)$ fulfills the recursion $\Tilde{c}(1)=1$ and \[ \Tilde{c}(n)=\sum\limits_{(n_a,n_b)\in A(n)} \Tilde{c}(n_a)\cdot\Tilde{c}(n_b)+\binom{\Tilde{c}(\frac{n}{2})+1}{2}\cdot \mathcal{I}(n\text{ mod }2=0). \] A closed formula is not known yet. In particular, if $n\in\{2^m-1,2^m,2^m+1\}$ for some $m\in\mathbb{N}_{\geq 1}$, there is exactly one minimal Colless tree in $\BTnstar$ (and if $n$ is a power of two this is precisely $\Tfb$). For all other $n$, there exist at least two trees in $\BTnstar$ with minimal Colless index.}
    \item \props{Expected value under the Yule model}{(\citet[p. 1820]{Heard1992}\footnote{Note that \citet{Heard1992} established the expected value of the \emph{normalized} Colless index under the Yule model, where $C(T)$ is divided by $(n-1)(n-2)/2$. By linearity of the expectation, above expression for the Colless index immediately follows.}; \citet[Theorem 2]{Blum2006a})}{Let $T_n$ be a phylogenetic tree with $n$ leaves sampled under the Yule model. Then, the expected value of $C$ of $T_n$ is \[ E_Y(C(T_n))= (n \text{ mod } 2) + n\cdot (H_{\lfloor \frac{n}{2} \rfloor}-1). \] Moreover, in the limit \[ E_Y(C(T_n))\sim n\cdot\ln(n)+(\gamma-1-\ln(2))\cdot n \] with $\gamma$ denoting Euler's constant.}
    \item \props{Variance under the Yule model}{(\citet[Corollary 6, Corollary 7]{Cardona2012})}{Let $T_n$ be a phylogenetic tree with $n$ leaves sampled under the Yule model. Then, the variance of $C$ of $T_n$ is
    \begin{equation*}
    \begin{split}
        V_Y(C(T_n)) &= \frac{5n^2+7n}{2}+(6n+1)\cdot\left\lfloor\frac{n}{2}\right\rfloor-4\left\lfloor\frac{n}{2}\right\rfloor^2+8\left\lfloor\frac{n+2}{4}\right\rfloor^2 -8(n+1)\cdot\left\lfloor\frac{n+2}{4}\right\rfloor-6n\cdot H_n\\
        &\quad +\left(2\cdot\left\lfloor\frac{n}{2}\right\rfloor-n(n-3)\right)\cdot H_{\left\lfloor\frac{n}{2}\right\rfloor}-n^2\cdot H_{\left\lfloor\frac{n}{2}\right\rfloor}^{(2)}+\left(n^2+3n-2\left\lfloor\frac{n}{2}\right\rfloor\right)\cdot H_{\left\lfloor\frac{n+2}{4}\right\rfloor}-2n\cdot H_{\left\lfloor\frac{n}{4}\right\rfloor}.
    \end{split}
    \end{equation*}
    Moreover, in the limit\footnote{This formula is a refinement of the limit $V_Y(C(T_n)) \sim \left(3-\frac{\pi^2}{6}-\ln(2)\right)\cdot n^2$, which was derived in \citep[Theorem 2]{Blum2006a}.}
    \begin{equation*}
    \begin{split}
        V_Y(C(T_n)) &\sim -\frac{8}{3}(-18+\pi^2+\ln(64))\cdot\left\lfloor\frac{n}{4}\right\rfloor^2 -8\left\lfloor\frac{n}{4}\right\rfloor\cdot\ln\left(\left\lfloor\frac{n}{4}\right\rfloor\right)\\
        &\quad + \left(20-8\gamma-32\ln(2)+\left(24-\frac{4}{3}\pi^2-8\ln(2)\right)(n \text{ mod } 4)\right)\cdot \left\lfloor\frac{n}{4}\right\rfloor 
    \end{split}
    \end{equation*} 
    with $\gamma$ denoting Euler's constant.\\
    In addition to these formulas, recursions for the moments of the Colless index under the Yule model are given in \citep{Rogers1994}.}
    \item \props{Expected value under the uniform model}{(\citet[p. 2029]{Rogers1994}; \citet[Theorem 4]{Blum2006a})}{\textit{Open problem.} Let $T_n$ be a phylogenetic tree with $n$ leaves sampled under the uniform model. Then, the expected value of $C$ of $T_n$ fulfills the recursion \[ E_U(C(T_n))=\frac{n!}{2\cdot (2n-3)!!}\cdot \sum\limits_{i=1}^{n-1} \frac{(2i-3)!!\cdot (2n-2i-3)!!}{i!\cdot(n-i)!}\cdot (2\cdot E_U(C(T_i))+|n-2i|). \] A closed formula is not known yet. Moreover, in the limit \[ E_U(C(T_n)) \sim \sqrt{\pi}\cdot n^{3/2}. \]}
    \item \props{Variance under the uniform model}{(\citet[pp. 2028-2029]{Rogers1994}; \citet[Theorem 4]{Blum2006a})}{\textit{Open problem.} Let $T_n$ be a phylogenetic tree with $n$ leaves sampled under the uniform model. Then, the variance of $C$ of $T_n$ fulfills the recursion
    \begin{equation*}
    \begin{split}
        V_U(C(T_n))&=E_U(C(T_n)^2)-E_U(C(T_n))^2\\
        &=\frac{n!}{2\cdot (2n-3)!!}\cdot \sum\limits_{i=1}^{n-1} \frac{(2i-3)!!\cdot (2n-2i-3)!!}{i!\cdot(n-i)!} \cdot \Big(2\cdot E_U(C(T_i)^2)+2\cdot E_U(C(T_i))\cdot E_U(C(T_{n-i}))\\
        &\qquad\qquad\qquad\qquad\qquad\qquad\qquad\qquad\qquad\qquad\qquad +4|n-2i|\cdot E_U(C(T_i))+|n-2i|^2\Big)\\
        &\qquad -\left(\frac{n!}{2\cdot (2n-3)!!}\cdot \sum\limits_{i=1}^{n-1} \frac{(2i-3)!!\cdot (2n-2i-3)!!}{i!\cdot(n-i)!} \cdot\Big(2\cdot E_U(C(T_i))+|n-2i|\Big)\right)^2.
    \end{split}
    \end{equation*}
    A closed formula is not known yet. Moreover, in the limit \[ V_U(C(T_n)) \sim \frac{10-3\pi}{3}\cdot n^3. \]}
    \item \textbf{Comments:} The Colless index is the most widely used balance index in phylogenetics. \citet{Bartoszek2021} attribute its popularity to several factors: it is one of the oldest balance indices, its intuitiveness of measuring \enquote{global imbalance} by adding up \enquote{local imbalances} of the vertices, its power in goodness-of-fit tests of probabilistic models and its power in discriminating tree shapes.\\
    Several methods have been suggested to make the Colless index applicable for arbitrary (i.e. not necessarily binary) trees, e.g. by ignoring multifurcating vertices \citep{Shao1990}. Moreover, \citet{Mir2018} introduced a family of Colless-like indices (Section\ref{factsheet_Colless-like}) for arbitrary trees.
\end{description}

\subsection{Colless-like indices} \label{factsheet_Colless-like}

The family of Colless-like indices introduced by \citet{Mir2018} is a generalization of the Colless index (\citep{Colless1982}, Section \ref{factsheet_Colless}) to rooted trees that are not necessarily binary. Each Colless-like index $\mathfrak{C}_{D,f}$ is parametrized by a \emph{dissimilarity} $D$ and a \emph{weight function} $f: \mathbb{N}_{\geq 0} \rightarrow \mathbb{R}_{\geq 0}$ (formal details below). Moreover, \citet{Mir2018} call a balance index \emph{sound} when the most balanced trees according to it are precisely those trees $T \in \Tnstar$ that have the property that for every inner vertex $v$ of $T$, the subtrees of $T$ rooted at the children of $v$ have the same shape. It turns out that not all Colless-like indices are sound in this sense, but \citet{Mir2018} discuss weight functions $f$, for which $\mathfrak{C}_{D,f}$ is sound for every dissimilarity $D$. The Colless-like index can be calculated for any given dissimilarity and weight function using the function \texttt{collesslikeI} from our \textsf{R} package \texttt{treebalance}.

\begin{description}
    \item \props{Definition}{(\citet{Mir2018}; in particular Definition 4 therein)}{Let $f: \mathbb{N}_{\geq 0} \rightarrow \mathbb{R}_{\geq 0}$ be a function that maps any natural number to a non-negative real number. Let $T \in \Tnstar$ be a rooted tree. Then, the \emph{$f$-size} of $T$ is defined as  \[ \Delta_f(T) \coloneqq \sum\limits_{v \in V(T)} f(deg^+(v)). \] In words, the $f$-size of $T$ is a weighted sum of the out-degrees of all vertices of $T$, where the out-degree of each vertex is weighted by means of the function $f$.\\
    Furthermore, let $\mathbb{R}^+ \coloneqq \{(x_1, \ldots, x_k) \, | \, k \geq 1, x_1, \ldots, x_k \in \mathbb{R}\}$ be the set of all non-empty finite-length sequences of real numbers. A \emph{dissimilarity} on $\mathbb{R}^+$ is any mapping $D: \mathbb{R}^+ \rightarrow \mathbb{R}_{\geq 0}$ satisfying the following conditions: For every $(x_1, \ldots, x_k) \in \mathbb{R}^+$, 
    \begin{itemize}
        \item $D(x_1, \ldots, x_k) = D(x_{\sigma(1)}, \ldots, x_{\sigma(k)})$ for every permutation $\sigma$ of $\{1, \ldots, k\}$.
        \item $D(x_1, \ldots, x_k)=0$ if and only if $x_1 = \ldots = x_k$.
    \end{itemize}
    Dissimilarities considered by \citet{Mir2018} are:
    \begin{itemize}
        \item the \emph{mean deviation from the median $\widetilde{x}$ of $(x_1,\ldots,x_k)$}
        $$ \textup{MDM}(x_1, \ldots, x_k) = \frac{1}{k}\cdot\sum\limits_{i=1}^k |x_i - \widetilde{x}|,$$
        \item the \emph{(sample) variance (where $\overline{x}$ denotes the mean of $(x_1,\ldots,x_k)$)}
        $$ \textup{var}(x_1, \ldots, x_k) = \frac{1}{k-1}\cdot\sum\limits_{i=1}^k (x_i - \overline{x})^2,$$
        \item and the \emph{(sample) standard deviation}
        $$ sd(x_1, \ldots, x_k) =  \sqrt{\textup{var}(x_1, \ldots, x_k)}.$$
    \end{itemize}
    Now, let $D$ be a dissimilarity on $\mathbb{R}^+$, $f: \mathbb{N}_{\geq 0} \rightarrow \mathbb{R}_{\geq 0}$ a function, $\delta_f$ the corresponding $f$-size and let $T \in \Tnstar$ be a rooted tree. Then, for every inner vertex $v$ of $T$ with children $v_1, \ldots, v_k$, the \emph{$(D,f)$-balance value} of $v$ is defined as \[ bal_{D,f}(v) \coloneqq D(\delta_f(T_{v_1}), \ldots, \delta_f(T_{v_k})). \]
    Finally, let $D$ be a dissimilarity on $\mathbb{R}^+$ and let $f: \mathbb{N}_{\geq 0} \rightarrow \mathbb{R}_{\geq 0}$ be a function. For every $T \in \Tnstar$, its \emph{Colless-like index relative to $D$ and $f$}, denoted by $\mathfrak{C}_{D,f}(T)$, is defined as the sum of the $(D,f)$-balance values of the inner vertices of $T$, i.e. \[ \mathfrak{C}_{D,f}(T) \coloneqq \sum\limits_{v \in \mathring{V}(T)} bal_{D,f}(v). \]
    In the following, a rooted tree $T \in \Tnstar$ is called \emph{fully symmetric} if and only if it has the property that for every inner vertex $v$ of $T$, the subtrees of $T$ rooted at the children of $v$ are isomorphic. Note that both the rooted star tree $\Tstar$ and the fully balanced tree $\Tfb$ are fully symmetric in this sense, but there may be other trees with this property. A Colless-like index $\mathfrak{C}_{D,f}$ is called \emph{sound} if for every $T \in \Tnstar$ we have $\mathfrak{C}_{D,f}(T)=0$ if and only if $T$ is fully symmetric \citep[Definition 9]{Mir2018}. A problem posed in \citep{Mir2018} is to find functions $f: \mathbb{N}_{\geq 0} \rightarrow \mathbb{R}_{\geq 0}$ such that $\mathfrak{C}_{D,f}$ is sound (for all dissimilarities $D$). The authors show that two such choices for $f$ are $f(n)=e^n$ \citep[Proposition 16]{Mir2018} and $f(n)=\ln(n+e)$ \citep[Proposition 17]{Mir2018}.}
    \item \props{Computation time}{(\citet[Proposition 8]{Mir2018})}{If the cost of computing $D(x_1, \ldots, x_k)$ is in $O(k)$ and the cost of computing each $f(k)$ is at most in $O(k)$, then for every tree $T \in \Tnstar$, the Colless-like index $\mathfrak{C}_{D,f}(T)$ relative to $D$ and $f$ can be computed in time $O(n)$.}
    \item \props{Recursiveness}{(\citet[p. 9]{Mir2018})}{Colless-like indices are recursive tree shape statistics. We have $\mathfrak{C}_{D,f}(T)=0$ for $T\in\mathcal{T}_1^\ast$, and for every $T \in \Tnstar$ with $n \geq 2$ and standard decomposition $T=(T_1, \ldots, T_k)$, we have \[ \mathfrak{C}_{D,f}(T) = \mathfrak{C}_{D,f}(T_1) + \ldots + \mathfrak{C}_{D,f}(T_k) + D(\delta_f(T_1), \ldots, \delta_f(T_k)). \]}
    \item \textbf{Locality:} \textit{Open problem.}
    \item \props{Maximal value and (number) of trees with maximal value on $\Tnstar$}{(\citet[Theorem 18 and 19]{Mir2018})}{The maximum value of $\mathfrak{C}_{D,f}$ on $\Tnstar$ clearly depends on the choices of $f$ and $D$. \citet{Mir2018} obtain the following two results:
    \begin{enumerate}
        \item Let $f:\mathbb{N}_{\geq 0} \rightarrow \mathbb{R}_{\geq 0}$ be a function such that $0 < f(k) < f(k-1)+f(2)$ for every $k \geq 3$ (note that $f(n)=\ln(n+e)$ satisfies these assumptions). Then, for every $n \geq 2$, the indices $\mathfrak{C}_{\textup{MDM},f}$, $\mathfrak{C}_{sd,f}$, and $\mathfrak{C}_{\textup{var},f}$ reach their maximum values on $\Tnstar$ precisely at the caterpillar tree $\Tcat$ (i.e. there is precisely one tree reaching the maximum). These maximum values are, respectively, 
        \begin{align*}
            \mathfrak{C}_{\textup{MDM},f}(\Tcat) &= \frac{f(0)+f(2)}{4} \cdot (n-1)(n-2) \\
            \mathfrak{C}_{sd,f}(\Tcat) &= \frac{f(0)+f(2)}{2\sqrt{2}} \cdot (n-1)(n-2) \\
            \mathfrak{C}_{\textup{var},f}(\Tcat) &= \frac{(f(0)+f(2))^2}{12} \cdot (n-1)(n-2)(2n-3).
        \end{align*}
        \item Let $\widehat{T}_n=(T_1^\mathit{star}, T_{n-1}^\mathit{star})$ be the unique tree on $n$ leaves whose two maximal pending subtrees are a rooted star tree on one leaf and a rooted star tree on $n-1$ leaves, respectively. Now, let $f(n)=e^n$. Then, for every $n \geq 2$:
            \begin{enumerate}
                \item If $n \neq 4$, then both $\mathfrak{C}_{\textup{MDM},f}$ and $\mathfrak{C}_{sd,f}$ reach their maximum value on $\Tnstar$ precisely at the tree $\widehat{T}_n$, and these maximum values are
                \begin{align*}
                    \mathfrak{C}_{\textup{MDM},f}(\widehat{T}_n) &= \frac{1}{2} (e^{n-1} +n-2) \\
                    \mathfrak{C}_{sd,f}(\widehat{T}_n) &= \frac{1}{\sqrt{2}}(e^{n-1} +n-2) .
                \end{align*}
            \item If $n=4$, then both $\mathfrak{C}_{\textup{MDM},f}$ and $\mathfrak{C}_{sd,f}$ reach their maximum value on $\mathcal{T}_4^\ast$ precisely at the caterpillar tree $T_4^{cat}$ and these maximum values are
                \begin{align*}
                    \mathfrak{C}_{\textup{MDM},f}(T_4^{cat}) &= \frac{3}{2} (e^2+1) \\
                     \mathfrak{C}_{sd,f}(T_4^{cat}) &= \frac{3}{\sqrt{2}} (e^2+1). 
                \end{align*}
            \item $\mathfrak{C}_{\textup{var},f}$ always reaches its maximum value on $\Tnstar$ precisely at the tree $\widehat{T}_n$, and the maximum value is
            \begin{align*}
                \mathfrak{C}_{\textup{var},f}(\widehat{T}_n) &= \frac{1}{2} (e^{n-1}+n-2)^2.
            \end{align*}
            \end{enumerate}
            Note that in all cases, there is precisely one maximal tree.
    \end{enumerate}}
    \item \textbf{Minimal value and (number) of trees with minimal value on $\Tnstar$:} By definition, $\mathfrak{C}_{D,f}$ is non-negative, but the minimum value clearly depends on the choice of $D$ and $f$.
    If $f$ is chosen such that $\mathfrak{C}_{D,f}$ is sound, then (by definition of soundness) for every $n \in \mathbb{N}_{\geq 1}$ and for every $T \in \Tnstar$, $\mathfrak{C}_{D,f}(T)=0$ if and only if $T$ is fully symmetric. Note that the number of fully symmetric trees on $n$ leaves equals the number of ordered factorizations of $n$ (sequence A074206 in Sloane's \emph{On-Line Encyclopedia of Integer Sequences} \citep{OEIS}). In particular, the minimal trees are not necessarily unique (for more details see \citet[p. 5]{Mir2018}).
    \item \props{Extremal values and extremal trees on $\BTnstar$}{(Consequence of Propositions 6 and 7 in \citet{Mir2018})}{The minimum and maximum value of $\mathfrak{C}_{D,f}$ restricted to binary trees clearly depend on the choice of $D$ and $f$. However, for $D=\textup{MDM}$ and $D=sd$, $\mathfrak{C}_{D,f}(T)$ is proportional to the Colless index $C(T)$ (Section \ref{factsheet_Colless}), and for $D=\textup{var}$, $\mathfrak{C}_{D,f}(T)$ is proportional to the quadratic Colless index $QC(T)$ (Section \ref{factsheet_qColless}). More precisely,
   \begin{itemize}
       \item Let $T \in \BTnstar$ and let $f: \mathbb{N} \rightarrow \mathbb{R}_{\geq 0}$ be any function. Then, 
       \begin{align*}
           \mathfrak{C}_{\textup{MDM},f}(T) = \frac{f(0)+f(2)}{2} \cdot C(T) \quad \text{and} \quad
           \mathfrak{C}_{sd,f}(T) = \frac{f(0)+f(2)}{\sqrt{2}} \cdot C(T)
       \end{align*}
       (Proposition 6 in \citet{Mir2018}).
       \item  Let $T \in \BTnstar$ and let $f: \mathbb{N} \rightarrow \mathbb{R}_{\geq 0}$ be any function. Then, 
       \begin{align*}
           \mathfrak{C}_{\textup{var},f}(T) = \frac{(f(0)+f(2))^2}{2} \cdot QC(T)
       \end{align*}
       (Proposition 7 in \citet{Mir2018}).
   \end{itemize}
   As the extremal values and extremal trees for the Colless index and the quadratic Colless index are fully characterized (see pages \pageref{factsheet_Colless} and \pageref{factsheet_qColless}), the same is true for $\mathfrak{C}_{D,f}$ restricted to binary trees, and $D \in \{\textup{MDM}, sd, \textup{var}\}$.}
   In particular (assuming $f(0)+f(2) \neq 0$):
   \begin{itemize}
       \item For any given $n \in \mathbb{N}_{\geq 1}$, $\mathfrak{C}_{D,f}$ (with $D \in \{\textup{MDM}, sd, \textup{var}\}$) reaches its maximum value on $\BTnstar$ precisely at the caterpillar tree $\Tcat$.
       \item For any given $n \in \mathbb{N}_{\geq 1}$, $\mathfrak{C}_{\textup{var,f}}$ reaches its minimum value on $\BTnstar$ precisely at the maximally balanced tree $\Tmb$.
       \item If $n \in \{2^{m-1}, 2^m, 2^{m+1}\}$ for some $m \in \mathbb{N}_{\geq 1}$, $\mathfrak{C}_{\textup{MDM},f}$ and $\mathfrak{C}_{sd,f}$ reach their minimum value on $\BTnstar$ precisely at the maximally balanced tree $\Tmb$. In all other cases, there are at least two trees in $\BTnstar$ with minimum $\mathfrak{C}_{\textup{MDM},d}$ and $\mathfrak{C}_{sd,f}$ index, respectively.
   \end{itemize}
    \item \textbf{Expected value under the Yule model:} \textit{Open problem.}
    \item \textbf{Variance under the Yule model:} \textit{Open problem.}
    \item \textbf{Expected value under the uniform model:} \textit{Open problem.}
    \item \textbf{Variance under the uniform model:} \textit{Open problem.}
    \item \textbf{Comments:} While \citet{Mir2018} prove that both $f(n)=\ln(n+e)$ and $f(n)=e^n$ yield sound Colless-like indices, they recommend using $f(n)=\ln(n+e)$. On the one hand, for this choice of $f$, the caterpillar tree is always the unique most imbalanced tree (in line with our definition of an (im)balance index). On the other hand, they report numerical difficulties when using $f(n)=e^n$. However, it remains an open problem to find further functions $f$ such that $\mathfrak{C}_{D,f}$ is sound. An interesting conjecture posed by \citet{Mir2018} states that there is no function $f: \mathbb{N} \rightarrow \mathbb{N}$ taking values in the set of natural numbers that yields a sound Colless-like index.\\
    Concerning the dissimilarity $D$, as mentioned above, \citet{Mir2018} show that MDM and $sd$ define indices that are proportional to the Colless index when applied to rooted binary trees (Proposition 6 in \citet{Mir2018}). Among these two options, they recommend using MDM, as it only involves linear operations and has less numerical precision problems than $sd$ (which uses a square root of a sum of squares).
\end{description}

\subsection{Corrected Colless index} \label{factsheet_cColless}

This index is only defined for rooted binary trees. It is an imbalance index, i.e. for a fixed $n\in\mathbb{N}_{\geq 1}$ it increases with decreasing balance of the tree. Just like with the Colless index, as the maximal possible balance value of a vertex increases with the number of its descendant leaves, the corrected Colless index gives more weight to vertices close to the root than to those closer to the leaves \citep{Kirkpatrick1993}. It can be calculated using the function \texttt{collessI} from our \textsf{R} package \texttt{treebalance} specifying \texttt{\enquote{corrected}} as the desired method.

\begin{description}
    \item \props{Definition}{(\citet{Heard1992})}{The corrected Colless index $I_C(T)$ of a binary tree $T\in\BTnstar$ is defined as \[ I_C(T) \coloneqq \frac{2\cdot C(T)}{(n-1)(n-2)} = \frac{2}{(n-1)(n-2)}\cdot\sum\limits_{v\in\mathring{V}(T)} bal_T(v) = \frac{2}{(n-1)(n-2)}\cdot\sum\limits_{v\in\mathring{V}(T)} |n_{v_1}-v_{v_2}|, \] in which $v_1$ and $v_2$ denote the children of $v$.}
    \item \props{Computation time}{(this manuscript, see Proposition \ref{runtime_corColless})}{For every binary tree $T\in\BTnstar$, the corrected Colless index $I_C(T)$ can be computed in time $O(n)$.}
    \item \props{Recursiveness}{(this manuscript, see Proposition \ref{recursiveness_corColless})}{The corrected Colless index is a binary recursive tree shape statistic. We have $I_C(T)=0$ if $T\in\mathcal{BT}_1^\ast$, and for every binary tree $T\in\BTnstar$ with $n \geq 2$ and standard decomposition $T=(T_1,T_2)$ we have \[ I_C(T)=\frac{(n_1-1)(n_1-2)\cdot I_C(T_1)}{(n_1+n_2-1)(n_1+n_2-2)}+\frac{(n_2-1)(n_2-2)\cdot I_C(T_2)}{(n_1+n_2-1)(n_1+n_2-2)}+\frac{2\cdot|n_1-n_2|}{(n_1+n_2-1)(n_1+n_2-2)}. \]}
    \item \props{Locality}{(this manuscript, see Proposition \ref{locality_corColless})}{The corrected Colless index is not local.}
    \item \props{Maximal value on $\BTnstar$}{(this manuscript, see Theorem \ref{max_corColless})}{For every binary tree $T\in\BTnstar$, the corrected Colless index fulfills $I_C(T)\leq 1$. This bound is tight for all $n\in\mathbb{N}_{\geq 3}$. For $n \in \{1,2\}$, we have $I_C(T)=0$.}
    \item \props{(Number of) trees with maximal value on $\BTnstar$}{(this manuscript, see Theorem \ref{max_corColless})}{For any given $n\in\mathbb{N}_{\geq 1}$, there is exactly one binary tree $T\in\BTnstar$ with maximal corrected Colless index, i.e. $I_C(T)=0$ if $n\in\{1,2\}$ and $I_C(T)=1$ if $n\geq 3$, namely the caterpillar tree $\Tcat$.}
    \item \props{Minimal value on $\BTnstar$}{(this manuscript, see Proposition \ref{min_corColless})}{Let $T\in\BTnstar$ be a binary tree. First, let $b_a b_{a-1} \ldots b_0$ denote the binary representation of $n$. Then, \[ I_C(T) \geq \frac{2}{(n-1)(n-2)}\cdot\left(2 \cdot (n \text{ mod } 2^a) + \sum\limits_{j=0}^{a-1} (-1)^{b_j} \cdot (n \text{ mod } 2^{j+1})\right). \]
    Second, write $n=\sum\limits_{j=1}^\ell 2^{d_j}$ with $\ell\geq 1$ and $d_1,\ldots,d_\ell\in\mathbb{N}_{\geq 0}$ such that $d_1>\ldots>d_\ell$. Then, \[ I_C(T)\geq \frac{2}{(n-1)(n-2)}\cdot\left(\sum\limits_{j=2}^\ell 2^{d_j}\cdot(d_1-d_j-2\cdot(j-2))\right). \]
    Third, let $s(x)$ denote the triangle wave, i.e. the distance from $x\in\mathbb{R}$ to its nearest integer. Then,  \[ I_C(T)\geq \frac{2}{(n-1)(n-2)}\cdot\left(\sum\limits_{j=1}^{\lceil\log_2(n)\rceil-1} 2^j\cdot s(2^{-j}\cdot n)\right). \] These three bounds are equivalent and tight for all $n\in\mathbb{N}_{\geq 1}$.}
    \item \props{Trees with minimal value on $\BTnstar$}{(this manuscript, see Theorem \ref{mintrees_corColless})}{Proposition 1 and 3 in \citep{Coronado2020a} provide a fully characterization of trees with minimum Colless index (and thus also with minimum corrected Colless index) for any given $n\in\mathbb{N}_{\geq 1}$, and Algorithm 1 in \citep{Coronado2020a} computes precisely those trees. In particular, each maximally balanced tree $\Tmb$ and each greedy from the bottom tree $\Tgfb$ has minimal corrected Colless index.}
    \item \props{Number of trees with minimal value on $\BTnstar$}{(this manuscript, see Proposition \ref{nummintrees_corColless})}{\textit{Open problem.} Let $d(n)$ denote the minimal corrected Colless index for a given $n\in\mathbb{N}_{\geq 1}$, let $B(n)$ denote the set of pairs $B(n)=\{(n_a,n_b)|n_a,n_b\in\mathbb{N},n_a>n_b\geq 1, n_a+n_b=n, \frac{(n_a-1)(n_a-2)d(n_a)+(n_b-1)(n_b-2)d(n_b)+2(n_a-n_b)}{(n_a+n_b-1)(n_a+n_b-2)}=d(n)\}$, and let $\widetilde{d}(n)$ denote the number of binary trees with $n$ leaves that have minimal corrected Colless index. Then, $\widetilde{d}(n)$ fulfills the recursion $\widetilde{d}(1)=1$ and \[ \widetilde{d}(n)=\sum\limits_{(n_a,n_b)\in B(n)} \widetilde{d}(n_a)\cdot\widetilde{d}(n_b)+\binom{\widetilde{d}(\frac{n}{2})+1}{2}\cdot \mathcal{I}(n\text{ mod }2=0). \] A closed formula is not known yet. If $n\in\{2^m-1,2^m,2^m+1\}$ for some $m\in\mathbb{N}_{\geq 1}$, there is exactly one tree in $\BTnstar$ with minimal corrected Colless index (and if $n$ is a power of two this is $\Tfb$). For all other $n$, there exist at least two trees in $\BTnstar$ that reach the minimum.}
    \item \props{Expected value under the Yule model}{(\citet{Heard1992}, this manuscript, see Proposition \ref{expY_corColless})}{Let $T_n$ be a phylogenetic tree with $n$ leaves sampled under the Yule model. Then, the expected value of $I_C$ of $T_n$ is \[ E_Y(I_C(T_n))=\begin{cases} \frac{2n}{(n-1)(n-2)}\cdot \left(H_{\left\lfloor\frac{n}{2}\right\rfloor}-1\right) & \text{if } n \text{ is even} \\ \frac{2n}{(n-1)(n-2)}\cdot\left(H_{\left\lfloor\frac{n}{2}\right\rfloor}-1+\frac{1}{n}\right) & \text{if } n \text{ is odd} \end{cases} = \frac{2(n\text{ mod }2)+2n\cdot(H_{\lfloor\frac{n}{2}\rfloor}-1)}{(n-1)(n-2)}. \] Moreover, in the limit $E_Y(I_C(T_n))\sim \frac{1}{n}\cdot \ln(n) \sim 0$.}
    \item \props{Variance under the Yule model}{(this manuscript, see Proposition \ref{varY_corColless})}{Let $T_n$ be a phylogenetic tree with $n$ leaves sampled under the Yule model. Then, the variance of $I_C$ of $T_n$ is
    \begin{equation*}
    \begin{split}
        &V_Y(I_C(T_n)) = \frac{4}{(n-1)^2(n-2)^2}\cdot\Bigg[ \frac{5n^2+7n}{2}+(6n+1)\cdot\left\lfloor\frac{n}{2}\right\rfloor-4\left\lfloor\frac{n}{2}\right\rfloor^2+8\left\lfloor\frac{n+2}{4}\right\rfloor^2 -8(n+1)\cdot\left\lfloor\frac{n+2}{4}\right\rfloor\\
        &\qquad -6n\cdot H_n+\left(2\cdot\left\lfloor\frac{n}{2}\right\rfloor-n(n-3)\right)\cdot H_{\left\lfloor\frac{n}{2}\right\rfloor}-n^2\cdot H_{\left\lfloor\frac{n}{2}\right\rfloor}^{(2)}+\left(n^2+3n-2\left\lfloor\frac{n}{2}\right\rfloor\right)\cdot H_{\left\lfloor\frac{n+2}{4}\right\rfloor}-2n\cdot H_{\left\lfloor\frac{n}{4}\right\rfloor}\Bigg].
    \end{split}
    \end{equation*}
    Moreover, in the limit 
    \begin{equation*}
    \begin{split}
        V_Y(I_C(T_n)) &\sim \frac{4}{(n-1)^2(n-2)^2} \cdot \Bigg[-\frac{8}{3}(-18+\pi^2+\ln(64))\cdot\left\lfloor\frac{n}{4}\right\rfloor^2 -8\left\lfloor\frac{n}{4}\right\rfloor\cdot\ln\left(\left\lfloor\frac{n}{4}\right\rfloor\right)\\
        &\quad + \left(20-8\gamma-32\ln(2)+\left(24-\frac{4}{3}\pi^2-8\ln(2)\right)(n \mod 4)\right)\cdot \left\lfloor\frac{n}{4}\right\rfloor \Bigg]
    \end{split}
    \end{equation*}
    with $\gamma$ denoting Euler's constant.}
    \item \props{Expected value under the uniform model}{(this manuscript, see Proposition \ref{exU_corColless})}{\textit{Open problem.} Let $T_n$ be a phylogenetic tree with $n$ leaves sampled under the uniform model. Then, the expected value of $I_C$ of $T_n$ fulfills the recursion \[ E_U(I_C(T_n))=\frac{n\cdot(n-3)!}{(2n-3)!!}\cdot \sum\limits_{i=1}^{n-1} \frac{(2i-3)!!\cdot (2n-2i-3)!!}{i!\cdot(n-i)!}\cdot ((i-1)(i-2)\cdot E_U(I_C(T_i))+|n-2i|). \] A closed formula is not known yet. Moreover, in the limit $E_U(I_C(T_n)) \sim \frac{2\pi}{\sqrt{n}} \sim 0 $.}
    \item \props{Variance under the uniform model}{(this manuscript, see Proposition \ref{varU_corColless})}{\textit{Open problem.} Let $T_n$ be a phylogenetic tree with $n$ leaves sampled under the uniform model. Then, the variance of $I_C$ of $T_n$ fulfills the recursion
    \begin{equation*}
    \begin{split}
        V_U(I_C(T_n)) &= \frac{2n(n-3)!}{(n-1)(n-2)(2n-3)!!}\cdot \sum\limits_{i=1}^{n-1} \frac{(2i-3)!!\cdot (2n-2i-3)!!}{i!\cdot(n-i)!} \cdot \bigg(\frac{(i-1)^2(i-2)^2}{2}E_U(I_C(T_i)^2)\\
        &\qquad + (i-1)(i-2)E_U(I_C(T_i))\cdot\frac{(n-i-1)(n-i-2)}{2}E_U(I_C(T_{n-i}))\\
        &\qquad + 2|n-2i|\cdot (i-1)(i-2)E_U(I_C(T_i))+|n-2i|^2\bigg)\\
        &\qquad -\left(\frac{n(n-3)!}{(2n-3)!!}\right)^2 \cdot \left(\sum\limits_{i=1}^{n-1} \frac{(2i-3)!!\cdot (2n-2i-3)!!}{i!\cdot(n-i)!} \cdot\Big((i-1)(i-2) E_U(I_C(T_i)) +|n-2i|\Big)\right)^2.
    \end{split}
    \end{equation*}
    A closed formula is not known yet. Moreover, in the limit \[ V_U(I_C(T_n)) \sim \frac{40-12\pi}{3n} \sim 0. \]}
    \item \textbf{Comments:} Although the sum of balance values over all inner vertices of a binary tree is commonly called Colless index, \citet{Colless1982} actually suggested to normalize this sum by dividing it by the score for complete asymmetry, which he erroneously stated as $\frac{n(n-3)+1}{2}$ and which was later corrected to $\frac{(n-1)(n-2)}{2}$ by \citet{Heard1992}. So, the corrected Colless index is actually closer to the original intention than the so-called Colless index.\\
    In addition note that \citet{Stam2002} used the difference between the corrected Colless index of a tree and the expected index under the Yule model for the same number of leaves, i.e. $\Delta I_C(T)=I_{C}(T)-E_Y(I_C(T_n))$ for a tree $T \in \Tnstar$, because the $\Delta I_C$ values were independent of tree size \citep{Stam2002}. $\Delta I_C$  thus simply represents a shifted $I_C$ index. Therefore, the combinatorial and stochastic properties can be easily obtained from the corrected Colless index. In particular, the extremal trees are exactly the same. However, strictly speaking we do not regard $\Delta I_C$ as an imbalance index as it also assigns negative values to certain trees (e.g. $\Delta I_C(T^{fb}_4)\approx-0.262$) and thus violates the 
non-negativity constraint of Definition \ref{def_imbalance}.  
\end{description}

\subsection{Equal weights Colless index / \texorpdfstring{$I_2$}{I\_2} index} \label{factsheet_ewColless}

The $I_2$ index is a version of the Colless index that weighs every inner vertex equally, whereas the normal Colless index gives more weight to vertices that are close to the root. More precisely, the balance value of each inner vertex $v$ with $n_v$ children is divided by its maximal possible balance value $n_v-2$. Thus, the imbalance of each vertex is roughly weighted by the inverse of the size of its pending subtree. Just like the normal Colless index, the $I_2$ index is only defined for binary trees, and it is an imbalance index, i.e. for a fixed $n\in\mathbb{N}_{\geq 1}$ it increases with decreasing balance of the tree. It can be calculated using the function \texttt{ewCollessI} from our \textsf{R} package \texttt{treebalance}.

\begin{description}
    \item \props{Definition}{(\citet{Mooers1997})}{The $I_2$ index $I_2(T)$ of a binary tree $T\in\BTnstar$ is defined as \[ I_2(T) \coloneqq \frac{1}{n-2}\cdot\sum\limits_{\substack{v \in \mathring{V}(T) \\ n_v > 2}} \frac{bal_T(v)}{n_v-2}. \]}
    \item \props{Computation time}{(this manuscript, see Proposition \ref{runtime_I2})}{For every binary tree $T\in\BTnstar$, the $I_2$ index $I_2(T)$ can be computed in time $O(n)$.}
    \item \props{Recursiveness}{(this manuscript, see Proposition \ref{recursiveness_I2})}{The $I_2$ index is a binary recursive tree shape statistic. We have $I_2(T)=0$ for $T\in\mathcal{BT}_1^\ast$, and for every binary tree $T\in\BTnstar$ with $n\geq 2$ and standard decomposition $T=(T_1,T_2)$ we have \[ I_2(T) = \frac{1}{n_1+n_2-2}\cdot\left( (n_1-2) \cdot I_2(T_1) + (n_2-2) \cdot I_2(T_2) + \frac{|n_1-n_2|}{n_1+n_2-2} \right). \]}
    \item \props{Locality}{(this manuscript, see Proposition \ref{locality_I2})}{The $I_2$ index is not local.}
    \item \props{Maximal value on $\BTnstar$}{(this manuscript, see Theorem \ref{max_I2})}{For every binary tree $T\in\BTnstar$, the $I_2$ index fulfills $I_2(T)\leq 1$. This bound is tight for all $n\in\mathbb{N}_{\geq 3}$. For $n \in \{1,2\}$, we have $I_2(T)=0$.}
    \item \props{(Number of) trees with maximal value on $\BTnstar$}{(this manuscript, see Theorem \ref{max_I2})}{For any given $n\in\mathbb{N}_{\geq 1}$, there is exactly one binary tree $T\in\BTnstar$ with maximal $I_2$ index, i.e. $I_2(T)=0$ if $n\in\{1,2\}$ and $I_2(T)=1$ if $n\geq 3$, namely the caterpillar tree $T_n^{cat}$.}
    \item \textbf{Minimal value on $\BTnstar$:} \textit{Open problem.} By definition, $I_2(T) \geq 0$ for every $n \in \mathbb{N}_{\geq 1}$ and for every $T \in \BTnstar$. However, this bound is tight if and only if $n$ is a power of two and $T$ is a fully balanced tree (see Proposition \ref{Prop_I2_FullyBalanced} in this manuscript). The minimum value of the $I_2$ index for arbitrary $n$ is (to our knowledge) not known in the literature.
    \item \textbf{(Number of) trees with minimal value on $\BTnstar$:} \textit{Open problem.} When $n$ is a power of two, i.e. $n=2^h$ for some $h \in \mathbb{N}_{\geq 0}$, there is precisely one rooted binary tree $T \in \BTnstar$ with minimum $I_2$ index, namely $T = \Tfb$ (see Proposition \ref{Prop_I2_FullyBalanced} in this manuscript). When $n$ is not a power of two, both a characterization of all trees in $\BTnstar$ that minimize the $I_2$ index as well as their number seems to be unknown in the literature.
    \item \textbf{Expected value under the Yule model:} \textit{Open problem.}
    \item \textbf{Variance under the Yule model:} \textit{Open problem.}
    \item \textbf{Expected value under the uniform model:} \textit{Open problem.}
    \item \textbf{Variance under the uniform model:} \textit{Open problem.}
    \item \textbf{Comments:} Seemingly, the name $I_2$ was initially chosen because it was second on the list of imbalance indices that were studied by \citet{Mooers1997} (similar to the names $B_1$ and $B_2$). We therefore suggest the term \enquote{equal weights Colless index} as a more descriptive name besides the term $I_2$ which, however, has already been used like this in the literature.\\
    Since the idea of this index is to divide the balance of each vertex by its maximum possible value, it is not surprising that the $I_2$ index is related to the $I$-based indices (Section \ref{factsheet_Ibased}). In fact, if $n_v$ is even (and $n_v\geq 4$), the summand $\frac{bal_T(v)}{n_v-2}$ of the former equals the $I_v$ value of the latter (Proposition \ref{prop_I2_Iv}).
\end{description}

\subsection{Furnas rank} \label{factsheet_Furnas}

In the original publication by \citet{Furnas1984}, the Furnas rank\footnote{The Furnas rank is sometimes also referred to as Furnas' $R$ statistic.} was introduced as a tool for sampling trees uniformly from the set of rooted binary unlabeled trees with a certain leaf number. \citet{Kirkpatrick1993} later suggested the Furnas rank as a measure of balance for rooted binary trees. As such it is a balance index, i.e. for a fixed $n\in\mathbb{N}_{\geq 1}$ it increases with increasing balance of the tree. Moreover, due to its original objective, it has the property that two trees with the same leaf number have the same index if and only if they are identical (stated in \cite{Furnas1984} and proven in Proposition \ref{Prop_Furnas_TotalOrdering} of this manuscript). This implies that the Furnas rank has the maximal possible resolution that a balance index can have. The function \texttt{furnasI} for the calculation of the Furnas rank as well as \texttt{furnasI\_inv}, its inverse function, can be found in our \textsf{R} package \texttt{treebalance}.

\begin{description}
    \item \props{Definition}{(\citet{Furnas1984, Kirkpatrick1993}, this manuscript, see Theorem \ref{th_Furnas_rank})}{The Furnas rank (Furnas balance index) $F(T)$ of a binary tree $T\in\BTnstar$ is defined as $F(T)=r_n(T)$ with $r_n(T)$ denoting the rank of $T$ in the LLR ordering (Definition \ref{def_LLR_ordering}) of all trees with the same number of leaves $n$. The rank $r_n(T)$ can be recursively computed as $r_1(T)=1$ if $n=1$ and otherwise 
    \begin{equation*}
        r_n(T) = 
        \begin{cases} 
        \displaystyle \sum\limits_{i=1}^{\alpha-1} we(i)\cdot we(n-i) +    (r_\alpha(T_L)-1)\cdot we(\beta) + r_\beta(T_R) & \text{if }\alpha<\beta\\
        \displaystyle \sum\limits_{i=1}^{\alpha-1} we(i)\cdot we(n-i) + (r_\alpha(T_L)-1)\cdot we(\beta)-\frac{r_\alpha(T_L)^2-r_\alpha(T_L)}{2} + r_\beta(T_R) & \text{if }\alpha=\beta
        \end{cases}
    \end{equation*}
    with $\alpha$ and $\beta$ denoting the leaf numbers of the two pending subtrees $T_L$ and $T_R$ of $T$ with $T_L\preceq T_R$.}
    \item \props{Computation time}{(this manuscript, see Proposition \ref{runtime_Furrank})}{For every binary tree $T\in\BTnstar$, the Furnas rank $F(T)$ can be computed in time $O(n^2)$.}
    \item \props{Recursiveness}{(this manuscript, see Proposition \ref{recursiveness_Furrank})}{The Furnas rank is a binary recursive tree shape statistic.
    We have $r_1(T)=1$ for $T \in \mathcal{BT}_1^\ast$, and for every binary tree $T\in\BTn$ with $n \geq 2$ and standard decomposition $T=(T_1,T_2)$ with ranks $r_1=r_{n_1}(T_1)$ and $r_2=r_{n_2}(T_2)$ and leaf numbers $n_1$ and $n_2$, we have 
    \begin{equation*}
    \begin{split} 
    r_n(T) &= \sum\limits_{i=1}^{\min\{n_1,n_2\}-1} we(i)\cdot we(n_1+n_2-i)\\
    &\qquad + \Big((r_1-1)\cdot we(n_2)+r_2\Big)\cdot\mathcal{I}(n_1<n_2)\\
    &\qquad + \Big((r_2-1)\cdot we(n_1)+r_1\Big)\cdot\mathcal{I}(n_2<n_1)\\
    &\qquad + \left((\min\{r_1,r_2\}-1)\cdot we(n_1)-\frac{\min\{r_1,r_2\}^2-\min\{r_1,r_2\}}{2}+\max\{r_1,r_2\}\right)\cdot\mathcal{I}(n_1=n_2).
    \end{split}
    \end{equation*}}
    \item \props{Locality}{(this manuscript, see Proposition \ref{locality_Furrank})}{The Furnas rank is not local.}
    \item \props{Maximal value on $\BTnstar$}{(this manuscript, see Remark \ref{rem_Rstat_vs_Rang} and Theorem \ref{Cor_Furnas_ResolutionExtrema})}{For every binary tree $T\in\BTnstar$, the Furnas rank fulfills $F(T)\leq we(n)$. This bound is tight for all $n\in\mathbb{N}_{\geq 1}$.}
    \item \props{(Number of) trees with maximal value on $\BTnstar$}{(this manuscript, see Remark \ref{rem_Rstat_vs_Rang} and Theorem \ref{Cor_Furnas_ResolutionExtrema})}{For any given $n\in\mathbb{N}_{\geq 1}$, there is exactly one tree $T\in\BTnstar$ with maximal Furnas rank, i.e. $F(T)=we(n)$, namely the maximally balanced tree $\Tmb$ (with $\Tfb=\Tmb$ if $n$ is a power of two).}
    \item \props{Minimal value on $\BTnstar$}{(this manuscript, see Remark \ref{rem_Rstat_vs_Rang} and Theorem \ref{Cor_Furnas_ResolutionExtrema})}{For every binary tree $T\in\BTnstar$, the Furnas rank fulfills $F(T)\geq 1$. This bound is tight for all $n\in\mathbb{N}_{\geq 1}$.}
    \item \props{(Number of) trees with minimal value on $\BTnstar$}{(this manuscript, see Remark \ref{rem_Rstat_vs_Rang} and Theorem \ref{Cor_Furnas_ResolutionExtrema})}{For any given $n\in\mathbb{N}_{\geq 1}$, there is exactly one tree $T\in\BTnstar$ with minimal Furnas rank, i.e. $F(T)=1$, namely the caterpillar tree $\Tcat$.}
    \item \textbf{Expected value under the Yule model:} \textit{Open problem.}
    \item \textbf{Variance under the Yule model:} \textit{Open problem.}
    \item \textbf{Expected value under the uniform model:} \textit{Open problem.}
    \item \textbf{Variance under the uniform model:} \textit{Open problem.}
    \item \textbf{Comments:} Since the rank $r_n(T)$ (and thus the Furnas rank) of a tree $T$ among all rooted binary trees with $n$ leaves is unique, it has an inverse function. This means that $T$ can be reconstructed from just its leaf number $n$ and rank $r$. The corresponding procedure can be found in either Algorithm \ref{alg:inverseFurnas} in the appendix of this manuscript or Section 2.5.2 of \citep{Furnas1984}. The computation time of both algorithms lies in $O(n^2)$ (see \citep{Furnas1984} and Proposition \ref{runtime_Furrank_alg} of this manuscript).\\
    While both Furnas' LLR-ordering (on which the rank is based) and the Colijn-Plazzotta labeling (Section \ref{factsheet_CP}) are bijective maps from the set of binary trees to the positive integers, the former considers all trees in $\BTnstar$ before moving onto $\mathcal{BT}_{n+1}^\ast$ while the latter can assign trees with the same leaf number quite distant ranks or assign trees with quite distant leaf numbers adjacent ranks \citep{Rosenberg2020}.
\end{description}

\subsection{\texorpdfstring{$I$}{I}-based indices} \label{factsheet_Ibased}

The $I$-based indices are a family of tree shape statistics defined for arbitrary trees. They are based on quantifying the imbalance of the nodes of a tree $T$ and using these local imbalance values to determine the global (im)balance of $T$ (e.g., by considering the mean of the imbalance values as in case of the Mean $I$ index $\overline{I}(T)$, or by simply considering the imbalance at the root of $T$ as in case of the $I$ value $I_\rho(T)$). Some of the $I$-based indices are (im)balance indices on $\BTnstar$, e.g. the Total $I$ index and the Mean $I$ index, and some of them are not, e.g. the $I$ value. Note that while the $I$-based indices are defined for arbitrary trees, the $I_v$ values can only be calculated for nodes $v$ with out-degree two and at least four descending leaves which is why they are only recommended for arbitrary trees with few polytomies \citep{Fusco1995}. In any case, $I_v$ values increase with a higher degree of asymmetry.

A wide variety of statistics can be applied to the $I_v$ values, e.g. the mean, sum, median and quartile deviation. All of these can be calculated using the function \texttt{IbasedI} from our \textsf{R} package \texttt{treebalance} (the desired statistic and correction method have to be specified). As representative examples, we discuss the Total $I$ index $\Sigma I(T)$, the Total $I'$ index $\Sigma I'(T)$, the Mean $I$ index $\overline{I}(T)$ and the Mean $I'$ index $\overline{I'}(T)$, because these have seen some popularity in the literature (e.g. \citep{Agapow2002,Blum2006b}). Furthermore, and despite the fact that it is not an (im)balance index, we will include the $I$ value $I_\rho(T)$ and the $I'$ value $I'_\rho(T)$ that calculates the $I_v$ value, respectively the $I'_v$ value, for the root only (see the comments below for more information on the origin of this version).

A major advantage of the family of $I$-based indices is that they can be used to compare trees of different sizes because some statistics, for instance the mean of the $I_v$ values, are independent of the number of leaves $n$ when using a suitable correction method.

\begin{description}
    \item \props{Definition}{(\citet{Fusco1995,Purvis2002})}{
    Let $T \in \Tnstar$ be a rooted tree and let $\mathring{V}_{bin, \geq 4}(T)$ denote the set of inner binary vertices $v$ of $T$ with $n_v \geq 4$.  
    Then, the Mean $I$ index $\overline{I}(T)$ and the Toal $I$ index $\Sigma I(T)$ are defined as the mean and total of the imbalance values $I_v$ over all vertices $v \in \mathring{V}_{bin, \geq 4}(T)$, i.e.
    \[ \overline{I}(T) \coloneqq \frac{1}{|\mathring{V}_{bin,\geq 4}(T)|} \cdot  \sum\limits_{v \in \mathring{V}_{bin,\geq 4}(T)} I_v \qquad \text{and} \qquad \Sigma I(T) \coloneqq  \sum\limits_{v \in \mathring{V}_{bin,\geq 4}(T)} I_v,\]
    where the imbalance value $I_v$ \cite{Fusco1995}\footnote{Note that \citet{Fusco1995} introduced the $I_v$ value in a slightly more general way by allowing each leaf of a tree to represent several species and then considering the number of descending terminal species instead of the number of descending leaves for each binary vertex $v$.} of a binary node $v$ with $n_v\geq 4$ is the ratio between the observed deviation of the number of descending leaves $n_{v_1}$ of the larger subtree rooted at $v$ from the minimum value possible and the maximum deviation possible, i.e.
    \[ I_v \coloneqq \frac{n_{v_1}-\lceil \frac{n_v}{2} \rceil}{(n_v-1)- \lceil \frac{n_v}{2} \rceil}. \] Note that $I_v \in [0,1]$ for each $v \in  \mathring{V}_{bin,\geq 4}(T)$. Moreover, note that $\overline{I}(T)=\Sigma I(T) = 0$ for each tree $T \in \Tnstar$ with $n \in \{1,2,3\}$ (since $\mathring{V}_{bin,\geq 4}(T)$ is the empty set for each such tree).
  
    Moreover, three correction methods related to the $I_v$ values have been proposed in the literature. In order to compare the frequency distributions of $I_v$ value for trees of different sizes \citet{Fusco1995} used a correction method to smooth the discrete distribution, such that it could be approximated with a uniform distribution on $[0,1]$ (see appendix of \citep{Fusco1995}). While this correction method is used when comparing the distribution of $I_v$ values, two further correction methods, $I'$ and $I^w$, are applied to the the $I_v$ values themselves. These methods have been defined by \citet{Purvis2002} and $I'$ has already been used in several studies \citep{Agapow2002,Blum2006b}. They are both designed to account for the fact that $I_v$ is not uniformly distributed on $\left\{0,\frac{1}{(n_v-1)- \lceil \frac{n_v}{2} \rceil},\ldots,1\right\}$ under the Yule model and thus the expected value under the Yule model is not independent of $n_v$. Both methods ensure $E_Y(I'_v)=\frac{1}{2}$, respectively $E_Y(I^w_v)=\frac{1}{2}$ for all $n_v \geq 4$. 
    
    The correction method $I'$ gives rise to the Mean $I'$ index $\overline{I'}(T)$ and the Total $I'$ index $\Sigma I'(T)$, which are defined as the mean and the total of the $I'_v$ values over all vertices $v \in \mathring{V}_{bin,\geq 4}(T)$, i.e.
    \[ \overline{I'}(T) \coloneqq \frac{1}{|\mathring{V}_{bin,\geq 4}(T)|} \cdot  \sum\limits_{v \in \mathring{V}_{bin,\geq 4}(T)} I_v' \qquad \text{and} \qquad \Sigma I'(T)\coloneqq\sum\limits_{v \in \mathring{V}_{bin,\geq 4}(T)} I_v',\]
    where the $I_v'$ value of a vertex $v \in \mathring{V}_{bin, \geq 4}$ is defined as
    \[ I_v' \coloneqq \begin{cases} 
    I_v &\text{if }n_v \text{ is odd} \\ 
    \frac{n_v-1}{n_v} \cdot I_v &\text{if }n_v \text{ is even.} \end{cases} \]
    Again, note that $I_v' \in [0,1]$ for each $v \in  \mathring{V}_{bin, \geq 4}$ and $\overline{I'}(T) = \Sigma I'(T)=0$ for each tree $T \in \Tnstar$ with $n \in \{1,2,3\}$. Also note that the disadvantage of this correction method is that the mean and even the median of the $I'_v$ values of a maximally imbalanced tree may be $<1$, i.e. it may not reach the maximal value. This could harm the power of these imbalance indices to successfully recognize and compare trees that are more imbalanced than expected under the Yule model.
    
    For the second method, $I^w$, first the $I_v$ values are computed according to their definition, but then a weighted mean with an expected value of 0.5 under the Yule model is calculated. More precisely, the $I^w_v$ value of a vertex $v\in\mathring{V}_{bin,\geq 4}$ is defined as \[I^w_v \coloneqq \frac{w(I_v) \cdot I_v}{\text{mean}_{v \in \mathring{V}_{bin,\geq 4}} w(I_v)} \qquad \text{with weights} \qquad w(I_v) \coloneqq \begin{cases} 
    1 & \text{if } n_v \text{ is odd} \\ 
    \frac{n_v-1}{n_v} & \text{if }n_v \text{ is even and } I_v>0 \\
    \frac{2\cdot(n_v-1)}{n_v} & \text{if } n_v \text{ is even and } I_v=0. \end{cases}\]
    
    Finally, the $I$ value $I_\rho(T)$ and the $I'$ value $I_\rho'(T)$ of a rooted tree $T\in \Tnstar$ with root $\rho$ such that $\rho$ has out-degree two and $n_{\rho}=n \geq 4$ are defined as $I_\rho$, respectively $I'_\rho$, i.e. \[I_{\rho}(T) \coloneqq I_{\rho} \quad \text{and} \quad I'_{\rho}(T) \coloneqq I_{\rho}.\]
    }
    \item \props{Computation time}{(this manuscript, see Proposition \ref{runtime_Ibased})}{For every tree $T\in\Tnstar$ the $I_v$ values for all binary inner nodes $v$ with $n_v\geq4$ can be computed in $O(n)$. The computation time of applying a statistic to the $I_v$ values (with or without correction $I'$ or $I^w$) of all binary vertices $v\in \mathring{V}(T)$ with $n_v\geq4$ only depends on the computation time of the respective statistic, but is at least linear.}
    \item \props{Recursiveness}{(this manuscript, see Proposition  \ref{recursiveness_I'}, \ref{recursiveness_SumI'} and \ref{recursiveness_meanI'}, and Remark \ref{recursiveness_I}, \ref{recursiveness_SumI} and \ref{recursiveness_meanI})}{Setting (for technical reasons) $I_\rho(T)=I'_\rho(T)=0$ for $n\in\{1,2,3\}$, the $I$ value and the $I'$ value are binary recursive tree shape statistics. We have $I_\rho(T)=I'_\rho(T)=0$ for $T \in \mathcal{BT}_1^\ast$, and for every binary tree $T\in\BTn$ with $n \geq 2$ and standard decomposition $T=(T_1,T_2)$, we have \[ I_\rho(T)=\frac{\max\{n_1,n_2\}-\lceil\frac{n_1+n_2}{2}\rceil}{n_1+n_2-1-\lceil\frac{n_1+n_2}{2}\rceil} \] and \[ I'_\rho(T)=\frac{\max\{n_1,n_2\}-\lceil\frac{n_1+n_2}{2}\rceil}{n_1+n_2-1-\lceil\frac{n_1+n_2}{2}\rceil}\cdot\left(\frac{n_1+n_2-1}{n_1+n_2}\right)^{\mathcal{I}(n_1+n_2\text{ mod }2=0)}. \] 
    The Total $I$ and the Total $I'$ index are binary recursive tree shape statistics. We have $\Sigma I(T)=\Sigma I'(T)=0$ for $T \in \mathcal{BT}_1^\ast$, and for every binary tree $T\in\BTnstar$ with $n \geq 2$ and standard decomposition $T=(T_1,T_2)$ we have \[ \Sigma I(T)=\Sigma I(T_1)+\Sigma I(T_2) +\frac{\max\{n_1,n_2\}-\lceil\frac{n_1+n_2}{2}\rceil}{n_1+n_2-1-\lceil\frac{n_1+n_2}{2}\rceil} \] and \[ \Sigma I'(T)=\Sigma I'(T_1)+\Sigma I'(T_2) +\frac{\max\{n_1,n_2\}-\lceil\frac{n_1+n_2}{2}\rceil}{n_1+n_2-1-\lceil\frac{n_1+n_2}{2}\rceil}\cdot\left(\frac{n_1+n_2-1}{n_1+n_2}\right)^{\mathcal{I}(n_1+n_2\text{ mod }2=0)} \]\\
    The Mean $I$ and the Mean $I'$ index are binary recursive tree shape statistics. We have $\overline{I}(T)=\overline{I'}(T)=0$ for $T \in \mathcal{BT}_1^\ast$, and for every binary tree $T\in\BTnstar$ with $n \geq 2$ and standard decomposition $T=(T_1,T_2)$ we have \[ \overline{I}(T)=\frac{\overline{I}(T_1)\cdot a(T_1)+\overline{I}(T_2)\cdot a(T_2)+\frac{\max\{n_1,n_2\}-\lceil\frac{n_1+n_2}{2}\rceil}{n_1+n_2-1-\lceil\frac{n_1+n_2}{2}\rceil}}{a(T_1)+a(T_2)+\mathcal{I}(n_1+n_2\geq 4)} \] and \[ \overline{I'}(T)=\frac{\overline{I'}(T_1)\cdot a(T_1)+\overline{I'}(T_2)\cdot a(T_2)+\frac{\max\{n_1,n_2\}-\lceil\frac{n_1+n_2}{2}\rceil}{n_1+n_2-1-\lceil\frac{n_1+n_2}{2}\rceil}\cdot\left(\frac{n_1+n_2-1}{n_1+n_2}\right)^{\mathcal{I}(n_1+n_2\text{ mod }2=0)}}{a(T_1)+a(T_2)+\mathcal{I}(n_1+n_2\geq 4)} \] in which $a(T_i)$ denotes the number of vertices $v$ in $T_i$ with $n_v\geq 4$.}
    \item \props{Locality}{(this manuscript, see Proposition \ref{locality_I'}, \ref{locality_SumI'} and \ref{locality_meanI'})}{The $I$ value and $I'$ value as well as the Mean $I$ index and the Mean $I'$ index are not local. The Total $I$ and Total $I'$ indices are local.}
    \item \props{Maximal value and (number of) maximal trees for the $I$ and $I'$ value on $\Tnstar$ and $\BTnstar$}{(this manuscript, see Theorem \ref{prop_Irho_max_ab})}{For every $n \in \mathbb{N}_{\geq 4}$ the maximal $I$ value $I_\rho(T)$ over all $T \in \mathcal{T}_n^\ast$ with binary roots or $T \in \mathcal{BT}_n^\ast$ is $I_\rho(T)=1$. Every tree whose (binary) root has a leaf as a child is a maximal tree. There are $we(n-1)$ maximal binary trees and $|\mathcal{T}_{(n-1)}^\ast|$ maximal arbitrary trees that are binary at the root. The results apply for the correction method $I'$ as well, except that the maximal value is $\frac{n-1}{n}$ if $n$ is even.}
    \item \props{Minimal value and (number of) minimal trees for the $I$ and $I'$ value on $\Tnstar$ and $\BTnstar$}{(this manuscript, see Theorem \ref{prop_Irho_min_ab})}{For every $n \in \mathbb{N}_{\geq 4}$ the minimal $I$ value $I_\rho(T)$ over all $T \in \mathcal{T}_n^\ast$ with binary roots or $T \in \mathcal{BT}_n^\ast$ is $I_\rho(T)=0$. Every tree whose binary root partitions the number of descending leaves $n$ into $\lceil \frac{n}{2} \rceil$ and $\lfloor \frac{n}{2} \rfloor$ is minimal. 
    There are $we(\lceil \frac{n}{2} \rceil) \cdot we(\lfloor \frac{n}{2} \rfloor)$ minimal binary trees if $n$ is odd and $\frac{1}{2} \cdot we(\frac{n}{2}) \cdot (we(\frac{n}{2})+1)$ if $n$ is even. Analogously, there are $|\mathcal{T}_{\lceil \frac{n}{2} \rceil}^\ast| \cdot  |\mathcal{T}_{\lfloor \frac{n}{2} \rfloor}^\ast|$ minimal arbitrary trees that are binary at the root if $n$ is odd and  $\frac{1}{2} \cdot |\mathcal{T}_{\frac{n}{2}}^\ast| \cdot \left( |\mathcal{T}_{\frac{n}{2}}^\ast| + 1\right)$  if $n$ is even. The same results hold for the correction method $I'$.}
    \item \props{Maximal value and (number of) maximal trees for the Total $I$ and Total $I'$ indices on $\Tnstar$ and $\BTnstar$}{(this manuscript, see Theorem \ref{prop_meanI_max_b} and \ref{prop_I_max_a} and Remark \ref{remark_I'_123})}{For every tree $T \in\mathcal{T}_{n\geq 4}^\ast$ and $T\in\mathcal{BT}_{n\geq 4}^\ast$, we have $\Sigma I(T)\leq n-3$. This bound is tight for all $n \geq 4$. Moreover, $T$ is a maximal tree if and only if $T=\Tcat$ or (when considering $\Tnstar$) $T$ can be constructed from $\Tcat$ by contracting the inner edge leading to its only cherry. In particular, on $\BTnstar$ there is precisely one maximal tree and on $\Tnstar$ there are precisely two maximal trees for $n\geq 4$. For $n\in\{1,2,3\}$ we have $\Sigma I(T)=0$ for every $T\in\Tnstar$ and $T\in\BTnstar=\{\Tcat\}$. The results hold for the correction method $I'$ as well except that the maximal value for $n\geq 4$ is $\left\lfloor\frac{n-3}{2}\right\rfloor+\sum\limits_{k=2}^{\left\lceil\frac{n-1}{2}\right\rceil} \frac{2k-1}{2k} < n-3$.}
    \item \props{Minimal value and (number of) minimal trees for the Total $I$ and Total $I'$ indices on $\Tnstar$ and $\BTnstar$}{(this manuscript, see Theorem \ref{prop_meanI_min_b} and Proposition \ref{prop_meanI_min_a} and Remark \ref{remark_I'_123})}{\textit{Open problem.} For every tree $T \in\BTnstar$ and $T\in\Tnstar$, we have $\Sigma I(T)\geq 0$. This bound is tight for all $n \geq 1$ in both the arbitrary and the binary case. Moreover, on $\BTnstar$ the unique tree reaching the minimum is the maximally balanced tree $\Tmb$, and on $\Tnstar$ a tree $T$ is minimal if and only if each binary node $v$ of $T$ with $n_v\geq 4$ has an $I_v$ value of zero. Note that the number of minimal trees in $\Tnstar$ is to our knowledge not known yet. For $n\in\{1,2,3\}$ we have $\Sigma I(T)=0$ for all $T\in\Tnstar$ and $T\in\BTnstar=\{\Tmb\}$. The results hold for the correction method $I'$ as well.}
    \item \props{Maximal value and (number of) maximal trees for the Mean $I$ and Mean $I'$ indices on $\Tnstar$}{(this manuscript, see Proposition \ref{prop_meanI_max_a})}{\textit{Open problem.} For every tree $T \in \mathcal{T}_{n \geq 4}^\ast$ with at least one binary node $v$ with $n_v \geq 4$, we have $\overline{I}(T) \leq 1$. This bound is tight for all $n \geq 4$. Moreover, any such tree $T \in \mathcal{T}_{n \geq 4}^\ast$ is a maximal tree if and only if each of its binary nodes $v \in \mathring{V}(T)$ with $n_v \geq 4$ has an $I_v$ value of one.\\
    Similarly, for every tree $T\in\mathcal{T}_{n\geq 4}^\ast$ with at least one binary node $v$ with $n_v\geq 4$, we have $\overline{I'}(T)\leq 1$. This bound is tight for $n=5$ and all $n\geq 7$. Moreover, any such tree $T\in\Tnstar$ with $n=5$ or $n\geq 7$ is a maximal tree if and only if for each of its binary nodes $v\in\mathring{V}(T)$ with $n_v\geq 4$ it holds that its $I_v$ value is one and $n_v$ is odd.
    The number of trees achieving the maximum in both cases is to our knowledge still unknown.}
    \item \props{Maximal value and (number of) maximal trees for the Mean $I$ and Mean $I'$ indices on $\BTnstar$}{(this manuscript, see Theorem \ref{prop_meanI_max_b})}{For every $n \in \mathbb{N}_{\geq 4}$ the maximal Mean $I$ index $\overline{I}(T)$ over all $T \in \mathcal{BT}_n^*$ is $\overline{I}(T)=1$ and $\Tcat$ is the unique maximal tree. The results hold for the correction method $I'$ as well, except that the maximal value is $\frac{1}{n-3}\cdot \left(\lfloor \frac{n-3}{2} \rfloor +  \sum\limits_{k=2}^{\lceil \frac{n-1}{2} \rceil}{\frac{2k-1}{2k}} \right)<1$. For $n\in\{1,2,3\}$ we have $\overline{I}(T)=\overline{I'}(T)=0$ for $T\in\BTnstar=\{\Tcat\}$.}
    \item \props{Minimal value and (number of) minimal trees for the Mean $I$ and Mean $I'$ indices on $\Tnstar$}{(this manuscript, see Proposition \ref{prop_meanI_min_a})}{\textit{Open problem.} For every tree $T \in \mathcal{T}_{n \geq 4}^\ast$ with at least one binary node $v$ with $n_v \geq 4$, we have $\overline{I}(T) \geq 0$. This bound is tight for all $n \geq 4$. Moreover, any such tree $T \in \mathcal{T}_{n \geq 4}^\ast$ is a minimal tree if and only if each of its binary nodes $v \in \mathring{V}(T)$ with $n_v \geq 4$ has an $I_v$ value of zero. The same results hold for the correction method $I'$.\\
    The number of trees achieving the minimum in both cases is to our knowledge still unknown.}
    \item \props{Minimal value and (number of) minimal trees for the Mean $I$ and Mean $I'$ indices on $\BTnstar$}{(this manuscript, see Theorem \ref{prop_meanI_min_b})}{For $n\in\{1,2,3\}$ we have $\overline{I}(T)=0$ for $T\in\BTnstar=\{\Tmb\}$. For every $n \in \mathbb{N}_{\geq 4}$ the minimal Mean $I$ index $\overline{I}(T)$ over all $T \in \BTnstar$ is $\overline{I}(T)=0$ and this minimum is uniquely achieved by the maximally balanced tree $\Tmb$. In particular, for $n=2^h$ with $h\in\mathbb{N}_{\geq 0}$, $\Tfb = T^\mathit{mb}_{2^h}$ is the unique minimal tree. The same results hold for the correction method $I'$.}
    \item \props{Expected value under the Yule model}{%\citep{Farris1976, Slowinski1990}
    (\citet{Purvis2002}; this manuscript, see Lemma \ref{lem_Iv_YuleEV} and 
    Proposition \ref{prop_Iv_YuleEV})}{\textit{Open problem.} Let $T_n \in \BTn$ be a phylogenetic tree with $n\geq4$ leaves sampled under the Yule model and let $v \in \mathring{V}(T_n)$ be an arbitrary node of $T_n$ with $n_v \geq 4$. Then, we have $E_Y(I_v)=\frac{1}{2}$ if $n_v$ is odd and $E_Y(I_v)=\frac{n_v/2}{(n_v-1)}>\frac{1}{2}$ if $n_v$ is even. In particular, $E_Y(I_\rho(T_n)) = E_Y(I_\rho)=\frac{1}{2}$ if $n$ is odd and $E_Y(I_\rho(T_n)) = E_Y(I_\rho)= \frac{n/2}{(n-1)} > \frac{1}{2}$ if $n$ is even. \\ 
    For the Mean $I$ index $\overline{I}(T_n)$, we have $\frac{1}{2} \leq E_Y(\overline{I}(T_n)) \leq \frac{n_{v,min}/2}{(n_{v,min}-1)} \leq 2/3$ with $n_{v,min}$ being the smallest even subtree size $> 3$ in $T_n$. Using the correction method $I'$ we have $E_Y(I_\rho'(T_n))=E_Y(\overline{I'}(T_n))=\frac{1}{2}$. \\ 
    Finally, for the Total $I'$ index, we have $E_Y(\Sigma I'(T_n))=\frac{n}{4}-\frac{1}{2}$. \\
    The exact value of $E_Y(\overline{I}(T_n))$ as well as the value of $E_Y(\Sigma I(T_n))$ are not known yet.}
    \item \props{Variance under the Yule model}{(this manuscript, see Proposition \ref{prop_Iv_YuleVar})}{\textit{Open problem.} Let $T_n \in \BTn$ be a phylogenetic tree with $n \geq 4$ leaves sampled under the Yule model and let $v \in \mathring{V}(T_n)$ be an arbitrary node of $T_n$ with $n_v \geq 4$. Then, we have
    $$V_Y(I_v)=\begin{cases} \frac{n_v^2-6n_v+17}{12(n_v-1)(n_v-3)} \xrightarrow{\ n_v \to \infty \ } \frac{1}{12}  & \textit{ if }n_v\textit{ is odd }\\
    \frac{n_v^4 - 6 n_v^3 + 12 n_v^2 - 4 n_v}{12 (n_v-1)^3 (n_v-2)} \xrightarrow{\ n_v \to \infty \ } \frac{1}{12} & \textit{ if }n_v\textit{ is even }
    \end{cases}$$
    and 
    $$V_Y(I_v')=\begin{cases} V_Y(I_v)=\frac{n_v^2-6n_v+17}{12(n_v-1)(n_v-3)} \xrightarrow{\ n_v \to \infty \ } \frac{1}{12}  & \textit{ if }n_v\textit{ is odd }\\
    \frac{n_v^3 - 6 n_v^2 + 12 n_v - 4}{12 (n_v-2) (n_v-1) n_v }
    \xrightarrow{\ n_v \to \infty \ } \frac{1}{12}
     & \textit{ if }n_v\textit{ is even. }
    \end{cases}$$
    In particular, $V_Y(I_\rho(T_n)) = V_Y(I_\rho)$ and $V_Y(I'_\rho(T_n)) = V_Y(I'_\rho)$ are obtained by substituting $n$ for $n_v$ in the expressions above. Expressions for the variances $V_Y(\overline{I}(T_n))$ and $V_Y(\overline{I'}(T_n))$ and $V_Y(\Sigma I(T_n))$ and $V_Y(\Sigma I'(T_n))$ are not known yet.}
    \item \textbf{Expected value under the uniform model:} \textit{Open problem.}
    \item \textbf{Variance under the uniform model:} \textit{Open problem.}
    \item \textbf{Comments:} Although not being an (im)balance index according to our definition, the $I$ value $I_\rho$ is an interesting tree shape statistic because it reflects the basic ideas and properties of the general $I_v$ values. Moreover, it is a more refined version of the \enquote{degree of unbalance}  at the root $n_{1}/n$, an older imbalance measurement introduced by \citet{Guyer1993}, that directly inspired Fusco and Cronk to define the $I_v$ values \citep[p. 236]{Fusco1995}.\\
    Note that a possible extension of the $I_v$ values to non-binary nodes by averaging over all possible binary pending trees was discussed by \citet{Fusco1995}, but ultimately rejected to avoid a higher influence of noise.
\end{description}

\subsection{Quadratic Colless index} \label{factsheet_qColless}

The quadratic Colless index was recently introduced by \citet{Bartoszek2021} as an alternative to the traditional Colless index (Section \ref{factsheet_Colless}) and is (just like the original) an imbalance index for rooted binary trees. The authors argue that it has certain advantages over the original version. On the one hand, the most balanced trees according to the quadratic Colless index are unique (namely, the maximally balanced trees), whereas for the original Colless index, the minimum value is almost always also reached at trees that are not maximally balanced. On the other hand, the distribution of the quadratic Colless index under probabilistic models for phylogenetic trees turns out to be easier to study analytically than the corresponding distribution of the original Colless index (e.g. an explicit formula for the expected value and variance of the latter under the uniform model are not known yet). Finally, the authors found that the quadratic Colless index is more powerful than the original one (and in fact other indices as well) in the sense that (for $n\in\{4,\ldots,20\}$) there are fewer ties (i.e. two trees obtaining the same index value) and in terms of discriminating between trees as it has been suggested in \citep{Hayati2019}. The quadratic Colless index can be calculated using the function \texttt{collessI} from our \textsf{R} package \texttt{treebalance} specifying \texttt{\enquote{quadratic}} as the desired method.

\begin{description}
    \item \props{Definition}{(\citet{Bartoszek2021})}{The quadratic Colless index $QC(T)$ of a binary tree $T\in\BTnstar$ is defined as the sum of the squared balance values of its inner vertices, i.e. \[ QC(T) \coloneqq \sum\limits_{v \in \mathring{V}(T)} bal_T(v)^2. \]}
    \item \props{Computation time}{(this manuscript, see Proposition \ref{runtime_Colless})}{For every binary tree $T\in\BTnstar$, the quadratic Colless index $QC(T)$ can be computed in time $O(n)$.}
    \item \props{Recursiveness}{(\citet[Lemma 1]{Bartoszek2021}; this manuscript, see Proposition \ref{recursiveness_qColless})}{The quadratic Colless index is a binary recursive tree shape statistic. We have $QC(T)=0$ for $T\in\mathcal{BT}_1^\ast$, and for every binary tree $T\in\BTnstar$ with $n\geq 2$ and standard decomposition $T=(T_1,T_2)$ we have \[ QC(T)=QC(T_1)+QC(T_2)+(n_1-n_2)^2. \]}
    \item \props{Locality}{(this manuscript, see Proposition \ref{locality_qColless})}{The quadratic Colless index is local.}
    \item \props{Maximal value on $\BTnstar$}{(\citet[Theorem 4]{Bartoszek2021})}{For every binary tree $T\in\BTnstar$, the quadratic Colless index fulfills \[ QC(T) \leq \binom{n}{3} + \binom{n-1}{3}. \] This bound is tight for all $n\in\mathbb{N}_{\geq 1}$.}
    \item \props{(Number of) trees with maximal value on $\BTnstar$}{(\citet[Theorem 4]{Bartoszek2021})}{For any given $n\in\mathbb{N}_{\geq 1}$, there is exactly one binary tree $T\in\BTnstar$ with maximal quadratic Colless index, i.e. $QC(T)=\binom{n}{3}+\binom{n-1}{3}$, namely the caterpillar tree $\Tcat$.}
    \item \props{Minimal value on $\BTnstar$}{(\citet[Theorem 3]{Bartoszek2021}; \citet[Theorem 4]{Hamoudi2017}; \citet[Theorem 2 and 3]{Coronado2020a})}{Let $T\in\BTnstar$ be a binary tree. First, let $b_a b_{a-1} \ldots b_0$ denote the binary representation of $n$. Then, \[ QC(T) \geq 2 \cdot (n \text{ mod } 2^a) + \sum\limits_{j=0}^{a-1} (-1)^{b_j} \cdot (n \text{ mod } 2^{j+1}). \]
    Second, consider the binary expansion of $n$, i.e. write $n=\sum\limits_{j=1}^\ell 2^{d_j}$ with $\ell\geq 1$ and $d_1,\ldots,d_\ell\in\mathbb{N}_{\geq 0}$ such that $d_1>\ldots>d_\ell$. Then, \[ QC(T)\geq \sum\limits_{j=2}^\ell 2^{d_j}\cdot(d_1-d_j-2\cdot(j-2)). \]
    Third, let $s(x)$ denote the triangle wave, i.e. the distance from $x\in\mathbb{R}$ to its nearest integer. Then,  \[ QC(T)\geq \sum\limits_{j=1}^{\lceil\log_2(n)\rceil-1} 2^j\cdot s(2^{-j}\cdot n). \] 
    These three bounds are equivalent\footnote{The equivalence follows directly from the fact that the quadratic and the traditional Colless index have the same minimum \citep{Bartoszek2021} and the fact that all three expressions on the right-hand side give the minimum value of the traditional Colless index \citep{Coronado2020a,Hamoudi2017} (see also Section \ref{factsheet_Colless}).} and are tight for all $n\in\mathbb{N}_{\geq 1}$. In particular, the minimum value of the quadratic Colless index coincides with the minimum value of the traditional Colless index.}
    \item \props{(Number of) trees with minimal value on $\BTnstar$}{(\citet[Theorem 3]{Bartoszek2021})}{For any given $n\in\mathbb{N}_{\geq 1}$, there is exactly one binary tree $T \in \BTnstar$ with minimal quadratic Colless index namely the maximally balanced tree $\Tmb$.}
    \item \props{Expected value under the Yule model}{(\citet[Theorem 6]{Bartoszek2021})}{Let $T_n$ be a phylogenetic tree with $n\geq 1$ leaves sampled under the Yule model. Then, the expected value of $QC$ of $T_n$ is \[ E_Y(QC(T_n))=n(n+1)-2n\cdot H_n. \] Moreover, in the limit \[ E_Y(QC(T_n)) \sim n^2. \]}
    \item \props{Variance under the Yule model}{(\citet[Theorem 6]{Bartoszek2021})}{Let $T_n$ be a phylogenetic tree with $n \geq 1$ leaves sampled under the Yule model. Then, the variance of $QC$ of $T_n$ is \[ V_Y(QC(T_n)) = \frac{1}{3}n\cdot \left( n^3 -8n^2 +50n - 1 - 30H_n - 12 n\cdot H_n^{(2)} \right). \] Moreover, in the limit \[ V_Y(QC(T_n)) \sim \frac{1}{\sqrt{3}}\cdot n^2. \]}
    \item \props{Expected value under the uniform model}{(\citet[Theorem 5]{Bartoszek2021})}{Let $T_n$ be a phylogenetic tree with $n\geq 1$ leaves sampled under the uniform model. Then, the expected value of $QC$ of $T_n$ is \[ E_U(QC(T_n)) = \binom{n+1}{2} \cdot \frac{(2n-2)!!}{(2n-3)!!} - n\cdot(2n-1). \] Moreover, in the limit \[ E_U(QC(T_n)) \sim \frac{\sqrt{\pi}}{2}\cdot n^{5/2}. \]}
    \item \props{Variance under the uniform model}{(\citet[Theorem 5]{Bartoszek2021})}{Let $T_n$ be a phylogenetic tree with $n\geq 1$ leaves sampled under the uniform model. Then, the variance of $QC$ of $T_n$ is
    \begin{align*}
        V_U(QC(T_n)) &= \frac{2}{15}\cdot (2n-1)(7n^2+9n-1) \binom{n+1}{2} - \frac{1}{8}\cdot(5n^2+n+2) \binom{n+1}{2}\cdot \frac{(2n-2)!!}{(2n-3)!!} \\
        &\qquad - \binom{n+1}{2}^2\cdot \left( \frac{(2n-2)!!}{(2n-3)!!} \right)^2.
    \end{align*}
    Moreover, in the limit \[ V_U(QC(T_n)) \sim \sqrt{\frac{14}{15}}\cdot n^{5/2}. \]}
    \item \textbf{Comments:} For each $T \in \BTnstar$, the quadratic Colless index of $T$ is bounded from below by the original Colless index of $T$, i.e. $QC(T) \geq C(T)$, where equality holds if and only if $T = \Tmb$ \citep[Lemma 2]{Bartoszek2021}.
\end{description}

\subsection{Rogers \texorpdfstring{$J$}{J} index} \label{factsheet_RogersU}

The Rogers $J$ index (also known as the \enquote{number of unbalanced vertices} $U$) is an imbalance index for binary trees that was introduced by \citet{Rogers1996}. Instead of measuring the degree of imbalance at each inner vertex and summing up these values like the Colless index (Section \ref{factsheet_Colless}), it simply counts the number of inner vertices that have a balance value greater than zero, i.e. that are not perfectly balanced. So, while the former measures the total imbalance of the tree, the Rogers $J$ index measures how this total imbalance is distributed across the inner vertices. It is thus less influenced by few highly unbalanced vertices while being equally influenced by numerous slightly unbalanced vertices \citep{Rogers1996}. It can be calculated using the function \texttt{rogersI} from our \textsf{R} package \texttt{treebalance}.

\begin{description}
    \item \props{Definition}{(\citet[p. 100]{Rogers1996})}{The Rogers $J$ index $J(T)$ of a binary tree $T\in\BTnstar$ is defined as the number of inner vertices that are not perfectly balanced, formally \[ J(T) \coloneqq \sum\limits_{v \in \mathring{V}(T)} \left(1 - \mathcal{I}(bal_T(v) = 0)\right) = \sum\limits_{v \in \mathring{V}(T)} \left(1- \mathcal{I}(n_{v_1}=n_{v_2})\right), \]
    where $n_{v_1}$ and $n_{v_2}$ denote the number of leaves of the subtrees of $T$ rooted at the children $v_1$ and $v_2$ of $v$ and $\mathcal{I}$ is the indicator function.}
    \item \props{Computation time}{(this manuscript, see Proposition \ref{runtime_Colless})}{For every binary tree $T \in \BTnstar$, the Rogers $J$ index $J(T)$ can be computed in time $O(n)$.}
    \item \props{Recursiveness}{(this manuscript, see Proposition \ref{recursiveness_U})}{The Rogers $J$ index is a binary recursive tree shape statistic. We have $J(T)=0$ for $T\in\mathcal{BT}_1^\ast$, and for every binary tree $T\in\BTnstar$ with $n\geq 2$ and standard decomposition $T=(T_1,T_2)$ we have \[ J(T) = J(T_1) + J(T_2) + 1-\mathcal{I}(n_1=n_2). \]}
    \item \props{Locality}{(this manuscript, see Proposition \ref{locality_U})}{The Rogers $J$ index is local.}
    \item \props{Maximal value on $\BTnstar$}{(\citet[Theorem 5.3]{Kersting2021})}{For every binary tree $T \in \BTnstar$ with $n \geq 2$, the Rogers $J$ index fulfills $J(T) \leq n-2$. This bound is tight for all $n \in \mathbb{N}_{\geq 2}$. For $n=1$, we have $J(T)=0$.}
    \item \props{(Number of) trees with maximal value on $\BTnstar$}{(\citet[Theorem 5.3]{Kersting2021})}{For any given $n \in \mathbb{N}_{\geq 1}$, there is exactly one binary tree $T\in\BTnstar$ with maximal Rogers $J$ index, namely the caterpillar tree $\Tcat$.}
    \item \props{Minimal value on $\BTnstar$}{(\citet[Theorem 5.4]{Kersting2021})}{For every binary tree $T\in\BTnstar$, the Rogers $J$ index fulfills $J(T)\geq wt(n)-1$, where $wt(n)$ denotes the binary weight of $n$. This bound is tight for all $n \in \mathbb{N}_{\geq 1}$.}
    \item \props{Trees with minimal value on $\BTnstar$}{(\citet[Theorem 5.4]{Kersting2021})}{For any given $n\in\mathbb{N}_{\geq 1}$, the minimal value of the Rogers $J$ index is reached precisely by the rooted binary weight trees $\widetilde{T}_n$.}
    \item \props{Number of trees with minimal value on $\BTnstar$}{(\citet[Theorem 5.4]{Kersting2021})}{For any given $n\in\mathbb{N}_{\geq 1}$, the number of binary trees $T \in \BTnstar$ with minimum Rogers $J$ index equals $|\mathcal{BT}_{wt(n)}|$ (i.e. the number of binary \emph{phylogenetic} trees on $wt(n)$ leaves), where $|\mathcal{BT}_{wt(n)}|=1 $ for $wt(n)=1$ and $|\mathcal{BT}_{wt(n)}|=(2 \cdot wt(n)-3)!!$ otherwise.}
    \item \textbf{Expected value under the Yule model:} \textit{Open problem.}
    \item \textbf{Variance under the Yule model:} \textit{Open problem.}
    \item \textbf{Expected value under the uniform model:} \textit{Open problem.}
    \item \textbf{Variance under the uniform model:} \textit{Open problem.}
    \item \textbf{Comments:} When compared to the Colless index $C(T)$, \citet{Rogers1996} found that the $J$ index measures the tree imbalance more coarsely, because the former can take a maximum of $1+\frac{(n-1)(n-2)}{2}$ distinct values on $\BTnstar$ while the latter can have a maximum of $n-1$ distinct values on $\BTnstar$.\\
    Note that \citet{Rogers1996} also analyzed the expected value, standard deviation, variance and skewness of the normalized $J$ index (where normalization corresponds to dividing $J(T)$ by its maximum value $n-2$), under the Yule and uniform models for trees of 1-50 leaves. However, general analytical results for all $n \in \mathbb{N}_{\geq 1}$  seem to be missing in the literature.\\
    There exists another modified version of the Rogers $J$ index, the index known as stairs or stairs1 (see Table  \ref{Table_bi_tss}) which divides $J(T)$ by $n-1$, i.e. it measures the fraction of inner nodes that are not perfectly balanced.\\
    Finally, while the Rogers $J$ index $J(T)$ and the symmetry nodes index $SNI(T)$ (Section \ref{factsheet_SNI}) are generally not the same (every symmetry vertex has a balance value of zero, but not every vertex with a balance value of zero is a symmetry vertex), their maximum and minimum values as well as the trees achieving these extremes coincide for all $n \in \mathbb{N}_{\geq 1}$ (see \citet[Section 5.2]{Kersting2021}).
\end{description}

\subsection{Rooted quartet index} \label{factsheet_rQuartet}

Defined for arbitrary trees the rooted quartet index is a balance index, i.e. for a fixed $n\in\mathbb{N}_{\geq 1}$ it increases with increasing balance of the tree. It is based on the idea that highly symmetrical evolutive processes will likely produce symmetrical evolutive histories when only small subsets of taxa are considered \citep{Coronado2019}. It thus quantifies the global balance of a tree based on the local balance of its induced subtrees on 4 leaves. As the range of values grows with the number of leaves, the rooted quartet index is not suitable to compare trees with differing leaf numbers. It can be calculated using the function \texttt{rQuartetI} from our \textsf{R} package \texttt{treebalance}.

\begin{description}
    \item \props{Definition}{(\citet{Coronado2019})}{The rooted quartet index $rQI(T)$ of a tree $T\in\Tnstar$ is defined as the sum of the $rQI$-values of its rooted quartets, i.e. \[ rQI(T) \coloneqq \sum\limits_{Q \in \mathcal{Q}(T)} rQI(Q) = \sum\limits_{i=1}^4 |\{Q \in \mathcal{Q}(T): Q \text{ has shape } Q_i^\ast \}| \cdot q_i, \]
    where $q_0 = 0$ and $0 < q_1 < q_2 < q_3 < q_4$. As stated in \citep{Coronado2019}, the specific numerical values of $q_1,q_2,q_3,q_4$ can be chosen in order to magnify the differences in symmetry between specific pairs of trees. For instance, \citet{Coronado2019} suggest to take $q_i=i$, $q_i=2^i$, the Sackin index $q_i=S(Q_i^\ast)$, or the total cophenetic index $q_i=\Phi(Q_i^\ast)$.\\
    If only binary trees $T \in \BTnstar$ are considered, \citet{Coronado2019} suggest the following alternative rooted quartet index that simply counts the number of rooted quartets of shape $Q_3^\ast$: \[ rQIB(T) \coloneqq \frac{1}{q_3}\cdot rQI(T)=|\{Q\in\mathcal{Q}(T):Q\text{ has shape }Q_3^\ast\}| = \sum\limits_{v\in\mathring{V}(T)} \binom{n_{v_1}}{2}\cdot\binom{n_{v_2}}{2},\] where $v_1$ and $v_2$ denote the children of $v$ and the last equality has been proven in \cite[Corollary 1]{Coronado2019}. Note that $rQIB(T)=rQI(T)$ for $q_3=1$.\footnote{For the sake of generality, we do not consider $rQIB$ any further, but state the results of \citet{Coronado2019} for general $q_3$.}}
    \item \props{Computation time}{(\citet[Proposition 1]{Coronado2019})}{For every tree $T\in\Tnstar$, the rooted quartet index $rQI(T)$ can be computed in time $O(n)$.}
    \item \props{Recursiveness}{(\citet[Lemma 7, Lemma 8]{Coronado2019})}{The rooted quartet index is a recursive tree shape statistic. We have $rQI(T)=0$ for $T\in\mathcal{T}_1^\ast$ and for every tree $T\in\Tnstar$ with $n\geq 2$ and standard decomposition $T=(T_1,\ldots,T_k)$ we have 
    \begin{equation} \label{rQI_recursion}
    \begin{split}
        rQI(T) &= \sum\limits_{i=1}^k rQI(T_i)  + q_4 \cdot \sum\limits_{1 \leq i_1 < i_2 < i_3 < i_4 \leq k} n_{i_1} n_{i_2} n_{i_3} n_{i_4} \\ \nonumber
        &\qquad + q_3 \cdot \sum\limits_{1 \leq i_1 < i_2 \leq k} \binom{n_{i_1}}{2} \binom{n_{i_2}}{2} + q_2 \cdot \sum\limits_{1 \leq i_1 < i_2 \leq k} \left( n_{i_1} \Upsilon(T_{i_2}) + n_{i_2} \Upsilon(T_{i_1}) \right) \\ \nonumber
        &\qquad + q_1 \cdot \sum\limits_{1 \leq i_1 < i_2 < i_3 \leq k} \left( \binom{n_{i_1}}{2} n_{i_2} n_{i_3} + \binom{n_{i_2}}{2} n_{i_1} n_{i_3} + \binom{n_{i_3}}{2} n_{i_1} n_{i_2} \right)
    \end{split}
    \end{equation}
    with $\Upsilon(T)$ denoting the number of strictly non-binary triples of $T$ (i.e. of restrictions of $T$ to sets of three leaves that are rooted star trees on three leaves). Note that $\Upsilon(T)=0$ if $n=1$ and for $n \geq 2$ with $T=(T_1, \ldots, T_k)$ we have \[ \Upsilon(T) = \sum\limits_{i=1}^k \Upsilon(T_i) + \sum\limits_{1 \leq i_1 < i_2 < i_3 \leq k} n_{i_1} n_{i_2} n_{i_3}. \]
    If only binary trees $T\in\BTnstar$ are considered, the recursion simplifies to $rQI(T)=0$ if $n=1$ and for $n\geq 2$ \[ rQI(T)=rQI(T_1)+rQI(T_2)+q_3\cdot\binom{n_1}{2}\cdot\binom{n_2}{2}. \]}
    \item \props{Locality}{(this manuscript, see Proposition \ref{locality_rQuartet_a} and Proposition \ref{locality_rQuartet_b})}{The rooted quartet index for arbitrary trees is not local. If only binary trees are considered, the rooted quartet index is local.}
    \item \props{Maximal value on $\Tnstar$}{(\citet[Theorem 1]{Coronado2019})}{For every tree $T\in \Tnstar$, the rooted quartet index fulfills $rQI(T)\leq q_4\cdot\binom{n}{4}$. This bound is tight for all $n\in\mathbb{N}_{\geq 1}$.}
    \item \props{Maximal value on $\BTnstar$}{(consequence of \citet[Lemma 11]{Coronado2019})}{\textit{Open problem.} For every $n\in\mathbb{N}_{\geq 1}$, let $a(n)$ denote the maximal value of the rooted quartet index on $\BTnstar$. Then, $a(n)$ fulfills the recursion $a(1)=a(2)=a(3)=0$ and for $n\geq 4$ \[ a(n)=a(\lceil n/2 \rceil)+a(\lfloor n/2 \rfloor)+q_3\cdot\binom{\lceil n/2 \rceil}{2}\cdot\binom{\lfloor n/2 \rfloor}{2}. \] A closed formula is not known yet.}
    \item \props{(Number of) trees with maximal value on $\Tnstar$}{(\citet[Theorem 1]{Coronado2019})}{For any given $n\in\mathbb{N}_{\geq 1}\setminus\{3\}$, there is exactly one tree $T\in\Tnstar$ with maximal rooted quartet index, i.e. $rQI(T)=q_4\cdot\binom{n}{4}$, namely the rooted star tree $\Tstar$. For $n=3$, there are two trees with maximal rooted quartet index, namely $T_3^{star}$ and $T_3^{cat}$.}
    \item \props{(Number of) trees with maximal value on $\BTnstar$}{(\citet[Theorem 2]{Coronado2019})}{For any given $n\in\mathbb{N}_{\geq 1}$, there is exactly one binary tree $T\in\BTnstar$ with maximal rooted quartet index, namely the maximally balanced tree $\Tmb$ (with $\Tmb=\Tfb$ if $n$ is a power of two).}
    \item \props{Minimal value on $\Tnstar$ and $\BTnstar$}{(\citet[Theorem 1]{Coronado2019})}{For every tree $T\in\Tnstar$, the rooted quartet index fulfills $rQI(T)\geq 0$. This bound is tight for all $n\in\mathbb{N}_{\geq 1}$ in both the arbitrary and the binary case.}
    \item \props{(Number of) trees with minimal value on $\Tnstar$ and $\BTnstar$}{(\citet[Theorem 1]{Coronado2019})}{For any given $n\in\mathbb{N}_{\geq 1}\setminus\{3\}$ there is exactly one tree $T\in\Tnstar$ with minimal rooted quartet index, i.e. $rQI(T)=0$, namely the caterpillar tree $\Tcat$. For $n=3$, there are two trees with minimal rooted quartet index, namely $T_3^{star}$ and $T_3^{cat}$.\footnote{Note that the fact that there are two minimal trees for $n=3$ strictly speaking implies that the rooted quartet index is only a balance index according to our definition on $\mathcal{T}_{n \in \mathbb{N}_{\geq 1} \setminus \{3\}}^\ast$. However, this is simply due to the fact that a tree with strictly fewer than four leaves cannot induce any quartets. We thus consider it as a valid balance index nonetheless.} On the space $\BTnstar$, the caterpillar tree $\Tcat$ is the unique tree with minimum rooted quartet index for all $n \in \mathbb{N}_{\geq 1}$.}
    \item \props{Expected value under the Yule model}{(consequence of \citet[Corollary 4]{Coronado2019})}{Let $T_n$ be a phylogenetic tree with $n$ leaves sampled under the Yule model. Then, the expected value of $rQI$ of $T_n$ is \[ E_Y(rQI(T_n)) = q_3 \cdot \frac{1}{3} \cdot \binom{n}{4}. \] Moreover, in the limit $E_Y(rQI(T_n))\sim\frac{q_3}{24}\cdot n^4$.}
    \item \props{Variance under the Yule model}{(consequence of \citet[Corollary 7]{Coronado2019})}{Let $T_n$ be a phylogenetic tree with $n$ leaves sampled under the Yule model. Then, the variance of $rQI$ of $T_n$ is \[ V_Y(rQI(T_n)) = q_3^2 \cdot \binom{n}{4}\cdot \frac{5n^4+30n^3+118n^2+408n+630}{33075}. \] Moreover, in the limit $V_Y(rQI(T_n))\sim\frac{q_3^2}{158760}\cdot n^8$.}
    \item \props{Expected value under the uniform model}{(consequence of \citet[Corollary 4]{Coronado2019})}{Let $T_n$ be a phylogenetic tree with $n$ leaves sampled under the uniform model. Then, the expected value of $rQI$ of $T_n$ is \[ E_U(rQI(T_n)) = q_3 \cdot \frac{1}{5} \cdot \binom{n}{4}. \]  Moreover, in the limit $E_U(rQI(T_n))\sim\frac{q_3}{120}\cdot n^4$.}
    \item \props{Variance under the uniform model}{(consequence of \citet[Corollary 7]{Coronado2019})}{Let $T_n$ be a phylogenetic tree with $n$ leaves sampled under the uniform model. Then, the variance of $rQI$ of $T_n$ is \[ V_U(rQI(T_n)) = q_3^2 \cdot \binom{n}{4} \cdot \frac{4(2n-1)(2n+1)(2n+3)(2n+5)}{225225}. \] Moreover, in the limit $V_U(rQI(T_n))\sim\frac{8\cdot q_3^2}{675675}\cdot n^8$.}
    \item \textbf{Comments:} Note that \citet{Coronado2019} established additional and more general results for the mean and variance of $rQI$. In particular, they provide exact formulas for the expected value and variance of $rQI$ for binary phylogenetic trees sampled under Ford's $\alpha$-model \citep{Ford2005} and Aldous' $\beta$-model \citep{Aldous1996}, as well as for arbitrary phylogenetic trees sampled under Chen-Ford-Winkel's $\alpha$-$\gamma$-model \citep{Chen2009}. We refer the reader to \citep{Coronado2019} for further details on this.\\
    Additionally, \citet{Coronado2019} point out that, compared to the Colless index (Section \ref{factsheet_Colless}) and Sackin index (Section \ref{factsheet_Sackin}), the rooted quartet index is also suitable to measure the balance of taxonomic trees, i.e. those trees that have a fixed depth, possibly inner vertices of out-degree 1 and bijectively labeled leaves.
\end{description}

\subsection{\texorpdfstring{$\widehat{s}$}{s}-shape statistic} \label{factsheet_s-shape}

The $\widehat{s}$-shape statistic is an imbalance index on $\BTnstar$, i.e. for a fixed $n\in\mathbb{N}_{\geq 1}$ it increases with decreasing balance of the tree. The logarithm base was originally not stated \citep{blum2006c}, however, it is common to use base 2 in binary (phylogenetic) trees. In any case, all results stated here apply regardless of the logarithm base used. Note, however, that while the $\widehat{s}$-shape statistic is defined for arbitrary trees, it does only satisfy our definition of an imbalance index when restricted to binary trees (as otherwise the caterpillar tree $\Tcat$ is not the unique maximal tree (see Theorem \ref{thm_s_star})). It can be calculated using the function \texttt{sShapeI} from our \textsf{R} package \texttt{treebalance}.

\begin{description}
    \item \props{Definition}{(\citet{blum2006c})}{The $\widehat{s}$-shape statistic $\widehat{s}(T)$ of a tree $T\in\Tnstar$ is defined as \[ \widehat{s}(T) \coloneqq \sum\limits_{v\in\mathring{V}(T)} \log(n_v-1). \]}
    \item \props{Computation time}{(this manuscript, see Proposition \ref{runtime_sShape})}{For every tree $T\in\Tnstar$, the $\widehat{s}$-shape statistic can be computed in time $O(n)$.}
    \item \props{Recursiveness}{(this manuscript, see Proposition \ref{recursiveness_sShape})}{The $\widehat{s}$-shape statistic is a recursive tree shape statistic. We have $\widehat{s}(T)=0$ for $T\in\mathcal{T}_1^\ast$, and for every tree $T\in\Tnstar$ with $n \geq 2$ and standard decomposition $T=(T_1,\ldots,T_k)$ we have \[ \widehat{s}(T)=\sum\limits_{i=1}^k \widehat{s}(T_i)+\log\left(-1+\sum\limits_{i=1}^k n_i\right)=\sum\limits_{i=1}^k \widehat{s}(T_i) + \log(n-1). \]}
    \item \props{Locality}{(this manuscript, see Proposition \ref{locality_sShape})}{The $\widehat{s}$-shape statistic is local.}
    \item \props{Maximal value on $\Tnstar$}{(this manuscript, see Theorem \ref{thm_sstat_max_a})}{For every tree $T\in\Tnstar$, the $\widehat{s}$-shape statistic fulfills $\widehat{s}(T)\leq \log((n-1)!)$. This bound is tight for all $n\in\mathbb{N}_{\geq 1}$.}
    \item \props{(Number of) trees with maximal value on $\Tnstar$}{(this manuscript, see Theorem \ref{thm_sstat_max_a})}{Let $T\in\Tnstar$ be a tree maximizing $\widehat{s}(T)$. Then, $T$ either equals $T_n^\mathit{cat}$ or it can be constructed by contracting the inner edge leading to the only cherry in $T_n^\mathit{cat}$. In particular, for each $n \geq 3$, there exist two distinct maximal trees.}
    \item \props{Maximal value on $\BTnstar$}{(this manuscript, see Theorem \ref{thm_sstat_cat})}{For every binary tree $T\in\BTnstar$, the $\widehat{s}$-shape statistic fulfills $\widehat{s}(T)\leq \log((n-1)!)$. This bound is tight for all $n\in\mathbb{N}_{\geq 1}$.}
    \item \props{(Number of) trees with maximal value on $\BTnstar$}{(this manuscript, see Theorem \ref{thm_sstat_cat})}{For any given $n\in\mathbb{N}_{\geq 1}$, there is exactly one rooted binary tree $T\in\BTnstar$ maximizing $\widehat{s}(T)$, namely the caterpillar tree $\Tcat$.}
    \item \props{Minimal value on $\Tnstar$}{(this manuscript, see Theorem \ref{thm_s_star})}{For every tree $T\in\Tnstar$, the $\widehat{s}$-shape statistic fulfills $\widehat{s}(T)\geq\log(n-1)$. This bound is tight for all $n\in\mathbb{N}_{\geq 1}.$}
    \item \props{(Number of trees) with minimal value on $\Tnstar$}{(this manuscript, see Corollary \ref{cor_s_star})}{For any given $n \in \mathbb{N}_{\geq 3} $, there are precisely $1 + \lfloor \frac{n}{2} \rfloor$ trees $T \in \Tnstar$ minimizing $\widehat{s}(T)$, namely the star tree $\Tstar$ and every tree $T \in \Tnstar$ with the property that all inner vertices of $T$ other than the root are parents of binary cherries. For $n \in \{1,2\}$, there is precisely one minimal tree.}
    \item \props{Minimal value on $\BTnstar$}{(this manuscript, see Theorem \ref{thm_s-minPower2})}{\textit{Open problem.} Let $n=2^h$ for some $h\in\mathbb{N}_{\geq 0}$. Then, for every tree $T\in\BTnstar$, the $\widehat{s}$-shape statistic fulfills \[ \widehat{s}(T)\geq \sum\limits_{i=0}^{h-1} 2^i\cdot\log(2^{h-i}-1). \] This bound is tight. For leaf numbers $n$ that are not powers of two, the minimal value is not known yet.}
    \item \props{(Number of) trees with minimal value on $\BTnstar$}{(this manuscript, see Theorem \ref{thm_s-minPower2})}{\textit{Open problem.} Let $n=2^h$ for some $h\in\mathbb{N}_{\geq 0}$. Then, there is exactly one rooted binary tree $T\in\BTnstar$ with minimal $\widehat{s}$-shape statistic, i.e. $\widehat{s}(T)=\sum\limits_{i=0}^{h-1} 2^i\cdot\log_2(2^{h-i}-1)$, namely the fully balanced tree $\Tfb$. For leaf numbers $n$ that are not powers of two, neither a characterization nor the number of minimal trees is known yet.}
    \item \props{Expected value under the Yule model}{}{\textit{Open problem.}}
    \item \props{Variance under the Yule model}{}{\textit{Open problem.}}
    \item \props{Expected value under the uniform model}{}{\textit{Open problem.}}
    \item \props{Variance under the uniform model}{}{\textit{Open problem.}}
    \item \textbf{Comments:} The $\widehat{s}$-shape statistic provides maximal power compared to other tree shape statistics for rejecting the ERM against the PDA in statistical tests \citep{blum2006c, Bortolussi2005}. Also note that a tree shape statistic related to the $\widehat{s}$-shape statistic has been studied in the literature for so-called \emph{binary search trees}, where similar results to the ones given here were obtained \citep{Fill1996}. 
\end{description}

\subsection{Sackin index} \label{factsheet_Sackin}

Originating from Sackin's idea to analyze the leaf depths of a tree \citep{Sackin1972}, the Sackin index is one of the oldest and  most widely applied (im)balance indices. It is defined for arbitrary rooted trees and it is an imbalance index, i.e. for a fixed $n\in\mathbb{N}_{\geq 1}$ it increases with decreasing balance of the tree. Since the possible range of values enlarges and the minimal value increases with $n$, it is only meaningful to compare the Sackin index of two trees if they have the same number of leaves. It can be calculated using the function \texttt{sackinI} from our \textsf{R} package \texttt{treebalance}.

\begin{description}
    \item \props{Definition}{(\citet{Shao1990, Fischer2021})}{The Sackin index $S(T)$ of a tree $T\in\Tnstar$ is defined as\footnote{The equivalence of these three definitions is shown in \citep[Lemma 1]{Fischer2021}.} \[ S(T)\coloneqq \sum\limits_{v\in V_L(T)} \delta_T(v) \ \ =\sum\limits_{v\in\mathring{V}(T)} n_v \ \ =\sum\limits_{v\in V(T)\setminus \{\rho\}} n_v. \]}
    \item \props{Computation time}{(this manuscript, see Proposition \ref{runtime_Sackin})}{For every tree $T\in\Tnstar$, the Sackin index $S(T)$ can be computed in time $O(n)$.}
    \item \props{Recursiveness}{(this manuscript, see Proposition \ref{recursiveness_Sackin})}{The Sackin index is a recursive tree shape statistic. We have $S(T)=0$ for $T\in\mathcal{T}_1^\ast$, and for every tree $T\in\Tnstar$ with $n \geq 2$ and standard decomposition $T=(T_1,\ldots,T_k)$ we have \[ S(T)=\sum\limits_{i=1}^k S(T_i)+\sum\limits_{i=1}^k n_i. \]}
    \item \props{Locality}{(this manuscript, see Proposition \ref{locality_Sackin})}{The Sackin index is local.}
    \item \props{Maximal value on $\Tnstar$ and $\BTnstar$}{(this manuscript, see Proposition \ref{max_Sackin_a_2}; \citet[Theorem 1, Theorem 4]{Fischer2021})}{For every tree $T\in\Tnstar$ with $m$ inner vertices, the Sackin index fulfills $S(T)\leq nm-\frac{(m-1)\cdot m}{2}$. This bound is tight for $n=1$ and $m=0$ and for all $n\in\mathbb{N}_{\geq 2}$ and $m\in\mathbb{N}_{\geq 1}$ with $m\leq n-1$. Moreover, we have $S(T)\leq\frac{n\cdot(n+1)}{2}-1=S(\Tcat)$.}
    \item \props{(Number of) trees with maximal value on $\Tnstar$ and $\BTnstar$}{(this manuscript, see Theorem \ref{max_Sackin_a_1} and Proposition \ref{max_Sackin_a_2}; \citet[Theorem 1, Theorem 4]{Fischer2021})}{For $n=1$ and $m=0$ or any given $n\in\mathbb{N}_{\geq 2}$ and $m\in\mathbb{N}_{\geq 1}$ with $m\leq n-1$, there is exactly one tree $T\in\Tnstar$ with $m$ inner vertices and maximal Sackin index, i.e. $S(T)=nm-\frac{(m-1)\cdot m}{2}$, namely the caterpillar tree on $m+1$ leaves that has $n-m-1$ additional leaves attached to the inner vertex with the largest depth. In particular, $S(T)<S(\Tcat)$ if $T\neq \Tcat$, i.e. for $n\geq 1$ the \emph{binary} caterpillar tree $\Tcat$ is the unique maximal tree on $\BTnstar$ and on $\Tnstar$ if $m$ is not fixed.}
    \item \props{Minimal value on $\Tnstar$ and $\BTnstar$}{(this manuscript, see Theorem \ref{min_Sackin_a}; \citet[Theorem 5]{Fischer2021}}{For every tree $T\in\Tnstar$ with $m$ inner vertices the Sackin index fulfills $S(T)=0$ if $n=1$ and $m=0$, and otherwise \[ S(T)\geq \left\lfloor\log_2\left(\frac{n}{k}\right)\right\rfloor\cdot n+3n-k\cdot 2^{\lfloor\log_2(\frac{n}{k})\rfloor+1} \] with $k=n-m+1$. The latter bound is tight for all $n\in\mathbb{N}_{\geq 2}$ and $m\in\mathbb{N}_{\geq 1}$ with $m\leq n-1$. Moreover, we have $S(T)>n=S(\Tstar)$ if $T\neq\Tstar$.}
    \item \props{Trees with minimal value on $\Tnstar$ and $\BTnstar$}{(this manuscript, see Lemma \ref{lem_Sackin_min_3}, Theorem \ref{min_Sackin_a}; \citet[Theorem 5]{Fischer2021})}{If $n=1$ and $m=0$, there is precisely one tree with minimal Sackin index, namely $T_1^\mathit{star}$. For any given $n\in\mathbb{N}_{\geq 2}$ and $m\in\mathbb{N}_{\geq 1}$ with $m\leq n-1$, a tree $T$ with $n$ leaves and $m$ inner vertices has minimal Sackin index if and only if it has $k=n-m+1$ maximal pending subtrees rooted in the children of the root $\rho$ and fulfills $|\delta_{T}(x)-\delta_{T}(y)|\leq 1$ for all $x,y\in V_L(T)$. Moreover, if $m$ is not fixed, the star tree $\Tstar$ is the unique tree minimizing the Sackin index on $\Tnstar$.}
    \item \props{Number of trees with minimal value on $\Tnstar$}{(this manuscript, see Corollary \ref{nummintrees_Sackin}; \citet[Theorem 5]{Fischer2021})}{Let $n \in \mathbb{N}_{\geq1}$ and $m \in \mathbb{N}_{\geq0}$, let $s(n,m)$ denote the number of Sackin minimal trees with $n$ leaves and $m$ inner vertices, and let $k=n-m+1$. Also, denote by $\mathcal{P}_k(n)$ the set of all sets of pairs $\{(a_1,\widetilde{n}_1),\ldots,(a_l,\widetilde{n}_l)\}$, where $a_i, \widetilde{n}_i\in\mathbb{N}_{\geq 1}$ are integers such that $\widetilde{n}_i\neq \widetilde{n}_j$ if $i\neq j$ and $2^{\delta-1}\leq \widetilde{n}_i\leq 2^{\delta}$ for $\delta=\lfloor\log_2\left(\frac{n}{k}\right)\rfloor+1$, and $a_1+\ldots+a_l=k$ and $a_1\cdot\widetilde{n}_1+\ldots+a_l\cdot\widetilde{n}_l=n$, i.e. each element in $\mathcal{P}_k(n)$ represents a specific unique integer partition of $n$. Then, we have $s(n,m)=0$ if $m>n-1$ or $m=0$ and $n>1$, $s(1,0)=1$ and otherwise: \[ s(n,m)=\sum\limits_{ \substack{\{(a_1,\widetilde{n}_1),\ldots,(a_l,\widetilde{n}_l)\} \\ \in \mathcal{P}_k(n)} } \ \prod\limits_{i=1}^l \binom{s(\widetilde{n}_i,\widetilde{n}_i-1)+a_i-1}{a_i}, \] where $s(n,n-1)$ corresponds to the number of Sackin minimal rooted binary trees with $n$ leaves, which can be calculated by the formula presented in \cite[Theorem 3]{Fischer2021} (see also Online Encyclopedia of Integer Sequences \cite[Sequence A299037]{OEIS} and next paragraph in this fact sheet).\\
    If the number of inner vertices is not fixed, the star tree $\Tstar$ is the unique tree minimizing the Sackin index.}
    \item \props{Number of trees with minimal value on $\BTnstar$}{(\citet[Theorem 3, Corollary 1]{Fischer2021})}{\textit{Open problem.} Let $\widehat{s}(n)$ denote the number of binary trees with $n$ leaves that have minimal Sackin index. Let $A(n)$ denote the set of pairs $A(n)=\{(n_a,n_b)|n_a,n_b\in\mathbb{N}_{\geq 1},n_a+n_b=n, \frac{n}{2}<n_a\leq 2^{\lceil\log_2(n)\rceil-1},n_b\geq2^{\lceil\log_2(n)\rceil-2}\}$, let \[ f(n)=\begin{cases} 0 & \text{if } n \text{ is odd} \\ \binom{\widehat{s}\left(\frac{n}{2}\right)+1}{2} & \text{if } n \text{ is even.} \end{cases} \] Then, $\widehat{s}(n)$ fulfills the recursion $\widehat{s}(1)=1$ and for $n\geq 2$ \[ \widehat{s}(n)=\sum\limits_{(n_a,n_b)\in A(n)} \widehat{s}(n_a)\cdot\widehat{s}(n_b)+f(n). \] A closed formula is not known yet. If $n\in\{2^m-1,2^m,2^m+1\}$ for some $m\in\mathbb{N}_{\geq 1}$, there is exactly one tree in $\BTnstar$ with minimal Sackin index (in particular, if $n$ is a power of two, the fully balanced tree $\Tfb$ is the unique minimal tree). For all other $n$, there exist at least two trees in $\BTnstar$ with minimal Sackin index.}
    \item \props{Expected value under the Yule model}{(\citet[Appendix]{Kirkpatrick1993}, \citet[p. 2198]{Blum2006a}, \citet[Remark 1]{Coronado2020b})}{Let $T_n$ be a phylogenetic tree with $n$ leaves sampled under the Yule model. Then, the expected value of $S$ of $T_n$ is \[ E_Y(S(T_n))=2n\cdot(H_n-1). \] Moreover, in the limit $E_Y(S(T_n))\sim 2n\cdot\ln(n)$}
    \item \props{Variance under the Yule model}{(\citet[Corollary 1]{Cardona2012}, \citet[p. 2198]{Blum2006a})}{Let $T_n$ be a phylogenetic tree with $n$ leaves sampled under the Yule model. Then, the variance of $S$ of $T_n$ is \[ V_Y(S(T_n))=7n^2-4n^2\cdot H_n^{(2)}-2n\cdot H_n-n. \] Moreover, in the limit we have  \[ V_Y(S(T_n)) \sim \left(7-\frac{2\pi^2}{3}\right)\cdot n^2. \]}
    \item \props{Expected value under the uniform model}{(\citet[Theorem 22]{Mir2013}; \citet[p. 2199]{Blum2006a})}{Let $T_n$ be a phylogenetic tree with $n$ leaves sampled under the uniform model. Then, the expected value of $S$ of $T_n$ is \[ E_U(S(T_n))=n\cdot\left(\frac{(2n-2)!!}{(2n-3)!!}-1\right). \] Moreover, in the limit \[ E_U(S(T_n)) \sim \sqrt{\pi}\cdot n^{3/2}. \]}
    \item \props{Variance under the uniform model}{(\citet[Theorem 6]{Coronado2020b}; \citet[Remark 3]{Blum2006a})}{Let $T_n$ be a phylogenetic tree with $n$ leaves sampled under the uniform model. Then, the variance of $S$ of $T_n$ is \[ V_U(S(T_n)) = \frac{n\cdot(10n^2 -3n -1)}{3} - \binom{n+1}{2} \cdot \frac{(2n-2)!!}{(2n-3)!!} - n^2 \cdot \left( \frac{(2n-2)!!}{(2n-3)!!} \right)^2. \] Moreover, in the limit \[ V_U(S(T_n)) \sim \left( \frac{10}{3} - \pi \right)\cdot n^3. \]}
    \item \textbf{Comments:} Despite the name of this index, \citet{Sackin1972} originally only considered the maximum and the \enquote{variation} of the leaf depths (rather than their sum), both of which he found to be larger in the caterpillar tree than in the fully balanced tree with the same number of leaves $n$. The definition as the summarized leaf depths was later introduced by \citet{Shao1990}.\\
    Also note that in the literature the term \enquote{Sackin index} or \enquote{Sackin's index} sometimes refers to the average leaf depth (Section \ref{factsheet_ALD}), which in fact is a normalization of the Sackin index.
\end{description}

\subsection{Symmetry nodes index} \label{factsheet_SNI}

This index only applies to binary trees. It uses the  already known concept of symmetry nodes to assess the degree of tree imbalance by counting the number of vertices that are \emph{not} symmetry vertices. As such it is an imbalance index, i.e. for a fixed $n\in\mathbb{N}_{\geq 1}$ it increases with decreasing balance of the tree. It can be calculated using the function \texttt{symNodesI} from our \textsf{R} package \texttt{treebalance}.

\begin{description}
    \item \props{Definition}{(\citet{Kersting2021})}{The symmetry nodes index $SNI(T)$ of a binary tree $T \in \BTnstar$ is defined as $SNI(T) \coloneqq n-1-s(T)$
    with $s(T)$ being the number of symmetry nodes in $T$.}
    \item \props{Computation time}{(\citet[Theorem 3.1]{Kersting2021})}{For every binary tree $T\in\BTnstar$, the symmetry nodes index $SNI(T)$ can be computed in time $O(n)$.}
    \item \props{Recursiveness}{(this manuscript, see Proposition \ref{recursiveness_Sym})}{The symmetry nodes index is a binary recursive tree shape statistic. We have $SNI(T)=0$ for $T\in\mathcal{BT}_1^\ast$, and for every tree $T\in\BTnstar$ with $n\geq 2$ and standard decomposition $T=(T_1,T_2)$ we have 
    \[ SNI(T) = SNI(T_1) + SNI(T_2) + \left(1-\mathcal{I}(CP(T_1)=CP(T_2))\right),\]
    where $CP(T_i)$ is the Colijn-Plazotta rank of $T_i$ \citep{Colijn2018}. Note that the Colijn-Plazzotta rank might be replaced by any other bijective map between the set of rooted binary trees and a set of real numbers that is itself a binary recursive tree shape statistic (for example the Furnas rank (Section \ref{factsheet_Furnas})).} 
    \item \props{Locality}{(this manuscript, see Proposition \ref{locality_Sym})}{The symmetry nodes index is not local.}
    \item \props{Maximal value on $\BTnstar$}{(\citet[Theorem 3.3]{Kersting2021})}{For every binary tree $T\in\BTnstar$ with $n \geq 2$ leaves, the symmetry nodes index fulfills $SNI(T)\leq n-2$. This bound is tight for all $n\in\mathbb{N}_{\geq 2}$. For $n=1$, we have $SNI(T)=0$.}
    \item \props{(Number of) trees with maximal value on $\BTnstar$}{(\citet[Theorem 3.3]{Kersting2021})}{For any given $n\in\mathbb{N}_{\geq 1}$, there is exactly one binary tree $T\in\BTnstar$ with maximal symmetry nodes index, namely the caterpillar tree $\Tcat$.}
    \item \props{Minimal value on $\BTnstar$}{(\citet[Theorem 3.5]{Kersting2021})}{For every binary tree $T\in\BTnstar$ with $n \geq 1$ leaves, the symmetry nodes index fulfills $SNI(T)\geq wt(n)-1$ with the binary weight $wt(n)$ being the number of 1's in the binary expansion of $n$. This bound is tight for all $n\in\mathbb{N}_{\geq 1}$.}
    \item \props{(Number of) trees with minimal value on $\BTnstar$}{(\citet[Theorem 3.5, Theorem 3.9]{Kersting2021})}{For any given $n\in\mathbb{N}_{\geq 1}$, the minimal value of the symmetry nodes index is reached precisely by the rooted binary weight trees $\widetilde{T}_n$. The number of minimal binary trees is thus $|\mathcal{BT}_{wt(n)}|$ with $|\mathcal{BT}_{wt(n)}|=1$ for $wt(n)=1$ and $|\mathcal{BT}_{wt(n)}|=(2\cdot wt(n) - 3)!!$ else.}
    \item \textbf{Expected value under the Yule model:} \textit{Open problem.}
    \item \textbf{Variance under the Yule model:} \textit{Open problem.}
    \item \props{Expected value and variance under the uniform model}{(this manuscript, see Proposition \ref{prop_sym_nodes_Unif} and Corollary \ref{Cor_sym_uniform})}{\textit{Open problem.} Let $T_n$ be a phylogenetic tree with $n$ leaves sampled under the uniform model. Exact but long recursions have been stated for the expected value und variance of $SNI$ of $T_n$, but there is a nearly perfect linear correlation between the number of leaves $n$ and both the expected value as well as the variance of  $SNI(T_n)$: For $n\geq 10$ we have the approximation
    \[E_U(SNI(T_n))\sim 0.72914 \cdot n-1.18545 \quad \text{ and } \quad  V_U(SNI(T_n))\sim 0.10491 \cdot n+0.02853\]
    based on exact values for $n=10,...,140$ as well as the following formulas based on $n=60,...,140$
    \[E_U(SNI(T_n))\sim 0.729 \cdot n-1.17208 \quad \text{ and } \quad  V_U(SNI(T_n))\sim 0.10494 \cdot n+0.02595.\]
   The second approximation should be prefered over the first for $n$ higher than $\approx 40$.} For $n<10$ we have exactly $E_U(SNI(T_{1}))=0$, $E_U(SNI(T_{2}))=0$, $E_U(SNI(T_{3}))=1$, $E_U(SNI(T_4)) =1.6$, $E_U(SNI(T_5)) \approx 2.43$, $E_U(SNI(T_6)) \approx 3.10$, $E_U(SNI(T_7)) \approx 3.88$, $E_U(SNI(T_8))$ $\approx 4.59$ and $E_U(SNI(T_9)) \approx 5.34$ as well as $V_U(SNI(T_{1,2,3})) = 0$, $V_U(SNI(T_4)) = 0.64$ and $V_U(SNI(T_5)) \approx0.53$, $V_U(SNI(T_6)) \approx 0.75 $, $V_U(SNI(T_7)) \approx 0.77 $, $V_U(SNI(T_8)) \approx 0.90 $ and  $V_U(SNI(T_9)) \approx 0.98$. A closed formula is not known yet.
    \item \textbf{Comments:} The symmetry nodes index is related to the Rogers $J$ index (Section \ref{factsheet_RogersU}), because a symmetry node is also a perfectly balanced node, as well as to the cherry index (Section \ref{factsheet_cherry}), because every cherry implies a symmetry node.
\end{description}

\subsection{Total cophenetic index} \label{factsheet_TCI}

This index is defined for arbitrary rooted trees. It is an imbalance index, i.e. for a fixed $n\in\mathbb{N}_{\geq 1}$ it increases with decreasing balance of the tree. While the Sackin index (Section \ref{factsheet_Sackin}) adds up the depths of the leaves, i.e. the depths of the lowest common ancestor of every leaf and itself, the total cophenetic index adds up the depths of the lowest common ancestor of every pair of different leaves \citep{Mir2013}. It can be calculated using the function \texttt{totCophI} from our \textsf{R} package \texttt{treebalance}.

\begin{description}
    \item \props{Definition}{(\citet[Definition 1, Lemma 2]{Mir2013})}{The total cophenetic index $\Phi(T)$ of a tree $T\in\Tnstar$ is defined as \[ \Phi(T)\coloneqq\sum\limits_{\substack{\{i,j\} \in V_L(T)^2 \\ i \neq j}} \varphi_T(i,j)=\sum\limits_{\substack{\{i,j\}\in V_L(T)^2\\ i\neq j}} \delta_T(LCA_T(i,j))=\sum\limits_{v\in\mathring{V}(T)\setminus\{\rho\}} \binom{n_v}{2}. \]}
    \item \props{Computation time}{(\citet[Lemma 2, Corollary 3]{Mir2013})}{For every tree $T\in\Tnstar$, the total cophenetic index $\Phi(T)$ can be computed in time $O(n)$ by using the alternative expression $\Phi(T)=\sum\limits_{v\in\mathring{V}(T)\setminus\{\rho\}} \binom{n_v}{2}$. }
    \item \props{Recursiveness}{(\citet[Lemma 4]{Mir2013})}{The total cophenetic index is a recursive tree shape statistic. We have $\Phi(T)=0$ for $T\in\mathcal{T}_1^\ast$, and for every tree $T\in\Tnstar$ with $n\geq 2$ and standard decomposition $T=(T_1,\ldots,T_k)$ we have \[ \Phi(T)=\sum\limits_{i=1}^k \Phi(T_i)+\sum\limits_{i=1}^k \binom{n_i}{2}. \]}
    \item \props{Locality}{(\citet[Lemma 5]{Mir2013})}{The total cophenetic index is local.}
    \item \props{Maximal value on $\Tnstar$ and $\BTnstar$}{(\citet[Proposition 10]{Mir2013})}{For every tree $T\in\Tnstar$, the total cophenetic index fulfills $\Phi(T)\leq\binom{n}{3}$. This bound is tight for all $n\in\mathbb{N}_{\geq 1}$ in both the arbitrary and binary case}.
    \item \props{(Number of) trees with maximal value on $\Tnstar$ and $\BTnstar$}{(\citet[Proposition 10]{Mir2013})}{For any given $n\in\mathbb{N}_{\geq 1}$, there is exactly one tree $T\in\Tnstar$ with maximal total cophenetic index, i.e. $\Phi(T)=\binom{n}{3}$, namely the caterpillar tree $\Tcat$.}
    \item \props{Minimal value on $\Tnstar$}{(\citet[p. 129]{Mir2013})}{For every tree $T\in\Tnstar$, the total cophenetic index fulfills $\Phi(T)\geq 0$. This bound is tight for all $n\in\mathbb{N}_{\geq 1}$.}
    \item \props{Minimal value on $\BTnstar$}{(\citet[Proposition 15]{Mir2013})}{For every binary tree $T\in\BTnstar$, the total cophenetic index fulfills \[ \Phi(T)\geq \sum\limits_{i=0}^{n-1} a(i) \] in which $a(i)$ denotes the highest power of 2 that divides $i!$. This bound is tight for all $n\in\mathbb{N}_{\geq 1}$.}
    \item \props{(Number of) trees with minimal value on $\Tnstar$}{(\citet[p. 129]{Mir2013})}{For any given $n\in\mathbb{N}_{\geq 1}$, there is exactly one tree $T\in\Tnstar$ with minimal total cophenetic index, i.e. $\Phi(T)=0$, namely the rooted star tree $\Tstar$.}
    \item \props{(Number of) trees with minimal value on $\BTnstar$}{(\citet[Theorem 13]{Mir2013})}{For any given $n\in\mathbb{N}_{\geq 1}$, there is exactly one binary tree $T\in\BTnstar$ with minimal total cophenetic index, i.e. $\Phi(T)=\sum\limits_{k=0}^{n-1} a(k)$, namely the maximally balanced tree $\Tmb$ (with $\Tmb=\Tfb$ is $n$ is a power of two).}
    \item \props{Expected value under the Yule model}{(\citet[Theorem 17, Corollary 18]{Mir2013})}{Let $T_n$ be a phylogenetic tree with $n$ leaves sampled under the Yule model. Then, the expected value of $\Phi$ of $T_n$ is \[ E_Y(\Phi(T_n))=n(n+1)-2n\cdot H_n. \] Moreover, in the limit \[ E_Y(\Phi(T_n))\sim n^2+(1-2\gamma)\cdot n-2n\cdot \ln(n) \] with $\gamma$ denoting Euler's constant.}
    \item \props{Variance under the Yule model}{(\citet[Corollary 3, Corollary 4]{Cardona2012}}{Let $T_n$ be a phylogenetic tree with $n$ leaves sampled under the Yule model. Then, the variance of $\Phi$ of $T_n$ is \[ V_Y(\Phi(T_n))= \frac{1}{12}(n^4-10n^3+131n^2-2n) - 4n^2 \cdot H_n^{(2)} -6n\cdot H_n. \] Moreover, in the limit \[ V_Y(\Phi(T_n))\sim \frac{1}{12}n^4-\frac{5}{6}n^3+\left(\frac{131}{12}-\frac{2\pi^2}{3}\right)n^2-6n\cdot\ln(n)+\left(\frac{23}{6}-6\gamma\right)n-5 \] with $\gamma$ denoting Euler's constant.}
    \item \props{Expected value under the uniform model}{(\citet[Theorem 23]{Mir2013})}{Let $T_n$ be a phylogenetic tree with $n$ leaves sampled under the uniform model. Then, the expected value of $\Phi$ of $T_n$ is \[ E_U(\Phi(T_n))=\frac{1}{2}\cdot\binom{n}{2}\cdot\left(\frac{(2n-2)!!}{(2n-3)!!}-2\right). \] Moreover, in the limit \[ E_U(\Phi(T_n)) \sim \frac{\sqrt{\pi}}{4}\cdot n^{5/2}. \]}
    \item \props{Variance under the uniform model}{(\citet[Theorem 6]{Coronado2020b})}{Let $T_n$ be a phylogenetic tree with $n$ leaves sampled under the uniform model. Then, the variance of $\Phi$ of $T_n$ is \[ V_U(\Phi(T_n)) = \binom{n}{2}\cdot\frac{(2n-1)(7n^2 -3n -2)}{30} - \binom{n}{2}\cdot\frac{5n^2-n-2}{32} \cdot \frac{(2n-2)!!}{(2n-3)!!} - \frac{1}{4} \binom{n}{2}^2\cdot\left( \frac{(2n-2)!!}{(2n-3)!!} \right)^2. \]
    Moreover, in the limit 
    \[ V_U(\Phi(T_n)) \sim \frac{56-15\pi}{240} \cdot n^5. \]}
    \item \textbf{Comments:} \citet{Mir2013} showed that the total cophenetic index has a larger range of values than the Sackin or Colless index. In a simulation study, they also found that compared to the latter two, the total cophenetic index has the lowest probability of a tie (two binary phylogenetic trees $T,T'\in\BTn$ having the same index).
\end{description}

\subsection{Variance of leaf depths} \label{factsheet_VLD}

The variance of leaf depths is an imbalance index defined for arbitrary (i.e. not necessarily binary) rooted trees. Originally proposed by \citet{Sackin1972} and later implemented by \citet{Kirkpatrick1993}, this index did not receive much attention in the literature until very recently when a mathematical analysis of the variance of the leaves' depth was performed by \citet{Coronado2020b}. Despite \citet{Sackin1972} being usually cited as the source for the Sackin index $S(T)$, i.e. the sum of leaf depths, the variance of leaf depths is in fact closer to the original intention of the author as he proposed considering the maximum and the \enquote{variation} of the leaf depths rather then their sum. It can be calculated using the function \texttt{varLeafDepI} from our \textsf{R} package \texttt{treebalance}.

\begin{description}
    \item \props{Definition}{(\citet{Sackin1972, Coronado2020b})}{The variance of leaf depths $\sigma_N^2(T)$ of a tree $T\in\Tnstar$ is defined as \[ \sigma_N^2(T) \coloneqq \frac{1}n \cdot \sum\limits_{x \in V_L(T)} \left( \delta_T(x) - \overline{N}(T) \right)^2, \] where $\overline{N}(T)$ denotes the average leaf depth of $T$. Setting \[ S^{(2)}(T) \coloneqq \sum\limits_{x \in V_L(T)}  \delta_T(x)^2 \]
    the variance of leaf depths $\sigma_N^2(T)$ can alternatively be expressed as 
    \[ \sigma_N^2(T) = \frac{1}{n}\cdot S^{(2)}(T)-\overline{N}(T)^2 = \frac{1}{n}\cdot S^{(2)}(T)-\frac{1}{n^2}\cdot S(T)^2. \]}
    \item \props{Computation time}{(this manuscript, see Proposition \ref{runtime_VLD})}{For every tree $T\in\Tnstar$, the variance of leaf depths $\sigma_N^2(T)$ can be computed in time $O(n)$.}
    \item \props{Recursiveness}{(this manuscript, see Proposition \ref{recursiveness_VLD})}{The variance of leaf depths is a recursive tree shape statistic. We have $\sigma_N^2(T)=0$ for $T\in\mathcal{T}_1^\ast$, and for every tree $T \in \Tnstar$ with $n\geq 2$ and standard decomposition $T=(T_1, \ldots, T_k)$, we have \[ \sigma_N^2(T) = \left(\sum\limits_{i=1}^k n_i\right)^{-1} \cdot \left( \sum\limits_{i=1}^k S^{(2)}(T_i) + 2\cdot\sum\limits_{i=1}^k S(T_i)\right) + 1 - \left(\sum\limits_{i=1}^k n_i\right)^{-2} \cdot \left( \sum\limits_{i=1}^k  S(T_i) + \sum\limits_{i=1}^k n_i \right)^2. \]}
    \item \props{Locality}{(this manuscript, see Proposition \ref{locality_VLD})}{The variance of leaf depths is not local.}
    \item \props{Maximal value on $\Tnstar$ and $\BTnstar$}{(\citet[Theorem 1]{Coronado2020b})}{For every tree $T \in \Tnstar$, the variance of leaf depths fulfills \[ \sigma_N^2(T) \leq \frac{(n-1)(n-2)(n^2+3n-6)}{12n^2}. \] This bound is tight for all $n \in \mathbb{N}_{\geq 1}$ in both the arbitrary and the binary case.}
    \item \props{(Number of) trees with maximal value on $\Tnstar$ and $\BTnstar$}{(\citet[Theorem 1]{Coronado2020b})}{For any given $n \in \mathbb{N}_{\geq 1}$, there is exactly one tree $T \in \Tnstar$ with maximal variance of leaf depths, i.e. $\sigma_N^2(T)=\frac{(n-1)(n-2)(n^2+3n-6)}{12n^2}$, namely the caterpillar tree $\Tcat$. As $\Tcat$ is binary, it is also the unique tree in $\BTnstar$ that maximizes $\sigma_N^2(T)$.}
    \item \props{Minimal value on $\Tnstar$}{(\citet[p. 6]{Coronado2020b})}{For every tree $T \in \Tnstar$, the variance of leaf depths fulfills $\sigma_N^2(T) \geq 0$. This bound is tight for all $n \in \mathbb{N}_{\geq 1}$.}
    \item \textbf{Minimal value on $\BTnstar$:} \textit{Open problem.} \citet{Coronado2020b} provide an algorithm to compute the minimum value of $\sigma_N^2$ on $\BTnstar$ for any $n\in\mathbb{N}_{\geq 1}$ in time $O(n\cdot\log(n))$. Note that for $n \leq 183$, we have in particular \[ \sigma_N^2(T)\geq\frac{2b\cdot(2^a-b)}{(2^a+b)^2} \] with $a=\lfloor\log_2(n)\rfloor$ and $b=n-2^a$ (see Remark 3 in \cite{Coronado2020b}), and if $n=2^h$ for some $h\in\mathbb{N}_{\geq 0}$ we have $\sigma_N^2(T)\geq 0$ (see \cite[p. 6]{Coronado2020b}). Both bounds are tight for all respective $n$ (because they are reached for $n\leq 183$ by $\Tmb$ and for $n=2^h$ by $\Tfb$, see \cite[p. 6]{Coronado2020b}). If $n\geq 184$ and not a power of $2$, an explicit formula for the minimal value of $\sigma_N^2$ on $\BTnstar$ has not yet been found (see the conclusion section in \cite{Coronado2020b}).
    \item \props{Trees with minimal value on $\Tnstar$}{(\citet[p. 6]{Coronado2020b})}{For each $n \in \mathbb{N}_{\geq 1}$ the minimum value of $\sigma_N^2(T)$ on $\Tnstar$ is reached precisely at those rooted trees, all of whose leaves have the same depth (such trees are sometimes called taxonomic trees). This includes in particular the star tree $\Tstar$ and (provided that $n$ is a power of 2) the fully balanced tree $\Tfb$.}
    \item \textbf{Trees with minimal value on $\BTnstar$:} \textit{Open problem.} \citet{Coronado2020b} derived a set of necessary conditions for a tree $T\in\BTnstar$ that minimizes $\sigma_N^2$ (\citet[Theorem 2]{Coronado2020b}) and provided an algorithm that uses these conditions to compute binary trees with minimal $\sigma_N^2$ for any given $n\in\mathbb{N}_{\geq 1}$ in time $O(n\cdot\log(n))$. However, a full characterization has yet to be obtained (see the conclusion section in \cite{Coronado2020b}). Note that for all $n\leq 183$, the maximally balanced tree is among those trees with minimum $\sigma_N^2$, while this property fails for almost every $n\geq 184$ \citep{Coronado2020b}. Moreover, for $n=2^h$ for some $h\in\mathbb{N}_{\geq 0}$, the minimum value of $\sigma_N^2(T)$ is zero and is achieved precisely at the fully balanced tree $\Tfb$ (since $\Tfb$ is the unique rooted binary tree with $n=2^h$ leaves such that all leaves have the same depth).
    \item \textbf{Number of trees with minimal value on $\Tnstar$:} \textit{Open problem.}
    \item \textbf{Number of trees with minimal value on $\BTnstar$:} \textit{Open problem.} If $n$ is a power of two, there is precisely one tree $T\in\BTnstar$ with minimal $\sigma_N^2$, namely the fully balanced tree $\Tfb$. For all other leaf numbers the trees in $\BTnstar$ that minimize $\sigma_N^2$ are not fully characterized yet (see also \cite{Coronado2020b}) and could thus not be counted yet.
    \item \props{Expected value under the Yule model}{(\citet[Theorem 4, Remark 1]{Coronado2020b})}{Let $T_n$ be a phylogenetic tree with $n\geq 1$ leaves sampled under the Yule model. Then, the expected value of $\sigma_N^2$ of $T_n$ is \[ E_Y(\sigma_N^2 (T_n)) = \frac{2(n+1)}{n} \cdot H_n + \frac{1}{n} - 5. \]
    Moreover, in the limit
    \[ E_Y(\sigma_N^2 (T_n)) \sim 2 \cdot \ln (n). \]}
    \item \textbf{Variance under the Yule model:} \textit{Open problem.}
    \item \props{Expected value under the uniform model}{(\citet[Theorem 5, Remark 2]{Coronado2020b})}{Let $T_n$ be a phylogenetic tree with $n\geq 1$ leaves sampled under the uniform model. Then, the expected value of $\sigma_N^2$ of $T_n$ is \[ E_U(\sigma_N^2(T_n)) = \frac{(2n-1)(n-1)}{3n} - \frac{n-1}{2n} \cdot \frac{(2n-2)!!}{(2n-3)!!}. \] Moreover, in the limit \[ E_U(\sigma_N^2(T_n)) \sim \frac{2}{3}\cdot n. \]}
    \item \textbf{Variance under the uniform model:} \textit{Open problem.}
    \item \textbf{Comments:} Although several mathematical results about the variance of leaf depths were established in \citep{Coronado2020b}, there are still various open questions about this index, in particular regarding its minimum value and the trees that achieve it (see the Discussion and Conclusion sections of \citep{Coronado2020b} for a list of open problems and conjectures). One striking property of the variance of leaf depths observed by \citet{Coronado2020b} is the fact that for $n \leq 183$, the minimum value of $\sigma_N^2$ on the space $\BTnstar$ is always achieved by the maximally balanced tree $\Tmb$ with $n$ leaves, but this property fails for almost every $n \geq 184$. In fact, when $h$ goes to infinity, the fraction of leaf numbers $n\in[2,2^h]$ for which the maximally balanced tree minimizes $\sigma_N^2$ tends to zero (Theorem 3 in \citet{Coronado2020b}). \citet{Coronado2020b} thus conclude that \enquote{[t]he phylogenetics community has been wise in preferring the sum $S(T)$ of the leaves’ depths of a phylogenetic tree $T$ over their variance [$\sigma_N^2(T)$] as a balance index, because the latter does not seem to capture correctly the notion of balance of large bifurcating rooted trees. But it is still a valid and useful shape index.}
\end{description}

\section*{Acknowledgement}

MF, SK, and LK were supported by the joint research project \textit{\textbf{DIG-IT!}}
funded by the European Social Fund (ESF), reference: ESF/14-BM-A55-
0017/19, and the Ministry of Education, Science and Culture of Mecklenburg-Vorpommerania, Germany. KW was supported by The Ohio State University's President's Postdoctoral Scholars Program. 
Moreover, all authors wish to thank Sebastian Brinkop for technical support in creating the website \url{treebalance.wordpress.com} and Volkmar Liebscher for helpful comments concerning the approximation of the expected value of the symmetry nodes index (SNI).

\bibliography{Sources}
\bibliographystyle{abbrvnat}

\newpage
\appendix

\section{Appendix: New results} \label{app_additionals}

In this section of the appendix, we present some additional results concerning the (im)balance indices in Tables \ref{Table_Imbalance} and \ref{Table_Balance}. To the best of our knowledge, these results are new and original. Moreover, we provide proofs for properties that have been mentioned before, but for which we could not find formal proofs anywhere in the literature.

\subsection{Established and confirmed (im)balance indices}

\begin{figure}[h] 
    \centering
    \includegraphics[scale=0.125]{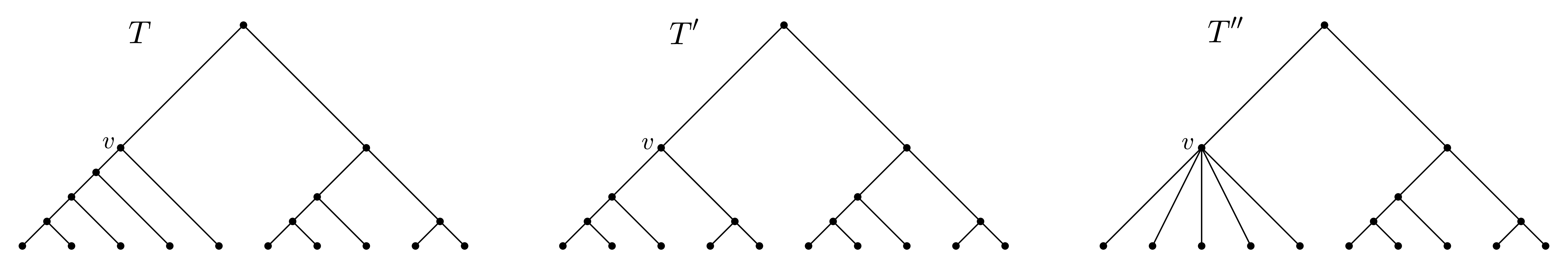}
    \caption{The trees $T$ and $T'$ can be used as a counterexample to show that the average leaf depth, the $B_1$ index, the $B_2$ index, the Furnas rank, the corrected Colless index, the $I_2$ index, the mean $I$ index, the mean $I'$ index, the symmetry nodes index, the variance of leaf depths and the Colijn-Plazzotta rank are not local. The trees $T$ and $T''$ can be used to show that the rooted quartet index for arbitrary trees is not local.}
    \label{fig_locality}
\end{figure}

\subsubsection{Average leaf depth and Sackin index}

In this subsection, we will provide some additional results on the Sackin index and the average leaf depth that -- to our knowledge -- have not yet been known or not yet been proven. For that, recall that the Sackin index \citep{Shao1990, Fischer2021} of a tree $T\in\Tnstar$ is defined as\footnote{The equivalence of the three definitions is shown in \citep[Lemma 1]{Fischer2021}.} \[ S(T)\coloneqq \sum\limits_{v\in V_L(T)} \delta_T(v)=\sum\limits_{v\in\mathring{V}(T)} n_v \ =\sum\limits_{v\in V(T)\setminus\{\rho\} n_v}. \] The average leaf depth \citep{Shao1990, Kirkpatrick1993} -- which is simply a normalized version of the Sackin index -- is defined as $\overline{N}(T)=\frac{1}{n}\cdot S(T)$.\\

The following result was mentioned by \citet{Cardona2012}, but without proof.

\begin{proposition} \label{runtime_Sackin}
For every tree $T\in\Tnstar$, the Sackin index $S(T)$ can be computed in time $O(n)$.
\end{proposition}
\begin{proof}
A vector containing the values $n_u$ for each $u\in V(T)$ can be computed in time $O(n)$ by traversing the tree in post order, setting $n_u=1$ if $u$ is a leaf and calculating $n_u=n_{u_1}+\ldots+n_{u_k}$ (where $u_1, \ldots, u_k$ denote the children of $u$) otherwise. Then, the Sackin index (i.e. the sum of the $n_u$ values for $u\in\mathring{V}(T)$) can be computed from this vector in time $O(n)$ since the cardinality of $\mathring{V}(T)$ is at most $n-1$.
\end{proof}

\begin{proposition} \label{runtime_ALD}
For every tree $T\in\Tnstar$, the average leaf depth $\overline{N}(T)$ can be computed in time $O(n)$.
\end{proposition}
\begin{proof}
The average leaf depth can be calculated from the Sackin index via $\overline{N}(T)=\frac{1}{n}\cdot S(T)$ \citep{Fischer2021}. Since $S(T)$ can be computed in time $O(n)$ (see Proposition \ref{runtime_Sackin}), it follows that $\overline{N}(T)$ can be computed in time $O(n)$ as well.
\end{proof}

In 2007, Matsen showed that the Sackin index is a binary recursive tree shape statistic \cite{Matsen2007}. The following proposition proves that it is also a recursive tree shape statistic when arbitrary trees are considered.

\begin{proposition} \label{recursiveness_Sackin}
The Sackin index is a recursive tree shape statistic. We have $S(T)=0$ for $T\in\mathcal{T}_1^\ast$, and for every tree $T\in\Tnstar$ with $n\geq 2$ and standard decomposition $T=(T_1,\ldots,T_k)$ we have \[ S(T)=\sum\limits_{i=1}^k S(T_i)+\sum\limits_{i=1}^k n_i. \]
\end{proposition}
\begin{proof}
The Sackin index fulfills the recursion \[ S(T) = \sum\limits_{u\in\mathring{V}(T)} n_u = \sum\limits_{u\in\mathring{V}(T_1)} n_u+\ldots+\sum\limits_{u\in\mathring{V}(T_k)} n_u+n_{\rho} = \sum\limits_{i=1}^k S(T_i)+\sum\limits_{i=1}^k n_i. \] Thus, it can be expressed as a recursive tree shape statistic of length $x=2$ with the recursions (where $S_i$ is the simplified notation of $S(T_i)$ and $n_i$ denotes the leaf number of $T_i$)
\begin{itemize}
    \item Sackin index: $\lambda_1=0$ and $r_1(T_1,\ldots,T_k)=S_1+\ldots+S_k+n_1+\ldots+n_k$
    \item leaf number: $\lambda_2=1$ and $r_2(T_1,\ldots,T_k)=n_1+\ldots+n_k$
\end{itemize}
It can easily be seen that $\lambda\in\mathbb{R}^2$ and $r_i: \underbrace{\mathbb{R}^2\times\ldots\times\mathbb{R}^2}_{k\text{ times}} \rightarrow \mathbb{R}$, and that all $r_i$ are independent of the order of subtrees. This completes the proof.
\end{proof}

We now use this result to show that the average leaf depth is also a recursive tree shape statistic.

\begin{proposition} \label{recursiveness_ALD}
The average leaf depth is a recursive tree shape statistic. We have $\overline{N}(T)=0$ for $T\in\mathcal{T}_1^\ast$, and for every tree $T\in\Tnstar$ with $n\geq 2$ and standard decomposition $T=(T_1,\ldots,T_k)$ we have \[ \overline{N}(T)=\left(\sum\limits_{i=1}^k n_i\right)^{-1}\cdot\left(\sum\limits_{i=1}^k n_i\cdot\overline{N}(T_i)\right)+1. \]
\end{proposition}
\begin{proof}
Let $T=(T_1,\ldots,T_k)$ be a tree, and let $n$, $n_1$, \ldots, $n_k$ denote the number of leaves in $T$, $T_1$, \ldots, $T_k$. 
Using $\overline{N}(T)=\frac{1}{n}\cdot S(T)$, and thus also $S(T)=n\cdot \overline{N}(T)$, and Proposition \ref{recursiveness_Sackin}, we have
\begin{equation*}
    \overline{N}(T) = \frac{1}{n}\cdot S(T) = \frac{1}{n}\cdot\left(\sum\limits_{i=1}^k S(T_i)+\sum\limits_{i=1}^k n_i\right) = \frac{1}{n}\cdot\sum\limits_{i=1}^k S(T_i) + \frac{1}{n}\cdot\sum\limits_{i=1}^k n_i = \left(\sum\limits_{i=1}^k n_i\right)^{-1}\cdot\left(\sum\limits_{i=1}^k n_i\cdot\overline{N}(T_i)\right)+1.
\end{equation*}
Thus, the average leaf depth can be expressed as a recursive tree shape statistic of length $x=2$ with the recursions (where $\overline{N}_i$ is the simplified notation of $\overline{N}(T_i)$ and $n_i$ denotes the leaf number of $T_i$)
\begin{itemize}
    \item average leaf depth: $\lambda_1=0$ and $r_1(T_1,\ldots,T_k)=(n_1+\ldots+n_k)^{-1}\cdot(n_1\cdot \overline{N}_1+\ldots+n_k\cdot \overline{N}_k)+1$
    \item leaf number: $\lambda_2=1$ and $r_2(T_1,\ldots,T_k)=n_1+\ldots+n_k$
\end{itemize}
It can easily be seen that $\lambda\in\mathbb{R}^2$ and $r_i:\underbrace{\mathbb{R}^2\times\ldots\times\mathbb{R}^2}_{k\text{ times}}\rightarrow\mathbb{R}$, and that all $r_i$ are independent of the order of subtrees. This completes the proof.
\end{proof}

Next, we will have a look at locality. The following result was mentioned by \citet{Mir2013}, but without proof.

\begin{proposition} \label{locality_Sackin}
The Sackin index is local.
\end{proposition}
\begin{proof}
Let $T'$ be the tree that we obtain from $T\in\Tnstar$ by exchanging a subtree $T_v$ of $T$ with a subtree $T_v'$ on the same number of leaves. Note that $\mathring{V}(T)\setminus\mathring{V}(T_v)=\mathring{V}(T')\setminus\mathring{V}(T_v')$ and $n_T(w)=n_{T'}(w)$ if $w\in \mathring{V}(T)\setminus\mathring{V}(T_v)$, because changing the shape of $T_v$ does not change the number of descendant leaves of $w\in \mathring{V}(T)\setminus\mathring{V}(T_v)$ as $T_v$ and $T_v'$ have the same number of leaves. Also note that $n_T(w)=n_{T_v}(w)$ if $w\in\mathring{V}(T_v)$ and $n_{T'}(w)=n_{T_v'}(w)$ if $w\in\mathring{V}(T'_v)$, because each descendant leaf of $v$, and thus of $w$, is in $T_v$ and $T_v'$. Hence, we have 
\begin{equation*}
\begin{split}
    S(T) - S(T') &= \sum\limits_{w\in\mathring{V}(T_v)} n_T(w) + \sum\limits_{w\in\mathring{V}(T)\setminus\mathring{V}(T_v)} n_T(w) - \sum\limits_{w\in\mathring{V}(T'_v)} n_{T'}(w) - \sum\limits_{w\in\mathring{V}(T')\setminus\mathring{V}(T'_v)} n_{T'}(w)\\
    &= \sum\limits_{w\in\mathring{V}(T_v)} n_{T_v}(w) + \sum\limits_{w\in\mathring{V}(T)\setminus\mathring{V}(T_v)} n_T(w) - \sum\limits_{w\in\mathring{V}(T'_v)} n_{T_v'}(w) - \sum\limits_{w\in\mathring{V}(T)\setminus\mathring{V}(T_v)} n_T(w)\\
    &= \sum\limits_{w\in\mathring{V}(T_v)} n_{T_v}(w) - \sum\limits_{w\in\mathring{V}(T'_v)} n_{T_v'}(w) = S(T_v) - S(T_v').
\end{split}
\end{equation*}
Thus, the Sackin index is local.
\end{proof}

\begin{proposition} \label{locality_ALD}
The average leaf depth is not local.
\end{proposition}
\begin{proof}
Consider the two trees $T$ and $T'$ in Figure \ref{fig_locality} on page \pageref{fig_locality}, which only differ in their subtrees rooted at $v$. Note that in both $T$ and $T'$ the vertex $v$ has exactly 5 descendant leaves. Nevertheless, we have $\overline{N}(T)-\overline{N}(T')=\frac{36}{10}-\frac{34}{10}=\frac{1}{5}\neq \frac{2}{5}=\frac{14}{5}-\frac{12}{5}=\overline{N}(T_v)-\overline{N}(T_v')$. Thus, the average leaf depth is not local. Note that this is due to the different normalization factors $\frac{1}{n}$ for $T$ and $T'$ and $\frac{1}{n_v}$ for $T_v$ and $T_v'$.
\end{proof}

In the following, we consider an aspect of tree balance that was already discussed by \citet{Shao1990}, namely the question of how maximally and minimally balanced trees look for a fixed number of inner vertices. At first, we will have a look at the maximal case. The results in Theorem \ref{max_Sackin_a_1} and Proposition \ref{max_Sackin_a_2} are adapted from \cite[Appendix I]{Shao1990}, but as the authors unfortunately did not provide a proof for their claims, we will prove the statements.

\begin{theorem} \label{max_Sackin_a_1}
For $n=1$ and $m=0$ or any given $n\in\mathbb{N}_{\geq 2}$ and $m\in\mathbb{N}_{\geq 1}$ with $m\leq n-1$, there is exactly one tree $T\in\Tnstar$ with $m$ inner vertices and maximal Sackin index, namely the caterpillar tree on $m+1$ leaves that has $n-m-1$ additional leaves attached to the inner vertex with the largest depth.
\end{theorem}
\begin{proof}
Seeking a contradiction, suppose that there is a tree $T$ with $m$ inner nodes and $n$ leaves and with maximum Sackin index, but that this tree cannot be constructed by taking a binary caterpillar on $m+1$ leaves and attaching $n-m-1$ additional leaves to the parent node of the unique cherry in said caterpillar. 

Then, there are two cases: If $T$ can be constructed from a binary caterpillar by attaching leaves to some inner nodes (but not only the one of maximal depth), then cutting all leaves that were attached to inner nodes that did not have maximal depth and reattaching them to the inner node of maximal depth obviously strictly increases their depths and thus the Sackin index of the tree, a contradiction (as $T$ is maximal). So the only remaining case to consider is that $T$ cannot be constructed by taking a binary caterpillar and attaching some more leaves to inner nodes. In this case, $T$ must contain two vertices $u$ and $v$ which are both such that their maximal pending subtrees all contain only one leaf each (i.e. $u$ and $v$ are parents of generalized cherries). Without loss of generality, let $\delta_T(u) \geq \delta_T(v)$, and denote the children of $v$ by $v_1,\ldots,v_l$ with $l\geq 2$. 

We now construct a tree $T'$ as follows: We delete all edges from $v$ to its children except for the edge leading to $v_1$, and we instead connect all these children to $u$. Then, we delete the edge leading to $v_1$ and subdivide the edge from $u$ to its ancestor by introducing a new node $p$. Then, we connect $p$ to $v_1$, so that now $p$ is the direct ancestor of $u$ and $v_1$. Finally, we delete leaf $v_l$ and the edge leading to it. 

The resulting tree, $T'$, has the following properties: The number of leaves is the same as for $T$, because while we deleted $v_l$, $v$ has now become a leaf. Moreover, the number of inner nodes is the same as for $T$, too, because while $v$ is no longer an inner node, the new inner node $p$ has been introduced. All other vertices, which are leaves (inner nodes) of $T$, are also leaves (inner nodes) of $T'$. Additionally, the depth of $p$ equals the previous depth of $u$ (and is thus at least as large as the depth of $v$), whereas the depth of $u$ (and thus of all its descendants) has strictly increased. In particular, as $v_1$ was previously a child of $v$ and is now a child of $p$, whose depth is at least as large as that of $v$, the depth of $v_1$ cannot have decreased (but it may have increased). On the other hand, the depths of all other leaves previously descending from $u$ have strictly increased by at least 1 (because the depth of $u$ has increased by 1). Thus, we have some leaves whose depth in $T'$ is strictly larger than in $T$, and we have no leaves whose depth in $T'$ is smaller than in $T$. This directly implies (using the definition of the Sackin index): \[ S(T')=\sum\limits_{v \in V_L(T')} \delta_{T'}(v) > \sum\limits_{v \in V_L(T)} \delta_T(v) =S(T), \] which is a contradiction. So such a tree cannot exist, which completes the proof.
\end{proof}

\begin{proposition} \label{max_Sackin_a_2}
For every tree $T\in\Tnstar$ with $m$ inner vertices, the Sackin index fulfills $S(T)\leq nm-\frac{(m-1)\cdot m}{2}$. This bound is tight for $n=1$ and $m=0$ and for all $n\in\mathbb{N}_{\geq 2}$, $m\in\mathbb{N}_{\geq 1}$ with $m\leq n-1$. Moreover, we have $S(T)\leq\frac{n\cdot(n+1)}{2}-1=S(\Tcat)$.
\end{proposition}
\begin{proof}
Let $T\in\Tnstar$ be an arbitrary tree with $n$ leaves and $m$ inner vertices, and let $\widetilde{T}$ be the caterpillar tree on $m+1$ leaves that has $n-m-1$ additional leaves attached to the inner vertex with the largest depth, so in particular $\widetilde{T}$ has $n$ leaves and $m$ inner vertices. In Theorem \ref{max_Sackin_a_1} it has already been proven that $S(\widetilde{T})\geq S(T)$ (and the inequality is strict if $T \neq \widetilde{T}$). Now, assume that the inner vertices in $\widetilde{T}$ are labeled $v_1,v_2,\ldots,v_m$ from the root towards the leaves  (i.e. $\rho_{\widetilde{T}}=v_1$). Then, we have $n_{v_1}=n$, $n_{v_2}=n-1, \ldots,n_{v_m}=n-m+1$ and thus \[ S(T)\leq S(\widetilde{T})\overset{\mbox{\tiny Def.}}{=}\sum\limits_{i=1}^m n_{v_i}=\sum\limits_{i=1}^m (n-i+1) = nm+m-\sum\limits_{i=1}^m i = nm+m-\frac{m\cdot(m+1)}{2} = \underbrace{nm-\frac{(m-1)\cdot m}{2}}_{=:f_n(m)}. \] For fixed $n$, we can regard the right-hand side term as a function $f_n(m)$ of $m$. For the first derivative of this function, we have: \[ f_n'(m)=n-m+\frac{1}{2}>0, \] where the latter inequality holds as $m \leq n-1$. This shows that $f_n$ is strictly increasing, which implies $f_n(m)$ is maximized at $m=n-1$. Thus, we conclude: \[ S(T)\leq f_n(n-1)=n(n-1)-\frac{((n-1)-1)\cdot(n-1)}{2}= \frac{n\cdot(n+1)}{2}-1=S(\Tcat), \] which completes the proof.
\end{proof}

These results about the maximal Sackin index for given $n$ and $m$ can be transformed onto the average leaf depth to give the following corollary.

\begin{corollary} \label{max_ALD}
For every tree $T\in\Tnstar$ with $n$ leaves and $m$ inner vertices, the average leaf depth fulfills $\overline{N}(T)\leq m-\frac{(m-1)\cdot m}{2n}$. This bound is tight for $n=1$ and $m=0$ and for all $n\in\mathbb{N}_{\geq 2}$ and $m\in\mathbb{N}_{\geq 1}$ with $m\leq n-1$, and it is reached precisely by the caterpillar tree on $m+1$ leaves that has $n-m-1$ additional leaves attached to the inner vertex with the largest depth. Moreover, we have $\overline{N}(T)\leq\frac{n+1}{2}-\frac{1}{n}=\overline{N}(\Tcat)$.
\end{corollary}
\begin{proof}
The statement follows immediately from the relation $\overline{N}(T)=\frac{1}{n}\cdot S(T)$ and the respective results for the maximal Sackin index (see Theorem \ref{max_Sackin_a_1} and Proposition \ref{max_Sackin_a_2}). 
\end{proof}

Next, we will turn our attention to the minimal value of the Sackin index for a given number of leaves $n$ and inner vertices $m$. In their manuscript, \citet[Appendix I]{Shao1990} presented a slightly complicated and unfortunately erroneous formula to calculate this minimum value. It was already pointed out in \cite[Footnote 4]{Fischer2021} that this formula even fails in the binary case. Thus, in the following, we give a formula for the minimal value of the Sackin index for given $n$ and $m$ and prove it. Moreover, we give a full characterization of all Sackin minimal trees for any leaf number $n$ and any number of inner vertices $m \in \{1,\ldots,n-1\}$, complementing the binary characterization presented in \cite{Fischer2021}. Note that our results show that contrary to the maximal case, the minimal tree need not be unique. In fact, for $n\geq 2$ the trees with minimal Sackin index are precisely those trees $T$ that have $k=n-m+1$ maximal pending subtrees rooted in the children of the root $\rho$ and fulfill $|\delta_T(x)-\delta_T(y)|\leq 1$ for each pair of leaves $x,y\in V_L(T)$.

\begin{lemma} \label{lem_Sackin_min_3}
Let $T$ be a tree with $n$ leaves and $m$ inner vertices such that $S(T)\leq S(T')$ for all trees $T'$ with $n$ leaves and $m$ inner vertices, i.e. $T$ has minimal Sackin index. Then, $T$ does not have any unresolved vertices other than possibly the root, i.e. for all $v \in \mathring{V}(T)\setminus \{\rho\}$ we have $|child(v)|=2$. Moreover, for all leaves $x$ and $y$ of $T$ we have $|\delta_T(x)-\delta_T(y)| \leq 1$.
\end{lemma}
\begin{proof}
Let $T$ have minimal Sackin index. Now, suppose that $T$ has a vertex $v \in \mathring{V}(T)\setminus \{\rho\}$ with  $|child(v)|\geq 3$. Let $v_1. \ldots, v_k$ denote the children of $v$. We take all but two of these children and move them from $v$ to $\rho$, i.e. we cut all edges from $v$ to its children other than $v_1$ and $v_2$ and instead connect them to $\rho$ with new edges. Note that this strictly decreases their depths, which at the same time strictly decreases the Sackin index (which is simply the sum of all leaf depths). So for the resulting tree $T'$, which by construction has the same number of inner vertices and leaves as $T$,  we have $S(T')<S(T)$, which contradicts the minimality of $T$. Therefore, such a vertex cannot exist.

Now assume that $T$ has two leaves $x$ and $y$ with $|\delta_T(x)-\delta_T(y)|\geq 2$. This in particular implies that $T$ is not the star tree. Without loss of generality, assume $\delta_T(x)>\delta_T(y)$, so we have $\delta_T(x)-\delta_T(y)\geq 2$. This implies that if we take a leaf $z$ of maximal depth in $T$, we have $\delta_T(z)\geq \delta_T(x)\geq \delta_T(y)+2$. However, note that by the same argument as presented in \cite[page 519]{Fischer2021}, $z$ must belong to a cherry $[z,z']$, because its parent, say $w$, necessarily has two children (Note that more than two is not possible as $w$ cannot be the root as $z$ has depth $\delta_T(z)\geq 2$. Thus, following the first part of the proof, $w$ must be binary.), and if the sibling $z'$ of $z$ itself had any children, $z$ would not have maximal depth. So now we construct a tree $T'$ from $T$ by exchanging the cherry $[z,z']$ and $y$ (i.e. the parent $w$ of $z$ and $z'$ gets connected to the parent of $y$ and $y$ gets connected to the parent of $w$).
Note that by construction, $T'$ has the same number of leaves and the same number of inner vertices as $T$. Now, we have\[S(T')=S(T)-2\delta_T(z)-\delta_T(y)+2(\delta_T(y)+1)+(\delta_T(z)-1)=S(T)-(\underbrace{\delta_T(z)-\delta_T(y)}_{\geq2}-1)\leq S(T)-1<S(T)\] and this contradicts the minimality of $T$ and thus completes the proof.
\end{proof}

\begin{observation}\label{obs_Sackin_min_1}
Note that Lemma \ref{lem_Sackin_min_3} implies in particular that all maximal pending subtrees $T_v$ with $v \in \mathring{V}(T)\setminus\{\rho\}$ of any Sackin minimal tree $T$ are binary. 
\end{observation}

We are now in a position to give a full characterization of Sackin minimal trees as well as the minimum value of the Sackin index for fixed $n$ and $m$. This directly generalizes \cite[Theorem 2]{Fischer2021}, where only binary trees were considered.

\begin{theorem}\label{min_Sackin_a}
Let $T$ be a tree with $n$ leaves and $m$ inner vertices and let $k=n-m+1$. Then, the following statements are equivalent:
\begin{enumerate}
    \item $T$ has minimal Sackin index, i.e. for all trees $T'$ with $n$ leaves and $m$ inner vertices we have $S(T)\leq S(T')$.
    \item  Either $n=1$ and $m=0$, or $T$ can be constructed as follows:
    \begin{enumerate}
    \item Take the star tree $T_k^{star}$.
    \item Replace each leaf of $T_k^{star}$ by $T_{\delta-1}^{fb}$, where $\delta=\lfloor\log_2\left(\frac{n}{k}\right)\rfloor+1.$
    \item Choose $n-k\cdot 2^{\delta-1}$ many leaves and replace them with cherries.
\end{enumerate}
\item $S(T)=\begin{cases} 0 & \text{if } n=1 \text{ and } m=0, \\ \left\lfloor\log_2\left(\frac{n}{k}\right)\right\rfloor\cdot n+3n-k\cdot 2^{\lfloor\log_2(\frac{n}{k})\rfloor+1} & \text{ else.} \end{cases} $
\end{enumerate}
\end{theorem}
\begin{proof}
First we show that 1) implies 2). In this regard, assume that $T$ has minimal Sackin index. If $n=1$ and $m=0$, $T$ consists of only one node and there remains nothing to show. Otherwise, by Lemma \ref{lem_Sackin_min_3} and Observation $\ref{obs_Sackin_min_1}$, $T$ has no unresolved inner vertices other than possibly the root, so all maximal pending subtrees $T_1,\ldots,T_k$ must be binary, and all leaves must have depth $\delta$ or $\delta+1$ for some $\delta$ (as leaf depth differences can vary by at most 1). Now, we first show that the number of subtrees $k$ descending from the root must equal $n-m+1$. In this regard, denote the number of leaves of maximal pending subtree  $T_i$ by $n_i$ for $i=1,\ldots,k$. Then, by Observation \ref{obs_Sackin_min_1}, as $T_i$ is binary for each $i$, we know that $T_i$ has $n_i-1$ inner vertices. This immediately implies: 
\begin{align*}
    m&=\sum\limits_{i=1}^k (n_i-1)+\underbrace{1}_{\mbox{\tiny due to $\rho$}}=-k+1+\sum\limits_{i=1}^k n_i = -k+1+n,
\end{align*}
which ultimately implies $k=n+1-m$.

Next, we calculate $\delta$, the minimal leaf depth in $T$. For the $k$ maximal pending subtrees $T_1,\ldots,T_k$ of $T$, each leaf $x\in V_L(T_i)$ fulfills either $\delta_{T_i}(x)=\delta-1$ or $\delta_{T_i}(x)=\delta$. Since $T_i$ is binary, this means in particular that $T_i$ has between $2^{\delta-1}$ and $2^{\delta}$ leaves, i.e. $2^{\delta-1}\leq n_i\leq 2^{\delta}$. Note that there must be at least one $i\in\{1,\ldots,k\}$ with $n_i<2^{\delta}$, because otherwise the minimal leaf depth of $T$ would be $\delta+1$. So in total, we now have $k\cdot 2^{\delta-1}\leq \sum\limits_{i=1}^k n_i=n< k\cdot 2^{\delta}$. Thus, $\delta$ is the largest integer such that $k\cdot 2^{\delta-1}\leq n$. Rearranging this term leads to $\delta\leq \log_2(\frac{n}{k})+1$. Since $\delta$ must be the largest integer fulfilling this inequality, we conclude $\delta=\lfloor\log_2\left(\frac{n}{k}\right)\rfloor+1$.

We now show that $T$ indeed can be constructed as described in Part 2 of Theorem \ref{min_Sackin_a}. In fact, as $T$ has $k=n-m+1$ maximal pending binary subtrees as shown above, it is clear that we can think of $T$ as a star tree, in which the leaves get exchanged by binary trees. Moreover, as all leaves in $T$ either have depth $\delta$ or $\delta+1$ with $\delta=\lfloor\log_2\left(\frac{n}{k}\right)\rfloor+1$ as explained above, it is clear that these binary maximal pending subtrees of $T$ all must be such that all their leaves have depth $\delta$ or $\delta-1$. As they are binary, this implies that we can think of them as fully balanced trees $T_{\delta-1}^{fb}$ in which some leaves get replaced by cherries (in order to construct the leaves of depth $\delta+1$ in $T$). Each replacement of such a leaf by a cherry increases the total leaf number by 1, which shows that we have to do that precisely $n-k\cdot 2^{\delta-1}$ times, because $T$ has $n$ leaves, but the star tree with $k$ leaves that get replaced by $T_{\delta-1}^{fb}$ only has $k\cdot 2^{\delta-1}$ leaves. This shows that Part 2) of the theorem holds.
\par\vspace{0.5cm}

Now we show that 2) implies 3). If $n=1$ and $m=0$ we clearly have $S(T)=0$. So, now assume $T$ can be constructed as explained in 2). Call the tree that results from steps (a) and (b) $T'$. Note that $T'$ has precisely $k\cdot 2^{\delta-1}$ leaves, and they all have depth $\delta$, which implies for the Sackin index, which is just the sum of all leaf depths that $S(T')=k\cdot 2^{\delta-1} \cdot \delta$. However, when we apply Step (c) to generate $T$ from $T'$, we increase the number of leaves by $n-k\cdot 2^{\delta-1}$ (note that this might be 0, so possibly no leaves of a greater depth are introduced if $n=k\cdot  2^{\delta-1}$). At the same time, as each added cherry increases the Sackin index by $2(\delta+1)-\delta=\delta+2$ (because two new leaves of depth $\delta+1$ get added and one leaf, namely the parent of the new cherry, is no longer a leaf, so we lose one leaf of depth $\delta$), the Sackin index increases by $(\delta+2)\cdot \left(n-k\cdot 2^{\delta-1}\right)$. Using $\delta=\lfloor\log_2\left(\frac{n}{k}\right)\rfloor+1,$ this shows that: 
\begin{align*}
    S(T)&=S(T')+(\delta+2)\cdot \left(n-k\cdot 2^{\delta-1}\right) = k\cdot 2^{\delta-1} \cdot \delta + (\delta+2)\cdot \left(n-k\cdot 2^{\delta-1}\right) = 2n+\delta n-k\cdot 2^{\delta} \\
    &=2n+\left(\left\lfloor\log_2\left(\frac{n}{k}\right)\right\rfloor +1\right) n-k\cdot 2^{\left\lfloor\log_2\left(\frac{n}{k}\right)\right\rfloor+1} = 3n+\left\lfloor\log_2\left(\frac{n}{k}\right)\right\rfloor \cdot n-k\cdot 2^{\left\lfloor\log_2\left(\frac{n}{k}\right)\right\rfloor+1}. \\
\end{align*}
This proves Part 3) of the theorem.
\par\vspace{0.5cm}
  
Next it only remains to show that Part 3) implies Part 1). If $n=1$ and $m=0$, there is only one possible tree, so if we take such a tree $T$, it has $S(T)=0$, which is clearly minimal, so there is nothing to show. 
So now assume $T$ has $n$ leaves and $m$ inner vertices with $(n,m)\neq (1,0)$ and such that $S(T)=3n+\left\lfloor\log_2\left(\frac{n}{k}\right)\right\rfloor \cdot n-k\cdot 2^{\left\lfloor\log_2\left(\frac{n}{k}\right)\right\rfloor+1}.$ We now need to show that no other tree with $n$ leaves and $m$ inner vertices can have a smaller Sackin index. Seeking a contradiction, assume there is a tree $T'$ with $n$ leaves and $m$ inner vertices with minimal Sackin index and such that $S(T')<S(T)$. As $T'$ has minimal Sackin index, we already know (because Part 1) implies Part 2) of the theorem) that $T'$ can be constructed by the procedure described in Part 2) of the theorem. However, this implies that $S(T')$ can be calculated by the formula given in Part 3) of the theorem (because Part 2) implies Part 3) as shown above). This, in turn, implies $S(T')=S(T)$, which contradicts $S(T')<S(T)$. This completes the proof.
\end{proof}

Part 2) of Theorem \ref{min_Sackin_a} can be regarded as an algorithm that constructs all Sackin minimal trees. It also leads to Corollary \ref{nummintrees_Sackin}, which counts the number of such trees. However, before stating the corollary, we need to introduce the notion of the set $\mathcal{P}_k(n)$. This  set is derived as follows: We consider all integer partitions of $n$ into precisely $k$ summands $n_1,\ldots,n_k$ with $2^{\delta-1} \leq n_i \leq 2^{\delta}$ for $\delta=\lfloor\log_2\left(\frac{n}{k}\right)\rfloor+1$. We then summarize equal summands by their multiplicities, i.e. we get $n=a_1 \widetilde{n}_1+\ldots+a_l \widetilde{n}_l$ with $\widetilde{n}_i\neq \widetilde{n}_j$ for $i\neq j$ and  with $\widetilde{n_i} \in \{n_1,\ldots,n_k\}$ for all $i=1,\ldots, l$ and $a_i=\sum\limits_{j=1}^k \mathcal{I}(n_j=\widetilde{n}_i)$, such that $\sum\limits_{i=1}^la_i = k$. Each such integer partition $n=a_1 \widetilde{n}_1+\ldots+a_l \widetilde{n}_l$ can be uniquely represented by a set of pairs $\{(a_1,\widetilde{n}_1),\ldots,(a_l,\widetilde{n}_l)\}$, and these sets of pairs form the elements of $\mathcal{P}_k(n)$.

\begin{corollary} \label{nummintrees_Sackin}
Let $n\in\mathbb{N}_{\geq 1}$ and $m\in\mathbb{N}_{\geq 0}$, and let $s(n,m)$ denote the number of Sackin minimal trees with $n$ leaves and $m$ inner vertices, and let $k=n-m+1$. Then, we have $s(n,m)=0$ if $m > n-1$ or $m=0$ and $n>1$, $s(1,0)=1$ and otherwise:
\[ s(n,m)=\begin{cases} 0 & \mbox{ if } m>n-1 \text{ or (} m=0 \text{ and } n>1 \text{)} \\ \sum\limits_{ \left\{\left(a_1,\widetilde{n}_1\right),\ldots,\left(a_l,\widetilde{n}_l\right)\right\} \in \mathcal{P}_k(n) } \ \prod\limits_{i=1}^l \binom{s\left(\widetilde{n}_i,\widetilde{n}_i-1\right)+a_i-1}{a_i}& \mbox{ else,}\end{cases}  \]
where $s(n,n-1)$ corresponds to the number of Sackin minimal rooted binary trees with $n$ leaves, which can be calculated by the formula presented in \cite[Theorem 3]{Fischer2021} (see also Online Encyclopedia of Integer Sequences \cite[Sequence A299037]{OEIS}). 
\end{corollary}
\begin{proof}
First of all note that there are no rooted trees with more than $n-1$ inner vertices since we consider only trees in which the root is the only vertex that is allowed to have degree 2. Moreover, there is precisely one tree with no inner vertex, namely the unique tree that consists of only one leaf, which at the same time is considered to be the root. This explains why $s(1,0)=1$ and why $s(n,m)=0$ if $m>n-1$ or if $m=0$ and $n>1$. 

For all other cases, we know 
by Theorem \ref{min_Sackin_a} and its proof that the set of Sackin minimal trees with $n$ leaves and $m$ inner vertices is characterized by the fact that they all can be constructed by attaching $k$ binary trees $T_1,\ldots,T_k$ to the root, and that each $T_i$ (with its number of leaves $n_i$) fulfils $2^{\delta-1}\leq n_i\leq 2^{\delta}$ for $\delta=\lfloor\log_2\left(\frac{n}{k}\right)\rfloor+1$. Now note that $T_1,\ldots,T_k$ are necessarily Sackin minimal, too, because using the recursiveness of the Sackin index (Proposition~\ref{recursiveness_Sackin}) it is easy to see that otherwise we could exchange a non-minimal maximal pending subtree by a minimal one with the same number of leaves and thus decrease the Sackin index, which would contradict the minimality of the given tree. The leaf numbers clearly fulfill $n=n_1+\ldots+n_k$. Now it may happen that some of the leaf numbers of the maximal pending subtrees coincide, which is why we now summarize equal ones by their multiplicities. This leads to $l$ distinct summands $\widetilde{n}_1,\ldots,\widetilde{n}_l$, where $l\leq k$ and where for each $i \in \{1,\ldots,l\}$ there is at least one $j\in\{1,\ldots,k\}$ such that  $\widetilde{n}_i=n_j$. Moreover, if we denote the multiplicity of $\widetilde{n}_i$ by $a_i$, we have $n=n_1+\ldots + n_k = a_1\widetilde{n}_1+ \ldots + a_l\widetilde{n}_l$ with $\sum\limits_{i=1}^l a_i = k$. So we know that we have to attach $a_i$ Sackin minimal binary subtrees with $\widetilde{n}_i$ leaves to the root, of which there are precisely $s(\widetilde{n}_i,\widetilde{n}_i-1)$ many. As there are $\binom{s(\widetilde{n}_i,\widetilde{n}_i-1)+a_i-1}{a_i}$ possibilities to choose $a_i$ trees from a set of $s(\widetilde{n}_i,\widetilde{n}_i-1)$ trees (unordered sampling with replacement) and as these choices for each $a_i$ can be combined with one another, this gives the recursion stated in the corollary and thus completes the proof.
\end{proof}

Just like in the beginning of this section, we use these properties of the Sackin index and adjust them for the average leaf depth. 

\begin{corollary} \label{mintrees_ALD}
For every tree $T\in\Tnstar$ with $n$ leaves and $m$ inner vertices the average leaf depth fulfills \[ \overline{N}(T) \geq \begin{cases} 0 & \text{if } n=1 \text{ and } m=0, \\ \left\lfloor\log_2\left(\frac{n}{k}\right)\right\rfloor+3-\frac{k}{n}\cdot 2^{\lfloor\log_2(\frac{n}{k})\rfloor+1} & \text{else,} \end{cases} \] with $k=n-m+1$. This bound is tight for $n=1$ and $m=0$ as well as for all $n\in\mathbb{N}_{\geq 2}$ and all $m\in\mathbb{N}_{\geq 1}$ with $m\leq n-1$. Moreover, the bound is reached by precisely those trees $T$ that fulfill either $T=T_1^{star}$ or have $k$ maximal pending subtrees and fulfill $|\delta_T(x)-\delta_T(y)|\leq 1$ for all leaves $x,y\in V_L(T)$. 
\end{corollary}
\begin{proof}
The statement follows immediately from the relation $\overline{N}(T)=\frac{1}{n}\cdot S(T)$ and the respective results for the minimal Sackin index (see Theorem \ref{min_Sackin_a}). 
\end{proof}

\begin{corollary} \label{nummintrees_ALD}
Let $n\in\mathbb{N}_{\geq 1}$ and $m\in\mathbb{N}_{\geq 0}$ and let $avl(n,m)$ denote the number of trees with minimal average leaf depth amongst the set of trees with $n$ leaves and $m$ inner vertices. Then, $avl(n,m)=s(n,m)$, where $s(n,m)$ denotes the number of trees with minimal Sackin index within the same set of trees, and $s(n,m)$ can be calculated with Corollary \ref{nummintrees_Sackin}. In particular, if $m=n-1$ (i.e. only binary trees are considered) $avl(n,m)$ can be computed with the results stated in \cite[Theorem 3, Corollary 1]{Fischer2021}.
\end{corollary}
\begin{proof}
The statement follows immediately from the relation $\overline{N}(T)=\frac{1}{n}\cdot S(T)$, which means in particular that a tree $T$ has minimal Sackin index if and only if it has minimal average leaf depth, and the respective results for the minimal Sackin index (see Corollary \ref{nummintrees_Sackin}). 
\end{proof}

In order to complete the fact sheet of the average leaf depth, we now add two results about its variance under the Yule and the uniform model. The following result is a consequence of the findings in \citep{Blum2006a} and \citep{Cardona2012}.

\begin{proposition} \label{varY_ALD}
Let $T_n$ be a phylogenetic tree with $n$ leaves sampled under the Yule model. Then, the variance of $\overline{N}$ of $T_n$ is \[ V_Y(\overline{N}(T_n))=7-4\cdot H_n^{(2)}-\frac{2}{n}\cdot H_n-\frac{1}{n}. \] Moreover, in the limit $V_Y(\overline{N}(T_n)) \sim 7-\frac{2\pi^2}{3}$.
\end{proposition}
\begin{proof}
Let $T_n$ be a phylogenetic tree with $n$ leaves sampled under the Yule model. From $\overline{N}(T)=\frac{1}{n}\cdot S(T)$ and the property that the variance fulfills $V(a\cdot X)=a^2\cdot V(X)$ for any constant $a\in\mathbb{R}$ and random variable $X$, we get $V_Y(\overline{N}(T_n))=V_Y\left(\frac{1}{n}\cdot S(T_n)\right)=\frac{1}{n^2}\cdot V_Y(S(T_n))$. Now, both statements follow immediately from the fact that $V_Y(S(T_n))=7n^2-4n^2\cdot H_n^{(2)}-2n\cdot H_n-n \stackrel{n\rightarrow\infty}{\sim} \left(7-\frac{2\pi^2}{3}\right)\cdot n^2$ (see \cite{Cardona2012,Blum2006a}).
\end{proof}

The following result is a consequence of the findings in \citep{Coronado2020b}.

\begin{proposition} \label{varU_ALD}
Let $T_n$ be a phylogenetic tree with $n$ leaves sampled under the uniform model. Then, the variance of $\overline{N}$ of $T_n$ is \[ V_U(\overline{N}(T_n))=\frac{10n^2-3n-1}{3n}-\frac{n+1}{2n}\cdot\frac{(2n-2)!!}{(2n-3)!!}-\left(\frac{(2n-2)!!}{(2n-3)!!}\right)^2. \] Moreover, in the limit $V_U(\overline{N}(T_n)) \sim \left(\frac{10}{3}-\pi\right)\cdot n$. 
\end{proposition}
\begin{proof}
Let $T_n$ be a phylogenetic tree with $n$ leaves sampled under the uniform model. From $\overline{N}(T)=\frac{1}{n}\cdot S(T)$ and the property that the variance fulfills $V(a\cdot X)=a^2\cdot V(X)$ for any constant $a\in\mathbb{R}$ and random variable $X$, we get $V_U(\overline{N}(T_n))= V_U\left(\frac{1}{n}\cdot S(T)\right)=\frac{1}{n^2}\cdot V_U(S(T_n))$. Now, both statements follow immediately from the fact that $V_U(S(T_n))=\frac{n\cdot(10n^2 -3n -1)}{3} - \binom{n+1}{2} \cdot \frac{(2n-2)!!}{(2n-3)!!} - n^2 \cdot \left( \frac{(2n-2)!!}{(2n-3)!!} \right)^2 \stackrel{n\rightarrow\infty}{\sim} \left(\frac{10}{3}-\pi\right)\cdot n^3$ (see \citep{Coronado2020b}).
\end{proof}

\subsubsection{\texorpdfstring{$B_1$}{B\_1} index}

After considering the Sackin index and the average leaf depth in the previous section, we will now shift our attention to the $B_1$ index. For this, recall that the $B_1$ index \citep{Shao1990} $B_1(T)$ of a tree $T\in\Tnstar$ is defined as \[ B_1(T)\coloneqq \sum\limits_{v\in\mathring{V} (T)\setminus\{\rho\}} \frac{1}{h(T_v)}.\] First, we will have a look at three general properties, namely computation time, recursiveness and locality.

\begin{proposition} \label{runtime_B1}
For every tree $T\in\Tnstar$, the $B_1$ index $B_1(T)$ can be computed in time $O(n)$.
\end{proposition}
\begin{proof}
A vector containing the values $h(T_u)$ for each $u\in V(T)$ can be computed in $O(n)$ by traversing the tree in post order, setting $h(T_u)=0$ if $u$ is a leaf and calculating $h(T_u)=\max\{h(T_{u_1}),\ldots,h(T_{u_k})\}+1$ otherwise (where $u_1, \ldots, u_k$ denote the children of $u$). Then, the $B_1$ index can be computed from this vector in time $O(n)$ since the cardinality of $\mathring{V}(T)\setminus\{\rho\}$ is at most $n-2$.
\end{proof}

In 2007, Matsen showed that the $B_1$ index is a binary recursive tree shape statistic \cite{Matsen2007}. The following proposition proves that it is also a recursive tree shape statistic when arbitrary trees are considered.

\begin{proposition} \label{recursiveness_B1}
The $B_1$ index is a recursive tree shape statistic. We have $B_1(T)=0$ for $T\in\mathcal{T}_1^\ast$, and for every tree $T\in\Tnstar$ with $n\geq 2$ and standard decomposition $T=(T_1,\ldots,T_k)$ we have \[ B_1(T)=\sum\limits_{i=1}^k B_1(T_i)+\sum\limits_{i=1}^k \frac{1-\mathcal{I}(h(T_i)=0)}{h(T_i)}, \] where we set $\frac{0}{0}\coloneqq 0$.
\end{proposition}
\begin{proof}
Let $T=(T_1,\ldots,T_k)$ be a tree, and let $\rho$, $\rho_1$, \ldots, $\rho_k$ denote the roots of $T$, $T_1$, \ldots, $T_k$. For $v\in\{\rho_1,\ldots,\rho_k\}$ there are two cases to consider: 1) If $v\in\mathring{V}(T)$ it contributes $\frac{1}{h(T_v)}$ to $B_1(T)$. 2) If $v\in V_L(T)$ it contributes $0$ to $B_1(T)$ (where we set $\frac{0}{0}\coloneqq 0$). Since $v$ is a leaf if and only if $h(T_v)=0$, we can combine the two cases and say that $v\in\{\rho_1,\ldots,\rho_k\}$ contributes $\frac{1-\mathcal{I}(h(T_v)=0)}{h(T_v)}$ to $B_1(T)$. Hence, we have 
\begin{equation*}
\begin{split}
   B_1(T) &= \sum\limits_{v\in\mathring{V}(T)\setminus\{\rho\}} \frac{1}{h(T_v)}
    = \sum\limits_{v\in\mathring{V}(T_1)\setminus\{\rho_1\}} \frac{1}{h(T_v)} + \ldots + \sum\limits_{v\in\mathring{V}(T_k)\setminus\{\rho_k\}} \frac{1}{h(T_v)} + \sum\limits_{i=1}^k \frac{1-\mathcal{I}(h(T_{\rho_i})=0)}{h(T_{\rho_i})}\\
   &= \sum\limits_{i=1}^k B_1(T_i) + \sum\limits_{i=1}^k \frac{1-\mathcal{I}(h(T_i)=0)}{h(T_i)}.
\end{split}
\end{equation*}
Thus, the $B_1$ index can be expressed as a recursive tree shape statistic of length $x=2$ with the recursions (where $B_i$ and $h_i$ are simplified notations of $B_1(T_i)$ and $h(T_i)$)
\begin{itemize}
    \item  $B_1$ index: $\lambda_1=0$ and $r_1(T_1,\ldots,T_k)=B_1+\ldots+B_k+\frac{1-\mathcal{I}(h_1=0)}{h_1}+\ldots+\frac{1-\mathcal{I}(h_k=0)}{h_k}$
    \item tree height: $\lambda_2=0$ and $r_2(T_1,\ldots,T_k)=1+\max\{h_1,\ldots,h_k\}$
\end{itemize}
It can easily be seen that $\lambda\in\mathbb{R}^2$ and $r_i:\underbrace{\mathbb{R}^2\times\ldots\times\mathbb{R}^2}_{k\text{ times}}\rightarrow\mathbb{R}$, and that all $r_i$ are independent of the order of subtrees. This completes the proof.
\end{proof}

\begin{proposition} \label{locality_B1}
The $B_1$ index is not local.
\end{proposition}
\begin{proof}
Consider the two trees $T$ and $T'$ in Figure \ref{fig_locality} on page \pageref{fig_locality}, which only differ on their subtrees rooted at $v$. Note that in both $T$ and $T'$ the vertex $v$ has exactly 5 descendant leaves. Nevertheless, we have $B_1(T)-B_1(T')= \frac{59}{12}-\frac{17}{3}=-\frac{3}{4}\neq -\frac{2}{3}=\frac{11}{6}-\frac{5}{2}=B_1(T_v)-B_1(T_v')$. Thus, the $B_1$ index is not local. Note that this property applies, because changing the subtree $T_v$ might change the height of a subtree $T_u$ with $u\in anc(v)$.
\end{proof}

In order for a tree shape statistic to be a balance index, the caterpillar tree must be the unique tree yielding the minimum value and -- provided that $n$ is a power of two -- the fully balanced tree must be the unique tree yielding the maximum value on $\BTnstar$ (see Definition \ref{def_balance}). In the following proposition we show that $B_1$ does fulfill the latter condition.

\begin{theorem} \label{Prop_B1_fb}
Let $n=2^h$ for some $h \in \mathbb{N}_{\geq 0}$. Let $\Tfb$ be the fully balanced tree on $n$ leaves. Then, $\Tfb$ is the unique tree in $\BTnstar$ that maximizes $B_1$, i.e. $B_1\left(\Tfb\right)> B_1(T)$ for all rooted binary trees $T$ on $n$ leaves such that $T \neq$ $\Tfb$. Moreover, we have $B_1\left(\Tfb\right)=\sum\limits_{i=1}^{h-1}\frac{2^i}{h-i}.$
\end{theorem}
\begin{proof} 
We start by proving the second part, i.e. by proving $B_1\left(\Tfb\right)=\sum\limits_{i=i}^{h-1}\frac{2^i}{h-i}.$ In order to see this, note that the fully balanced tree $\Tfb$ of height $h$ has precisely $2^{h-1}$ cherries, i.e. $2^{h-1}$ subtrees of height $1$, $2^{h-2}$ subtrees of height 2, \ldots, $1=2^0=2^{h-h}$ subtrees of height $h$. However, the last subtree, the one of height $h$, is $T_\rho=\Tfb$ itself, so this is not considered by $B_1$. So in summary, we have $B_1\left(\Tfb\right)=\sum\limits_{i=1}^{h-1} 2^i \cdot \frac{1}{h-i}$. This completes the proof of the second statement of the proposition.

It remains to show that for all other rooted binary trees $T \in \mathcal{BT}_{2^h}^\ast$, we have $B_1(T) < B_1\left(\Tfb\right)$. 

Assume for the sake of a contradiction that $T$ is a rooted binary tree with $n=2^h$ leaves and maximal $B_1$ index that is not fully balanced, i.e. $T \neq \Tfb$.
Without loss of generality we may assume that $h \geq 2$ (since there is nothing to show for $h \in \{0,1\}$). 
We now show that transforming $T$ into a tree $\widehat{T}$ more similar to $\Tfb$  strictly increases its $B_1$ index which yields the contradiction. As $T \neq \Tfb$, there must exist a smallest value $h^\ast \in \{1,2, \ldots, h-1 \}$ such that $T$ has $2^{h-\widetilde{h}}$ subtrees of height $\widetilde{h}$ for all $\widetilde{h} < h^\ast$, but $T$ has fewer than $2^{h-h^\ast}$ subtrees of height $h^\ast$ (for instance, $T$ might have $2^{h-1}$ cherries, i.e. subtrees of height 1, but fewer than $2^{h-2}$ subtrees of height 2). As $h^\ast$ is the smallest such value, $T$ in particular has $2^{h-(h^\ast-1)}=2^{h - h^\ast+1}$ fully balanced subtrees of height $h^\ast-1$. In particular, there is an even number of fully balanced subtrees of height $h^\ast-1$, of which some (again, an even number) might be paired into a fully balanced subtree of height $h^\ast$. Thus, there must be an even number of fully balanced subtrees of height $h^\ast-1$ that are not maximal pending subtrees of a fully balanced subtree of height $h^\ast$ of $T$. Let $T_u$ and $T_v$ be two such ``unpaired'' fully balanced subtrees of height $h^\ast-1$ and let $(u',u)$ and $(v',v)$ be the edges directed into $u$ and $v$, respectively. Note that by the choice of $u$ and $v$ it cannot be the case that both $u$ and $v$ are adjacent to the root $\rho$ of $T$.
Without loss of generality, we may assume that $v$ is not adjacent to $\rho$ (note that $u$ might or might not be adjacent to $\rho$).
We now construct a tree $\widehat{T}$ with $n=2^h$ leaves by deleting edge $(v',v)$,  suppressing $v'$, subdividing edge $(u',u)$ with a new node $w$, and adding an edge $(w,v)$ (i.e., we graft $T_v$ onto the edge $(u',u)$ whereby ``pairing'' $T_v$ and $T_u$ to form a fully balanced subtree of height $h^\ast$ rooted at node $w$). 
We now compare $B_1(T)$ and $B_1(\widehat{T}$) and show that $B_1(\widehat{T}) > B_1(T)$. 
First, note that in $T$ we have $h_T(T_{u'}) \geq h^\ast+1$ and $h_T(T_{v'}) \geq h^\ast+1$. This is due to the fact that $T_u$ is a fully balanced subtree of height $h^\ast-1$ of $T$, which by assumption is not a maximal pending subtree of a fully balanced tree of height $h^\ast$. At the same time all subtrees of $T$ of height less or equal to $h^\ast-1$ are already fully balanced trees, and thus the other maximal pending subtree of $T_{u'}$ apart from $T_u$ must have strictly more than $2^{h^\ast-1}$ leaves. This implies that the height of this subtree must be at least $h^\ast$, which in turn implies that the height of $T_{u'}$ must be at least $h^\ast+1$. Analogously, this holds for $T_{v'}$.

Second, note that the contribution to the $B_1$ index of vertices that are \emph{not} contained on a path from the root of $T$, respectively $\widehat{T}$, to $u$ or $v$, is the same for $T$ and $\widehat{T}$. Thus, it suffices to analyze the contributions of vertices  (excluding $u$ and $v$) on the paths from the root of $T$, respectively $\widehat{T}$, to $u$ and $v$. Here, we distinguish three cases: 
\begin{enumerate}[(i)]
    \item Nodes $u'$ and $v'$ are not ancestors of one another (Figure~\ref{Fig_B1_Case1}).
    Let $uv$ denote the lowest common ancestor of $u'$ and $v'$ in $T$ (which might be the root of $T$). Then, the paths in $T$ from $uv$ to $u'$ and $v'$, respectively, are disjoint, and we have:
        \begin{enumerate}[(a)]
            \item The contribution of nodes on the path from $uv$ to $u'$ to the $B_1$ index (excluding $uv$) is the same for $T$ and $\widehat{T}$. This is due to the fact that grafting $T_v$ onto the edge $(u',u)$ does not change the height of $T_{u'}$, i.e. $h_T(T_{u'}) = h_{\widehat{T}}(T_{u'}) \geq h^\ast+1$, and thus it does also not change the heights of any other subtrees induced by nodes on this path.
            \item The contribution of nodes on the path from the root of $T$, respectively $\widehat{T}$, via $uv$ to the parent of $v'$, say $v''$, to the $B_1$ index might be larger in $\widehat{T}$ than in $T$ but it cannot be smaller (note that $v'$ no longer exists in $\widehat{T}$). This is due to the fact that suppressing $v'$ might decrease the height of some or all subtrees induced by nodes on this path in $\widehat{T}$ by one compared to $T$ (which increases the $B_1$ index), but it cannot increase any heights (which would result in a decrease of the $B_1$ index). 
            \item Finally and most importantly, $h_T(T_{v'}) \geq h^\ast+1$, whereas $h_{\widehat{T}}(T_w) = h^\ast$. In particular, $h_{\widehat{T}}(T_w) < h_T(T_{v'})$.
         \end{enumerate}
         In total, this implies
         \begin{align*}
             B_1(\widehat{T}) &= B_1(T) + \underbrace{\left( \frac{1}{h_{\widehat{T}}(T_w)} -  \frac{1}{h_T(T_{v'})} \right)}_{> 0} + \underbrace{\left( \delta^{\widehat{T}}_{v''} - \delta^T_{v''} \right)}_{\geq 0} 
             > B_1(T),
         \end{align*}
         where $\delta^{\widehat{T}}_{v''}$ and $\delta^T_{v''}$ denote the contribution of nodes on the path from the root to the parent $v''$ of $v'$ in $T$, respectively $\widehat{T}$, discussed in (b). 
         \begin{figure}[htbp]
             \centering
             \includegraphics[scale=0.45]{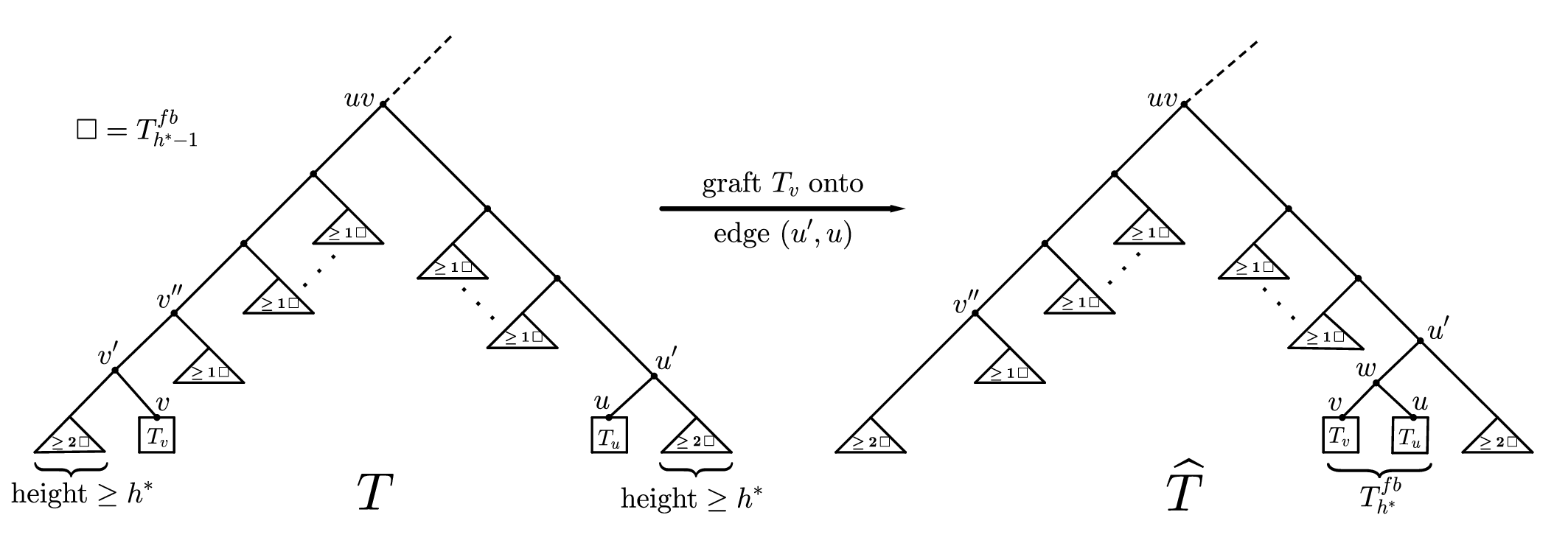}
             \caption{Case (i) in the proof of Theorem~\ref{Prop_B1_fb}.}
             \label{Fig_B1_Case1}
         \end{figure}
    \item Node $u'$ is an ancestor of node $v'$ (Figure~\ref{Fig_B1_Case2}). Let $v''$ denote the parent of $v'$ (note that we might have $u'=v''$).
    In this case, $h_T(T_{u'}) > h_T(T_{v'}) \geq h^\ast+1$. In particular, $h_T(T_{u'}) \geq h^\ast+2$. If we now graft $T_v$ onto the edge $(u',u)$, we have:
    \begin{enumerate}[(a)]
        \item The contribution of nodes on the path from the root via $u'$ to  $v''$ (including $v''$) to the $B_1$ index might be larger in $\widehat{T}$ than in $T$ but it cannot be smaller. Again, this is due to the fact that suppressing $v'$ might decrease the heights of some or all subtrees induced by nodes on this path by one. Note that this is in particular true for $T_{u'}$ despite the fact that we introduce a new node $w$ below $u'$. 
        \item We have that $h_T(T_{v'}) \geq h^\ast + 1 > h^\ast = h_{\widehat{T}}(T_w)$.
     \end{enumerate}
     Letting $\delta^{\widehat{T}}_{v''}$ and $\delta^T_{v''}$ denote the contribution of nodes on the path from the root to $v''$ in $T$, respectively $\widehat{T}$, we get
     \begin{align*}
             B_1(\widehat{T}) &= B_1(T) + \underbrace{\left( \frac{1}{h_{\widehat{T}(w)}} -  \frac{1}{h_T(v')} \right)}_{> 0} + \underbrace{\left( \delta^{\widehat{T}}_{v''} - \delta^T_{v''} \right)}_{\geq 0} 
             > B_1(T).
         \end{align*}
     \begin{figure}[htbp]
             \centering
             \includegraphics[scale=0.5]{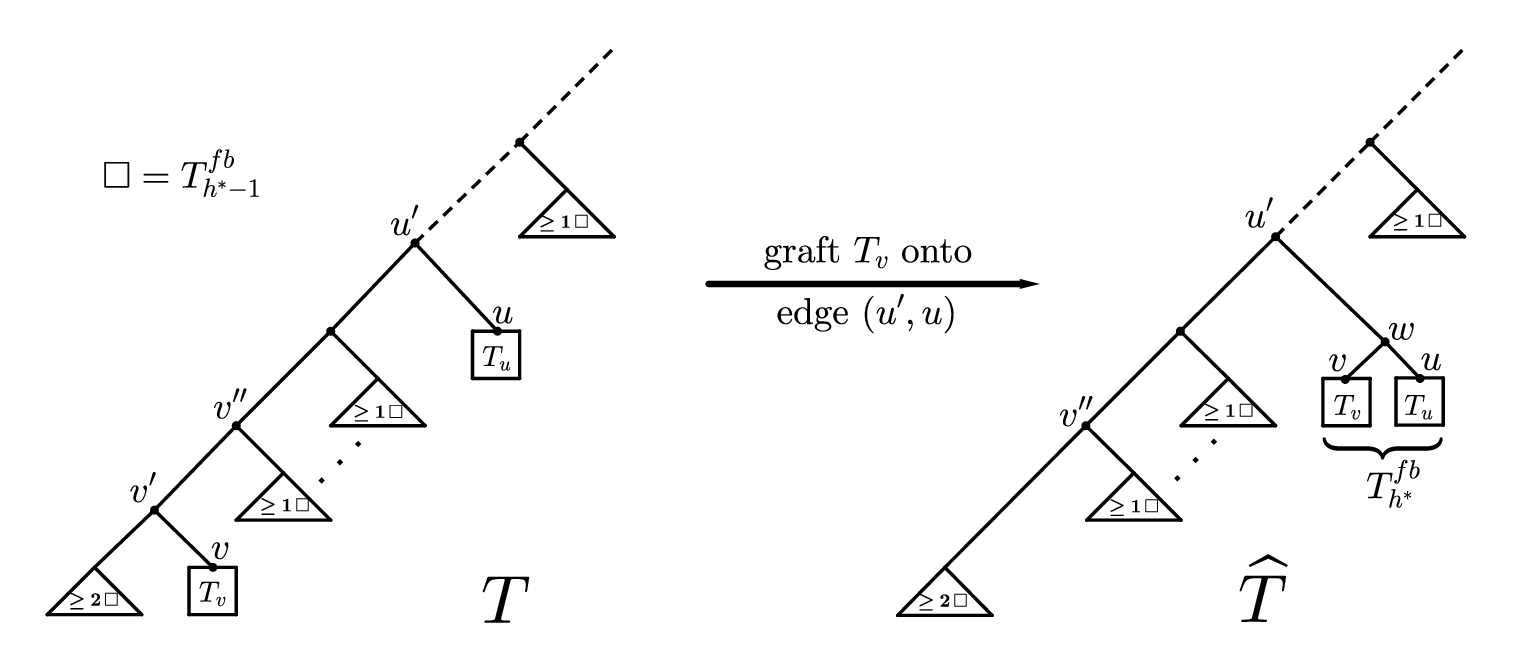}
             \caption{Case (ii) in the proof of Theorem~\ref{Prop_B1_fb}.}
             \label{Fig_B1_Case2}
         \end{figure}
    \item Node $v'$ is an ancestor of node $u'$. In this case, we simply exchange the roles of $v'$ and $u'$, which yields Case (ii) described above.
\end{enumerate}
Thus, in all cases the $B_1$ index strictly increases by transforming $T$ into $\widehat{T}$. However, as we assumed that $T$ had maximal $B_1$ index, this is a contradiction, which shows that the assumption was wrong. This completes the proof.
\end{proof}

\begin{remark}\label{Rem_B1_Shao} Note that \citet[Equation (A2) in the appendix]{Shao1990} presented a formula for the maximum value of the $B_1$ index for all $n$, which is unfortunately erroneous. The authors claim that the maximal value of $B_1$ can be calculated as: \[ \sum\limits_{i=1}^{l-1} \sum\limits_{j=1}^{m_i} i^{-1} = \sum\limits_{i=1}^{l-1} \frac{m_i}{i}, \] where $m_i=\left\lfloor\frac{t+2^{i-1}-1}{2^i}\right\rfloor$, $t=n$, and $l=\left\lfloor \frac{\log t}{\log 2}+0.9999\right\rfloor$ (the value $0.9999$ seems somewhat obscure). The value calculated by this formula for $n=16,385$ leaves would give a maximum $B_1$ value of $\frac{204619418}{18018}$. However, there is a rooted binary tree with 16,385 leaves, namely $T_{16385}^{gfb}$, which gives $B_1\left(T_{16385}^{gfb}\right)=\frac{204620705}{18018}$, which is by  $\frac{1}{14}$ larger than the term suggested as the maximum in \cite{Shao1990}. For reference, the tree can be downloaded in Nexus file format (see \cite{nexus}) from our supplementary material website (\cite{FischerEtAlSupp}). 
\end{remark}

After considering the maximum, we now turn our attention to the minimum. The following proposition provides the minimal value of the $B_1$ index when the number of leaves $n$ and the number of inner vertices $m$ is fixed. \citet{Shao1990} have already mentioned this formula, but without proof. Moreover, we also provide a full characterization of the trees that achieve this minimum.

\begin{theorem} \label{min_B1_a}
For every tree $T\in\Tnstar$ with $m$ inner vertices, the $B_1$ index fulfills $B_1(T)\geq H_{m-1}$. This bound is tight for $n=1$ and $m=0$ and for all $n\in\mathbb{N}_{\geq 2}$ and $m\in\mathbb{N}_{\geq 1}$ with $m\leq n-1$, and it is reached precisely by caterpillar trees on $m+1$ leaves that have $n-m-1$ additional leaves attached to their inner vertices. 
\end{theorem}
\begin{proof}
Seeking a contradiction, assume that $T$ has minimal $B_1$ index but does not have the described shape. Then, $T$ contains at least one vertex $w$ with children $w_1,\ldots,w_k$ of which at least two are inner vertices. Without loss of generality, assume that $w_1,w_2\in\mathring{V}(T)$. Construct $T'$ by exchanging $T_{w_2}$ with a leaf $l$ of maximal depth in $T_{w_1}$, i.e. we delete the edge from $l$ to its parent $p(l)$ (where $p(l)$ might be $w_1$) and reconnect $l$ to $w$, then delete the edge $(w,w_2)$ and instead insert the edge $(p(l),w_2)$. Thus, $T'$ by construction has the same number of leaves and the same number of inner vertices as $T$.

Going from $T$ to $T'$, the height of the pending subtrees rooted at the ancestors $anc_T(w)$ of $w$ or $w$ itself might have increased by at most $h(T_{w_2})\geq1$ but cannot have decreased since $h(T_w')=max\{h(T_{w_1})+h(T_{w_2}),h(T_{w_3}),\ldots,h(T_{w_k})\}+1\geq max\{h(T_{w_1}),h(T_{w_2}),\ldots,h(T_{w_k})\}+1=h(T_w)$. Moreover, the heights of subtrees rooted at nodes in $anc_T(l)\cap\mathring{V}(T_{w_1})$ have strictly increased because their heights in $T$ were defined by the depth $\delta_{T_{w_1}}(l)$ of $l$. For $w_1$, for instance, we have $h(T_{w_1}')=\delta_{T_{w_1}}(l)+h(T_{w_2})=h(T_{w_1})+h(T_{w_2})\geq h(T_{w_1})+1>h(T_{w_1})$. 

Thus, we have $h(T_v')\geq h(T_v)$ for all $v\in\mathring{V}(T)\setminus\{\rho\}$ and for at least one of those $v$, namely $w_1$, the inequality holds as $w_1$ is a child of $w$ and thus $w_1\neq \rho$. All other subtree heights are not affected by the transformation from $T$ to $T'$.
So in summary, $B_1(T')=\sum\limits_{v \in \mathring{V}(T')\setminus\{\rho\}} \frac{1}{h(T_v')} < \sum\limits_{v \in \mathring{V}(T)\setminus\{\rho\}} \frac{1}{h(T_v)}=B_1(T)$. This contradicts the minimality of $T$. So the assumption was wrong and such a tree cannot exist.

So all $B_1$ minimal trees with $n$ leaves and $m$ inner nodes can be constructed from a binary caterpillar on $m+1$ leaves by attaching $n-m-1$ additional leaves to their inner vertices. Let $T$ be such a tree. It remains to show that $B_1(T)=H_{m-1}$. Note that adding the extra $n-m-1$ leaves does not change the heights of the subtrees of $T_{m+1}^{cat}$, and this tree has precisely 1 subtree of height 1, 1 subtree of height 2, and so forth, up to 1 subtree of height $m$, but the last one is not considered by $B_1$ (as the sum does not include $\rho$). In summary, we get: $$B_1(T)=\frac{1}{1}+\frac{1}{2}+\ldots + \frac{1}{m-1} = \sum\limits_{i=1}^{m-1}\frac{1}{i}=H_{m-1}.$$ This completes the proof.
\end{proof}

After fully characterizing the trees with $n$ leaves, $m$ inner vertices and minimal $B_1$ index in Theorem \ref{min_B1_a}, we can also compute the number of such trees, which will be done in the following proposition.

\begin{proposition} \label{nummintrees_B1}
Let $n \in \mathbb{N}_{\geq 1}$ and $m \in \mathbb{N}_{\geq 0}$, and let $b(n,m)$ denote the number of trees in $\Tnstar$ that have $m$ inner vertices and minimal $B_1$ index. Then, $b(n,m)=0$ if $m > n-1$ or if $m=0$ and $n > 1$, $b(1,0)=1$, and otherwise
 $b(n,m)=\binom{n-2}{n-m-1}$ for $n\in\mathbb{N}_{\geq 2}$ and $1\leq m\leq n-1$.
\end{proposition}

\begin{proof}
First, note that there are no rooted trees in $\Tnstar$ with $m > n-1$ or with $m=0$ and $n>1$, and thus $b(n,m)=0$ in these cases.
Moreover, for $n=1$, we have $|\Tnstar|=1$ and thus $b(1,0)=1$. Now, consider the case where $n\geq 2$. In Theorem \ref{min_B1_a} it has been shown that each caterpillar tree on $m+1$ leaves that has $n-m-1$ additional leaves attached to its inner vertices has minimal $B_1$ index. Note that there is one possibility to choose a caterpillar tree on $m+1$ leaves. Attaching $n-m-1$ indistinguishable leaves to $m$ distinguishable inner vertices while allowing multiple leaves to be attached to the same inner vertex is like drawing with replacement and without order. So, there are $\binom{m+(n-m-1)-1}{n-m-1}=\binom{n-2}{n-m-1}$ different possibilities to attach the additional leaves. In total, we have $b(n,m)=1\cdot\binom{n-2}{n-m-1}=\binom{n-2}{n-m-1}$. This completes the proof.
\end{proof}

The following corollary follows directly from Theorem \ref{min_B1_a} and Proposition \ref{nummintrees_B1} and the fact that each binary tree with $n$ leaves has precisely $m=n-1$ inner vertices. 

\begin{corollary} \label{nummintrees_B1_b}
For every binary tree $T\in\BTnstar$, the $B_1$ index fulfills $B_1(T)\geq H_{n-2}$. This bound is tight for all $n\in\mathbb{N}_{\geq 1}$ and is reached only by the caterpillar tree $\Tcat$.
\end{corollary}

\begin{remark} \label{remark_B1_restriction}
Note that if we consider arbitrary trees, it can easily be seen that the star tree is the unique tree that minimizes $B_1$. This is due to the fact that the star tree has only one inner vertex, namely the root, so the set $\mathring{V}(\Tstar)\setminus\{\rho\}$ is empty. This, in turn, implies that the sum in the definition of $B_1$ is empty, too, so it equals 0. All other trees have at least one inner vertex other than the root, so none of them can achieve the value of 0, which makes the star tree the unique minimum. In particular, it makes the caterpillar tree less balanced than the fully balanced tree and more balanced than the rooted star tree. This is not only counterintuitive as most indices consider the star tree as very balanced, it also contradicts our definition of an (im)balance index as we require the caterpillar to be either the unique minimal or the unique maximal tree. Thus, the $B_1$ index is only a balance index when restricted to binary trees.
\end{remark}

\subsubsection{\texorpdfstring{$B_2$}{B\_2} index}

In addition to the $B_1$ index, \citet{Shao1990} introduced the $B_2$ index as a measure of tree balance. Being suitable for binary and arbitrary trees the $B_2$ index $B_2(T)$ of a tree $T\in\Tnstar$ is defined as \[ B_2(T)\coloneqq -\sum\limits_{x\in V_L(T)} p_x\cdot\log(p_x) \text{\quad\quad with \quad\quad} p_x\coloneqq\prod\limits_{v\in anc(x)} \frac{1}{|child(v)|}.\] Note that in a binary tree $T$, we have $p_x=(1/2)^{|anc(x)|}=(1/2)^{\delta_T(x)}$, because each inner vertex has exactly two children.

Again, we will start with some general properties, namely the computation time and the locality.

\begin{proposition} \label{runtime_B2}
For every tree $T\in\Tnstar$, the $B_2$ index $B_2(T)$ can be computed in time $O(n)$ (regardless of the logarithm base).
\end{proposition}
\begin{proof}
A vector containing the values $\widetilde{p}_u=\prod\limits_{v\in anc(u)} \frac{1}{|child(v)|}$ for each $u\in V(T)$ can be computed in time $O(n)$ by traversing the tree in pre order, i.e. from the root towards the leaves, setting $\widetilde{p}_\rho=1$ and calculating $\widetilde{p}_u=\widetilde{p}_{p(u)}\cdot\frac{1}{|child(p(u))|}$ otherwise (with $p(u)$ denoting the parent of $u$). Then, the $B_2$ index can be computed as $B_2(T)=-\sum\limits_{x\in V_L(T)} \widetilde{p}_{x}\cdot log(\widetilde{p}_{x})$, which can be done in $O(n)$. Thus, the total computation time is in $O(n)$.
\end{proof}

\begin{proposition} \label{locality_B2}
The $B_2$ index is not local (regardless of the logarithm base).
\end{proposition}
\begin{proof}
Consider the two trees $T$ and $T'$ in Figure \ref{fig_locality} on page \pageref{fig_locality}, which only differ in their subtrees rooted at $v$. Note that in both $T$ and $T'$ the vertex $v$ has exactly 5 descendant leaves. Nevertheless, we have $B_2(T)-B_2(T')=-\log(\frac{1}{2})\cdot\frac{49}{16}+\log(\frac{1}{2})\cdot\frac{13}{4}= \log(\frac{1}{2})\cdot\frac{3}{16} \neq \log(\frac{1}{2})\cdot\frac{3}{8} =-\log(\frac{1}{2})\cdot\frac{15}{8}+\log(\frac{1}{2})\cdot\frac{9}{4}=B_2(T_v)-B_2(T_v')$. Thus, the $B_2$ index is not local (regardless of the logarithm base). 
\end{proof}

\begin{remark} \label{remark_recursiveness_B2}
Although \citet{Bienvenu2020} use logarithm base two in the definition of $B_2$, the proof that they present for the recursiveness of $B_2$ does not depend on the logarithm base (see \citep[Corollary 1.12 and Proposition 1.10]{Bienvenu2020}). Thus, $B_2(T)=\frac{1}{2}\cdot(B_2(T_1)+B_2(T_2))+1$ holds regardless of the logarithm base.
\end{remark}

In the remainder of this section, we add statements concerning the trees with maximal and minimal $B_2$ index that have been missing in the literature until now. Firstly, in proposition \ref{nummaxtrees_B2} we present a formula for the number of binary trees reaching the maximal $B_2$ index on $\BTnstar$ by using results for the Sackin index.

\begin{proposition} \label{nummaxtrees_B2}
Assume that the logarithm base is 2. Let $n\in\mathbb{N}_{\geq 1}$ and let $g(n)$ denote the number of binary trees with $n$ leaves that have maximal $B_2$ index. Let $A(n)$ denote the set of pairs $A(n)=\{(n_a,n_b)|n_a,n_b\in\mathbb{N}_{\geq 1},n_a+n_b=n, \frac{n}{2}<n_a\leq 2^{\lceil\log_2(n)\rceil-1},n_b\geq2^{\lceil\log_2(n)\rceil-2}\}$. Then, $g(n)$ fulfills the recursion $g(1)=1$ and for $n\geq 2$ \[ g(n)=\sum\limits_{(n_a,n_b)\in A(n)} g(n_a)\cdot g(n_b)+f(n) \text{\qquad\qquad with \qquad\qquad} f(n)=\begin{cases} 0 & \text{if } n \text{ is odd} \\ \binom{g\left(\frac{n}{2}\right)+1}{2} & \text{if } n \text{ is even.} \end{cases} \]  If $n\in\{2^m-1,2^m,2^m+1\}$ for some $m\in\mathbb{N}_{\geq 1}$, there is exactly one tree in $\BTnstar$ with maximal $B_2$ index. For all other $n$, there exist at least two trees in $\BTnstar$ with maximal $B_2$ index.
\end{proposition}
\begin{proof}
In \cite[Theorem 2.3]{Bienvenu2020} it has been shown that a binary tree $T$ has maximal $B_2$ index if and only if it fulfills $\max\limits_{v,w\in V_L(T)} |\delta_T(v)-\delta_T(w)|\leq 1$, which is equivalent to $|\delta_T(v)-\delta_T(w)|\leq 1$ for all $v,w\in V_L(T)$. As stated in \cite[Theorem 2]{Fischer2021} those are precisely the trees with minimal Sackin index. Thus, the statements in Proposition \ref{nummaxtrees_B2} follow immediately from the recursion of the number of trees with minimal Sackin index (see \cite[Theorem 3]{Fischer2021}) and the fact that the tree with minimal Sackin index is unique if and only if $n\in\{2^m-1,2^m,2^m+1\}$ for some $m\in\mathbb{N}_{\geq 1}$ (see \cite[Corollary 1]{Fischer2021}).
\end{proof}

The following main theorem of this section generalizes the result of \cite[Theorem 2.3]{Bienvenu2020}, which states that the caterpillar tree is the unique minimum concerning $B_2$ amongst binary trees, to arbitrary trees.

\begin{theorem}\label{thm_B2_min} Let $T \in \Tnstar$ be a rooted tree with minimal $B_2$ index (for any logarithm base). Then, $T$ is a binary caterpillar, and $B_2(T)= \log(2) \cdot \left( 2-2^{-n+2}\right)$. In particular, the tree with minimal $B_2$ index is unique.
\end{theorem}

Before we can prove Theorem \ref{thm_B2_min}, we need to derive a few preliminary results. We start by stating the next proposition, which corresponds to \cite[Proposition 1.10]{Bienvenu2020}, where it is stated for logarithm base 2 only. However, the proof presented there does not at all depend on the logarithm base, which is why this crucial proposition is still valid when being generalized to an arbitrary logarithm base.

\begin{proposition}[Proposition 1.10 of \cite{Bienvenu2020}, adapted]  \label{prop_bienvenu}
Let $T$ and $T'$ be two rooted trees, and let $T''$ be the
rooted tree obtained by grafting $T'$ on a leaf $x^*$ of $T$, i.e. by turning the
ancestor of $x^*$ in $T$ into the ancestor of the root of $T'$ instead or, if $x^*$ is the only leaf of $T$, by replacing $x^*$ by $T'$. Then, we have (for any logarithm base): 
$B_2(T'')=B_2(T)+p_{x^*} \cdot B_2(T')$.
\end{proposition}

The following theorem uses the previous proposition to generalize the findings of \citet{Bienvenu2020}, who only considered logarithm base 2.

\begin{theorem}\label{prop_B2cat} Let $T \in \Tnstar$ be a rooted binary tree with minimal $B_2$ index (for any logarithm base). Then, $T$ is a caterpillar, and $B_2(T)= \log(2) \cdot \left( 2-2^{-n+2}\right)$.
\end{theorem}
\begin{proof} The fact that the caterpillar is the unique minimum amongst all binary trees is a direct consequence of Theorem 2.3 (i) in \cite{Bienvenu2020}, which is based on Proposition \ref{prop_bienvenu}, as the proofs given there do not at all depend on the logarithm base, which can easily be verified.

So we only need to show the $B_2$ value of the caterpillar tree.  Note that for any binary tree $T$ we have \begin{align*}
B_2(T)&=- \sum\limits_{x \in V_L(T)} \left(\frac{1}{2} \right)^{\delta_T(x)} \cdot \log \left(\left(\frac{1}{2}\right)^{\delta_T(x)} \right) = \sum\limits_{x \in V_L(T)} \left(\frac{1}{2} \right)^{\delta_T(x)} \cdot \delta_T(x) \log(2). \\
\end{align*} 
Now for a caterpillar, as we have precisely one leaf of depths $1,\ldots,n-2$ and two leaves of depth $n-1$ (namely the ones of the unique cherry), we get:
\begin{align*}
B_2(T_n^{cat})&=  \left(\sum\limits_{i=1}^{n-1} \left(\frac{1}{2} \right)^{i} \cdot i \cdot \log(2)\right) +   \left(\frac{1}{2} \right)^{n-1} \cdot (n-1) \cdot \log(2)= \log(2) \cdot \left( \frac{n-1}{2^{n-1}}+\sum\limits_{i=1}^{n-1}\frac{i}{2^i}\right).
\end{align*} 

So it only remains to show that $\frac{n-1}{2^{n-1}}+\sum\limits_{i=1}^{n-1}\frac{i}{2^i}=2-2^{-n+2}$ for all $n \in \mathbb{N}_{\geq 1}$. We show this by induction on $n$. For $n=1$, both terms equal 0, which shows the assertion. Now assume that the statement holds for $n$ and consider $n+1$. In this case, we have: 
\begin{align*} \frac{(n+1)-1}{2^{(n+1)-1}} +  \sum\limits_{i=1}^{(n+1)-1} \frac{i}{2^i} &=  \frac{n}{2^{n}} +  \sum\limits_{i=1}^{n} \frac{i}{2^i} = \underbrace{\frac{n-1}{2^{n-1}}+\sum\limits_{i=1}^{n-1}\frac{i}{2^i}}_{\overset{\mbox{\tiny ind.}}{=} 2-2^{-n+2}} +2 \cdot \frac{n}{2^n} - \frac{n-1}{2^{n-1}} \\ 
&= 2-2^{-n+2} + \frac{n}{2^{n-1}} - \frac{n-1}{2^{n-1}} = 2  -  2^{-n+2}+2^{-n+1} \\
&=2-2\cdot2^{-n+1}+2^{-n+1} = 2-2^{-(n+1)+2},
\end{align*} 
which completes the proof.
\end{proof}

The following corollary is a direct consequence of Proposition \ref{prop_bienvenu}.

\begin{corollary} \label{cor_B2_subtrees} Let $T$ be a (not necessarily binary) rooted tree with $n$ leaves and with minimum (maximum) $B_2$ index. Then all pending subtrees of $T$ have minimum (maximum) $B_2$ value, too.
\end{corollary}
\begin{proof} Let $T'$ be a pending subtree of $T$, and let $v$ be its root. Let $\widetilde{T}$ be the tree resulting from $T$ by replacing $T'$ by a single leaf, i.e. by deleting all descendants of $v$ and thus turning $v$ into a leaf. By Proposition \ref{prop_bienvenu}, we have for any logarithm base:
\begin{equation*}B_2(T)= B_2(\widetilde{T}) + p_v \cdot B_2(T').\end{equation*}

This clearly shows that if $T$ has minimal (maximal) $B_2$ index, so does $T'$, because else we could replace $T'$ in $T$ by another rooted tree with the same number of leaves and a smaller (larger) $B_2$ value  and thus derive a smaller (larger) $B_2$ value for $T$, which would contradict its minimality (maximality). This completes the proof.
\end{proof}

The following lemma and the subsequent corollary will turn out to be the most crucial ingredients for the proof of Theorem  \ref{thm_B2_min}. 

\begin{lemma} \label{lem_B2_unresolved} Let $T$ be a rooted (not necessarily binary) tree with $n_1+n_2$ many leaves, such that the root of $T$ has precisely two maximal pending subtrees $T_1$ and $T_2$ with $n_1$ and $n_2$ leaves, respectively. Let $T'$ be a rooted tree with $n_1+n_2+\ldots+n_k$ many leaves such that the root of $T'$ has $k>2$ maximal pending subtrees, two of which are $T_1$ and $T_2$. Then we have: If $B_2(T')$ is minimal (maximal), so is $B_2(T)$.
\end{lemma}
\begin{proof} Denote the $k-2$ maximal pending subtrees of $T'$ other than $T_1$ and $T_2$ by $T_3,\ldots, T_k$. Then, we have: 
\begin{equation*} B_2(T)=-\sum\limits_{x \in V_L(T_1) \cup V_L(T_2)} p_x(T) \cdot \log (p_x(T)),
\end{equation*}
and 
\begin{equation*} B_2(T')=-\sum\limits_{x \in V_L(T_1) \cup V_L(T_2)} p_x(T') \cdot \log (p_x(T'))-\sum\limits_{x \in  \bigcup_{i=3}^k V_L(T_i)} p_x(T') \cdot \log (p_x(T')).
\end{equation*}

Note that $p_x(T')=p_x(T) \cdot \frac{2}{k}$ for all $x \in V_L(T_1) \cup V_L(T_2)$, which gives:
\begin{align*} B_2(T') &= - \sum\limits_{x \in V_L(T_1) \cup V_L(T_2)} \left(p_x(T) \cdot \frac{2}{k}  \right) \cdot \log\left(p_x(T) \cdot \frac{2}{k} \right) - \sum\limits_{x \in  \bigcup_{i=3}^k V_L(T_i)} p_x(T') \cdot \log (p_x(T')) \\
&= - \frac{2}{k} \cdot \underbrace{\sum\limits_{x \in V_L(T_1) \cup V_L(T_2)} p_x(T) \log(p_x(T)) }_{=-B_2(T)}-\frac{2}{k} \log\left(\frac{2}{k}\right)\cdot \underbrace{\sum\limits_{x \in V_L(T_1) \cup V_L(T_2)}p_x(T)}_{=1} \\&\quad- \sum\limits_{x \in  \bigcup_{i=3}^k V_L(T_i)} p_x(T') \cdot \log (p_x(T')) \\
&=\frac{2}{k}\cdot  B_2(T) \underbrace{- \frac{2}{k}\log \left( \frac{2}{k}\right)- \sum\limits_{x \in  \bigcup_{i=3}^k V_L(T_i)} p_x(T') \cdot \log (p_x(T'))}_{\mbox{\tiny independent of $T$}}.\\
\end{align*}
The latter term clearly shows that if $B_2(T')$ is minimal (maximal), so is $B_2(T)$, as otherwise we could substitute $T=(T_1,T_2)$ in $T'$ by a tree with the same number of leaves but a smaller (larger) $B_2$ value, which would then also decrease (increase) $B_2(T')$ and thus contradict the minimality (maximality) of $T'$. This completes the proof.
\end{proof}

The following corollary is the last ingredient needed to prove Theorem \ref{thm_B2_min}. 

\begin{corollary} \label{cor_B2_nonbin} Let $T$ be a rooted tree with $k\geq 2$ maximal pending \emph{binary} subtrees $T_1,\ldots,T_k$ such that $B_2(T)$ is minimal. Then we have:
\begin{itemize}
\item $T_1,\ldots,T_k$ are caterpillars and
\item at most one of the trees $T_1,\ldots,T_k$ consists of more than one leaf.
\end{itemize}
\end{corollary}
\begin{proof} By Lemma \ref{lem_B2_unresolved}, as $B_2(T)$ is minimal, so are the trees $T_{i,j}:=(T_i,T_j)$ for $i,j \in \{1,\ldots,k\}$, $i\neq j$. Note that each $T_{i,j}$ is binary as $T_i$ and $T_j$ are binary by assumption. 

Now suppose that there exist two values $i,j \in \{1,\ldots,k\}$, $i\neq j$, such that $T_i$ and $T_j$ both contain more than one leaf. Then, $T_{i,j}$ is not a binary caterpillar, which, by Theorem \ref{prop_B2cat} contradicts the minimality of $B_2(T_{i,j})$. So there can be at most one tree in the set $\{T_1,\ldots,T_k\}$ which contains more than one leaf. Without loss of generality, assume this tree is $T_1$. Then, again by Theorem \ref{prop_B2cat}, $T_{1,j}$ is a binary caterpillar for all $j \in \{2,\ldots,k\}$, and thus in particular, $T_1$ is also a binary caterpillar (note that the 1-leaf trees $T_2,\ldots,T_k$ are caterpillars by definition, too). This completes the proof.
\end{proof}

\par\vspace{0.5cm}
Now we are finally in a position to prove Theorem \ref{thm_B2_min}.

\begin{proof}[Proof of Theorem \ref{thm_B2_min}]
By Theorem \ref{prop_B2cat} it suffices to show that each $B_2$-minimal tree is binary. So, for the sake of a contradiction, assume that there are $B_2$-minimal trees that are not binary. Let $n$ be the smallest number of leaves for which such a tree $T$ exists, i.e. $T$ has minimal $B_2$ value amongst all trees with $n$  leaves and $T$ is not binary. Let $T_1,\ldots,T_k$ denote the maximal pending subtrees of $T$. 

Due to Corollary \ref{cor_B2_subtrees}, all pending subtrees of $T$ are $B_2$-minimal as $B_2(T)$ is minimal, and, more importantly, as $n$ was chosen to be the minimal number of leaves permitting an arbitrary tree that has minimal $B_2$ value, \emph{all} maximal pending subtrees $T_1,\ldots,T_k$  of $T$ are necessarily binary as they are $B_2$-minimal. In particular, as $T$ is not binary, the root $\rho$ of $T$ is the unique inner vertex of $T$ inducing more than two maximal pending subtrees, and thus $k>2$. By Corollary \ref{cor_B2_nonbin}, at most one of the maximal pending subtrees of $T$, say $T_1$, can have more than one leaf, and this tree must be a binary caterpillar. Denote the number of leaves of $T_1$ with $n_1$ (and note that $n_1=1$ is possible). 

\par\vspace{0.5cm}
We now construct a tree $T'$ as follows:
\begin{itemize}
\item The root $\rho'$ of $T'$ has $k-1$ maximal pending subtrees, $k-2$ of which are $T_3,\ldots, T_k$.
\item The last maximal pending subtree $T_1'$ is defined as $T_1':=(T_1,T_2)$, i.e. it is a rooted binary caterpillar on $n_1+1$ leaves (as $T_2$ contains only one leaf).
\end{itemize}

We now analyze $B_2(T)$ and $B_2(T')$ seperately in order to simplify the respective terms. The goal is to show that assuming $B_2(T)\leq B_2(T')$ leads to a contradiction, so that $T$ cannot have minimal $B_2$ value.  \par\vspace{0.5cm}

By definition of $B_2$, we have $B_2(T)=- \sum\limits_{x \in V_L(T_1)} p_x(T) \cdot \log (p_x(T)) - \sum\limits_{x \in \bigcup_{i=2}^k V_L(T_i)} p_x(T) \cdot \log (p_x(T))$. Note that as each tree $T_2,\ldots,T_k$ contains only one leaf each, we have $p_x(T) = \frac{1}{k}$ for each $x \in  \bigcup\limits_{i=2}^k V_L(T_i)$, and thus also $\log (p_x(T)) = - \log (k)$. Using additionally that $p_x(T)=\frac{1}{k} \cdot p_x(T_1)$ for all $x \in V_L(T_1)$, this leads to:
\begin{align*}
B_2(T) &= - \sum\limits_{x \in V_L(T_1)} \left( \frac{1}{k} \cdot p_x(T_1) \right)\cdot \log\left( \frac{1}{k} \cdot p_x(T_1) \right) + (k-1) \frac{1}{k} \log (k)\\
&= \underbrace{-\frac{1}{k}\log \left(\frac{1}{k}\right)}_{=\frac{1}{k}\log(k)} \cdot \underbrace{\sum\limits_{x \in V_L(T_1)}p_x(T_1)}_{=1} - \frac{1}{k} \cdot \underbrace{\sum\limits_{x \in V_L(T_1)} p_x(T_1) \cdot \log (p_x(T_1))}_{=-B_2(T_1)} + \frac{k-1}{k} \log (k)\\
&= \frac{1}{k}\cdot B_2(T_1) + \log(k) \cdot \underbrace{\left(\frac{1}{k}+\frac{k-1}{k}\right)}_{=1}\\ 
&\overset{\mbox{\tiny Th. \ref{prop_B2cat}}}{=} \frac{1}{k} \cdot \log(2) \cdot  \left(2-2^{-n_1+2}\right) + \log(k),
\end{align*}
where the last step is true because $T_1$ is a binary caterpillar. \par\vspace{0.5cm}

Analogously, as $T'$ is like $T$ except that it only has $k-1$ maximal pending subtrees, all but one of which consist of only one leaf and the remaining one is a binary caterpillar with $n_1+1$ leaves, we derive:
\begin{equation*} B_2(T')= \frac{1}{k-1} \cdot \log(2) \cdot  \left(2-2^{-n_1+1}\right) + \log(k-1). \end{equation*}

Now, as by assumption $T$ has minimal $B_2$ index, we know that $B_2(T)\leq B_2(T')$. This leads to:
\begin{align*}
\frac{1}{k} \cdot \log(2) \cdot  \left(2-2^{-n_1+2}\right) + \log(k) &\leq  \frac{1}{k-1} \cdot \log(2) \cdot  \left(2-2^{-n_1+1}\right) + \log(k-1)\\
&\Leftrightarrow\\
\log(2) \cdot\left( \frac{2}{k}-\frac{2}{k-1}\right)+\log\left(\frac{k}{k-1}\right) &\leq \log(2) \cdot \left(\frac{2^{-n_1+2}}{k}-\frac{2^{-n_1+1}}{k-1}\right).\\
\text{Using $k>2$ in the following term rearrangements, }&\text{the latter holds if and only if}  \\
\frac{\log(2) \cdot\left( -\frac{2}{k(k-1)}\right)\cdot k(k-1)}{(2k-4)\cdot \log(2)}+ \frac{\log\left(\frac{k}{k-1}\right)\cdot k(k-1)}{(2k-4)\cdot \log(2)} &\leq \frac{1}{2^{n_1}}\\
&\Leftrightarrow\\
2^{n_1} &\leq \frac{(2k-4)\cdot \log(2)}{-2\log(2) + \log\left(\frac{k}{k-1}\right)\cdot k(k-1)} \\
&\Leftrightarrow\\
n_1 &\leq \log_2\left(\frac{(2k-4)\cdot \log(2)}{-2\log(2) + \log\left(\frac{k}{k-1}\right)\cdot k(k-1)} \right)\\
&\Leftrightarrow\\
n_1 &\leq \log_2\left(\frac{\log\left(4^{k-2}\right)}{ \log\left(\frac{1}{4} \cdot \left(\frac{k}{k-1} \right)^{k(k-1)}\right)} \right).\\
\end{align*}

Since the logarithm is continuous and as the logarithm base in the fraction \enquote{cancels out}, i.e. is irrelevant, the following holds for $k \rightarrow \infty$:
\begin{align*}
\frac{\log\left(4^{k-2}\right)}{ \log\left(\frac{1}{4} \cdot \left(\frac{k}{k-1} \right)^{k(k-1)}\right)} &= \frac{(k-2) \cdot \ln(4)}{-\ln\left(4 \cdot \left(\frac{k-1}{k} \right)^{k(k-1)}\right)} = \frac{k-2}{k-1} \cdot \frac{1}{-\ln\left(4^\frac{1}{k-1} \cdot \left(\frac{k-1}{k} \right)^{k}\right)} \cdot \ln(4) \\
&= \underbrace{\frac{k-2}{k-1}}_{\to 1} \cdot \frac{1}{-\ln\left(\underbrace{4^\frac{1}{k-1}}_{\to 1} \cdot \underbrace{\left(1+\frac{-1}{k} \right)^{k}}_{\to e^{-1}}\right)} \cdot \ln(4) \xrightarrow[k \to \infty]{} \ln(4).
\end{align*}

So this fraction converges to $\ln(4)$ as $k \rightarrow \infty$, and thus we can conclude that in particular, this value is $<2$ for all logarithm bases, as the fraction is monotonically increasing\footnote{In order to verify that the fraction is monotonically increasing, we used the computer algebra system Mathematica \cite{Mathematica} to verify that the first derivative is strictly positive for all values of $k$, and thus also in particular for all integers 
$k>2$.}  and converging from below to $\ln(4)<2$. Thus the right hand side of the above inequality is strictly smaller than 1. However, this implies that $n_1$ must be strictly smaller than 1, which in turn is not possible as $T_1$ employs at least one leaf. This gives the desired contradiction and thus completes the proof.
\end{proof}

\subsubsection{Colijn-Plazzotta rank}

The Colijn-Plazzotta rank has been analyzed in \citep{Colijn2018} and \citep{Rosenberg2020}. In addition to these results, we add a statement about the locality of the Colijn-Plazzotta rank. Recall that given $\ell$ as the maximal number of children of any vertex, the Colijn-Plazzotta rank $CP(T)$ of an arbitrary tree $T\in\Tnstar$ is recursively defined as $CP(T)=0$ if $T$ is the empty tree (with no vertices), $CP(T)=1$ if $T$ consists of only one leaf and \[CP(T)\coloneqq\sum\limits_{i=1}^\ell \binom{CP(T_i)+i-1}{i}\] (with $CP(T_\ell)\geq CP(T_{\ell-1})\geq\ldots\geq CP(T_1)$) if $T$ has at least two leaves and the standard decomposition $T=(T_1,\ldots,T_k)$ with $k\leq\ell$ (recall that if $k < l$, $T_{k+1}, \ldots, T_\ell$ are empty trees and thus $CP(T_{k+1})= \ldots= CP(T_\ell)=0$ in the above sum). In the binary case, the recursion simplifies to \[ CP(T)\coloneqq \frac{1}{2}\cdot CP(T_1)\cdot (CP(T_1)-1)+CP(T_2)+1 \] (with $CP(T_1)\geq CP(T_2)$).

\begin{proposition} \label{locality_CP}
The Colijn-Plazzotta ranking is not local.
\end{proposition}
\begin{proof}
Consider the two binary trees $T$ and $T'$ in Figure \ref{fig_locality} on page \pageref{fig_locality}, which only differ in their subtrees rooted at $v$. Note that in both $T$ and $T'$ the vertex $v$ has exactly 5 descendant leaves. Nevertheless, we have $CP(T)-CP(T')=73-22=51\neq 6=12-6=CP(T_v)-CP(T_v')$. Thus, the Colijn-Plazzotta ranking is not local when considering binary trees. Since this is a special case of the arbitrary trees, the Colijn-Plazzotta ranking is also not local when considering arbitrary trees.
\end{proof}

\subsubsection{Colless index, quadratic Colless index and Rogers \texorpdfstring{$J$}{J} index}

One of the oldest and most widely applied imbalance indices is (next to the Sackin index) the Colless index. The Colless index \citep{Shao1990} $C(T)$ of a binary tree $T\in\BTnstar$ is defined as \[ C(T) \coloneqq \sum\limits_{v\in\mathring{V}(T)} bal_T(v)=\sum\limits_{v\in\mathring{V}(T)} |n_{v_1}-n_{v_2}|, \] where $v_1$ and $v_2$ denote the children of $v$. 

Similar to the Colless index, with respect to its definition, is the quadratic Colless index with the only difference that it sums the quadratic balance values instead of the normal balance values. To be more precise, the quadratic Colless index \citep{Bartoszek2021} $QC(T)$ of a binary tree $T\in\BTnstar$ is defined as \[ QC(T) \coloneqq \sum\limits_{v \in \mathring{V}(T)} bal_T(v)^2. \]

The Rogers $J$ index is also closely related as it counts the number of those inner vertices whose balance value is unequal to zero. Formally, the Rogers $J$ index \citep{Rogers1996} $J(T)$ of a binary tree $T\in\BTnstar$ is defined as \[ J(T) \coloneqq \sum\limits_{v \in \mathring{V}(T)} \left( 1- \mathcal{I}(bal_T(v)=0) \right). \]

Like most indices in this manuscript, the Colless index, the quadratic Colless index and the Rogers $J$ index can be computed in linear time as shown in the following proposition.

\begin{proposition} \label{runtime_Colless}
For every binary tree $T\in\BTnstar$, the Colless index $C(T)$, the quadratic Colless index $QC(T)$ and the Rogers index $J(T)$ can be computed in time $O(n)$.
\end{proposition}
\begin{proof}
A vector containing the values $n_u$ for each $u\in V(T)$ can be computed in time $O(n)$ by traversing the tree in post order, setting $n_u=1$ if $u$ is a leaf and calculating $n_u=n_{u_1}+n_{u_2}$ otherwise (where $u_1$ and $u_2$ denote the children of $u$). Then, the Colless index (i.e. the sum of the absolute values $|n_{u_1}-n_{u_2}|$), the quadratic Colless index (i.e. the sum of the quadratic values $(n_{u_1}-n_{u_2})^2$) and the Rogers index (i.e. the sum of the boolean values $1-\mathcal{I}(bal_T(v)=0)=\mathcal{I}(n_{u_1}-n_{u_2}\neq 0)$) over all inner vertices can be computed from this vector in time $O(n)$ since the cardinality of $\mathring{V}(T)$ is $n-1$. Thus, the total computation time of all three indices is in $O(n)$.
\end{proof}

Just like the Colless index (see \citep{Matsen2007}), the quadratic Colless index and the Rogers index are binary recursive tree shape statistics as it is shown in Propositions \ref{recursiveness_qColless} and \ref{recursiveness_U}. 

\begin{proposition} \label{recursiveness_qColless}
	The quadratic Colless index is a binary recursive tree shape statistic. We have $QC(T)=0$ for $T\in\mathcal{BT}_1^\ast$, and for every binary tree $T\in\BTnstar$ with $n\geq 2$ and standard decomposition $T=(T_1,T_2)$ we have \[ QC(T)=QC(T_1)+QC(T_2)+(n_1-n_2)^2. \]
\end{proposition}
\begin{proof}
	The quadratic Colless index fulfills the recursion $QC(T)=QC(T_1)+QC(T_2)+(n_1-n_2)^2$, see \citep{Bartoszek2021}. Thus, it can be expressed as a binary recursive tree shape statistic of length $x=2$ with the recursions (where $QC_i$ is the simplified notation of $QC(T_i)$ and $n_1$ and $n_2$ denote the leaf numbers of $T_1$ and $T_2$) 
	\begin{itemize}
	    \item quadratic Colless index: $\lambda_1=0$ and $r_1(T_1,T_2)=QC_1+QC_2+(n_1-n_2)^2$
	    \item leaf number: $\lambda_2=1$ and $r_2(T_1,T_2)=n_1+n_2$
	\end{itemize}
	It can easily be seen that $\lambda\in\mathbb{R}^2$ and $r_i:\mathbb{R}^2\times\mathbb{R}^2\rightarrow\mathbb{R}$, and that all $r_i$ are independent of the order of subtrees. This completes the proof.
\end{proof}

The following recursion was already mentioned by \citet{Rogers1996}, but no formal proof was given.

\begin{proposition}\label{recursiveness_U}
The Rogers $J$ index is a binary recursive tree shape statistic. We have $J(T)=0$ for $T \in \mathcal{BT}_1^\ast$, and for every binary tree $T \in \BTnstar$ with $n \geq 2$ and standard decomposition $T=(T_1,T_2)$ we have \[ J(T) = J(T_1) + J(T_2) + (1-\mathcal{I}(n_1=n_2)). \]
\end{proposition}
\begin{proof}
Let $T=(T_1,T_2)$ be a binary tree with root $\rho$ and $n \geq 2$ leaves, and let $n, n_1$, and $n_2$ denote the numbers of leaves of $T$, $T_1$, and $T_2$. Then,
\begin{align*}
    J(T) &= \sum\limits_{v \in \mathring{V}(T)} (1-\mathcal{I}(bal_T(v)=0)) \\
    &= \sum\limits_{v \in \mathring{V}(T_1)} (1-\mathcal{I}(bal_{T_1}(v)=0))  + \sum\limits_{v \in \mathring{V}(T_2)} (1-\mathcal{I}(bal_{T_2}(v)=0)) + (1-\mathcal{I}(bal_T(\rho)=0)) \\
    &= J(T_1) + J(T_2) + (1-\mathcal{I}(n_1=n_2)).
\end{align*}
Thus, the Rogers $J$ index can be expressed as a binary recursive tree shape statistic of length $x=2$ with the recursions (where $J_i$ is the simplified notation of $J(T_i)$ and $n_1$ and $n_2$ denote the leaf numbers of $T_1$ and $T_2$)
\begin{itemize}
    \item Rogers $J$ index: $\lambda_1=0$ and $r(T_1,T_2) = J_1+J_2+(1-\mathcal{I}(n_1=n_2))$
    \item leaf number: $\lambda_2=1$ and $r(T_1,T_2)=n_1+n_2$
\end{itemize}
It can easily be seen that $\lambda \in \mathbb{R}^2$ and $r_i:\mathbb{R}^2 \times \mathbb{R}^2 \rightarrow \mathbb{R}$, and that all $r_i$ are independent of the order of subtrees. This completes the proof.
\end{proof}

\citet{Mir2013} stated that the Colless index is local, but did not provide a proof. We will thus prove the statement in Proposition \ref{locality_Colless}. The fact that the quadratic Colless index and the Rogers index are also local can be shown using similar argumentations, see Propositions \ref{locality_qColless} and \ref{locality_U}.

\begin{proposition} \label{locality_Colless}
The Colless index is local.
\end{proposition}
\begin{proof}
Let $T'$ be the binary tree that we obtain from $T\in\BTnstar$ by exchanging a subtree $T_v$ of $T$ with a binary subtree $T_v'$ on the same number of leaves. First, note that $\mathring{V}(T)\setminus\mathring{V}(T_v)=\mathring{V}(T')\setminus\mathring{V}(T'_v)$. Let $w$ be a vertex and let $w_1$ and $w_2$ be its children. For $i\in\{1,2\}$ note that $n_T(w_i)=n_{T'}(w_i)$ if $w\in\mathring{V}(T)\setminus\mathring{V}(T_v)$, because changing the shape of $T_v$ does not change the number of descendant leaves of $w_i\in (V(T)\setminus V(T_v))\cup\{v\}$ (and thus $w\in\mathring{V}(T)\setminus\mathring{V}(T_v)$) as $T_v$ and $T_v'$ have the same number of leaves. Also note that $n_T(w_i)=n_{T_v}(w_i)$ if $w\in\mathring{V}(T_v)$ and $n_{T'}(w_i)=n_{T_v'}(w_i)$ if $w\in\mathring{V}(T'_v)$, because each descendant leaf of $v$, and thus of $w$ and $w_i$, is in $T_v$ and $T_v'$. Hence, we can write
\begin{equation*}
\begin{split}
    C(T) - C(T') &= \sum\limits_{w\in\mathring{V}(T_v)} |n_T(w_1)-n_T(w_2)| + \sum\limits_{w\in\mathring{V}(T)\setminus\mathring{V}(T_v)} |n_T(w_1)-n_T(w_2)|\\
    &\quad- \sum\limits_{w\in\mathring{V}(T'_v)} |n_{T'}(w_1)-n_{T'}(w_2)| - \sum\limits_{w\in\mathring{V}(T')\setminus\mathring{V}(T'_v)} |n_{T'}(w_1)-n_{T'}(w_2)|\\
    &= \sum\limits_{w\in\mathring{V}(T_v)} |n_{T_v}(w_1)-n_{T_v}(w_2)| + \sum\limits_{w\in\mathring{V}(T)\setminus\mathring{V}(T_v)} |n_T(w_1)-n_T(w_2)|\\ 
    &\quad- \sum\limits_{w\in\mathring{V}(T'_v)} |n_{T_v'}(w_1)-n_{T_v'}(w_2)| - \sum\limits_{w\in\mathring{V}(T)\setminus\mathring{V}(T_v)} |n_T(w_1)-n_T(w_2)|\\
    &= \sum\limits_{w\in\mathring{V}(T_v)} |n_{T_v}(w_1)-n_{T_v}(w_2)| - \sum\limits_{w\in\mathring{V}(T'_v)} |n_{T_v'}(w_1)-n_{T_v'}(w_2)| = C(T_v) - C(T_v').
\end{split}
\end{equation*}
Thus, the Colless index is local.
\end{proof}

\begin{proposition} \label{locality_qColless}
The quadratic Colless index is local.
\end{proposition}
\begin{proof}
The proof is completely analogous to the proof of Proposition \ref{locality_Colless} showing that the Colless index is local (by replacing each occurrence of an absolute difference of the form $|n_T(v_1)-n_T(v_2)|$ by the expression $(n_T(v_1)-n_T(v_2))^2$).
\end{proof}

\begin{proposition} \label{locality_U}
The Rogers $J$ index is local.
\end{proposition}
\begin{proof}
The proof is completely analogous to the proof of Proposition \ref{locality_Colless} showing that the Colless index is local (by replacing each occurrence of an absolute difference of the form $|n_T(v_1)-n_T(v_2)|$ by the expression $(1-\mathcal{I}(n_T(v_1)=n_T(v_2))$).
\end{proof}

\subsubsection{Corrected Colless index}

Recall that the corrected Colless index \citep{Heard1992} $I_C(T)$ of a binary tree $T\in\BTnstar$ is defined as \[ I_C(T) \coloneqq \frac{2\cdot C(T)}{(n-1)(n-2)} = \frac{2}{(n-1)(n-2)}\cdot\sum\limits_{v\in\mathring{V}(T)} bal_T(v) = \frac{2}{(n-1)(n-2)}\cdot\sum\limits_{v\in\mathring{V}(T)} |n_{v_1}-v_{v_2}|, \] where $v_1$ and $v_2$ denote the children of $v$. It can thus be seen as a normalized version of the Colless index. 

Again, we start the section with some additional results on the computation time, recursiveness and locality of the index.

\begin{proposition} \label{runtime_corColless}
For every binary tree $T\in\BTnstar$, the corrected Colless index $I_C(T)$ can be computed in time $O(n)$.
\end{proposition}
\begin{proof}
The corrected Colless index can be calculated from the Colless index via $I_C(T)=\frac{2}{(n-1)(n-2)}\cdot C(T)$. Since $C(T)$ can be computed in time $O(n)$ (see Proposition \ref{runtime_Colless}), it follows that $I_C(T)$ can be computed in time $O(n)$ as well.
\end{proof}

\begin{proposition} \label{recursiveness_corColless}
The corrected Colless index is a binary recursive tree shape statistic. We have $I_C(T)=0$ for $T\in\mathcal{BT}_1^\ast$, and for every binary tree $T\in\BTnstar$ with $n\geq 2$ and standard decomposition $T=(T_1,T_2)$ we have \[ I_C(T)=\frac{(n_1-1)(n_1-2)\cdot I_C(T_1)}{(n_1+n_2-1)(n_1+n_2-2)}+\frac{(n_2-1)(n_2-2)\cdot I_C(T_2)}{(n_1+n_2-1)(n_1+n_2-2)}+\frac{2\cdot|n_1-n_2|}{(n_1+n_2-1)(n_1+n_2-2)}. \]
\end{proposition}
\begin{proof}
Let $T=(T_1,T_2)$ be a binary tree, and let $n$, $n_1$ and $n_2$ denote the number of leaves in $T$, $T_1$ and $T_2$. Since $I_C(T)=\frac{2}{(n-1)(n-2)}\cdot C(T)$, and thus also $C(T)=\frac{I_C(T)\cdot(n-1)(n-2)}{2}$, and using the recursiveness of the Colless index, which has been proven by \citet{Matsen2007}, we have
\begin{equation*}
\begin{split}
    I_C(T) &= \frac{2}{(n-1)(n-2)}\cdot C(T) = \frac{2}{(n-1)(n-2)}\cdot (C(T_1)+C(T_2)+|n_1-n_2|)\\
    &= \frac{2}{(n_1+n_2-1)(n_1+n_2-2)}\cdot\left(\frac{(n_1-1)(n_1-2)\cdot I_C(T_1)}{2}+\frac{(n_2-1)(n_2-2)\cdot I_C(T_2)}{2}+|n_1-n_2|\right)\\
    &= \frac{(n_1-1)(n_1-2)\cdot I_C(T_1)}{(n_1+n_2-1)(n_1+n_2-2)}+\frac{(n_2-1)(n_2-2)\cdot I_C(T_2)}{(n_1+n_2-1)(n_1+n_2-2)}+\frac{2\cdot |n_1-n_2|}{(n_1+n_2-1)(n_1+n_2-2)}.
\end{split}
\end{equation*}
Thus, the corrected Colless index can be expressed as a binary recursive tree shape statistic of length $x=2$ with the following recursions (where $I_i$ is the simplified notation of $I_C(T_i)$ and $n_1$ and $n_2$ denote the leaf numbers of $T_1$ and $T_2$):
\begin{itemize}
    \item cor. Colless index: $\lambda_1=0$ and $r_1(T_1,T_2)=\frac{(n_1-1)(n_1-2)\cdot I_{1}}{(n_1+n_2-1)(n_1+n_2-2)}+\frac{(n_2-1)(n_2-2)\cdot I_{2}}{(n_1+n_2-1)(n_1+n_2-2)}+\frac{2\cdot |n_1-n_2|}{(n_1+n_2-1)(n_1+n_2-2)}$
    \item leaf number: $\lambda_2=1$ and $r_2(T_1,T_2)=n_1+n_2$
\end{itemize}
It can easily be seen that $\lambda\in\mathbb{R}^2$ and $r_i:\mathbb{R}^2\times\mathbb{R}^2\rightarrow\mathbb{R}$, and that all $r_i$ are independent of the order of subtrees. This completes the proof.
\end{proof}

\begin{proposition} \label{locality_corColless}
The corrected Colless index is not local.
\end{proposition}
\begin{proof}
Consider the two trees $T$ and $T'$ in Figure \ref{fig_locality} on page \pageref{fig_locality}, which only differ in their subtrees rooted at $v$. Note that in both $T$ and $T'$ the vertex $v$ has exactly 5 descendant leaves. Nevertheless, we have $I_C(T)-I_C(T')=\frac{2}{9}-\frac{1}{9}=\frac{1}{9}\neq\frac{2}{3}=1-\frac{1}{3}=I_C(T_v)-I_C(T_v')$. Thus, the corrected Colless index is not local. Note that this is due to the different normalization factors $\frac{2}{(n-1)(n-2)}$ for $T$ and $T'$ and $\frac{2}{(n_v-1)(n_v-2)}$ for $T_v$ and $T_v'$.
\end{proof}

After considering those general properties, we will now have a look at the maximal and minimal value of $I_C$ for a given $n$. The following result was mentioned without proof by \citet{Heard1992, Kirkpatrick1993, Hitchin1997} and \citet{Heard2007}.

\begin{theorem} \label{max_corColless}
For any given $n\in\mathbb{N}_{\geq 1}$, there is exactly one binary tree $T\in\BTnstar$ with maximal corrected Colless index, namely the caterpillar tree $\Tcat$. Also, for every binary tree $T\in\BTnstar$, the corrected Colless index fulfills $I_C(T)=0$ for $n\in\{1,2\}$ and $I_C(T)\leq 1$ for $n\geq 3$. This bound is tight for all $n\in\mathbb{N}_{\geq 3}$.
\end{theorem}
\begin{proof}
By definition, we have the relation $I_C(T)=\frac{2\cdot C(T)}{(n-1)(n-2)}$. Then, the first property follows directly from the fact that $\Tcat$ is the unique tree with maximal Colless index, and the second property follows from the fact that $C(T)=0$ for $n\in\{1,2\}$ (and using $\frac{0}{0}=0$) and $C(T)\leq C(\Tcat)=\frac{(n-1)(n-2)}{2}$ for $n\geq 3$ (cf.~\citet[Lemma 1]{Mir2018}).
\end{proof}

It has been stated by \citet{Heard1992, Kirkpatrick1993, Hitchin1997} and \citet{Heard2007} that the minimal corrected Colless index of a tree $T$ is 0, which is obtained if and only if $T$ is fully balanced (in particular its number of leaves must be a power of 2). Using results from \citet{Hamoudi2017} and \citet{Coronado2020a}, this bound can be extended to all $n\geq1$ as shown in the following proposition.

\begin{proposition} \label{min_corColless}
    Let $T\in\BTnstar$ be a binary tree and let $b_a b_{a-1} \ldots b_0$ denote the binary representation of $n$. Write $n=\sum\limits_{j=1}^\ell 2^{d_j}$ with $\ell\geq 1$ and $d_1,\ldots,d_\ell\in\mathbb{N}_{\geq 0}$ such that $d_1>\ldots>d_\ell$ and let $s(x)$ denote the triangle wave, i.e. the distance from $x\in\mathbb{R}$ to its nearest integer. Then,
\begin{equation*}
\begin{split}
    I_C(T) &\geq \frac{2}{(n-1)(n-2)}\cdot \left(2\cdot(n\text{ \textup{mod} }2^a)+\sum\limits_{j=0}^{a-1} (-1)^{b_j}\cdot(n\text{ \textup{mod} }2^{j+1})\right)\\
    &= \frac{2}{(n-1)(n-2)}\cdot\left(\sum\limits_{j=2}^\ell 2^{d_j}\cdot(d_1-d_j-2\cdot(j-2))\right) = \frac{2}{(n-1)(n-2)}\cdot\left(\sum\limits_{j=1}^{\lceil\log_2(n)\rceil-1} 2^j\cdot s(2^{-j}\cdot n)\right).
\end{split}
\end{equation*}
This bound is tight for all $n\in\mathbb{N}_{\geq 1}$.
\end{proposition}
\begin{proof}
    The equation follows immediately from the relation $I_C(T)=\frac{2\cdot C(T)}{(n-1)(n-2)}$ and the minimal value of the Colless index (see \cite[Theorem 4]{Hamoudi2017} and \cite[Theorem 2 and 3]{Coronado2020a}). Also, since the lower bound on $C(T)$ is tight for all $n\in\mathbb{N}_{\geq 1}$ (see \cite[Theorem 4]{Hamoudi2017} and \cite[Theorem 1]{Coronado2020a}), the lower bound on $I_C(T)$ is tight for all $n\in\mathbb{N}_{\geq 1}$ as well.
\end{proof}

After having stated the minimal value of the corrected Colless index, we now have a look at the trees achieving this value and their number. The following result is a consequence of the findings of \citet{Coronado2020a}.

\begin{theorem} \label{mintrees_corColless}
    Proposition 1 and 3 in \citep{Coronado2020a} provide a full characterization of trees with minimum corrected Colless index for any given $n\in\mathbb{N}_{\geq 1}$, and Algorithm 1 in \citep{Coronado2020a} generates precisely those trees. In particular, each maximally balanced tree $\Tmb$ and each greedy from the bottom tree $\Tgfb$ has minimal corrected Colless index.
\end{theorem}
\begin{proof}
    Since the relation $I_C(T)=\frac{2\cdot C(T)}{(n-1)(n-2)}$ implies that a tree $T\in\BTnstar$ has minimal corrected Colless index if and only if it has minimal Colless index, the stated properties follow directly from the respective properties of the Colless index (see \citep[Proposition 1, 3 and 6, Algorithm 1, Theorem 1]{Coronado2020a}).
\end{proof}

The following result is also a consequence of the findings of \citet{Coronado2020a}.

\begin{proposition} \label{nummintrees_corColless}
    Let $d(n)$ denote the minimal corrected Colless index for a given $n\in\mathbb{N}_{\geq 1}$, let $B(n)$ denote the set of pairs $B(n)=\{(n_a,n_b)|n_a,n_b\in\mathbb{N},n_a>n_b\geq 1, n_a+n_b=n, \frac{(n_a-1)(n_a-2)d(n_a)+(n_b-1)(n_b-2)d(n_b)+2(n_a-n_b)}{(n_a+n_b-1)(n_a+n_b-2)}=d(n)\}$, and let $\widetilde{d}(n)$ denote the number of binary trees with $n$ leaves that have minimal corrected Colless index. Then, $\widetilde{d}(n)$ fulfills the recursion $\widetilde{d}(1)=1$ and \[ \widetilde{d}(n)=\sum\limits_{(n_a,n_b)\in B(n)} \widetilde{d}(n_a)\cdot\widetilde{d}(n_b)+\binom{\widetilde{d}(\frac{n}{2})+1}{2}\cdot \mathcal{I}(n\textup{ mod }2=0). \] In particular, if $n\in\{2^m-1,2^m,2^m+1\}$ for some $m\in\mathbb{N}_{\geq 1}$, there is exactly one tree in $\BTnstar$ with minimal corrected Colless index. For all other $n$, there exist at least two trees in $\BTnstar$ that reach the minimum.
\end{proposition}
\begin{proof} Let $c(n)$ and $d(n)$ denote the minimal Colless index and minimal corrected Colless index for a given $n$. Let the set $B(n)$ be defined as above. And let $\widetilde{c}(n)$ and $\widetilde{d}(n)$ be the number of binary trees $T\in\BTnstar$ with minimal Colless index and minimal corrected Colless index, respectively. Since the relation $I_C(T)=\frac{2\cdot C(T)}{(n-1)(n-2)}$ implies that a tree $T\in\BTnstar$ has minimal corrected Colless index if and only if it has minimal Colless index, $\widetilde{d}(n)$ must have the same start value and must follow the same recursion as $\widetilde{c}(n)$. It also implies that the set $B(n)$ is identical to the set $A(n)=\{(n_a,n_b)|n_a,n_b\in\mathbb{N},n_a>n_b\geq 1,n_a+n_b=n,c(n_a)+c(n_b)+n_a-n_b=c(n)\}$, because any bipartition of the leaf set that yields a tree with minimal Colless index yields a tree with minimal corrected Colless index, namely the same tree, and the other way around. Thus, using the results from \citep[Proposition 4]{Coronado2020a} we have $\widetilde{d}(1)=\widetilde{c}(1)=1$ and
\begin{equation*}
\begin{split}
    \widetilde{d}(n) = \widetilde{c}(n) &= \sum\limits_{(n_a,n_b)\in A(n)} \widetilde{c}(n_a)\cdot\widetilde{c}(n_b)+\binom{\widetilde{c}(\frac{n}{2})+1}{2}\cdot \mathcal{I}(n\text{ mod }2=0)\\
    &= \sum\limits_{(n_a,n_b)\in B(n)} \widetilde{d}(n_a)\cdot\widetilde{d}(n_b)+\binom{\widetilde{d}(\frac{n}{2})+1}{2}\cdot \mathcal{I}(n\text{ mod }2=0).
\end{split}
\end{equation*}
Since the set of binary trees with minimal Colless index (which equals the set of binary trees with minimal corrected Colless index) contains only $\Tmb$ whenever $n\in\{2^m-1,2^m,2^m+1\}$ for some $m\in\mathbb{N}_{\geq 1}$ \cite[Corollary 7]{Coronado2020a}, and contains $\Tmb$ and $\Tgfb$ with $\Tmb \neq \Tgfb$ for all other $n$ \cite[Corollary 7]{Coronado2020a}, the second part of the statement holds as well.
\end{proof}

For the sake of completeness, we also provide formulas for the expected value and variance of the corrected Colless index under the Yule and Uniform model. The following result is a consequence of the findings of \citet{Heard1992}.

\begin{proposition} \label{expY_corColless}
Let $T_n$ be a phylogenetic tree with $n$ leaves sampled under the Yule model. Then, the expected value of $I_C$ of $T_n$ has the limit distribution $E_Y(I_C(T_n)) \sim \frac{1}{n}\cdot \ln(\lfloor\frac{n}{2}\rfloor) \sim 0$.
\end{proposition}
\begin{proof}
This property follows immediately from \[ E_Y(I_C(T_n))=\begin{cases} \frac{2n}{(n-1)(n-2)}\cdot (H_{\lfloor n/2\rfloor}-1) & \text{if } n \text{ is even} \\ \frac{2n}{(n-1)(n-2)}\cdot (H_{\lfloor n/2\rfloor}-1+1/n) & \text{if } n \text{ is odd} \end{cases} \] (see \cite{Heard1992}) and the fact that $H_n\sim\ln(n)$ (see, for instance, \cite[Section~6.3]{Graham1994}).
\end{proof}

The following result is a consequence of the findings of \citet{Cardona2012}.

\begin{proposition} \label{varY_corColless}
    Let $T_n$ be a phylogenetic tree with $n$ leaves sampled under the Yule model. Then, the variance of $I_C$ of $T_n$ is
\begin{equation*}
\begin{split}
    &V_Y(I_C(T_n)) = \frac{4}{(n-1)^2(n-2)^2}\cdot\Bigg[ \frac{5n^2+7n}{2}+(6n+1)\cdot\left\lfloor\frac{n}{2}\right\rfloor-4\left\lfloor\frac{n}{2}\right\rfloor^2+8\left\lfloor\frac{n+2}{4}\right\rfloor^2 -8(n+1)\cdot\left\lfloor\frac{n+2}{4}\right\rfloor\\
    &\qquad -6n\cdot H_n+\left(2\cdot\left\lfloor\frac{n}{2}\right\rfloor-n(n-3)\right)\cdot H_{\left\lfloor\frac{n}{2}\right\rfloor}-n^2\cdot H_{\left\lfloor\frac{n}{2}\right\rfloor}^{(2)}+\left(n^2+3n-2\left\lfloor\frac{n}{2}\right\rfloor\right)\cdot H_{\left\lfloor\frac{n+2}{4}\right\rfloor}-2n\cdot H_{\left\lfloor\frac{n}{4}\right\rfloor}\Bigg].
\end{split}
\end{equation*}
Moreover, in the limit 
\begin{equation*}
\begin{split}
    V_Y(I_C(T_n)) &\sim \frac{4}{(n-1)^2(n-2)^2} \cdot \Bigg[-\frac{8}{3}(-18+\pi^2+\ln(64))\cdot\left\lfloor\frac{n}{4}\right\rfloor^2 -8\left\lfloor\frac{n}{4}\right\rfloor\cdot\ln\left(\left\lfloor\frac{n}{4}\right\rfloor\right)\\
    &\quad + \left(20-8\gamma-32\ln(2)+\left(24-\frac{4}{3}\pi^2-8\ln(2)\right)(n \textup{ mod } 4)\right)\cdot \left\lfloor\frac{n}{4}\right\rfloor \Bigg]
\end{split}
\end{equation*}
with $\gamma$ denoting Euler's constant.
\end{proposition}
\begin{proof}
Both properties follow immediately from the relation $I_C(T)=\frac{2\cdot C(T)}{(n-1)(n-2)}$, the fact that the variance fulfills $V(a\cdot X)=a^2\cdot V(X)$ for any constant $a\in\mathbb{R}$ and the respective formulas for the variance of the Colless index under the Yule model (see \cite[Corollary 6, Corollary 7]{Cardona2012}).
\end{proof}

The following results are a consequence of the findings of \citet{Rogers1994} and \citet{Blum2006a}.

\begin{proposition} \label{exU_corColless}
Let $T_n$ be a phylogenetic tree with $n$ leaves sampled under the uniform model. Then, the expected value of $I_C$ of $T_n$ fulfills the recursion \[ E_U(I_C(T_n))=\frac{n\cdot(n-3)!}{(2n-3)!!}\cdot \sum\limits_{i=1}^{n-1} \frac{(2i-3)!!\cdot (2n-2i-3)!!}{i!\cdot(n-i)!}\cdot ((i-1)(i-2)\cdot E_U(I_C(T_i))+|n-2i|). \] Moreover, in the limit $E_U(I_C(T_n)) \sim \frac{2\pi}{\sqrt{n}} \sim 0$.
\end{proposition}
\begin{proof}
Both properties follow immediately from the relation $I_C(T)=\frac{2\cdot C(T)}{(n-1)(n-2)}$, the fact that the expected value fulfills $E(a\cdot X)=a\cdot E(X)$ for any constant $a\in\mathbb{R}$ and the respective formulas for the expected value of the Colless index under the uniform model (see \cite[p. 2029]{Rogers1994} and \cite[Theorem 4]{Blum2006a}).
\end{proof}

\begin{proposition} \label{varU_corColless}
Let $T_n$ be a phylogenetic tree with $n$ leaves sampled under the uniform model. Then, the variance of $I_C$ of $T_n$ fulfills the recursion
\begin{equation*}
\begin{split}
    V_U(I_C(T_n)) &= \frac{2n(n-3)!}{(n-1)(n-2)(2n-3)!!}\cdot \sum\limits_{i=1}^{n-1} \frac{(2i-3)!!\cdot (2n-2i-3)!!}{i!\cdot(n-i)!} \cdot \bigg(\frac{(i-1)^2(i-2)^2}{2}E_U(I_C(T_i)^2)\\
    &\qquad + (i-1)(i-2)E_U(I_C(T_i))\cdot\frac{(n-i-1)(n-i-2)}{2}E_U(I_C(T_{n-i}))\\
    &\qquad + 2|n-2i|\cdot (i-1)(i-2)E_U(I_C(T_i))+|n-2i|^2\bigg)\\
    &\qquad -\left(\frac{n(n-3)!}{(2n-3)!!}\right)^2 \cdot \left(\sum\limits_{i=1}^{n-1} \frac{(2i-3)!!\cdot (2n-2i-3)!!}{i!\cdot(n-i)!} \cdot\Big((i-1)(i-2) E_U(I_C(T_i)) +|n-2i|\Big)\right)^2.
\end{split}
\end{equation*}
Moreover, in the limit $V_U(I_C(T_n)) \sim \frac{40-12\pi}{3n} \sim 0$.
\end{proposition}
\begin{proof}
Both properties follow immediately from the relation $I_C(T)=\frac{2\cdot C(T)}{(n-1)(n-2)}$, the fact that the variance fulfills $V(a\cdot X)=a^2\cdot V(X)$ for any constant $a\in\mathbb{R}$ and the respective formulas for the variance of the Colless index under the uniform model (see \cite[p. 2028]{Rogers1994} and \cite[Theorem 4]{Blum2006a}).
\end{proof}

\subsubsection{Equal weights Colless index / \texorpdfstring{$I_2$}{I\_2} index}

Now, we shift our attention to the equal weights Colless index (or $I_2$ index). The $I_2$ index \citep{Mooers1997} $I_2(T)$ of a binary tree $T\in\BTnstar$ is defined as \[ I_2(T) \coloneqq \frac{1}{n-2}\cdot\sum\limits_{\substack{v \in \mathring{V}(T) \\ n_v > 2}} \frac{bal_T(v)}{n_v-2}. \]

Before having a look at its maximum and minimum, we first provide statements on its computation time, recursiveness and locality.

\begin{proposition} \label{runtime_I2}
For every binary tree $T\in\BTnstar$, the $I_2$ index $I_2(T)$ can be computed in time $O(n)$.
\end{proposition}
\begin{proof}
A vector containing the values $n_u$ for each $u\in V(T)$ can be computed in time $O(n)$ by traversing the tree in post order, setting $n_u=1$ if $u$ is a leaf and calculating $n_u=n_{u_1}+n_{u_2}$ otherwise (where $u_1$ and $u_2$ denote the children of $u$). Then, the $I_2$ index can be computed from this vector in time $O(n)$ since the cardinality of $\{v\in\mathring{V}(T): n_v>2\}$ is at most $n-2$.
\end{proof}

\begin{proposition} \label{recursiveness_I2}
The $I_2$ index is a binary recursive tree shape statistic. We have $I_2(T)=0$ for $T\in\mathcal{BT}_1^\ast$, and for every binary tree $T\in\BTnstar$ with $n\geq 2$ and standard decomposition $T=(T_1,T_2)$ we have \[ I_2(T) = \frac{1}{n_1+n_2-2}\cdot\left( (n_1-2) \cdot I_2(T_1) + (n_2-2) \cdot I_2(T_2) + \frac{|n_1-n_2|}{n_1+n_2-2} \right). \]
\end{proposition}
\begin{proof}
Let $T=(T_1,T_2)$ be a binary tree with root $\rho$, and let $n$, $n_1$ and $n_2$ denote the number of leaves in $T$, $T_1$ and $T_2$. By definition of $I_2$ we have
\begin{equation*}
\begin{split}
	I_2(T) &= \frac{1}{n-2} \cdot \sum\limits_{\substack{v \in \mathring{V}(T) \\ n_v > 2}} \frac{bal_T(v)}{n_v-2} = \frac{1}{n-2} \cdot \Bigg( \sum\limits_{\substack{v \in \mathring{V}(T_1) \\ n_v > 2}} \frac{bal_T(v)}{n_v-2} + \sum\limits_{\substack{v \in \mathring{V}(T_2) \\ n_v > 2}} \frac{bal_T(v)}{n_v-2} + \frac{bal_T(\rho)}{n_{\rho}-2} \Bigg) \\
	&= \frac{1}{n-2} \cdot \Bigg( (n_1-2) \cdot \underbrace{\frac{1}{n_1-2} \cdot \sum\limits_{\substack{v \in \mathring{V}(T_1) \\ n_v > 2}} \frac{bal_T(v)}{n_v-2}}_{= I_2(T_1)} + (n_2-2) \cdot \underbrace{\frac{1}{n_2-2} \cdot \sum\limits_{\substack{v \in \mathring{V}(T_2) \\ n_v > 2}} \frac{bal_T(v)}{n_v-2}}_{= I_2(T_2)} + \frac{bal_T(\rho)}{n-2} \Bigg) \\
	&= \frac{1}{n_1+n_2-2} \cdot \left( (n_1-2) \cdot I_2(T_1) + (n_2-2) \cdot I_2(T_2) + \frac{|n_1-n_2|}{n_1+n_2-2} \right).
\end{split}
\end{equation*}
Thus, the $I_2$ index can be expressed as a binary recursive tree shape statistic of length $x=2$ with the recursions (where $I_i$ is the simplified notation of $I_2(T_i)$ and $n_1$ and $n_2$ denote the leaf numbers of $T_1$ and $T_2$)
\begin{itemize}
    \item $I_2$ index: $\lambda_1=0$ and $r_1(T_1,T_2)=\frac{1}{n_1+n_2-2} \cdot \left( (n_1-2) \cdot I_1 + (n_2-2) \cdot I_2 + \frac{|n_1-n_2|}{n_1+n_2-2} \right)$
    \item leaf number: $\lambda_2=1$ and $r_2(T_1,T_2)=n_1+n_2$
\end{itemize}
It can easily be seen that $\lambda\in\mathbb{R}^2$ and $r_i:\mathbb{R}^2\times\mathbb{R}^2\rightarrow\mathbb{R}$, and that all $r_i$ are independent of the order of the subtrees. This completes the proof.
\end{proof}

\begin{proposition} \label{locality_I2}
The $I_2$ index is not local.
\end{proposition}
\begin{proof}
Consider the two trees $T$ and $T'$ in Figure \ref{fig_locality} on page \pageref{fig_locality}, which only differ in their subtrees rooted at $v$. Note that in both $T$ and $T'$ the vertex $v$ has exactly 5 descendant leaves. Nevertheless, we have $I_2(T)-I_2(T')=\frac{13}{24}-\frac{1}{3}=\frac{5}{24}\neq \frac{5}{9}=1-\frac{4}{9}=I_2(T_v)-I_2(T_v')$. Thus, the $I_2$ index is not local. Note that this is due to the different normalization factors $\frac{1}{n-2}$ for $T$ and $T'$ and $\frac{1}{n_v-2}$ for $T_v$ and $T_v'$.
\end{proof}

As one expects from an imbalance index, the caterpillar is the only tree achieving its maximum as stated in the following proposition.

\begin{theorem} \label{max_I2}
For any given $n\in\mathbb{N}_{\geq 1}$, there is exactly one tree $T\in\BTnstar$ with maximal $I_2$ index, namely $I_2(T)=0$ if $n\in\{1,2\}$ and $I_2(T)=1$ if $n\geq 3$, namely the caterpillar tree $\Tcat$.
\end{theorem}
\begin{proof}
At first, note that for $n\in\{1,2\}$ the caterpillar tree $\Tcat$ is the only tree in $\BTnstar$ and thus the only one with maximal $I_2$ index. Also note that in this case $I_2(\Tcat)=0$, because the sum in the definition of the $I_2$ index is empty and evaluates to zero by convention.\\
Now, consider the case $n\geq 3$. In order to prove that $I_2(T)\leq 1$ for every tree $T\in\BTnstar$, we require the following two facts:
\begin{enumerate}
    \item [(a)] Let $v \in \mathring{V}(T)$ be an inner vertex of $T$ with children $v_1$ and $v_2$. Then, $bal_T(v) \leq n_v-2$. To see this, assume without loss of generality that $n_{v_1} \geq n_{v_2}$. Then, $bal_T(v) = |n_{v_1}-n_{v_2}| = n_{v_1} - n_{v_2} \leq (n_v-1)-1 = n_v-2$, because the difference between $n_{v_1}$ and $n_{v_2}$ is maximized when $n_{v_2}$ is as small as possible (i.e. $n_{v_2}=1$) and $n_{v_1}$ is as large as possible (i.e. $n_{v_1}=n_v-1$).
    \item [(b)] $|\{v \in \mathring{V}(T): n_v > 2\}| \leq n-2$. To see this, recall that a rooted binary tree with $n$ leaves has $n-1$ inner vertices, i.e. $|\mathring{V}(T)|=n-1$. Moreover, recall that every rooted binary tree with $n \geq 2$ leaves has at least one cherry, i.e. there exists at least one inner vertex $u \in \mathring{V}(T)$ that is the parent of two leaves. In particular, $n_u=2$. Thus, $\{v \in \mathring{V}(T): n_v > 2\} = \{v \in \mathring{V}(T) \setminus \{u\}: n_v > 2\}$ and the claim follows.
\end{enumerate}
Using these facts, we have for every $T\in\BTnstar$ that
\begin{equation} \label{eq_max_I2}
\begin{split}
    I_2(T) &= \frac{1}{n-2} \cdot \sum\limits_{\substack{v \in \mathring{V}(T) \\ n_v > 2}} \frac{bal_T(v)}{n_v-2} 
    \overset{\text{(a)}}{\leq} \frac{1}{n-2} \cdot \sum\limits_{\substack{v \in \mathring{V}(T) \\ n_v > 2}} \frac{n_v-2}{n_v-2}
   \overset{\text{(b)}}{\leq} \frac{1}{n-2} \cdot (n-2) \cdot 1  = 1.
\end{split}
\end{equation}
The second step is to prove that $\Tcat$ is indeed the only tree in $\BTnstar$ that reaches the maximum. By definition of the caterpillar tree (in particular, by the fact that each inner vertex $v$ with $n_v>2$ is incident to precisely one leaf and one subtree of size $n_v-1$), we have that $bal_{\Tcat}(v)=n_v-2$ for each $v\in\mathring{V}(\Tcat)$ with $n_v>2$. Moreover, by the fact that the caterpillar tree contains precisely one cherry, there are $n-2$ such inner vertices $v$. This implies that both \enquote{$\leq$} in Equation \eqref{eq_max_I2} are indeed equal signs when $T=\Tcat$, i.e. $\Tcat$ has maximal $I_2$ index also for $n\geq 3$. Additionally, note that the caterpillar tree is the only tree in $\BTnstar$ that has precisely one cherry, i.e. each $T\in\BTnstar$ with $T\neq\Tcat$ fulfills $I_2(T)<\frac{1}{n-2}\cdot (n-2)\cdot 1=1$ (because for $T \neq \Tcat$, there are at most $(n-1)-2=n-3$ vertices $v \in \mathring{V}(T)$ with $n_v > 2$ and thus using (a) and (b) $I_2(T) \leq (n-3)/(n-2) < 1$). Thus, also for $n\geq 3$ we have that $\Tcat$ is the unique tree with maximal $I_2$ index. This completes the proof.
\end{proof}

The following lemma will later be needed for the minimal value of the $I_2$ index.

\begin{lemma}\label{Lemma_I2_subtrees}
Let $T=(T_1,T_2)$ be a rooted binary tree with $n \geq 2$ leaves. If $T$ has maximum (minimum) $I_2$ index in $\mathcal{BT}_n^\ast$, then $T_1$ and $T_2$ have maximum (minimum) $I_2$ index in $\mathcal{BT}_{n_1}^\ast$ and $\mathcal{BT}_{n_2}^\ast$, respectively.
\end{lemma}
\begin{proof}
First, let $T=(T_1,T_2)$ be a rooted binary tree with $n \geq 2$ leaves and maximum $I_2$ index in $\mathcal{BT}_n^\ast$. For the sake of a contradiction, assume that $I_2(T_1)$ is not maximal (the case when $I_2(T_2)$ is not maximal follows analogously). Then, there exists a tree $\widehat{T}$ in $\mathcal{BT}_{n_1}^\ast$ with $I_2(\widehat{T}) > I_2(T_1)$. Consider the tree $\widetilde{T}=(\widehat{T}, T_2) \in \mathcal{BT}_n^\ast$ obtained by replacing in $T$ the rooted subtree $T_1$ by $\widehat{T}$. Then, by Proposition \ref{recursiveness_I2},
\begin{align*}
    I_2(\widetilde{T}) &= \frac{1}{n-2} \left( (n_1-2) \cdot I_2(\widehat{T}) + (n_2-2) \cdot I_2(T_2) + \frac{|n_1-n_2|}{n-2} \right) \\
    &> \frac{1}{n-2} \left( (n_1-2) \cdot I_2(T_1) + (n_2-2) \cdot I_2(T_2) + \frac{|n_1-n_2|}{n-2} \right) = I_2(T),
\end{align*}
which implies that $I_2(T)$ is not maximal. Thus, if $I_2(T)$ is maximal, $I_2(T_1)$ and $I_2(T_2)$ must be maximal, too. A similar argument shows that if $I_2(T)$ is minimal, $I_2(T_1)$ and $I_2(T_2)$ must be minimal, too. This completes the proof.
\end{proof}

\begin{remark}
Note that Proposition \ref{recursiveness_I2} and Lemma \ref{Lemma_I2_subtrees} recursively imply that \emph{every} rooted subtree of a tree with maximum (minimum) $I_2$ index also has maximum (minimum) $I_2$ index.
\end{remark}

By definition, the equal weights Colless index fulfills $I_2(T)\geq 0$ for any binary tree $T\in\BTnstar$. The following proposition shows that this lower bound will only be reached by fully balanced trees.

\begin{proposition}\label{Prop_I2_FullyBalanced}
For every $n \in \mathbb{N}_{\geq 1}$ and for every $T \in\BTnstar$, we have $I_2(T)=0$ if and only if $n$ is a power of two and $T$ is a fully balanced tree.
\end{proposition}
\begin{proof}
The \enquote{if} implication is a direct consequence of the fact that in a fully balanced tree, the balance values of all inner vertices are equal to zero, and thus $I_2(\Tfb)=0$ for each $h \in \mathbb{N}_{\geq 0}$.\\
We now prove the \enquote{only if} implication by induction on $n$. For $n=1$, $\mathcal{BT}_1^\ast$ contains precisely one tree. This tree $T$ fulfills $I_2(T)=0$ (note that an empty sum evaluates to zero) and is fully balanced, so there is nothing left to show. Now, let $n \geq 2$ and assume that the assertion is true for every $1 \leq n' < n$. Let $T=(T_1,T_2) \in \BTnstar$ be such that $I_2(T)=0$. By Proposition \ref{recursiveness_I2}, $I_2(T)=0$ is equivalent to either
\begin{enumerate}
    \item [(i)] $n_1=n_2=2$, or 
    \item [(ii)] $n_1=n_2\neq 2$ and $I_2(T_1)=I_2(T_2)=0$.
\end{enumerate}
In Case (i), as there is only one element in $\mathcal{BT}_2^\ast$, namely $T_1^{\mathit{fb}}$, we can immediately conclude that $T=(T_1^{\mathit{fb}}, T_1^{\mathit{fb}}) = T_2^{\mathit{fb}}$ is a fully balanced tree.\\
In Case (ii), it follows from the inductive hypothesis that $n_1=n_2$ is a power of two and hence $n=n_1+n_2 = 2 \cdot n_1$ also is a power of two. Moreover, it follows from the inductive hypothesis that $T_1$ and $T_2$ are both fully balanced trees. In summary, this implies that $n$ is a power of two and $T=(T_1,T_2)$ is a fully balanced tree. This completes the proof.
\end{proof}

Last but not least, the following proposition connects the equal weights Colless index to the $I_v$ values.

\begin{proposition} \label{prop_I2_Iv}
Let $T$ be a rooted binary tree and let $v \in \mathring{V}(T)$ be an inner node with $n_v>3$ descendant leaves. If $n_v$ is even, the corresponding summand in the formula of $I_2$ equals the $I_v$ value: \[ \frac{bal_T(v)}{n_v-2} = I_v =\frac{n_{v_1}-\left\lceil\frac{n_v}{2}\right\rceil}{(n_v-1)-\left\lceil\frac{n_v}{2}\right\rceil}. \]
\end{proposition}
\begin{proof}
Let $T$ be a rooted binary tree and let $v \in \mathring{V}(T)$ be an inner node with $n_v>3$ descendant leaves that are partitioned into $n_{v_1}$ and $n_v-n_{v_1}$ leaves (without loss of generality $n_{v_1}\geq n_v-n_{v_1}$). Then, $bal_T(v)=n_{v_1}-(n_v-n_{v_1})=2n_{v_1}-n_v$. Using additionally that $n_v$ is even (and thus $\left\lceil\frac{n_v}{2}\right\rceil=\frac{n_v}{2}$) we have:
\begin{align*}
    \frac{bal_T(v)}{n_v-2} &= \frac{n_{v_1}-\left\lceil\frac{n_v}{2}\right\rceil}{(n_v-1)-\left\lceil\frac{n_v}{2}\right\rceil}=\frac{n_{v_1}-\frac{n_v}{2}}{n_v-1-\frac{n_v}{2}}\\
 \iff  (2n_{v_1}-n_v) \cdot  \left(n_v-1-\frac{n_v}{2}\right) &= \left(n_{v_1}-\frac{n_v}{2}\right) \cdot (n_v-2) \\
 \iff  -\frac{1}{2}n_v^2 + n_v + n_{v_1}n_v - 2n_{v_1} &= -\frac{1}{2}n_v^2 + n_v + n_{v_1}n_v - 2n_{v_1}.
\end{align*}
This completes the proof.
\end{proof}

 \subsubsection{Furnas rank}

In this section, we will list some properties of the Furnas rank \citep{Furnas1984, Kirkpatrick1993} that have been mentioned but not explicitly proven yet and provide proofs for the statements. Note that in the following statements we will use $|T|$ as a shorthand for $|V_L(T)|$.

At first, recall from Definition \ref{def_LLR_ordering} that for two rooted binary trees $T,T'\in\BTnstar$ we have $T'\prec T$ if and only if 1) $|T'|<|T|$, or 2) $|T'|=|T|$ and $T_L'\prec T_L$, or 3) $|T'|=|T|$ and $T_L'=T_L$ and $T_R'\prec T_R$, where $T_L$ and $T_R$ with $T_L\preceq T_R$ (and $T_L'$ and $T_R'$ with $T_L'\preceq T_R'$) denote the two maximal pending subtrees of $T$ (and $T'$) provided that $T$ (and $T'$) has at least two leaves. Also recall that the rank $r_n(T)$ of a tree $T$ in this LLR ordering is precisely one more than the number of trees $T'$ with $|T'|=|T|=n$ and $T'\prec T$.

\begin{remark} \label{rem_Rstat_vs_Rang}
Since the Furnas rank of a tree is by definition identical to its rank in the LLR ordering, i.e. $F(T)=r_n(T)$, we can use them interchangeably. In particular, each of the following statements for $r_n(T)$ holds also for $F(T)$.
\end{remark}

At first, we show that two trees have the same Furnas rank if and only if they are identical. Thus, the LLR ordering really induces a complete ordering on the set of binary trees $\BTnstar$.

\begin{proposition} \label{Prop_Furnas_TotalOrdering} \leavevmode
\begin{enumerate}
    \item Let $T, T' \in \BTnstar$ be distinct. Then, $T \prec T'$ or $T' \prec T.$ In particular, $r_n(T) \neq r_n(T')$.
    \item If $T = T' \in \BTnstar$, then $r_n(T)=r_n(T')$.
\end{enumerate}
\end{proposition}
\begin{proof} \leavevmode
\begin{enumerate}
    \item We prove this statement by induction on $n$. 
    For $n \in \{1,2,3\}$, there is only one tree in $\BTnstar$, and thus there is nothing to show. 
    Assume that the assertion is true for all positive integers $1 \leq n' < n$ and consider $T, T'  \in \BTnstar$ with $n \geq 4$ and $T \neq T'$. Let $T_L$ and $T_R$ denote the two maximal pending subtrees of $T$ such that $T_L \preceq T_R$ (i.e. $T_L=T_R$ or $T_L \prec T_R$). Analogously, let $T'_L$ and $T'_R$ denote the two maximal pending subtrees of $T'$ such that $T_L' \preceq T_R'$. 
    As $T \neq T'$, one of the following three cases must hold:
    \begin{enumerate}
        \item $T_L \neq T_L'$ and $T_R \neq T_R'$:
            \begin{itemize}
                \item If $|T_L| < |T_L'|$ or $|T_L'| < |T_L|$, then by Definition \ref{def_LLR_ordering}, Part 1, $T_L \prec T_L'$ or $T_L' \prec T_L$. By Definition \ref{def_LLR_ordering}, Part 2, this implies $T \prec T'$ or $T' \prec T$.
                \item If $|T_L| = |T_L'|$, by the inductive hypothesis, either $T_L \prec T_L'$ or $T_L' \prec T_L$. By Part 2 of Definition \ref{def_LLR_ordering}, this implies $T \prec T'$ or $T' \prec T$.
            \end{itemize}
        \item $T_L \neq T_L'$ and $T_R = T_R'$: This case is completely analogous to the previous case.
        \item $T_L = T_L'$ and $T_R \neq T_R'$: As $T_L=T_L'$, we can conclude that $|T_R|=|T_R'| < n$. Thus, by the inductive hypothesis either $T_R \prec T_R'$ or $T_R' \prec T_R$. By Definition \ref{def_LLR_ordering}, Part 3, $T \prec T'$ or $T' \prec T$.
    \end{enumerate}
    In all cases, either $T \prec T'$ or $T' \prec T$. Thus, either $T$ is ranked before $T'$, or $T'$ is ranked before $T$ in the LLR ordering on $\BTnstar$. In particular, the number of trees ranked before $T$ is not identical to the number of trees ranked before $T'$, and thus $r_n(T) \neq r_n(T')$.
    \item Consider $T=T' \in \BTnstar$. By definition, $r_n(T)$ is one more than the number of trees ranked before $T$ in the LLR ordering on $\BTnstar$, and $r_n(T')$ is one more than the number of trees ranked before $T'$. As $T=T'$, those numbers clearly coincide, and thus $r_n(T)=r_n(T')$. This completes the proof.
\end{enumerate}
\end{proof}

Now, we will have a look at the minimal and maximal Furnas rank for a given $n\in\mathbb{N}_{\geq 1}$ and the trees achieving these values. The following property has already been stated by \citet{Furnas1984} and \citet{Kirkpatrick1993} (but without explicit proof). 

\begin{proposition} \label{Prop_Furnas_Cat}
For each $n \in \mathbb{N}_{\geq 1}$, there exists no tree $T \in \BTnstar$ with $T \prec \Tcat$.
\end{proposition}
\begin{proof}
We prove this statement by induction on $n$.
For $n=1$, $\Tcat$ is the only element in $\BTnstar$, and there is nothing to show.
Now, assume that the assertion is true for all positive integers $1 \leq n' < n$ and consider $\BTnstar$ with $n \geq 2$. 
Let $\Tcat$ be the caterpillar tree on $n$ leaves. Since $(\Tcat)_L=(\Tcat)_R$ or $(\Tcat)_L\prec (\Tcat)_R$ must apply (see Definition \ref{def_LLR_ordering}), we have $(\Tcat)_L = T_1^\mathit{cat}$, and $(\Tcat)_{R} = T_{n-1}^\mathit{cat}$. 
Suppose there exists a tree $T \in \BTnstar$ with $T \prec \Tcat$. Let $T_L$ and $T_R$ denote the two maximal pending subtrees of $T$ such that $T_L \preceq T_R$. As $|T| = |\Tcat|$, by Definition \ref{def_LLR_ordering}, $T \prec \Tcat$ implies that
\begin{enumerate}[(i)]
    \item $T_L \prec (\Tcat)_L$; or
    \item $T_L = (\Tcat)_L$ and $T_R \prec (\Tcat)_R$.
\end{enumerate}
First, consider Case (i). As $|(\Tcat)_L| = |T^\mathit{cat}_1|=1$ and $|T_L| \geq 1$, $T_L \prec (\Tcat)_L$ cannot happen (as $\mathcal{BT}_1^\ast$ contains only one element). Now, consider Case (ii). As $T_L = (\Tcat)_L$, we can conclude that $|T_R| = |(\Tcat)_R|$. However, as $|T_R| = |(\Tcat)_R| < n$, $T_R \prec (\Tcat)_R$ contradicts the inductive hypothesis. Thus, this case cannot happen either.
In particular, there exists no $T \in \BTnstar$ with $T \prec \Tcat$. This completes the proof.
\end{proof}

\begin{proposition}\label{Prop_Furnas_MB}
For each $n \in \mathbb{N}_{\geq 1}$, there exists no tree $T \in \BTnstar$ with $\Tmb \prec T$.
\end{proposition}
\begin{proof}
We prove this statement by induction on $n$. 
For $n=1$, $\Tmb$ is the only element in $\BTnstar$, and there is nothing to show.
Assume that the assertion is true for all positive integers $1 \leq n' < n$, and consider $\BTnstar$ with $n \geq 2$. Let $\Tmb$ be the maximally balanced tree on $n$ leaves. Note that by definition of the maximally balanced tree and the fact that $(\Tmb)_L=(\Tmb)_R$ or $(\Tmb)_L\prec (\Tmb)_R$ must apply (see Definition \ref{def_LLR_ordering}), we have $(\Tmb)_L = T_{\lfloor \frac{n}{2} \rfloor}^\mathit{mb}$ and $(\Tmb)_R = T_{\lceil \frac{n}{2} \rceil}^\mathit{mb}$. 
Suppose there exists a tree $T \in \BTnstar$ with $\Tmb \prec T$. Let $T_L$ and $T_R$ denote the two maximal pending subtrees of $T$ such that $T_L \preceq T_R$. As $|T| = |\Tmb|$, by Definition \ref{def_LLR_ordering}, the assumption $\Tmb \prec T$ implies that
\begin{enumerate}[(i)]
    \item $(\Tmb)_L \prec T_L$; or
    \item $(\Tmb)_L = T_L$ and $(\Tmb)_R \prec T_R$.
\end{enumerate}
First, consider Case (i), where we need to distinguish two sub-cases:
\begin{itemize}
    \item If $|(\Tmb)_L| = |T_L| < n$, $(\Tmb)_L \prec T_L$ contradicts the inductive hypothesis. Thus, this case cannot happen.
    \item If $|(\Tmb)_L| \neq |T_L|$, the relation $(\Tmb)_L \prec T_L$ implies that $|(\Tmb)_L| < |T_L|$ (Definition \ref{def_LLR_ordering}, Part 1). Furthermore, we know $|T_L|\leq|T_R|$ and thus $|T_L|\leq \lfloor\frac{n}{2}\rfloor$, which leads to $|(\Tmb)_L| < \lfloor\frac{n}{2}\rfloor$, a contradiction. Thus, this case cannot happen.
\end{itemize}
Now, consider Case (ii). As $(\Tmb)_L = T_L$, we can conclude that $|(\Tmb)_R| = |T_R| < n$. However, in this case $(\Tmb)_R \prec T_R$ contradicts the inductive hypothesis. Thus, this case cannot happen, either. 
In particular, there exists no $T \in \BTnstar$ with $\Tmb \prec T$. This completes the proof.
\end{proof}

Now, Propositions \ref{Prop_Furnas_TotalOrdering}, \ref{Prop_Furnas_Cat} and \ref{Prop_Furnas_MB} lead to the following theorem.

\begin{theorem} \label{Cor_Furnas_ResolutionExtrema}
Let $T \in \BTnstar$ be a rooted binary tree with $n$ leaves. Then, $1 \leq r_n(T) \leq we(n)$. Moreover, $T=\Tcat$ is the unique tree with $r_n(T)=1$, and $T = \Tmb$ is the unique tree with $r_n(T) = we(n)$.
\end{theorem}
\begin{proof}
Let $T\in\BTnstar$ be a binary tree. Recall that by definition, $r_n(T)$ is one more than the number of trees ranked before $T$ in the LLR ordering on $\BTnstar$. Since the number of trees ranked before $T$ is at least zero, we have $r_n(T)\geq 0+1=1$. And since the number of trees ranked before $T$ is at most $we(n)-1$, because $\BTnstar\setminus\{T\}$ contains exactly $we(n)-1$ trees, we have $r_n(T)\leq we(n)-1+1=we(n)$.\\
Moreover, since there is no tree $T$ with $T\prec \Tcat$ (see Proposition \ref{Prop_Furnas_Cat}), we have $r_n(\Tcat)=0+1=1$. And since each of the $we(n)-1$ trees $T\neq\Tmb$ fulfills $T\prec\Tmb$, because Proposition \ref{Prop_Furnas_TotalOrdering} implies that either $T\prec\Tmb$ or $\Tmb\prec T$ must be fulfilled and Proposition \ref{Prop_Furnas_MB} implies that $\Tmb\nprec T$, we have $r_n(\Tmb)=we(n)-1+1=we(n)$. In particular, $\Tcat$ and $\Tmb$ are the unique trees reaching the bounds, because $r_n(T)=r_n(T')$ if and only if $T=T'$ (see Proposition \ref{Prop_Furnas_TotalOrdering}).
\end{proof}

\begin{remark} \label{remark_Furnas_bijection}
As $|\BTnstar| = we(n)$, and $r_n(T)$ takes integer values in $\{1, \ldots, we(n)\}$ (see Theorem \ref{Cor_Furnas_ResolutionExtrema}), and distinct trees receive distinct ranks while identical trees receive identical ranks (see Proposition \ref{Prop_Furnas_TotalOrdering}), we can conclude that $r_n$ induces a bijection between $\BTnstar$ and the set $\{1, \ldots, we(n)\}$. This has been mentioned by \citet{Rosenberg2020} before, but without explicit proof.
\end{remark}

The rank of an arbitrary tree $T$ with $n$ leaves in the LLR ordering can be computed recursively using the results in the following proposition. As we will later see in Proposition \ref{prop_furnas_original}, this recursion is identical to the one that was stated in \citep{Furnas1984} without explicit proof.

\begin{theorem} \label{th_Furnas_rank}
Let $T$ be a rooted binary tree with $n$ leaves. Then, the rank $r_n(T)$ of $T$ in the left-light rooted ordering of all trees with $n$ leaves is $r_1(T)=1$ if $n=1$ and otherwise
\begin{equation*} \label{eq_Furnas_rank}
    r_n(T) = 
    \begin{cases} 
    \displaystyle \sum\limits_{i=1}^{\alpha-1} we(i)\cdot we(n-i) + (r_\alpha(T_L)-1)\cdot we(\beta) + r_\beta(T_R) & \text{if }\alpha<\beta\\[12pt]
    \displaystyle \sum\limits_{i=1}^{\alpha-1} we(i)\cdot we(n-i) + (r_\alpha(T_L)-1)\cdot we(\beta)-\frac{r_\alpha(T_L)^2-r_\alpha(T_L)}{2} + r_\beta(T_R) & \text{if }\alpha=\beta
    \end{cases}
\end{equation*}
with $\alpha$ and $\beta$ denoting the leaf numbers of the two maximal pending subtrees $T_L$ and $T_R$ of $T$ with $T_L\preceq T_R$. 
\end{theorem}

Before we can prove Theorem \ref{th_Furnas_rank}, we need the following lemma. 

\begin{lemma}\label{lem_Furnas_rank}
Let $T=(T_L,T_R)$ be a rooted binary tree with $n\geq 2$ leaves, and let $\alpha$ and $\beta$ denote the leaf numbers of the two maximal pending subtrees $T_L$ and $T_R$, respectively. Then, we have $T_L\preceq T_R$ if and only if 
\begin{enumerate}
    \item [\textup{(i)}] $\alpha < \beta$; or 
    \item [\textup{(ii)}] $\alpha = \beta$ and  $r_\alpha(T_L)\leq r_\beta(T_R)$. 
\end{enumerate}
In addition, each distinct choice of a pair $(r_\alpha(T_L), r_\beta(T_R))$ with $1 \leq r_\alpha(T_L) \leq we(\alpha)$ and $1 \leq r_\beta(T_R) \leq we(\beta)$ that satisfies \textup{(i)} or \textup{(ii)}  yields a distinct unique tree $T=(T_L,T_R)$ with $T_L \preceq T_R$.
\end{lemma}

\begin{proof}
First, suppose that $T_L \preceq T_R$ holds. Assume for the sake of a contradiction that $\alpha > \beta$. Then, Definition \ref{def_LLR_ordering}, Part 1, immediately implies that $T_R\prec T_L$ which is a contradiction to $T_L\preceq T_R$. So, we must have $\alpha \leq \beta$. Now, if $\alpha < \beta$, Condition (i) clearly holds. If $\alpha = \beta$, then $T_L\preceq T_R$ implies that either $T_L=T_R$, in which case $r_\alpha(T_L)=r_\alpha(T_R)=r_\beta(T_R)$, or $T_L$ comes before $T_R$ in the LLR ordering of trees with $\alpha=\beta$ leaves, in which case $r_\alpha(T_L)<r_\alpha(T_R)=r_\beta(T_R)$. In particular, Condition (ii) holds. 

Now, suppose that either Condition (i) or (ii) hold. If $\alpha < \beta$, it follows from Definition \ref{def_LLR_ordering}, Part 1, that $T_L \prec T_R$, and so in particular $T_L \preceq T_R$. If $\alpha = \beta$ and $r_\alpha(T_L) =  r_\beta(T_R)$, then $T_L = T_R$, and so in particular $T_L \preceq T_R$. Finally, if $\alpha = \beta$ and $r_\alpha(T_L) < r_\beta(T_R)$, $T_L$ comes before $T_R$ in the LLR ordering of trees with $\alpha=\beta$ leaves, and thus $T_L \prec T_R$. In particular, $T_L \preceq T_R$. This completes the first part of the proof.

For the second part, first note that the conditions $1 \leq r_\alpha(T_L) \leq we(\alpha)$ and $1 \leq r_\beta(T_R) \leq we(\beta)$ simply ensure that $r_\alpha(T_L)$ and $r_\beta(T_R)$ are valid ranks (i.e. they ensure that there exists a rooted binary tree $T_L$ on $\alpha$ leaves and rank $r_\alpha(T_L)$, and analogously there exists a rooted binary tree $T_R$ on $\beta$ leaves and rank $r_\beta(T_R)$). Moreover, recall that two rooted binary trees $T=(T_L,T_R)$ and $T'=(T_L',T_R')$ are identical if and only if they have the same maximal pending subtrees, i.e. $\{T_L,T_R\}=\{T_L',T_R'\}$. Since a tree is uniquely defined by its leaf number and rank (as $\prec$ is a strict partial order, see Proposition~\ref{Prop_Furnas_TotalOrdering} and Remark~\ref{remark_Furnas_bijection}), this is equivalent to the claim that the leaf numbers and corresponding ranks of the maximal pending subtrees are identical, i.e. $\{(\alpha,r_\alpha(T_L)),(\beta,r_\beta(T_R)\}=\{(\alpha',r_{\alpha'}(T_L')),(\beta',r_{\beta'}(T_R'))\}$. Hence, for fixed $\alpha,\beta$ each distinct valid choice of $(r_\alpha(T_L),r_\beta(T_R))$ (valid in the sense that (i) or (ii) are satisfied and $r_\alpha(T_L)$ and $r_\beta(T_R)$ are valid ranks) yields a distinct unique tree $T=(T_L,T_R)$ with $T_L \preceq T_R$. This completes the proof.
\end{proof}
 
We are now in a position to prove Theorem \ref{th_Furnas_rank}. 

\begin{proof} [Proof of Theorem \ref{th_Furnas_rank}]
We prove this by induction on $n$. For $n=1$, there is only one tree in $\BTnstar$ and hence its rank is 1. This equals the starting condition $r_1(T)=1$ in Theorem \ref{th_Furnas_rank}. Assume that the assertion is true for all positive integers $1\leq n'<n$ and consider a tree $T$ with $n\geq 2$ leaves whose two maximal pending subtrees $T_L$ and $T_R$ with $T_L\preceq T_R$ have $\alpha$ and $\beta$ leaves. Since $1\leq\alpha\leq\beta\leq n-1$ the ranks of $T_L$ and $T_R$ are according to the induction assertion $r_\alpha(T_L)$ and $r_\beta(T_R)$, respectively. Since $r_n(T)$ is precisely one more than the number of trees with $n$ leaves coming before $T$ in the left-light rooted ordering, we have to count the number of different possibilities to choose a tree $T'$ with $n$ leaves and $T'\prec T$. According to the definition of the left-light rooted ordering, we have $T'\prec T$ if and only if 1.) $|T'|<|T|$, or 2.) $|T'|=|T|$ and $T_L'\prec T_L$, or 3.) $|T'|=|T|$ and $T_L'=T_L$ and $T_R'\prec T_R$ (always assuming that $T_L'\preceq T_R'$). The first case does not have to be considered, because we are only interested in the rank of $T$ among trees with the same number of leaves, i.e. $|T'|=|T|$ is trivially fulfilled. Now, consider the other two cases. For this, let $T'=(T_L',T_R')$ be a rooted binary tree with $|T'|=|T|=n\geq 2$ leaves and let $\alpha'$ and $\beta'$ denote the leaf numbers of its two maximal pending subtrees $T_L'$ and $T_R'$ with $T_L'\preceq T_R'$.\\

\noindent\textit{Assume that $\alpha<\beta$.}
\begin{enumerate}
\setcounter{enumi}{1}
\item 
\begin{enumerate}[a)]
    \item $|T'|=|T|$ and $|T_L'|<|T_L|$. Since $\alpha<\beta$ and $\alpha'=|T_L'|<|T_L|=\alpha$, we also have $\alpha'<\alpha<\beta=n-\alpha<n-\alpha'=\beta'$, i.e. in particular $\alpha'<\beta'$. Moreover, by Theorem \ref{Cor_Furnas_ResolutionExtrema} we must have $1\leq r_{\alpha'}(T_L')\leq we(\alpha')$ and $1\leq r_{\beta'}(T_R')\leq we(\beta')$ and by Lemma \ref{lem_Furnas_rank} each combination of $r_{\alpha'}(T_L')$ and $r_{\beta'}(T_R')$ yields a different unique tree. This means that for a fixed $\alpha'$, there are $we(\alpha')$ ways to choose $T_L'$ and $we(\beta')=we(n-\alpha')$ ways to choose $T_R'$ and thus $we(\alpha')\cdot we(n-\alpha')$ ways to choose $T'$. Since all $\alpha'$ with $1\leq \alpha'<\alpha$ have to be considered, we have a total of \[ \sum\limits_{i=1}^{\alpha-1} we(i)\cdot we(n-i) \] different possibilities for $T'$.
    \item $|T'|=|T|$ and $|T_L'|=|T_L|$ and $T_L'\prec T_L$. Since $\alpha<\beta$ and $\alpha'=|T_L'|=|T_L|=\alpha$ we also have $\alpha'=\alpha<\beta=n-\alpha=n-\alpha'=\beta'$, i.e. in particular $\alpha'<\beta'$. Using the fact that $T_L'\prec T_L$ and Theorem \ref{Cor_Furnas_ResolutionExtrema}, we get $1\leq r_{\alpha'}(T_L')<r_\alpha(T_L)$ and $1\leq r_{\beta'}(T_R')\leq we(\beta')$ with each combination of $r_{\alpha'}(T_L')$ and $r_{\beta'}(T_R')$ yielding a different unique tree by Lemma \ref{lem_Furnas_rank}. This means that there are $r_\alpha(T_L)-1$ ways to choose $T_L'$ and $we(\beta')=we(\beta)$ ways to choose $T_R'$ leading to a total of $(r_\alpha(T_L)-1)\cdot we(\beta)$ different possibilities for $T'$.
\end{enumerate}
\item [3.)] $|T'|=|T|$ and $T_L'=T_L$ and $T_R'\prec T_R$. Since $\alpha<\beta$ and $T_L'=T_L$ (implying that $\alpha'=\alpha)$ we also have $\alpha'=\alpha<\beta=n-\alpha=n-\alpha'=\beta'$, i.e. in particular $\alpha'<\beta'$. Using Theorem \ref{Cor_Furnas_ResolutionExtrema} as well as $T_L'=T_L$ and $T_R'\prec T_R$, we must have $r_{\alpha'}(T_L')=r_\alpha(T_L)$ and $1\leq r_{\beta'}(T_R')<r_\beta(T_R)$ with each combination of $r_{\alpha'}(T_L')$ and $r_{\beta'}(T_R')$ yielding a different unique tree by Lemma \ref{lem_Furnas_rank}. This means that there is one way to choose $T_L'$ and there are $r_\beta(T_R)-1$ ways to choose $T_R'$ leading to a total of $1\cdot(r_\beta(T_R)-1)=r_\beta(T_R)-1$ different possibilities for $T'$.
\end{enumerate}
Note that 2.a) and 2.b) combined form case 2), and that case 2.a), 2.b) and 3.) are mutually exclusive, i.e. any tree $T'$ that is generated in one of the cases cannot be generated in any of the other cases. Additionally, 2.a), 2.b) and 3.) cover all possibilities of how a tree $T'$ with $n$ leaves and $T'\prec T$ can be constructed. Since $r_n(T)$ is precisely one more than the number of trees $T'$ with $n$ leaves and $T'\prec T$, we have for $\alpha<\beta$ that
\begin{equation} \label{eq_Furnas_rank_1}
\begin{split}
    r_n(T) &= \underbrace{\sum\limits_{i=1}^{\alpha-1} we(i)\cdot we(n-i)}_{\text{case 2.a)}}+\underbrace{(r_\alpha(T_L)-1)\cdot we(\beta)}_{\text{case 2.b)}}+\underbrace{r_\beta(T_R)-1}_{\text{case 3.)}}+1\\
    &= \sum\limits_{i=1}^{\alpha-1} we(i)\cdot we(n-i) + (r_\alpha(T_L)-1)\cdot we(\beta) + r_\beta(T_R).
\end{split}
\end{equation}

\noindent\textit{Assume that $\alpha=\beta$.}
\begin{enumerate}\setcounter{enumi}{1}
\item \begin{enumerate}[a)]
\item $|T'|=|T|$ and $|T_L'|<|T_L|$. Since $\alpha=\beta$ and $\alpha'=|T_L'|<|T_L|=\alpha$ we also have $\alpha'<\alpha=\beta=n-\alpha<n-\alpha'=\beta'$, i.e. in particular $\alpha'<\beta'$. The reasoning is thus analogous to case 2.a) above (where $\alpha<\beta$), leading to a total of \[ \sum\limits_{i=1}^{\alpha-1} we(i)\cdot we(n-i) \] different possibilities for $T'$.\\
\item $|T'|=|T|$ and $|T_L'|=|T_L|$ and $T_L'\prec T_L$. Since $\alpha=\beta$ and $\alpha'=|T_L'|=|T_L|=\alpha$ we also have $\alpha'=\alpha=\beta=n-\alpha=n-\alpha'=\beta'$, i.e. in particular $\alpha'=\beta'$. Using Theorem \ref{Cor_Furnas_ResolutionExtrema} and $T_L'\prec T_L$, we get $1\leq r_{\alpha'}(T_L')< r_\alpha(T_L)$.  Moreover, as $T_L' \preceq T_R'$ and again using Theorem \ref{Cor_Furnas_ResolutionExtrema}, we also have $r_{\alpha'}(T_L')\leq r_{\beta'}(T_R')\leq we(\beta')$ with each combination of $r_{\alpha'}(T_L')$ and $r_{\beta'}(T_R')$ yielding a different unique tree by Lemma \ref{lem_Furnas_rank}. This means that for a fixed $r_{\alpha'}(T_L')$ there is one way to choose $T_L'$ and there are $we(\beta')-r_{\alpha'}(T_L')+1=we(\beta)-r_{\alpha'}(T_L')+1$ ways to choose $T_R'$ and thus there are $1\cdot(we(\beta)-r_{\alpha'}(T_L')+1)=we(\beta)-r_{\alpha'}(T_L')+1$ ways to choose $T'$. Since all $r_{\alpha'}(T_L')$ with $1\leq r_{\alpha'}(T_L')<r_\alpha(T_L)$ have to be considered, we have a total of 
\begin{equation*}
\begin{split}
    &\sum\limits_{i=1}^{r_\alpha(T_L)-1} (we(\beta)-i+1) = (r_\alpha(T_L)-1)\cdot we(\beta)+r_\alpha(T_L)-1-\sum\limits_{i=1}^{r_\alpha(T_L)-1} i \\
    &= (r_\alpha(T_L)-1)\cdot we(\beta)+r_\alpha(T_L)-1-\frac{r_\alpha(T_L)^2-r_\alpha(T_L)}{2}
\end{split}
\end{equation*}
different possibilities for $T'$.
\end{enumerate}
\item [3.)] $|T'|=|T|$ and $T_L'=T_L$ and $T_R'\prec T_R$. Since $\alpha=\beta$ and $T_L'=T_L$ (implying that $\alpha'=\alpha$) we also have $\alpha'=\alpha=\beta=n-\alpha=n-\alpha'=\beta'$, i.e. in particular $\alpha'=\beta'$. As $T_L'=T_L$ and $T_R'\prec T_R$ we must have $r_{\alpha'}(T_L')=r_\alpha(T_L)$ and $r_{\alpha'}(T_L')\leq r_{\beta'}(T_R')<r_\beta(T_R)$ with each combination of $r_{\alpha'}(T_L')$ and $r_{\beta'}(T_R')$ yielding a different unique tree by Lemma \ref{lem_Furnas_rank}. This means that there is one way to chose $T_L'$ and there are $r_\beta(T_R)-r_{\alpha'}(T_L')=r_\beta(T_R)-r_\alpha(T_L)$ ways to choose $T_R'$ leading to a total of $1\cdot(r_\beta(T_R)-r_\alpha(T_L))=r_\beta(T_R)-r_\alpha(T_L)$ different possibilities for $T'$.
\end{enumerate}
Note again that 2.a) and 2.b) combined form case 2), and that case 2.a), 2.b) and 3) are mutually exclusive. Additionally, 2.a), 2.b) and 3) cover all possibilities of how a tree $T'$ with $n$ leaves and $T'\prec T$ can be constructed. Since $r_n(T)$ is precisely one more than the number of trees $T'$ with $n$ leaves and $T'\prec T$, we have for $\alpha=\beta$ that 
\begin{equation*}
\begin{split}
    r_n(T) &= \underbrace{\sum\limits_{i=1}^{\alpha-1} we(i)\cdot we(n-i)}_{\text{case 2.a)}} + \underbrace{(r_\alpha(T_L)-1)\cdot we(\beta)+r_\alpha(T_L)-1-\frac{r_\alpha(T_L)^2-r_\alpha(T_L)}{2}}_{\text{case 2.b)}} + \underbrace{r_\beta(T_R)-r_\alpha(T_L)}_{\text{case 3.)}} + 1\\
    &= \sum\limits_{i=1}^{\alpha-1} we(i)\cdot we(n-i) + (r_\alpha(T_L)-1)\cdot we(\beta)-\frac{r_\alpha(T_L)^2-r_\alpha(T_L)}{2} + r_\beta(T_R).
\end{split}
\end{equation*}
In total, we have shown that the rank of a tree $T$ in the LLR ordering of all trees with the same leaf number $n$ equals $r_n(T)$ as given in Theorem \ref{th_Furnas_rank}. This completes the proof.
\end{proof}

In the following proposition, we show that the formula for the rank $r_n(T)$ that we presented in Theorem \ref{th_Furnas_rank} is indeed identical to the one given by \citet{Furnas1984}.

\begin{proposition}\label{prop_furnas_original}
The recursion stated in Theorem \ref{th_Furnas_rank} is equivalent to the original recursion stated by Furnas \citep[Section~2.5.1.2]{Furnas1984}, which is defined as follows:\\
Let $T$ be a rooted binary tree with $n$ leaves. Then, the rank $r_n(T)$ of $T$ in the left-light rooted ordering of all trees with $n$ leaves is $r_1(T)=1$ if $n=1$ and otherwise
\begin{equation*}
    r_n(T) = 
    \begin{cases} 
    \displaystyle \sum\limits_{i=1}^{\alpha-1} we(i)\cdot we(n-i) + (r_\alpha(T_L)-1)\cdot we(\beta) + (r_\beta(T_R)-1)+1 & \text{if }\alpha<\beta\\
    \displaystyle \sum\limits_{i=1}^{\alpha-1} we(i)\cdot we(n-i) +  \frac{we(\alpha)\cdot (we(\alpha)+1)}{2} & \\
    \quad - \frac{(we(\alpha)-r_\alpha(T_L)+1)\cdot(we(\alpha)-r_\alpha(T_L)+2)}{2} + (r_\beta(T_R) - 1)  - (r_\alpha(T_L) - 1) + 1&\text{if }\alpha=\beta
    \end{cases}
\end{equation*}
with $\alpha$ and $\beta$ denoting the leaf numbers of the two maximal pending subtrees $T_L$ and $T_R$ of $T$ with $T_L\preceq T_R$.
\end{proposition}
\begin{proof}
In order to show that the original recursion stated by Furnas equals the one given in Theorem \ref{th_Furnas_rank}, we need to show that their start values as well as their recursions are equal. The start values are both at $n=1$ with $r_1(T)=1$ and their recursions for the case $\alpha<\beta$ are clearly equal as well. Now, consider the case $\alpha=\beta$. The original Furnas recursion can be rewritten as follows:
\begin{equation*}
\begin{split}
    r_n(T) &= \sum\limits_{i=1}^{\alpha-1} we(i)\cdot we(n-i) + \frac{we(\alpha)\cdot (we(\alpha)+1)}{2}\\
    &\quad - \frac{(we(\alpha)-r_\alpha(T_L)+1)\cdot(we(\alpha)-r_\alpha(T_L)+2)}{2} + (r_\beta(T_R) - 1)  - (r_\alpha(T_L) - 1) + 1\\
    &= \sum\limits_{i=1}^{\alpha-1} we(i)\cdot we(n-i) + \frac{we(\alpha)^2}{2} + \frac{we(\alpha)}{2} - \frac{we(\alpha)^2}{2} + \frac{we(\alpha)\cdot r_\alpha(T_L)}{2} - \frac{2\cdot we(\alpha)}{2}\\
    &\quad + \frac{r_\alpha(T_L)\cdot we(\alpha)}{2} - \frac{r_\alpha(T_L)^2}{2} + \frac{2\cdot r_\alpha(T_L)}{2} - \frac{we(\alpha)}{2} + \frac{r_\alpha(T_L)}{2} - \frac{2}{2} + r_\beta(T_R) - r_\alpha(T_L) + 1\\
    &= \sum\limits_{i=1}^{\alpha-1} we(i)\cdot we(n-i) + r_\alpha(T_L)\cdot we(\alpha) - we(\alpha) - \frac{r_\alpha(T_L)^2}{2} + \frac{r_\alpha(T_L)}{2} + r_\beta(T_R)\\
    &\quad \\
    &= \sum\limits_{i=1}^{\alpha-1} we(i)\cdot we(n-i) + (r_\alpha(T_L)-1)\cdot \underbrace{we(\alpha)}_{\text{use }\alpha=\beta} - \frac{r_\alpha(T_L)^2-r_\alpha(T_L)}{2} + r_\beta(T_R)\\
    &= \sum\limits_{i=1}^{\alpha-1} we(i)\cdot we(n-i) + (r_\alpha(T_L)-1)\cdot we(\beta) - \frac{r_\alpha(T_L)^2-r_\alpha(T_L)}{2} + r_\beta(T_R)\\
\end{split}
\end{equation*}
which equals the recursion in Theorem \ref{th_Furnas_rank}. Thus, if $\alpha=\beta$ the recursions are equal as well. The proof is complete.
\end{proof}

Next, we show a procedure to invert Furnas' rank function. The procedure described here, which is formally stated by Algorithm \ref{alg:inverseFurnas}, is very basic, as it is mainly based on the Euclidian division theorem. It differs only slightly from the procedure stated in \citep{Furnas1984} (namely in the case where the left and the right subtrees have the same size), and in the light of Proposition \ref{prop_furnas_original}, it is obvious that both inversions are equivalent. However, our proof is both more basic and more detailed, which makes the algorithm more easily accessible.

\begin{proposition}\label{prop_inverseFurnas} 
Given $n$ and $r_n(T)$, the recursion given by Theorem \ref{th_Furnas_rank} can be inverted using Algorithm \ref{alg:inverseFurnas}. In particular, we have $\mathtt{inv}(n,r_n(T))=T$. 
\end{proposition}
\begin{proof} If $n=1$, there is only one tree, namely the one consisting of only one vertex, and this is recovered by Algorithm \ref{alg:inverseFurnas}, so there is nothing more to show. Now, given a leaf number $n\geq 2$ and a Furnas tree rank $r$, we know by Theorem \ref{th_Furnas_rank} that for the corresponding tree $T=(T_L,T_R)$ with $n_L$ and $n_R$ (with $n_L \leq n_R$) leaves in its maximal pending subtrees, respectively, at least $\sum\limits_{i=1}^{n_L-1} we(i)\cdot we(n-i)$ trees must have a rank smaller than $r$. This is due to the fact that all trees in $\BTnstar$ whose left maximal pending subtree has strictly fewer than $n_L$ leaves get a smaller rank than $T$. This leads to \begin{equation}\label{eq_furnas_1}r > \sum\limits_{i=1}^{n_L-1} we(i)\cdot we(n-i). \end{equation} On the other hand, we know that $r$ is bounded by the number of all trees whose left subtrees have at most $n_L$ leaves, which gives \begin{equation}\label{eq_furnas_2}r\leq \sum\limits_{i=1}^{n_L} we(i)\cdot we(n-i) .\end{equation}

In summary, this shows that $n_L=\alpha$ as stated by Algorithm \ref{alg:inverseFurnas}, as $\alpha$ is the unique number fulfilling both properties. This also immediately shows that $n_R=n-n_L=n-\alpha=\beta$. In particular, we have $\alpha\leq \beta$ because $n_L\leq n_R$. So we already know that Algorithm \ref{alg:inverseFurnas} correctly reconstructs the sizes of the maximal pending subtrees. 

It only remains to show that also their ranks $r_{\alpha}$ and $r_{\beta}$ are correctly reconstructed by Algorithm \ref{alg:inverseFurnas}. 

 We now distinguish two cases: $\alpha<\beta$ and $\alpha=\beta$. First consider the case $\alpha<\beta$:
Let $t= r-\sum\limits_{i=1}^{\alpha-1} we(i)\cdot we(n-i)>0$, where the latter inequality is true due to Equation \eqref{eq_furnas_1}. Then, using the Euclidean division theorem, there are unique integers $a$ and $b$ such that $t=a\cdot we(\beta) + b$ with $0\leq b < we(\beta)$, and these integers are precisely the ones defined in Algorithm \ref{alg:inverseFurnas}. Note that by the definition of $\alpha$, we have $0< t= r-\sum\limits_{i=1}^{\alpha-1}we(i)\cdot we(n-i) \leq \sum\limits_{i=1}^{\alpha}we(i)\cdot we(n-i)-\sum\limits_{i=1}^{\alpha-1}we(i)\cdot we(n-i)=we(\alpha)\cdot we(n-\alpha)=we(\alpha)\cdot we(\beta) $, which in turn shows that $a=\frac{t-b}{we(\beta)} \leq \frac{we(\alpha) \cdot we(\beta)-b}{we(\beta)} \leq \frac{we(\alpha)\cdot we(\beta) }{we(\beta)}=we(\alpha)$. Moreover, we clearly have $a\geq 0$ as $a$ is a non-negative integer by the Euclidean division theorem. 

We now consider the following subcases:
\begin{itemize}
    \item If $\alpha<\beta$ and additionally $b>0$, then Algorithm \ref{alg:inverseFurnas} sets $r_{\beta}=b$ and $r_{\alpha}=a+1$. This leads to $t=a\cdot we(\beta) + b=(r_{\alpha}-1) \cdot we(\beta) + r_{\beta}$, and thus, by definition of $t$, to $r=\sum\limits_{i=1}^{\alpha-1} we(i)\cdot we(n-i)+(r_{\alpha}-1) \cdot we(\beta) + r_{\beta}$. Note that as $b\in \{1,\ldots,we(\beta)-1\}$ and $r_{\beta}=b$, we have $1\leq  r_{\beta}< we(\beta)$. Moreover,  as $a\geq 0$, we have $1 \leq a+1 = r_{\alpha}$.   
    
    Now suppose $r_{\alpha}>we(\alpha)$. As $a\leq we(\alpha)$ and $r_{\alpha}=a+1$, this implies $a=we(\alpha)$ and thus $t=we(\alpha)\cdot we(\beta) + b$. As we have $b>0$, this contradicts $t\leq we(\alpha)\cdot we(\beta)$, so the assumption was wrong and we must have $r_{\alpha}\leq we(\alpha)$.
    
    \item If $\alpha<\beta$ and additionally $b=0$, then Algorithm \ref{alg:inverseFurnas} sets $r_{\beta}=we(\beta)$ (and thus, in particular, $1 \leq r_\beta \leq we(\beta)$) and $r_{\alpha}=a$. This leads to $t=a\cdot we(\beta) =r_{\alpha} \cdot we(\beta) =(r_{\alpha}-1) \cdot we(\beta)+we(\beta)=(r_{\alpha}-1) \cdot we(\beta)+r_{\beta}$. Thus, by definition of $t$, this gives $r=\sum\limits_{i=1}^{\alpha-1} we(i)\cdot we(n-i)+(r_{\alpha}-1) \cdot we(\beta) + r_{\beta}.$ 
    Note that as $b=0$, we clearly have $a>0$ as (by the choice of $\alpha$) $t=r-\sum\limits_{i=1}^{\alpha-1}we(i)\cdot we(n-i)>0$. Thus, $r_{\alpha}=a\geq 1$. 
    On the other hand, however, as $a\leq we(\alpha)$ as explained before, we clearly have $r_{\alpha}\leq we(\alpha)$. 
\end{itemize}

So when $\alpha<\beta$, in both subcases the unique values $a$ and $b$ with $0\leq b<we(\beta)$ provided by the Euclidean algorithm lead to unique values $r_{\alpha}$ and $r_{\beta}$ with the property that  $1\leq r_{\beta}\leq we(\beta)$ as well as  $1\leq r_{\alpha}\leq we(\alpha)$ and which fulfill the recursion stated by Theorem \ref{th_Furnas_rank}. These values are the ones recovered by Algorithm \ref{alg:inverseFurnas} and they must be fulfilled by any tree $T$ with $r_n(T)=r$. This shows that $\mathtt{inv}(n,r_n(T))=T$ for all rooted binary trees $T$ with $\alpha<\beta$.

It remains to consider the case $\alpha=\beta$. In this case, we have $we(\alpha)=we(\beta)$. We first show that a value of $m$ with the following property always exists:

$$m=\min\left\{z \in \mathbb{N}_{\geq 0}: \sum\limits_{j=0}^z \left(we(\beta)-j\right)\geq r - \sum\limits_{i=1}^{\alpha-1}we(i)\cdot we(n-i)\right\}.$$

Moreover, we will show that this value coincides with $m$ as stated in Algorithm \ref{alg:inverseFurnas}.

Note that if $\alpha=\beta=\frac{n}{2}$, then a valid Furnas rank will always ensure that $r_{\alpha}\leq r_{\beta}$. However, this means that for each tree for which the left subtree has rank $r_{\alpha}$, the rank of the right subtree must be contained in $\{r_{\alpha}, r_{\alpha}+1,\ldots, we(\beta)\}$. So for $r_{\alpha}=1$, we have $we(\beta)$ many trees, for $r_{\alpha}=2$, we have $we(\beta)-1$ many trees and so forth. Moreover, we note that all trees in which the left subtree has fewer leaves than the right one, of which there are $\sum\limits_{i=1}^{\alpha-1}we(i)\cdot we(n-i) $ many, are ranked before the trees with $\alpha=\beta$.  Thus, we can guarantee that $\underbrace{\sum\limits_{j=0}^{we(\beta)-1}(we(\beta)-j)}_{\text{number of trees with } \alpha = \beta}+\underbrace{\sum\limits_{i=1}^{\alpha-1}we(i)\cdot we(n-i)}_{\text{number of trees with } \alpha < \beta}\geq r$, which shows that an integer $m$ as required must exist and that $m\leq we(\beta)-1$.  

We will now show that we can simplify the calculation of $m$ in order to make it more efficient. Using the minimality of $m$, we know that $m$ is the smallest integer $z$ for which we have $\sum\limits_{j=0}^z (we(\beta)-j) \geq r-\sum\limits_{i=1}^{\alpha-1} we(i)\cdot we(n-i)$. 

Re-arranging the inequality and using the shorthand $rsum$ for $r-\sum\limits_{i=1}^{\alpha-1} we(i)\cdot we(n-i)$ leads to: \begin{align*}\sum\limits_{j=0}^z (we(\beta)-j) =(z+1)\cdot we(\beta)-\frac{z(z+1)}{2} &\geq rsum \\ 
\Leftrightarrow&\\
-\frac{1}{2}z^2+\left(we(\beta)-\frac{1}{2}\right)z +\left( we(\beta)-rsum\right)&\geq 0\\
\Leftrightarrow&\\
z^2+\left(1-2we(\beta)\right)z +2\left( rsum-we(\beta)\right)&\leq 0
\end{align*}

The latter holds whenever the left-hand side term of this inequality is contained in the interval between the at most two roots of the corresponding quadratic equation $z^2+\left(1-2we(\beta)\right)z +2\left( rsum-we(\beta)\right)= 0$. As we know that $z=we(\beta)-1$ is a valid solution, this point must be contained in this interval, so the interval cannot be empty. Thus, the equation has at least one root. Now as we are searching for the smallest non-negative integer in this interval, this can only be the smaller root rounded up to the next integer, or it is 0 in case this integer is negative (as 0 is then clearly the smallest non-negative integer in the described interval).

In summary, we get:

\begin{equation*}
m=\max\left\{0, \left\lceil we(\beta)-\frac{1}{2} -\sqrt{2(we(\beta)-rsum) + \frac{(1-2we(\beta))^2}{4}} \ \right\rceil\right\},
\end{equation*}

which is precisely the term stated in Algorithm \ref{alg:inverseFurnas}. 

We now show that the values of $r_{\alpha}$ and $r_{\beta}$ as assigned by Algorithm \ref{alg:inverseFurnas} are in the range of $\{1,\ldots,we(\alpha)=we(\beta)\}$:
\begin{itemize}
    \item We start with $r_{\alpha}=m+1$. As $m\in \mathbb{N}_{\geq 0}$ by definition, we clearly have $r_{\alpha}\geq 1$. Moreover, as we have already seen that $m\leq we(\beta)-1$, this shows that $r_{\alpha}\leq we(\beta)$.  So in summary, we have $r_{\alpha}\in \{1,\ldots,we(\alpha)=we(\beta)\}$.
   \item Now we consider $r_{\beta}=r - \sum\limits_{i=1}^{\alpha-1}we(i)\cdot we(n-i)-(r_{\alpha}-1) \cdot we(\beta)+\frac{(r_{\alpha}-2)(r_{\alpha}-1)}{2}+r_{\alpha}-1$, which can be re-arranged to 
    $r_{\beta}=r - \sum\limits_{i=1}^{\alpha-1}we(i)\cdot we(n-i)-\sum\limits_{j=0}^{r_{\alpha}-2}\left(we(\beta)-j\right)+r_{\alpha}-1$,  which in turn equals $
      r - \sum\limits_{i=1}^{\alpha-1}we(i)\cdot we(n-i)-\sum\limits_{j=0}^{m-1}\left(we(\beta)-j\right)+m$. Note that as explained above, for $m$ we have:  

      $$ \sum_{j=0}^{m-1} \left(we(\beta)-j\right) < r -\sum\limits_{i=1}^{\alpha-1}we(i)\cdot we(n-i) \leq \sum_{j=0}^{m} \left(we(\beta)-j\right). $$
      In particular,
      $$0< \underbrace{r - \sum\limits_{i=1}^{\alpha-1}we(i)\cdot we(n-i) - \sum\limits_{j=0}^{m-1}(we(\beta)-j)}_{= r_\beta - m} \leq \sum\limits_{j=0}^{m}(we(\beta)-j)-\sum\limits_{j=0}^{m-1}(we(\beta)-j)=we(\beta)-m,$$ and thus, we clearly have: 
      $$0<r_{\beta}=\underbrace{r - \sum\limits_{i=1}^{\alpha-1}we(i)\cdot we(n-i)-\underbrace{\sum\limits_{j=0}^{m-1}\left(we(\beta)-j\right)}_{<r-\sum\limits_{i=1}^{\alpha-1}we(i)\cdot we(n-i)}}_{>0} +m\leq we(\beta)-m+m=we(\beta). $$ So in summary, we have $r_{\beta}\in \{1,\ldots,we(\alpha)=we(\beta)\}$ as required. Moreover, we have $r_{\beta}>m$ and thus $r_{\beta} \geq m+1=r_{\alpha}$ as required in the LLR ordering.
\end{itemize}
So both $r_{\alpha}$ and $r_{\beta}$ are in the correct range and we have $r_{\alpha} \leq r_{\beta}$. It only remains to show that these values give us the correct term for $r$.
Re-arranging the definition of $r_{\beta}$ from Algorithm \ref{alg:inverseFurnas}, we derive \begin{align*}r&=r_{\beta}+\sum\limits_{i=1}^{\alpha-1} we(i)\cdot we(n-i) + \sum\limits_{j=0}^{r_{\alpha}-2}(we(\beta)-j) - r_{\alpha}+1\\
&=r_{\beta}+\sum\limits_{i=1}^{\alpha-1} we(i)\cdot we(n-i) + \underbrace{\left(\sum\limits_{j=0}^{r_{\alpha}-2} we(\beta)\right)}_{=(r_{\alpha}-1) \cdot we(\beta)} - \underbrace{\left(\sum\limits_{j=0}^{r_{\alpha}-2} j\right)}_{= \frac{(r_{\alpha}-2)(r_{\alpha}-1)}{2}} -\frac{2r_{\alpha}}{2}+\frac{2}{2} \\ &= r_{\beta}+\sum\limits_{i=1}^{\alpha-1} we(i)\cdot we(n-i) + (r_{\alpha}-1) \cdot we(\beta) +  \frac{-(r_{\alpha}^2-3r_{\alpha}+2)-2r_{\alpha}+2}{2}  \\ &= r_{\beta}+\sum\limits_{i=1}^{\alpha-1} we(i)\cdot we(n-i) + (r_{\alpha}-1) \cdot we(\beta) -  \frac{r_{\alpha}^2-r_{\alpha}}{2}.
\end{align*}
Note that the latter term coincides with $r_n(T)$ as stated by Theorem \ref{th_Furnas_rank}, so our values of $r_{\alpha}$ and $r_{\beta}$ fulfill the required recursion. As by Proposition \ref{Prop_Furnas_TotalOrdering} the value $r_n(T)$ is unique for $T$, this shows that the values of $r_{\alpha}$ and $r_{\beta}$ given by Algorithm \ref{alg:inverseFurnas} are correct. So we have  $\mathtt{inv}(n,r_n(T))=T$ for all rooted binary trees $T$ with $\alpha=\beta$. This completes the proof.
\end{proof}

\vspace{0.5 cm}
{\small
\begin{algorithm}[H]
\caption{Inversion $\mathtt{inv}(n,r)$ of the recursion $r_n$ (Theorem \ref{th_Furnas_rank})}\label{alg:inverseFurnas}
 \SetKwInOut{Input}{Input}\SetKwInOut{Output}{Output}
 \vspace{0.15 cm}
 \Input{number of leaves $n$, tree rank $r$}
 \Output{tree $T$ with $n$ leaves and with $r_n(T)=r$} 
 \If {$n=1$} {\Return single node}
 \Else 
 {$\alpha:= \min\left\{j \in \mathbb{N}_{\geq 1}: \ \ \sum\limits_{i=1}^j we(i)\cdot we(n-i) \geq r\right\}$ \\
$\beta:=n-\alpha$ \\ 
\If{$\alpha < \beta$}{  $b:= \left( r-\sum\limits_{i=1}^{\alpha-1} we(i)\cdot we(n-i)\right) \mod we(\beta)$ \\ 
     $a:=\frac{1}{we(\beta)}\cdot\left(\left(r-\sum\limits_{i=1}^{\alpha-1} we(i)\cdot we(n-i)\right)-b\right)$\\
     \If{$b>0$}{ $r_{\beta}:=b$ \\ $r_{\alpha}:=a+1$}
     \Else{$r_{\beta}:=we(\beta)$ \\ $r_{\alpha}:=a$}}
     \If{$\alpha=\beta$}{
     $m:=\max\left\{0, \left\lceil we(\beta)-\frac{1}{2} -\sqrt{2\left(we(\beta)-r+\sum\limits_{i=1}^{\alpha-1}we(i)\cdot we(n-i)\right) + \frac{(1-2we(\beta))^2}{4}} \ \right\rceil\right\}$\\[8pt]
      $r_{\alpha}:=m+1$\\
     $r_{\beta}:=r - \sum\limits_{i=1}^{\alpha-1}we(i)\cdot we(n-i)-(r_{\alpha}-1) \cdot we(\beta)+\frac{(r_{\alpha}-2)(r_{\alpha}-1)}{2}+r_{\alpha}-1
     $}
      \Return{$T:=(\mathtt{inv}(\alpha,r_{\alpha}),\mathtt{inv}(\beta,r_{\beta}))$}
 }
\end{algorithm}
}

\vspace{0.3 cm}

\begin{proposition} \label{runtime_Furrank_alg}
Algorithm \ref{alg:inverseFurnas} has a computation time in $O(n^2)$. 
\end{proposition}

\begin{proof}
For an efficient computation, a vector containing the Wedderburn-Etherington numbers $we(i)$ for $i=1,\ldots,n$ has to be computed first. This can be done in time $O(n^2)$ \citep{Furnas1984}. The actual algorithm requires $O(n)$ recursive calls. Within each call, $\alpha$ can be computed in time $O(n)$ based on the vector of Wedderburn-Etherington numbers. If the sum in the computation of $\alpha$ is saved, $\beta$, $b$ and $a$ can be computed from this information in constant time. This already shows that in all cases where $\alpha<\beta$ (and thus particularly in all cases where $n$ is odd), we can definitely guarantee a run time in $O(n^2)$.

Moreover, in case $\alpha=\beta$, once we have calculated $m$, which can be done in constant time given that the Wedderburn-Etherington numbers have been pre-calculated, it is obvious that the calculation of $r_{\alpha}$ and $r_{\beta}$ from this information also only takes constant time.  This leads to a total computation time in $O(n^2)$, which completes the proof. 
\end{proof}

\begin{remark} Note that the Furnas rank can be easily modified such that it can be used to enumerate the space of all binary trees by first listing all trees with one, then two, then three, etc. leaves. This way, the number assigned to a caterpillar with $n$ leaves would not be 1, but instead $1+ \sum\limits_{i=1}^{n-1}we(i)$. This enumeration can then be inverted by a slight modification of Algorithm \ref{alg:inverseFurnas}, too. Using this tree space enumeration as a balance index, however, would require some sort of normalization in order not to regard trees as more imbalanced simply because they have fewer leaves. This is why we decided to consider the Furnas rank depending on $n$ instead. Nevertheless, the function \texttt{treenumber} computing the rank of a tree in the LLR ordering of all trees (not just the ones with the same leaf number) and its inverse function \texttt{treenumber\_inv} can be found in our \textsf{R} package \texttt{treebalance} as well.
\end{remark}

Last but not least, we want to add three results about the computation time, the recursiveness in the sense of \citet{Matsen2007} and the locality of the Furnas rank.

\begin{proposition} \label{runtime_Furrank}
For every binary tree $T\in\BTnstar$, the Furnas rank $F(T)$ can be computed in time $O(n^2)$.
\end{proposition}
\begin{proof}
Firstly, compute a vector containing the values $n_u$ for each $u\in V(T)$. This can be done in time $O(n)$ by traversing the tree in post order, setting $n_u=1$ if $u$ is a leaf and calculating $n_u=n_{u_1}+\ldots+n_{u_k}$ (where $u_1, \ldots, u_k$ denote the children of $u$) otherwise. Secondly, compute a vector containing the Wedderburn-Etherington numbers $we(1),\ldots,we(n)$, which can be done in time $O(n)$ (see \cite[Section 2.5.3]{Furnas1984}). Thirdly, compute a matrix containing the values $z(\widetilde{n},k)=\sum\limits_{i=1}^k we(i)\cdot we(\widetilde{n}-i)$ for all $\widetilde{n}=1,\ldots,\lfloor\frac{n}{2}\rfloor$ and all $k=1,\ldots,\widetilde{n}$. This can be done in time $O(n^2)$ if the vector of Wedderburn-Etherington numbers is used and $z(\widetilde{n},k)$ is calculated recursively as $z(\widetilde{n},1)=we(n-1)$ and $z(\widetilde{n},k)=z(\widetilde{n},k-1)+we(k)\cdot we(n-k)$ for $k>1$. Then, the Furnas rank of $T$ can be computed from the vectors and matrix in time $O(n)$ by using the recursion in Theorem \ref{th_Furnas_rank}. Since the most time consuming step is the computation of the $z(\widetilde{n},k)$-table, which can be done in time $O(n^2)$, the total computation time is in $O(n^2)$.
\end{proof}

\begin{proposition} \label{recursiveness_Furrank}
The Furnas rank is a binary recursive tree shape statistic. 
\end{proposition}
\begin{proof}
The Furnas rank is already defined recursively (see Theorem~\ref{th_Furnas_rank}). In order to prove that it is also a binary recursive tree shape statistic in the sense of \citet{Matsen2007}, we have to show that each recursion has a single real start value and is independent of the order of subtrees. To do this, we will use indicator, maximum and minimum functions to bypass the condition $T_L\preceq T_R$. Let $T=(T_1,T_2)$ be a rooted binary tree and denote by $n\geq 2$, $n_1$ and $n_2$ the leaf numbers of $T$, $T_1$ and $T_2$. Ordering $T_1$ and $T_2$ according to the left-light rooted ordering leads to $T=(T_L,T_R)$ with $T_L\preceq T_R$, where either $T_1=T_L$ and $T_2=T_R$ or $T_1=T_R$ and $T_2=T_L$. Let $\alpha$ denote the leaf number of $T_L$. Now, using Lemma \ref{lem_Furnas_rank} we know that $\alpha=\min\{n_1,n_2\}$, and that $n_1<n_2$ implies $T_L=T_1$ and $T_R=T_2$, and that $n_2<n_1$ implies $T_L=T_2$ and $T_R=T_1$, and that $n_1=n_2$ implies $r_\alpha(T_L)=\min\{r_{n_1}(T_1),r_{n_2}(T_2)\}$ and $r_\beta(T_R)=\max\{r_{n_1}(T_1),r_{n_2}(T_2)\}$.

Using this, the Furnas rank can be expressed as a binary recursive tree shape statistic of length $x=2$ with the following recursions (where $r_i$ and $n_i$ denote the rank and leaf number of $T_i$):
\begin{itemize}
    \item Furnas rank: $\lambda_1=1$ and 
    \begin{equation*}
    \begin{split} 
    r_1(T_1,T_2) &= \sum\limits_{i=1}^{\min\{n_1,n_2\}-1} we(i)\cdot we(n_1+n_2-i)\\
    &\qquad + \Big((r_1-1)\cdot we(n_2)+r_2\Big)\cdot\mathcal{I}(n_1<n_2)\\
    &\qquad + \Big((r_2-1)\cdot we(n_1)+r_1\Big)\cdot\mathcal{I}(n_2<n_1)\\
    &\qquad + \left((\min\{r_1,r_2\}-1)\cdot we(n_1)-\frac{\min\{r_1,r_2\}^2-\min\{r_1,r_2\}}{2}+\max\{r_1,r_2\}\right)\cdot\mathcal{I}(n_1=n_2)
    \end{split}
    \end{equation*}
    \item leaf number: $\lambda_2=1$ and $r_2(T_1,T_2)=n_1+n_2$
\end{itemize}
It can easily be seen that $\lambda\in\mathbb{R}^2$ and $r_i:\mathbb{R}^2\times\mathbb{R}^2\rightarrow\mathbb{R}$. Since the last line in the recursion of the Furnas rank is only relevant for $n_1=n_2$ we can use $we(n_1)$ without contradicting that the $r_i$ are independent of the order of subtrees. This completes the proof.
\end{proof}

\begin{proposition} \label{locality_Furrank}
The Furnas rank is not local.
\end{proposition}
\begin{proof}
Consider the two trees $T$ and $T'$ in Figure \ref{fig_locality}, which only differ in their subtrees rooted at $v$. Note that in both $T$ and $T'$ the vertex $v$ has exactly 5 descendant leaves. Nevertheless, we have $F(T)-F(T')=95-98=-3\neq -2=1-3=F(T_v)-F(T_v')$. Thus, the Furnas rank is not local.
\end{proof}

\subsubsection{\texorpdfstring{$I$}{I}-based indices}

This section summarizes results on the class of $I$-based indices. Do note that not all $I$-based indices fulfill our definition of an (im)balance index. However, due to their close relatedness, we group them here. Also note that originally \citet{Fusco1995} allowed each leaf of a tree to represent one or more species and considered the number of descendant species rather than the number of descending leaves of a vertex. However, to stay in line with the other indices we assume in the following that each leaf represents precisely one species.

Now, recall that the imbalance value $I_v$ \cite{Fusco1995} of a binary node $v$ with $n_v\geq 4$ and its correction $I'_v$ \citep{Purvis2002} are defined as
    \[ I_v \coloneqq \frac{n_{v_1}- \left \lceil \frac{n_v}{2} \right \rceil}{(n_v-1)- \left \lceil \frac{n_v}{2} \right \rceil} \qquad  \text{and} \qquad I_v' \coloneqq \begin{cases} 
    I_v &\text{if }n_v \text{ is odd} \\ 
    \frac{n_v-1}{n_v} \cdot I_v &\text{if }n_v \text{ is even.} \end{cases} \]
The $I^w_v$ value of a vertex $v$ is another correction method (namely a weighted version) of the $I_v$ value and is defined as  
\[I^w_v \coloneqq \frac{w(I_v) \cdot I_v}{\text{mean}_{v \in \mathring{V}_{bin,\geq 4}} w(I_v)} \qquad \text{with weights} \qquad w(I_v) \coloneqq \begin{cases} 
    1 & \text{if } n_v \text{ is odd} \\ 
    \frac{n_v-1}{n_v} & \text{if }n_v \text{ is even and } I_v>0 \\
    \frac{2\cdot(n_v-1)}{n_v} & \text{if } n_v \text{ is even and } I_v=0 
\end{cases}\]
with $\mathring{V}_{bin,\geq 4}$ denoting the set of inner vertices $v$ of $T$ that have precisely two children and $n_v\geq 4$.

Recall that based on these imbalance values, the $I$ value $I_\rho(T)$, $I'$ value $I'_\rho(T)$, Total $I$ index $\Sigma I(T)$, Total $I'$ index $\Sigma I'(T)$, Mean $I$ index $\overline{I}(T)$, and Mean $I'$ index $\overline{I'}(T)$ are defined as follows. First, if $T \in \Tnstar$ is a rooted tree with binary root $\rho$,
\begin{align*}
    I_\rho(T) \coloneqq I_{\rho} \quad \text{and} \quad I'_\rho(T) \coloneqq I'_{\rho}.
\end{align*}
Second, for $T \in \Tnstar$, we have
\begin{align*}
    \Sigma I(T) \coloneqq \sum\limits_{v \in \mathring{V}_{bin,\geq 4}(T)} I_v \qquad \text{and} \qquad \Sigma I'(T) \coloneqq \sum\limits_{v \in \mathring{V}_{bin,\geq 4}(T)} I_v',
\end{align*}
and finally
\begin{align*}
    \overline{I}(T) \coloneqq \frac{1}{|\mathring{V}_{bin,\geq 4}(T)|} \cdot  \sum\limits_{v \in \mathring{V}_{bin,\geq 4}(T)} I_v \quad \text{and} \quad \overline{I'}(T) \coloneqq \frac{1}{|\mathring{V}_{bin,\geq 4}(T)|} \cdot  \sum\limits_{v \in \mathring{V}_{bin,\geq 4}(T)} I_v'.
\end{align*}

Although the above definitions are for arbitrary trees, these measures are only meaningful for binary trees or such arbitrary trees that have a small percentage of non-binary vertices \citep{Fusco1995}.

While most of the results given below only concern the $I$ value, $I'$ value, Total $I$ index, Total $I'$ index, Mean $I$ index and Mean $I'$ index, we can make statements about the computation time of applying statistics to the (corrected) $I_v$ values in general.

\begin{proposition} \label{runtime_Ibased}
For every tree $T\in\Tnstar$ the computation time of applying a statistic to the $I_v$ values (with or without correction $I'$ or $I^w$) of all binary vertices $v\in \mathring{V}(T)$ with $n_v\geq4$ only depends on the computation time of the respective statistic, but is at least linear.
\end{proposition}
\begin{proof}
Let $n_u$ denote the number of leaves descending from vertex $u$. A vector containing the values $n_u$ for each $u\in V(T)$ can be computed in time $O(n)$ by traversing the tree in post order, setting $n_u=1$ if $u$ is a leaf and calculating $n_u=n_{u_1}+\ldots+n_{u_k}$ otherwise (where $u_1, \ldots, u_k$ denote the children of $u$). Using this information, vectors containing the $I_v$ values, the corrected $I'_v$ values and the weights $w(I_v)$ can be computed in linear time allowing a subsequent computation of the mean $w(I_v)$ and thus of the corrected $I_v^w$ values also in linear time. Now, the computation time only depends on the computation time of the statistics that are applied to either the $I_v$, $I_v'$ or $I_v^w$ vector (each of length $\approx n$).
\end{proof}

\begin{remark}
Proposition \ref{runtime_Ibased} implies, for instance, that the mean, the sum, the variance (and, on average, the median) as well as single values (e.g. $I_\rho$) of the $I_v$ values can be computed in time $O(n)$ since the mean, the sum, the variance and, on average, also the median can all be computed in linear time.
\end{remark}

Next, we will have a look at the recursiveness of the $I$ and $I'$ value, the Total $I$ and Total $I'$ index as well as the Mean $I$ and Mean $I'$ index. Note that the following propositions consider only binary trees. Also note that in Proposition \ref{recursiveness_I'} and Remark \ref{recursiveness_I} we set $I'_\rho(T)=I_\rho(T)=0$ if $T\in\BTnstar$ with $n\in\{1,2,3\}$, because the recursiveness requires values for $n\in\{1,2,3\}$. This choice is sensible because if $T\in\mathcal{BT}_{n\in\{1,2,3\}}^\ast$ the partition of $n_v$ into $n_{v_1}$ and $n_{v_2}$ is as balanced as it can get for each $v\in V(T)$ with children $v_1$ and $v_2$. In addition recalls that $\overline{I}(T)=\Sigma I(T) = 0$ for each tree $T \in \Tnstar$ with $n \in \{1,2,3\}$ (since $\mathring{V}_{bin,\geq 4}(T)$ is the empty set for each such tree).

\begin{proposition} \label{recursiveness_I'}
Setting $I'_\rho(T)=0$ if $n\in\{1,2,3\}$ the $I'$ value $I'_\rho(T)$ is a binary recursive tree shape statistic. We have $I'_{\rho}(T)=0$ for $T\in\mathcal{BT}_1^\ast$, and for every binary tree $T\in\BTnstar$ with $n\geq 2$ and standard decomposition $T=(T_1,T_2)$ we have \[ I'_{\rho}(T)=\frac{\max\{n_1,n_2\}-\left \lceil\frac{n_1+n_2}{2}\right\rceil}{n_1+n_2-1-\left\lceil\frac{n_1+n_2}{2}\right\rceil}\cdot\left(\frac{n_1+n_2-1}{n_1+n_2}\right)^{\mathcal{I}(n_1+n_2\textup{ mod }2=0)}. \]
\end{proposition}
\begin{proof}
Since $n_1\geq n_2$ with $n=n_1+n_2$, we have
\begin{equation*}
\begin{split}
    I'_{\rho}(T) &= \begin{cases}\frac{n_1-\left\lceil\frac{n}{2}\right\rceil}{n-1-\left\lceil\frac{n}{2}\right\rceil} & \text{if } n \text{ is odd} \\ \frac{n_1-\left\lceil\frac{n}{2}\right\rceil}{n-1-\left\lceil\frac{n}{2}\right\rceil}\cdot\frac{n-1}{n} & \text{if } n \text{ is even} \end{cases} = \begin{cases}\frac{\max\{n_1,n_2\}-\left\lceil\frac{n_1+n_2}{2}\right\rceil}{n_1+n_2-1-\left\lceil\frac{n_1+n_2}{2}\right\rceil} & \text{if } n_1+n_2 \text{ is odd} \\ \frac{\max\{n_1,n_2\}-\left\lceil\frac{n_1+n_2}{2}\right\rceil}{n_1+n_2-1-\left\lceil\frac{n_1+n_2}{2}\right\rceil}\cdot\frac{n_1+n_2-1}{n_1+n_2} & \text{if } n_1+n_2 \text{ is even} \end{cases}\\ 
    &=  \frac{\max\{n_1,n_2\}-\left\lceil\frac{n_1+n_2}{2}\right\rceil}{n_1+n_2-1-\left\lceil\frac{n_1+n_2}{2}\right\rceil}\cdot\left(\frac{n_1+n_2-1}{n_1+n_2}\right)^{\mathcal{I}(n_1+n_2\text{ mod }2=0)}.
\end{split}
\end{equation*}
Thus, the $I'$ value can be expressed as a binary recursive tree shape statistic of length $x=2$ with the recursions (where $I'_i$ is the simplified notation of $I'_\rho(T_i)$, and $n_1$ and $n_2$ denote the leaf numbers of $T_1$ and $T_2$)
\begin{itemize}
    \item  $I'$ value: $\lambda_1=0$ and $r_1(T_1,T_2)=\frac{\max\{n_1,n_2\}-\left\lceil\frac{n_1+n_2}{2}\right\rceil}{n_1+n_2-1-\left\lceil\frac{n_1+n_2}{2}\right\rceil}\cdot\left(\frac{n_1+n_2-1}{n_1+n_2}\right)^{\mathcal{I}(n_1+n_2\text{ mod }2=0)}$
    \item leaf number: $\lambda_2=1$ and $r_2(T_1,T_2)=n_1+n_2$
\end{itemize}
It can easily be seen that $\lambda\in\mathbb{R}^2$ and $r_i:\mathbb{R}^2\times\mathbb{R}^2\rightarrow\mathbb{R}$, and that all $r_i$ are independent of the order of subtrees. This completes the proof.
\end{proof}

\begin{remark} \label{recursiveness_I}
Similarly to the proof of Proposition \ref{recursiveness_I'} one can show that setting $I_\rho(T)=0$ for $n\in\{1,2,3\}$ the $I$ value $I_\rho(T)$ is a binary recursive tree shape statistic of length $x=2$ where the recursions are identical to the ones stated in the proof of Proposition \ref{recursiveness_I'} except that the case where $n$ is even and thus the term $\left(\frac{n_1+n_2-1}{n_1+n_2}\right)^{\mathcal{I}(n_1+n_2\textup{ mod }2=0)}$ is omitted.
\end{remark}

\begin{proposition} \label{recursiveness_SumI'}
The Total $I'$ index is a binary recursive tree shape statistic. We have $\Sigma I'(T)=0$ for $T\in\mathcal{BT}_1^\ast$, and for every binary tree $T\in\BTnstar$ with $n\geq 2$ and standard decomposition $T=(T_1,T_2)$ we have \[ \Sigma I'(T)=\Sigma I'(T_1) + \Sigma I'(T_2) + \frac{\max\{n_1,n_2\}-\left\lceil\frac{n_1+n_2}{2}\right\rceil}{n_1+n_2-1-\left\lceil\frac{n_1+n_2}{2}\right\rceil}\cdot\left(\frac{n_1+n_2-1}{n_1+n_2}\right)^{\mathcal{I}(n_1+n_2\text{ mod }2=0)}. \]
\end{proposition}
\begin{proof}
Let $T=(T_1,T_2)$ be a binary tree. Then, using Propositon \ref{recursiveness_I'}, we have
\begin{equation*}
\begin{split}
    \Sigma I'(T) &= \sum\limits_{\substack{v\in\mathring{V}(T) \\ n_v\geq 4}} I'_v = \sum\limits_{\substack{v\in\mathring{V}(T_1) \\ n_v\geq 4}} I'_v + \sum\limits_{\substack{v\in\mathring{V}(T_2) \\ n_v\geq 4}} I'_v + I'_{\rho} = \Sigma I'(T_1) + \Sigma I'(T_2) + I'_\rho(T)\\
    &= \Sigma I'(T_1) + \Sigma I'(T_2) + \frac{\max\{n_1,n_2\}-\left\lceil\frac{n_1+n_2}{2}\right\rceil}{n_1+n_2-1-\left\lceil\frac{n_1+n_2}{2}\right\rceil}\cdot\left(\frac{n_1+n_2-1}{n_1+n_2}\right)^{\mathcal{I}(n_1+n_2\text{ mod }2=0)}.
\end{split}
\end{equation*}
Thus, the Total $I'$ index can be expressed as a binary recursive tree shape statistic of length $x=2$ with the recursions (where $\Sigma'_i$ is the simplified notation of $\Sigma I'(T_i)$, and $n_1$ and $n_2$ denote the leaf numbers of $T_1$ and $T_2$)
\begin{itemize}
    \item $\Sigma I'$ index: $\lambda_1=0$ and $r_1(T_1,T_2)=\Sigma'_1+\Sigma'_2+\frac{\max\{n_1,n_2\}-\left\lceil\frac{n_1+n_2}{2}\right\rceil}{n_1+n_2-1-\left\lceil\frac{n_1+n_2}{2}\right\rceil}\cdot\left(\frac{n_1+n_2-1}{n_1+n_2}\right)^{\mathcal{I}(n_1+n_2\text{ mod }2=0)}$
    \item leaf number: $\lambda_2=1$ and $r_2(T_1,T_2)=n_1+n_2$
\end{itemize}
It can easily be seen that $\lambda\in\mathbb{R}^2$ and $r_i:\mathbb{R}^2\times\mathbb{R}^2\rightarrow\mathbb{R}$, and that all $r_i$ are independent of the order of subtrees. This completes the proof. 
\end{proof}

\begin{remark} \label{recursiveness_SumI}
Similar to the proof of Proposition \ref{recursiveness_SumI'} one can show that the Total $I$ value $\Sigma I(T)$ is a binary recursive tree shape statistic of length $x=2$ where the recursions are identical to the ones stated in the proof of Proposition \ref{recursiveness_SumI'} except that the term $\left(\frac{n_1+n_2-1}{n_1+n_2}\right)^{\mathcal{I}(n_1+n_2\textup{ mod }2=0)}$ is omitted as the $I_v$ values are not corrected for $n_v$ even or odd.
\end{remark}

\begin{proposition} \label{recursiveness_meanI'}
The Mean $I'$ index is a binary recursive tree shape statistic. We have $\overline{I'}(T)=0$ for $T\in\mathcal{BT}_1^\ast$, and for every binary tree $T\in\BTnstar$ with $n\geq 2$ and standard decomposition $T=(T_1,T_2)$ we have \[ \overline{I'}(T)=\frac{\overline{I'}(T_1)\cdot a(T_1)+\overline{I'}(T_2)\cdot a(T_2)+\frac{\max\{n_1,n_2\}-\left\lceil\frac{n_1+n_2}{2}\right\rceil}{n_1+n_2-1-\left\lceil\frac{n_1+n_2}{2}\right\rceil}\cdot\left(\frac{n_1+n_2-1}{n_1+n_2}\right)^{\mathcal{I}(n_1+n_2\textup{ mod }2=0)}}{a(T_1)+a(T_2)+\mathcal{I}(n_1+n_2\geq 4)} \] in which $a(T_i)$ denotes the number of vertices $v$ in $T_i$ with $n_v\geq 4$.
\end{proposition}
\begin{proof}
Let $T=(T_1,T_2)$ be a binary tree. Let $a(T)$ denote the number of vertices $v$ in $V(T)$ with $n_v\geq 4$; then, we have $a(T)=a(T_1)+a(T_2)+\mathcal{I}(n_1+n_2\geq 4)$. Now, since $\overline{I'}(T)=\frac{\Sigma I'(T)}{a(T)}$ (with $\frac{0}{0}=0$) and thus $\Sigma I'(T)=\overline{I'}(T)\cdot a(T)$ and using Proposition \ref{recursiveness_SumI'} we have 
\begin{equation*}
\begin{split}
    \overline{I'}(T) &= \frac{\Sigma I'(T)}{a(T)}\\
    &= \frac{\Sigma I'(T_1) + \Sigma I'(T_2) + \frac{\max\{n_1,n_2\}-\left\lceil\frac{n_1+n_2}{2}\right\rceil}{n_1+n_2-1-\left\lceil\frac{n_1+n_2}{2}\right\rceil}\cdot\left(\frac{n_1+n_2-1}{n_1+n_2}\right)^{\mathcal{I}(n_1+n_2\text{ mod }2=0)}}{a(T_1)+a(T_2)+\mathcal{I}(n_1+n_2\geq 4)}\\[10pt]
    &=\frac{\overline{I'}(T_1)\cdot a(T_1) + \overline{I'}(T_2)\cdot a(T_2) + \frac{\max\{n_1,n_2\}-\left\lceil\frac{n_1+n_2}{2}\right\rceil}{n_1+n_2-1-\left\lceil\frac{n_1+n_2}{2}\right\rceil}\cdot\left(\frac{n_1+n_2-1}{n_1+n_2}\right)^{\mathcal{I}(n_1+n_2\text{ mod }2=0)}}{a(T_1)+a(T_2)+\mathcal{I}(n_1+n_2\geq 4)}.
\end{split}
\end{equation*}
Thus, the Mean $I'$ index can be expressed as a binary recursive tree shape statistic of length $x=3$ with the recursions (where $\overline{I'}_i$ and $a_i$ are the simplified notations of $\overline{I'}(T_i)$ and $a(T_i)$, and $n_i$ denotes the leaf number of $T_i$) 
\begin{itemize}
    \item $\overline{I'}$ index: $\lambda_1=0$ and\\ $r_1(T_1,T_2)=\left(\overline{I'}_1\cdot a_1+\overline{I'}_2\cdot a_2+\frac{\max\{n_1,n_2\}-\left\lceil\frac{n_1+n_2}{2}\right\rceil}{n_1+n_2-1-\left\lceil\frac{n_1+n_2}{2}\right\rceil}\cdot\left(\frac{n_1+n_2-1}{n_1+n_2}\right)^{\mathcal{I}(n_1+n_2\text{ mod }2=0)}\right)\cdot \frac{1}{a_1+a_2+\mathcal{I}(n_1+n_2\geq 4)}$
    \item number of nodes $v$ with $n_v\geq 4$: $\lambda_2=0$ and $r_2(T_1,T_2)=a_1+a_2+\mathcal{I}(n_1+n_2\geq 4)$
    \item leaf number: $\lambda_3=1$ and $r_3(T_1,T_2)=n_1+n_2$
\end{itemize}
It can easily be seen that $\lambda\in\mathbb{R}^3$ and $r_i:\mathbb{R}^3\times\mathbb{R}^3\rightarrow\mathbb{R}$, and that all $r_i$ are independent of the order of subtrees. This completes the proof. 
\end{proof}

\begin{remark} \label{recursiveness_meanI}
Similar to the proof of Proposition \ref{recursiveness_meanI'} it can be shown that the Mean $I$ index $\overline{I}(T)$ is a binary recursive tree shape statistic of length $x=3$ where the recursions are identical to the ones stated in the proof of Proposition \ref{recursiveness_meanI'} except that the term $\left(\frac{n_1+n_2-1}{n_1+n_2}\right)^{\mathcal{I}(n_1+n_2\textup{ mod }2=0)}$ is omitted.
\end{remark}

In the following three propositions we consider the locality of the $I$ and $I'$ value, the Total $I$ and Total $I'$ index, and the Mean $I$ and Mean $I'$ index.

\begin{proposition} \label{locality_I'}
The $I$ value $I_\rho(T)$ and the $I'$ value $I'_\rho(T)$ are not local.
\end{proposition}
\begin{proof}
Consider the two trees $T$ and $T'$ in Figure \ref{fig_locality} on page \pageref{fig_locality}, which only differ in their subtrees rooted at $v$. Note that in both $T$ and $T'$ the vertex $v$ has exactly 5 descendant leaves. Nevertheless, we have $I'_{\rho}(T)-I'_{\rho}(T')=0-0=0\neq 1=1-0=I'_{\rho=v}(T_v)-I'_{\rho=v}(T_v')$. Thus, the $I'$ index is not local.\\
We also have $I_{\rho}(T)-I_{\rho}(T')=0-0=0\neq 1=1-0=I_{\rho=v}(T_v)-I_{\rho=v}(T_v')$. Thus, the $I$ value is not local either.\\
This applies, because as long as $v\neq\rho$ we always have $I'_{\rho}(T)=I'_{\rho}(T')$, but if $T_v$ and $T_v'$ have different leaf numbers in their maximal pending subtrees we have $I'_{\rho=v}(T_v)\neq I'_{\rho=v}(T_v')$.
\end{proof}

\begin{proposition} \label{locality_SumI'}
The Total $I$ index $\Sigma I(T)$ and the Total $I'$ index $\Sigma I'(T)$ are local.
\end{proposition}
\begin{proof}
Let $T'$ be the tree that we obtain from $T\in\Tnstar$ by exchanging a subtree $T_v$ of $T$ with a subtree $T_v'$ on the same number of leaves. Firstly, note that $\mathring{V}_{bin,\geq 4}(T)\setminus\mathring{V}_{bin,\geq 4}(T_v)=\mathring{V}_{bin,\geq 4}(T')\setminus\mathring{V}_{bin,\geq 4}(T_v')$. Secondly, we have $n_T(w)=n_{T_v}(w)$ if $w\in\mathring{V}_{bin,\geq 4}(T_v)$ and $n_{T'}(w)=n_{T_v'}(w)$ if $w\in\mathring{V}_{bin,\geq 4}(T'_v)$, because each descendant leaf of $v$, and thus of $w$, is in $T_v$ and $T_v'$. This, in turn, implies that the $I_v$ values of $w\in\mathring{V}_{bin,\geq 4}(T_v)$ in $T$ and $T_v$ are equal, and that the $I_v$ values of $w\in\mathring{V}_{bin,\geq 4}(T_v')$ in $T'$ and $T_v'$ are equal. Thirdly, note that $n_T(w)=n_{T'}(w)$ if $w\in \mathring{V}_{bin,\geq 4}(T)\setminus\mathring{V}_{bin,\geq 4}(T_v)$, because changing the shape of $T_v$ does not change the number of descendant leaves of $w\in \mathring{V}_{bin,\geq 4}(T)\setminus\mathring{V}_{bin,\geq 4}(T_v)$ as $T_v$ and $T_v'$ have the same leaf number. This, in turn implies that the $I_v$ values of $w\in\mathring{V}_{bin,\geq 4}(T)\setminus\mathring{V}_{bin,\geq 4}(T_v)$ in $T$ and $T'$ are equal. Hence, we have 
\begin{equation*}
\begin{split}
    \Sigma I(T)-\Sigma I(T') &= \sum\limits_{w\in\mathring{V}_{bin,\geq 4}(T)} I_w - \sum\limits_{w\in\mathring{V}_{bin,\geq 4}(T')} I_w\\[10pt]
    &= \sum\limits_{w\in\mathring{V}_{bin,\geq 4}(T_v)} I_w + \sum\limits_{w\in\mathring{V}_{bin,\geq 4}(T)\setminus\mathring{V}_{bin,\geq 4}(T_v)} I_w - \sum\limits_{w\in\mathring{V}_{bin,\geq 4}(T_v')} I_w - \sum\limits_{w\in\mathring{V}_{bin,\geq 4}(T')\setminus\mathring{V}_{bin,\geq 4}(T_v')} I_w\\[10pt]
    &= \sum\limits_{w \in\mathring{V}_{bin,\geq 4}(T_v)} I \sum\limits_{w\in\mathring{V}_{bin,\geq 4}(T_v')} I_w = \Sigma I(T_v)-\Sigma I(T_v').
\end{split}
\end{equation*}
Thus, the Total $I$ index is local. The proof for the locality of the Total $I'$ index is analogous.
\end{proof}

\begin{proposition} \label{locality_meanI'}
The Mean $I$ index and the Mean $I'$ index are not local.
\end{proposition}
\begin{proof}
Consider the two trees $T$ and $T'$ in Figure \ref{fig_locality} on page \pageref{fig_locality}, which only differ in their subtrees rooted at $v$. Note that in both $T$ and $T'$ the vertex $v$ has exactly 5 descendant leaves. Nevertheless, we have $\overline{I'}(T)-\overline{I'}(T')=\frac{7}{16}-0=\frac{7}{16}\neq \frac{7}{8}=\frac{7}{8}-0=\overline{I'}(T_v)-\overline{I'}(T_v')$. Thus, the $\overline{I'}$ index is not local.\\
We also have $\overline{I}(T)-\overline{I}(T')=\frac{1}{2}-0=\frac{1}{2}\neq 1=1-0=\overline{I}(T_v)-\overline{I}(T_v')$. Thus, the $\overline{I}$ index is not local either. Both statements are due to the different normalization factors $\frac{1}{|\{u\in\mathring{V}(T):n_u\geq4\}|}$ for $T$ and $T'$ and $\frac{1}{|\{u\in\mathring{V}(T_v):n_u\geq4\}|}$ for $T_v$ and $T_v'$.
\end{proof}

Next, we will concentrate on the minima and maxima of the $I$ and $I'$ values, the Total $I$ and Total $I'$ indices, and the Mean $I$ and Mean $I'$ indices. In particular, in Theorem \ref{prop_Irho_max_ab} and \ref{prop_Irho_min_ab} we will show that the $I$ and $I'$ values are neither balance nor imbalance indices according to our definitions. Afterwards, we will show that the Total $I$ and Total $I'$ indices as well as the Mean $I$ and Mean $I'$ indices fulfill the definition of imbalance indices when restricted to $\BTnstar$.

\begin{theorem} \label{prop_Irho_max_ab}
For every $n \in \mathbb{N}_{\geq 4}$ the maximal $I$ value $I_\rho(T)$ over all $T \in \mathcal{T}_n^\ast$ with binary roots or $T \in \mathcal{BT}_n^\ast$ is $I_\rho(T)=1$. Every tree whose (binary) root has a leaf as a child is a maximal tree. There are $we(n-1)$ maximal binary trees and $|\mathcal{T}^\ast_{n-1}|$ maximal arbitrary trees that are binary at the root. The results hold for the correction method $I'$ as well, except that the maximal value is $\frac{n-1}{n}$ if $n$ is even.
\end{theorem}
\begin{proof}
First, recall that by definition, $I_v \in [0,1]$ for any binary node $v \in \mathring{V}(T)$ with $n_v \geq 4$. Now, for any such node $v$ with children $v_1$ and $v_2$, we have $I_v=1$ if and only $n_{v_1} = n_v-1$ and $n_{v_2}=1$. It immediately follows that the trees $T$ with maximal $I_\rho$ value are precisely those in which the root has two children one of which is a leaf and the other is the ancestor of $n-1$ leaves. The number of such trees solely depends on the number of topologies for the pending subtree with $n-1$ leaves and is thus equal to $we(n-1)$ for $T \in \BTnstar$ and $|\mathcal{T}_{n-1}^\ast|$ for $T \in \Tnstar$ with binary root.
Using the correction method $I'$, the maximal trees are not affected. However, for $n$ even the values are scaled by $\frac{n-1}{n}$ resulting in a maximum of $\frac{n-1}{n}$. This completes the proof.
\end{proof}

\begin{theorem} \label{prop_Irho_min_ab}
For every $n \in \mathbb{N}_{\geq 4}$ the minimal $I$ value $I_\rho(T)$ over all $T \in \mathcal{T}_n^*$ with binary roots or $T \in\BTnstar$ is $I_\rho(T)=0$. Every tree whose binary root partitions the number of descending leaves $n$ into $\left \lceil \frac{n}{2} \right \rceil$ and $\left \lfloor \frac{n}{2} \right \rfloor$ is minimal. 
There are $we \left( \left\lceil \frac{n}{2} \right\rceil \right) \cdot we(\lfloor \frac{n}{2} \rfloor)$ minimal binary trees if $n$ is odd and $\frac{1}{2} \cdot we\left (\frac{n}{2}\right) \cdot (we\left(\frac{n}{2})+1\right)$ if $n$ is even. Analogously, there are $\left \lvert \mathcal{T}_{\left \lceil \frac{n}{2} \right \rceil}^\ast \right\rvert \cdot  \left|\mathcal{T}_{\left \lfloor \frac{n}{2} \right \rfloor}^\ast\right|$ minimal arbitrary trees that are binary at the root if $n$ is odd and $\frac{1}{2} \cdot \left|\mathcal{T}_{\frac{n}{2}}^\ast\right| \cdot \left( \left|\mathcal{T}_{\frac{n}{2}}^\ast\right| + 1\right)$ if $n$ is even. 
The same results hold for the correction method $I'$.
\end{theorem}
\begin{proof}
Again, recall that by definition, $I_v \in [0,1]$ for any binary node $v \in \mathring{V}(T)$ with $n_v \geq 4$. Moreover, for any such node $v$ with children $v_1$ and $v_2$, we have $I_v=0$ if and only if $n_{v_1} = \left \lceil \frac{n_v}{2} \right \rceil $ and $n_{v_2} = \left \lfloor \frac{n_v}{2} \right \rfloor$. It immediately follows that the trees with minimal $I_\rho$ value are precisely those trees $T$ that have two maximal pending subtrees $T_1$ and $T_2$ with $n_1=\left \lceil \frac{n}{2} \right \rceil$ and $n_2 = \left \lfloor \frac{n}{2} \right \rfloor$. The number of such trees $T$ depends on the number of topologies for the two maximal pending subtrees with $\left \lceil \frac{n}{2} \right \rceil$ and $\left \lfloor \frac{n}{2} \right \rfloor$ leaves, respectively. 
For $T \in \BTnstar$ and $n$ odd, there are $we\left( \left\lceil \frac{n}{2} \right\rceil\right) \cdot we\left(\left\lfloor \frac{n}{2} \right\rfloor\right)$ minimal trees, and for $n$ even, there are $\binom{we\left(\frac{n}{2}\right)+1}{2} = \frac{1}{2} \cdot we\left(\frac{n}{2}\right) \cdot (we\left(\frac{n}{2}\right)+1)$ minimal trees. Analogously, for $T \in \Tnstar$ with binary root, there are $\left|\mathcal{T}_{\left\lceil \frac{n}{2} \right\rceil}^\ast\right| \cdot  \left|\mathcal{T}_{\left\lfloor \frac{n}{2} \right \rfloor}^\ast\right|$ minimal trees if $n$ is odd, and $\frac{1}{2} \cdot \left|\mathcal{T}_{\frac{n}{2}}^\ast\right| \cdot \left( \left|\mathcal{T}_{\frac{n}{2}}^\ast\right| + 1\right)$ minimal trees if $n$ is even.
The correction method $I'$ does not affect nodes $v$ with $I_v=0$ and therefore has no effect on the minimal value and minimal trees. This completes the proof.
\end{proof}

From the previous two propositions it is obvious that the $I$ and $I'$ values are neither balance nor imbalance indices on $\Tnstar$ and $\BTnstar$, because for $n\geq 5$ there are trees $T\neq\Tcat$ that have the same value as the caterpillar tree (see Theorem \ref{prop_Irho_max_ab}) and because for $h\geq 3$ (and $n=2^h$) there are trees $T\neq\Tfb$ that have the same value as the fully balanced tree (see Theorem \ref{prop_Irho_min_ab}).

Now, we show that -- restricted to $\BTnstar$ -- the Total $I$ and Total $I'$ indices as well as the Mean $I$ and Mean $I'$ indices fulfill our criteria and are in fact imbalance indices on $\BTnstar$.

\begin{remark} \label{remark_I'_123}
Note that for $n\in\{1,2\}$ the set $\BTnstar$ consists only of the caterpillar tree which equals the fully balanced tree, and that for $n=3$, the caterpillar tree is the only tree in $\BTnstar$. This means in particular, that the caterpillar tree is the unique tree yielding the minimum value of $\Sigma I$, $\Sigma I'$, $\overline{I}$ and $\overline{I'}$ on $\BTnstar$ for $n\in\{1,2,3\}$ and the fully balanced tree is the unique tree yielding the maximum value on $\BTnstar$ for $n\in\{1,2\}$. Since any tree with $n\leq 3$ has no vertices in $\mathring{V}_{bin,\geq 4}$, this minimal and maximal value is $\Sigma I(T)=\Sigma I'(T)=\overline{I}(T)=\overline{I'}(T)=0$. So, in the following propositions we will only consider the case $n\geq 4$.
\end{remark}

\begin{theorem} \label{prop_meanI_max_b}
For every $n \in \mathbb{N}_{\geq 4}$ the maximal Total $I$ index $\Sigma I(T)$ over all $T\in\BTnstar$ is $\Sigma I(T)=n-3$ and $\Tcat$ is the unique maximal tree. The results hold for the correction method $I'$ as well, except that the maximal value is $\left \lfloor \frac{n-3}{2} \right \rfloor +  \sum\limits_{k=2}^{\left\lceil \frac{n-1}{2} \right\rceil}{\frac{2k-1}{2k}} < n-3$.\\
Also, for every $n \in \mathbb{N}_{\geq 4}$ the maximal Mean $I$ index $\overline{I}(T)$ over all $T \in \mathcal{BT}_n^*$ is $\overline{I}(T)=1$ and $\Tcat$ is the unique maximal tree. The results hold for the correction method $I'$ as well, except that the maximal value is $\frac{1}{n-3}\cdot \left(\left \lfloor \frac{n-3}{2} \right \rfloor +  \sum\limits_{k=2}^{\left\lceil \frac{n-1}{2} \right\rceil}{\frac{2k-1}{2k}} \right)<1$.
\end{theorem}
\begin{proof}
By definition, $I_v \in [0,1]$ for each $v \in \mathring{V}(T)$ with $n_v \geq 4$. This immediately implies that $\Sigma I(T)\leq n-3$ (as $|\mathring{V}(T)|=n-1$, but there are at least two pending subtrees with $n_v\leq 3$ in $T$, namely either one cherry and one with exactly three leaves or two cherries; thus, $|\{v \in \mathring{V}(T): n_v \geq 4\}| \leq n-3$) and $\overline{I}(T) \leq 1$ for each $T \in \BTnstar$ with $n \geq 4$. Moreover, for each $v \in \mathring{V}(T)$ with $n_v \geq 4$ and children $v_1$ and $v_2$, we have $I_v=1$ if and only if $n_{v_1} = n_v-1$ and $n_{v_2} = 1$. 
In particular, the maxima $\Sigma I(T)=n-3$ and $\overline{I}(T) = 1$ are reached if and only if $n_{v_1} = n_v-1$ and $n_{v_2} = 1$ for each $v \in \mathring{V}(T)$ with $n_v \geq 4$ and children $v_1$ and $v_2$. However, this uniquely characterizes the caterpillar tree $\Tcat$. Thus, $\Tcat$ is the unique rooted binary tree $T \in \BTnstar$ with $n \geq 4$ maximizing the Total $I$ and Mean $I$ index and we have $\Sigma I(\Tcat)=n-3$ and $\overline{I}(\Tcat)=1$.

Now, consider the correction method $I'$. We first show that 
\begin{equation} \label{eq_max_SumMeanI'}
    \Sigma I'(\Tcat)=\left \lfloor \frac{n-3}{2} \right \rfloor + \sum\limits_{k=2}^{\left \lceil \frac{n-1}{2} \right \rceil}{\frac{2k-1}{2k}} < n-3 \quad \text{and}
\end{equation}
\begin{equation} \label{eq_max_SumMeanI'second}
    \overline{I'}(\Tcat) = \frac{1}{n-3}\cdot \left(\left \lfloor \frac{n-3}{2} \right \rfloor +  \sum\limits_{k=2}^{\left \lceil \frac{n-1}{2} \right \rceil}{\frac{2k-1}{2k}} \right) < 1.
\end{equation}
To see this, note that for $n\geq 4$, the caterpillar tree $\Tcat$ has exactly $n-3$ inner nodes $v$ with $n_v \geq 4$. The $I_v$ values of these inner nodes sorted from largest to smallest depth are $\left(\frac{3}{4},1,\frac{5}{6},1,\frac{7}{8}, \ldots ,1\right)$ if $n$ is odd and $\left(\frac{3}{4},1,\frac{5}{6},1,\frac{7}{8},\ldots,\frac{n-1}{n}\right)$ if $n$ is even. In any case, there are $\left \lfloor \frac{n-3}{2} \right \rfloor$ ones and $\left \lceil \frac{n-3}{2} \right \rceil$ fractions that are smaller than one. Thus, the sum and the mean of these values are precisely the ones stated in Equation \eqref{eq_max_SumMeanI'} and \eqref{eq_max_SumMeanI'second}. 

It remains to show that $\Sigma I'(T)<\Sigma I'(\Tcat)$ and $\overline{I'}(T) < \overline{I'} (\Tcat)$ for every rooted binary tree $T \in \mathcal{BT}_{n \geq 4}^\ast$ with $T\neq\Tcat$, i.e. $\Tcat$ is the unique tree maximizing $\Sigma I'(T)$ and $\overline{I'}(T)$. Assume for the sake of a contradiction that there exists a tree $T \in  \mathcal{BT}_{n \geq 4}^\ast$  with $T \neq \Tcat$ that has maximum $\Sigma I'$ or $\overline{I'}$ index. As $T \neq \Tcat$, there must exist at least one subtree $T_u$ of $T$ that is not a caterpillar tree (note that $T_u = T$ is possible). Let $T_u$ be minimal with this property, i.e. if we let $u_1$ and $u_2$ denote the children of $u$, then $T_{u_1}$ and $T_{u_2}$ are caterpillar trees with $n_T(u_1)$ and $n_T(u_2)$ leaves, respectively. Without loss of generality $n_T(u_1) \geq n_T(u_2)$. Note that we must have $n_T(u_2) \geq 2$ (as otherwise $T_u = (T_{u_1},T_{u_2})$ would be a caterpillar tree) and thus in particular $n_T(u) \geq 4$.
We now construct a tree $T' \in \mathcal{BT}_{n \geq 4}^\ast$ by modifying $T$ and show that this strictly increases the $\Sigma I'$ and $\overline{I'}$ index thereby contradicting the maximality of $T$. We distinguish the following cases:
\begin{enumerate}[(i)]
    \item If $n_{T}(u_1)$ is even, let $x$ denote a leaf that is adjacent to $u_2$. We now delete the edge $(u_2,x)$, suppress $u_2$, subdivide the edge $(u,u_1)$ with a new degree-2 vertex $\widetilde{u}_2$, and add the edge $(\widetilde{u}_2,x)$ to obtain $T'$ (see Figure~\ref{Fig_MeanIPrime_Part1}). We now note that the $I_v'$ values of all nodes of $T$, respectively $T'$, except for $u$, $u_2$, and $\widetilde{u}_2$ are unaffected by this procedure. However, we clearly have increased the $I'$ value of node $u$, i.e. $I'_u(T') > I'_u(T)$ (since we have increased the difference in the sizes of the two maximal pending subtrees adjacent to $u$). We now distinguish the following three subcases:
        \begin{itemize}
            \item If $n_{T}(u_1) = n_{T}(u_2)=2$, then neither does node $u_2$ (which is present in $T$ but not in $T'$) contribute to $\Sigma I'(T)$ or $\overline{I'}(T)$ nor does node $\widetilde{u}_2$ (which is present in $T'$ but not in $T$) contribute to $\Sigma I'(T')$ or $\overline{I'}(T')$. 
            In particular, $T$ and $T'$ contain the same number of nodes contributing to $\Sigma I'(T)$ and $\Sigma I'(T')$ as well as to $\overline{I'}(T)$ and $\overline{I'}(T')$, and the only difference between them is the contribution of node $u$. As $I'_u(T') > I'_u(T)$, we thus have $\Sigma I'(T')>\Sigma I'(T)$ and $\overline{I'}(T') > \overline{I'}(T)$ contradicting the maximality of $T$.
            \item If $n_{T}(u_1) \geq 3$ (which implies $n_{T}(u_1) \geq 4$ since we assume $n_{T}(u_1)$ even) and $n_{T}(u_2) \in \{2,3\}$, then node $u_2$ does not contribute to $\Sigma I'(T)$ or $\overline{I'}(T)$, but node $\widetilde{u}_2$ contributes to $\Sigma I'(T')$ and $\overline{I'}(T')$. This, together with $I'_u(T')>I_u'(T)$, already leads to $\Sigma I'(T')>\Sigma I'(T)$. Moreover, note that we have $I'_{\widetilde{u}_2}(T') = I_{\widetilde{u}_2}(T') = 1$ (where the first equality follows from the fact that $n_{T'}(\widetilde{u}_2)=n_{T'}(u_1)+1 = n_T(u_1)+1$ is odd), i.e. the contribution of node $\widetilde{u}_2$ is maximal.
            Together with the fact that $I'_u(T') > I'_u(T)$ this implies that the mean of the $I'$ values in $T'$ must be strictly larger than the mean of the $I'$ values in $T$, i.e. $\overline{I'}(T') > \overline{I'}(T)$. This, again, contradicts the maximality of $T$. 
            \item Finally, if $n_T(u_1), n_T(u_2) \geq 4$, then node $u_2$ contributes only to $\Sigma I'(T)$ and $\overline{I'}(T)$, whereas node $\widetilde{u}_2$ contributes only to $\Sigma I'(T')$ and $\overline{I'}(T')$. In this case, $T$ and $T'$ contain the same number of nodes contributing to $\Sigma I'(T)$ and $\Sigma I'(T')$ and to $\overline{I'}(T)$ and $\overline{I'}(T')$. However, as $I'_{\widetilde{u}_2}(T') = I_{\widetilde{u}_2}(T') = 1$, we clearly have $I'_{\widetilde{u}_2}(T') \geq I'_{u_2}(T)$. Together with the fact that $I'_u(T') > I'_u(T)$ this again implies $\Sigma I'(T')>\Sigma I'(T)$ and $\overline{I'}(T') > \overline{I'}(T)$, thereby contradicting the maximality of $T$.
        \end{itemize}
        \begin{figure}[htbp]
        \centering
        \includegraphics[scale=0.3]{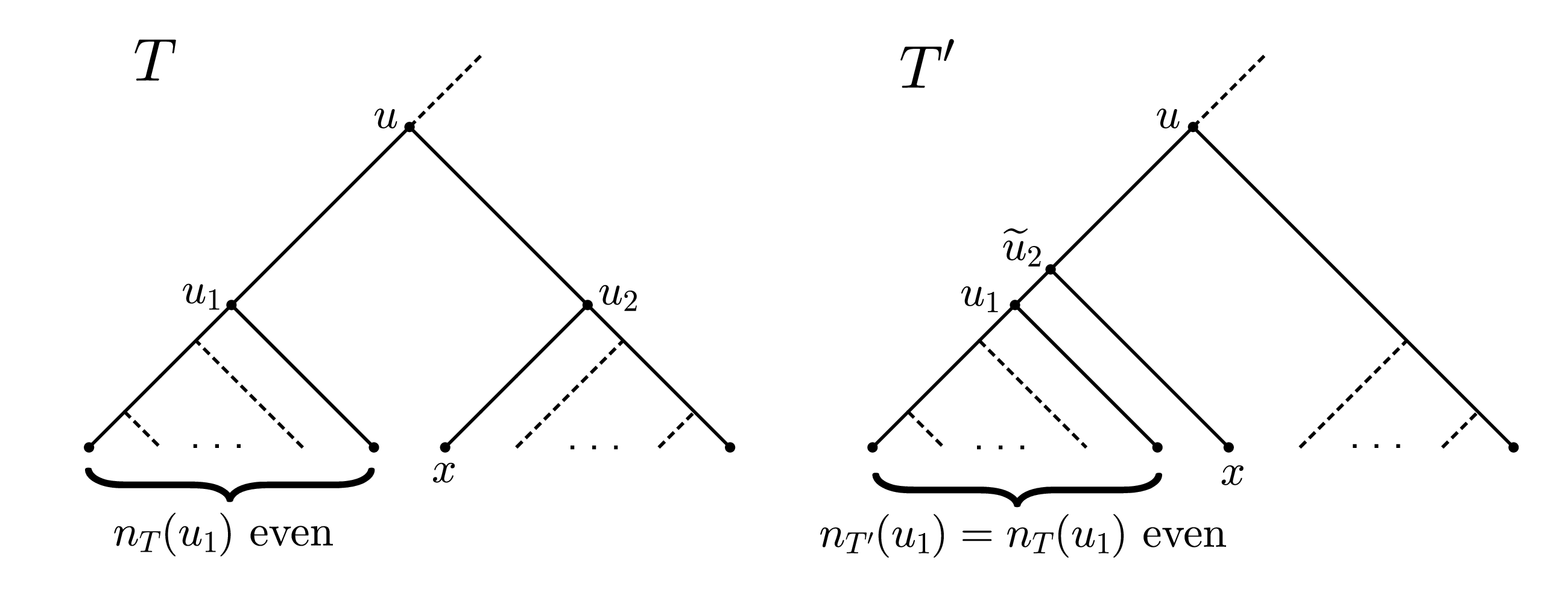}
        \caption{Trees $T$ and $T'$ as needed in Part~(i) of the proof of Theorem~\ref{prop_meanI_max_b}.}
        \label{Fig_MeanIPrime_Part1}
        \end{figure}
    \item If $n_T(u_1)$ is odd, the approach is similar; however, we distinguish different subcases. 
        \begin{enumerate}[(a)]
            \item If $n_T(u_2) \geq 5$, let $x$ and $y$ denote the two leaves of $T_{u_2}$ closest to $u_2$ and let $u_x$ and $u_y$ denote their parents (note that $u_2 = u_x$ or $u_2 = u_y$; without loss of generality $u_2 = u_x$). We now delete the edges $(u_x,x)$ and $(u_y,y)$, suppress $u_x$ and $u_y$, subdivide the edge $(u,u_1)$ with two degree-2 vertices, say $\widetilde{u}_x$ and $\widetilde{u}_y$, such that $\widetilde{u}_x$ is the parent of $\widetilde{u}_y$ which in turn is the parent of $u_1$, and add the edges $(\widetilde{u}_x,x)$ and $(\widetilde{u}_y,y)$ to obtain $T'$. Now, we note that the $I'_v$ values of all nodes of $T$, respectively $T'$, except for $u, u_x, u_y, \widetilde{u}_x$, and $\widetilde{u}_y$ are unaffected by this procedure. Moreover, we have:
            \begin{itemize}
                \item The $I'$ value of $u$ is strictly larger in $T'$ than in $T$, i.e. $I'_u(T') > I'_u(T)$.
                \item Nodes $u_x$ and $u_y$ contribute to $\Sigma I'(T)$ and $\overline{I'}(T)$ but not to $\Sigma I'(T')$ and $\overline{I'}(T')$, whereas $\widetilde{u}_x$ and $\widetilde{u}_y$ contribute to $\Sigma I'(T')$ and $\overline{I'}(T')$ but not to $\Sigma I'(T)$ and $\overline{I'}(T)$. However, the number of nodes contributing to $\Sigma I'(T)$ and $\Sigma I'(T')$ as well as to $\overline{I'}(T)$ and $\overline{I'}(T')$ is the same. In addition, note that in $T'$ we have $I'_{\widetilde{u}_x}(T')=1$ (since $n_{T'}(\widetilde{u}_x)$ is odd) and in $T$ either $I'_{u_x}(T)=1$ or $I'_{u_y}(T)=1$ (depending on the parity of $n_T(u_2)$). Without loss of generality $I'_{u_x}(T)=1$ (else, swap the roles of $x$ and $y$). Moreover, as $n_{T'}(u_1) = n_T(u_1) \geq n_T(u_2)$ and $x$ and $y$ were attached above node $u_1$,  we also have $I'_{\widetilde{u}_y}(T') > I'_{u_y}(T)$. 
                \end{itemize}
            In total, this implies $\Sigma I'(T')>\Sigma I'(T)$ and $\overline{I'}(T') > \overline{I'}(T)$, thereby contradicting the maximality of $T$.
            \begin{figure}[htbp]
                \centering
                \includegraphics[scale=0.275]{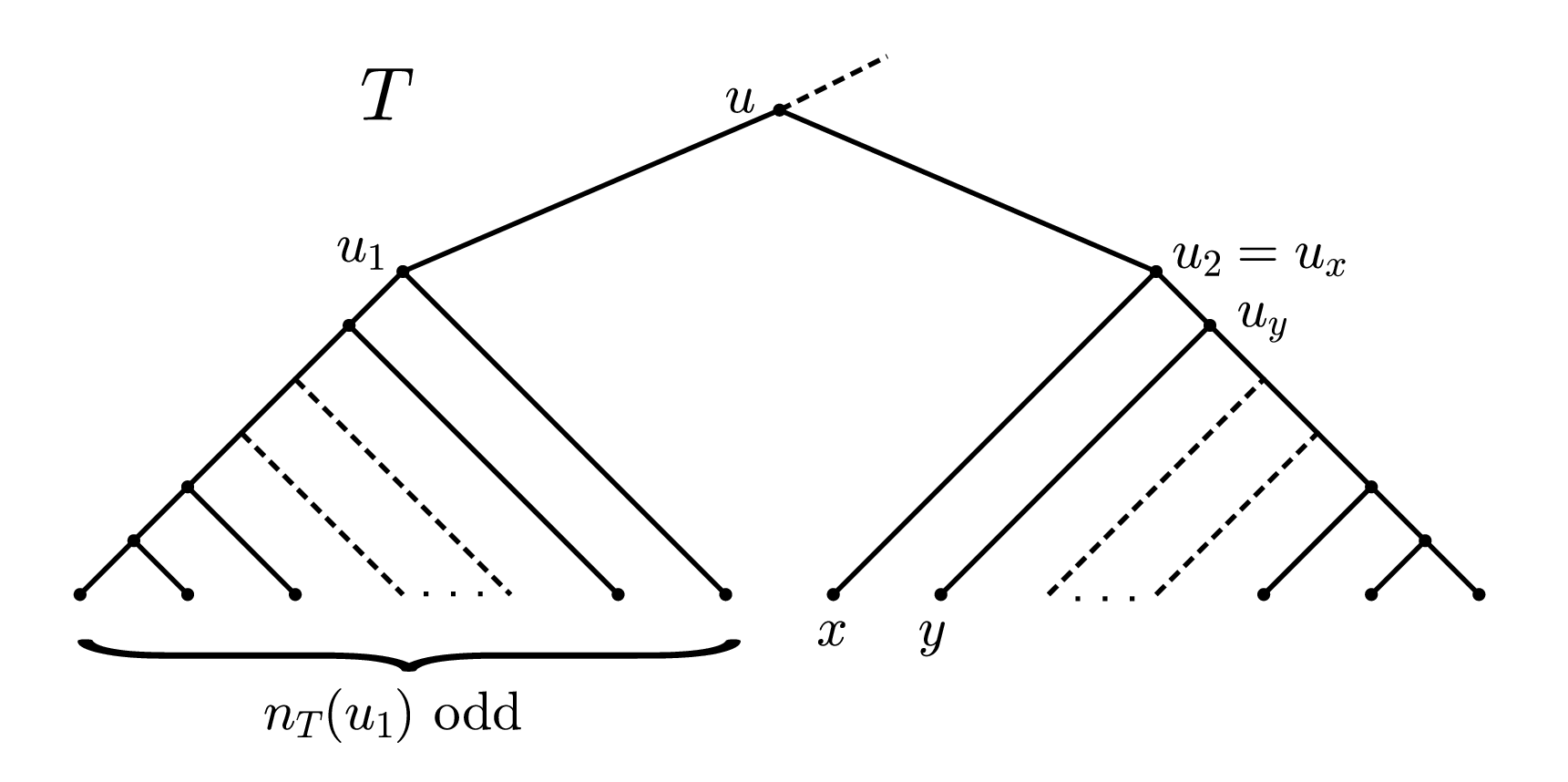}
                \includegraphics[scale=0.275]{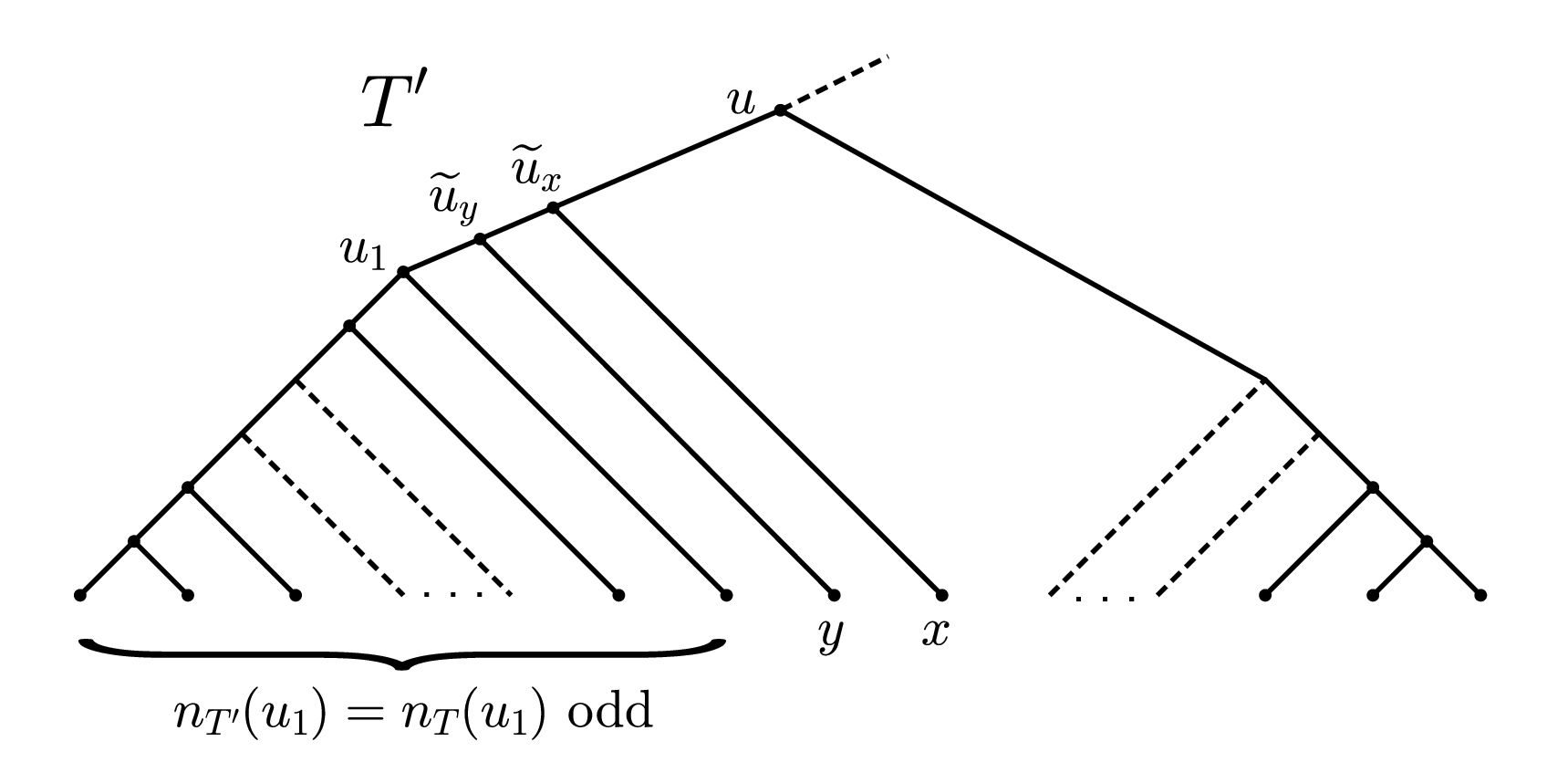}
                \caption{Trees $T$ and $T'$ as needed in Part~(ii-a) of the proof of Theorem~\ref{prop_meanI_max_b} for the case $n_T(u_2)=n_T(u_x)$ being odd (else swap $u_x$ and $u_y$).} 
                \label{Fig_MeanIPrime_Part2a}
            \end{figure}
        \item If $n_T(u_2) = 4$, let $x$ denote the leaf of $T_{u_2}$ adjacent to $u_2$. We now delete the edge $(u_2,x)$, suppress $u_2$, subdivide the edge $(u,u_1)$ with a degree-2 vertex $\widetilde{u}_2$, and add the edge $(\widetilde{u}_2,x)$ to obtain $T'$ (see Figure~\ref{Fig_MeanIPrime_Part2b}). We note that the $I'_v$ values of all nodes $v$ of $T$, respectively $T'$, except for $u, u_2$, and $\widetilde{u}_2$ are unaffected by this procedure. For $u, u_2$, and $\widetilde{u}_2$, we have:
            \begin{itemize}
                \item The $I'$ value of $u$ is strictly larger in $T'$ than in $T$, i.e. $I'_u(T') > I'_u(T)$. 
                \item Node $u_2$ contributes to $\Sigma I'(T)$ and $\overline{I'}(T)$ but not to $\Sigma I'(T')$ and $\overline{I'}(T')$, whereas node $\widetilde{u}_2$ contributes to $\Sigma I'(T')$ and $\overline{I'}(T')$ but not to $\Sigma I'(T)$ and $\overline{I'}(T)$. However, comparing the corresponding $I'$ values, we have $I'_{u_2}(T) = \frac{3}{4}$, whereas $I'_{\widetilde{u}_2}(T') = \frac{(n_{T'}(u_1)+1)-1}{(n_{T'}(u_1)+1)} = \frac{(n_{T}(u_1)+1)-1}{(n_{T}(u_1)+1)} = \frac{n_T(u_1)}{n_T(u_1)+1} \geq \frac{5}{6}$ since $n_T(u_1) \geq n_T(u_2) = 4$ and $n_T(u_1)$ is odd, thus $n_T(u_1) \geq 5$.  
            \end{itemize}
        These two observations immediately imply that $\Sigma I'(T')>\Sigma I'(T)$ and $\overline{I'}(T') > \overline{I'}(T)$, thereby contradicting the maximality of $T$.
        \begin{figure}[htbp]
            \centering
            \includegraphics[scale=0.3]{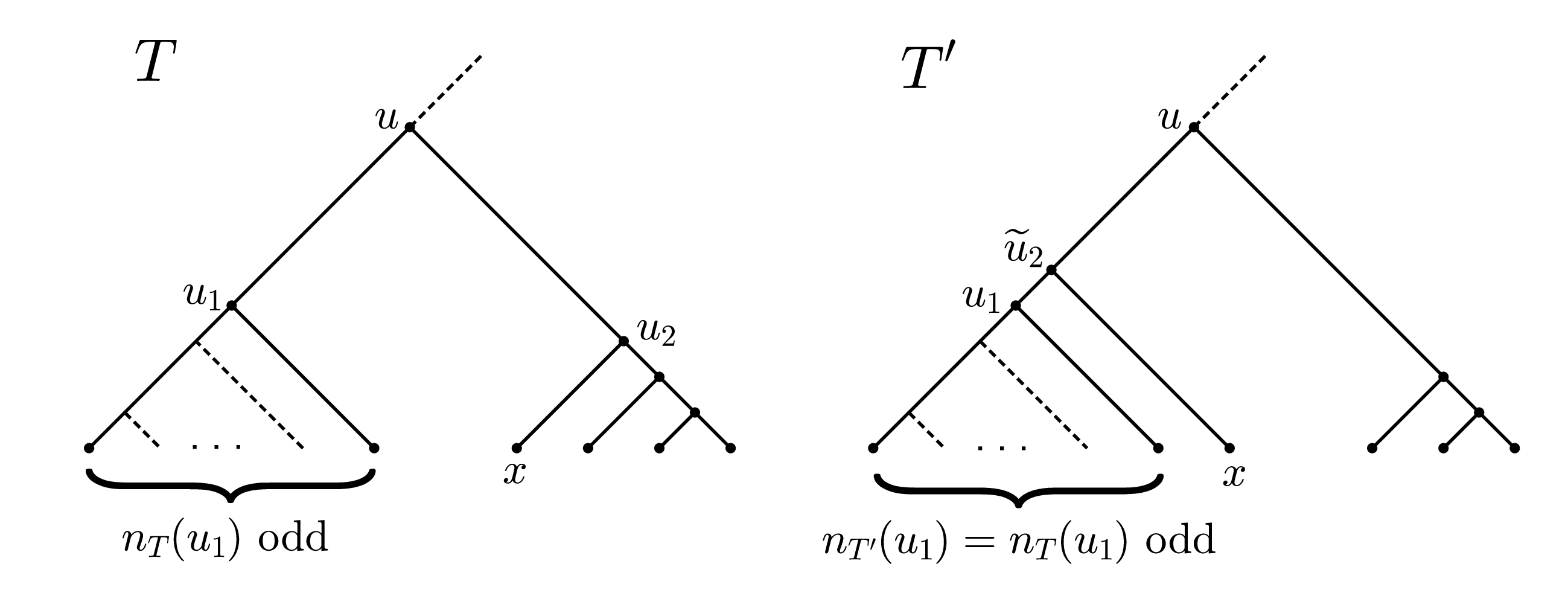}
            \caption{Trees $T$ and $T'$ as needed in Part~(ii-b) of the proof of Theorem~\ref{prop_meanI_max_b}.}
            \label{Fig_MeanIPrime_Part2b}
        \end{figure}
        \item If $n_T(u_2) = 3$, let $x$ and $y$ denote the two leaves of $T_{u_2}$ closest to $u_2$ and let $u_x$ and $u_y$ denote their parents (without loss of generality $u_2 = u_x$ and $u_y$ is a child of $u_2$). We now delete the edges $(u_x,x)$ and $(u_y,y)$, suppress $u_x$ and $u_y$, subdivide the edge $(u,u_1)$ with two degree-2 vertices, say $\widetilde{u}_x$ and $\widetilde{u}_y$, and add the edges $(\widetilde{u}_x,x)$ and $(\widetilde{u}_y,y)$ to obtain $T'$ (see Figure \ref{Fig_MeanIPrime_Part2c}). We now note that only the $I'$ values of the nodes $u$, $u_x$, $u_y$, $\widetilde{u}_x$ and $\widetilde{u}_y$ are affected by this procedure. In particular, note that $\widetilde{u}_x$ and $\widetilde{u}_y$ contribute to $\Sigma I'(T')$ and $\overline{I'}(T')$ because $n_T(u_1)\geq n_T(u_2)=3$ and thus $n_{T'}(\widetilde{u}_x), n_{T'}(\widetilde{u}_y)\geq 4$, whereas the nodes $u_x$ and $u_y$ did not contribute to $\Sigma I'(T)$ and $\overline{I'}(T)$, i.e. there are two more nodes contributing to $\Sigma I'(T')$ and $\overline{I'}(T')$ than to $\Sigma I'(T)$ and $\overline{I'}(T)$.
        \begin{itemize}
            \item First, note that $I'_{\widetilde{u}_x}(T') = 1$ (since $n_{T'}(\widetilde{u}_x)$ is odd) and $I'_{\widetilde{u}_y}(T') = \frac{n_T(u)-3}{n_T(u)-2}$ since
             \begin{align*}
                I'_{\widetilde{u}_y}(T') &= \frac{n_{T'}(\widetilde{u}_y)-1}{n_{T'}({\widetilde{u}_y)}} \cdot I_{\widetilde{u}_{y}}(T') \quad \text{(as $n_{T'}(\widetilde{u}_y)$ is even)} \\
                &= \frac{(n_T(u)-2)-1}{(n_T(u)-2)} \cdot \frac{(n_T(u)-3) - \left \lceil \frac{(n_T(u)-2)}{2}\right \rceil}{((n_T(u)-2)-1) -  \left\lceil \frac{(n_T(u)-2)}{2} \right \rceil} \\ &\qquad \text{(as $n_{T'}(\widetilde{u}_y) = n_T(u)-2$ and $n_{T'}(u_1) = n_T(u_1) = n_T(u)-3$)} \\
                &= \frac{n_T(u)-3}{n_T(u)-2} \cdot \frac{n_T(u)-3 - \frac{n_T(u)-1}{2}}{n_T(u)-3-\frac{n_T(u)-1}{2}} \\
                &= \frac{n_T(u)-3}{n_T(u)-2}.
            \end{align*}
            \item Second, note that $I'_u(T') = \frac{n_{T'}(u)-1}{n_{T'}(u)} \cdot 1 = \frac{n_{T}(u)-1}{n_{T}(u)}$ (since $n_{T'}(u)= n_T(u)$ is even).
            \item Finally, note that 
            \begin{align*}
                I'_u(T) &= \frac{n_T(u)-1}{n_T(u)} \cdot I_u(T) \quad \text{(as $n_T(u)$ is even)} \\
                &= \frac{n_T(u)-1}{n_T(u)} \cdot \frac{n_T(u_1) - \left \lceil \frac{n_T(u)}{2} \right \rceil}{(n_T(u)-1)- \left\lceil \frac{n_T(u)}{2} \right\rceil} \\
                &= \frac{n_T(u)-1}{n_T(u)} \cdot \frac{(n_T(u)-3) - \frac{n_T(u)}{2}}{n_T(u)-1-\frac{n_T(u)}{2}} \quad \text{(as $n_T(u_1) = n_T(u)-3$)} \\
                &= \frac{n_T(u)-1}{n_T(u)} \cdot \frac{\frac{n_T(u)-6}{2}}{\frac{n_T(u)-2}{2}} \\
                &= \frac{n_T(u)-1}{n_T(u)} \cdot \frac{n_T(u)-6}{n_T(u)-2}.
            \end{align*}
        \end{itemize}
        From $I_u'(T')>I'_u(T)$ and the fact that $\widetilde{u}_x$ and $\widetilde{u}_y$ contribute to $\Sigma I'(T')$, whereas $u_x$ and $u_y$ do not contribute to $\Sigma I'(T)$ it already follows that $\Sigma I'(T')>\Sigma I'(T)$ contradicting the maximality of $T$. We now argue why we must also have $\overline{I'}(T') > \overline{I'}(T)$. To this end, let $\mathring{V}_{\geq 4}(T)$ denote the set of interior vertices $v$ of $T$ with $n_v \geq 4$ (i.e. $\mathring{V}_{\geq 4}(T)$ is the set of interior vertices of $T$ that contribute to $\overline{I'}(T)$), and let $n_T(u) = n_{T'}(u)$ simply be denoted as $n_u$. Then, we have
        \begin{align*}
            \overline{I'}(T') &= \frac{1}{|\mathring{V}_{\geq 4}(T)|+2} \cdot \left(\sum\limits_{v \in \mathring{V}_{\geq 4}(T) \setminus \{u\}} I'_v(T) + I'_u(T') + I'_{\widetilde{u}_x}(T') + I'_{\widetilde{u}_y}(T') \right) \\
            &= \frac{1}{|\mathring{V}_{\geq 4}(T)|+2} \cdot \left(\sum\limits_{v \in \mathring{V}_{\geq 4}(T)} I'_v(T) + I'_u(T') - I'_u(T)+  I'_{\widetilde{u}_x}(T') + I'_{\widetilde{u}_y}(T') \right) \\
            &= \frac{1}{|\mathring{V}_{\geq 4}(T)|+2} \cdot \left(\sum\limits_{v \in \mathring{V}_{\geq 4}(T)} I'_v(T) +  \frac{n_u-1}{n_u} \left(1- \frac{n_u-6}{n_u-2} \right) + 1 + \frac{n_u-3}{n_u-2} \right) \\
            &= \frac{1}{|\mathring{V}_{\geq 4}(T)|+2} \cdot \left(\sum\limits_{v \in \mathring{V}_{\geq 4}(T)} I'_v(T) +  \frac{n_u-1}{n_u} \cdot \frac{n_u-2-(n_u-6)}{n_u-2} + 1 + \frac{n_u-3}{n_u-2} \right) \\
            &= \frac{1}{|\mathring{V}_{\geq 4}(T)|+2} \cdot \left(\sum\limits_{v \in \mathring{V}_{\geq 4}(T)} I'_v(T) +  \frac{4(n_u-1)  }{n_u(n_u-2)} + 1 + \frac{n_u-3}{n_u-2} \right) \\
            &= \frac{1}{|\mathring{V}_{\geq 4}(T)|+2} \cdot \left(\sum\limits_{v \in \mathring{V}_{\geq 4}(T)} I'_v(T) +  \frac{4n_u-4+n_u^2 -2n_u+n_u^2-3n_u}{n_u(n_u-2)} \right) \\
            &= \frac{1}{|\mathring{V}_{\geq 4}(T)|+2} \cdot \left(\sum\limits_{v \in \mathring{V}_{\geq 4}(T)} I'_v(T) +  \underbrace{\frac{2n_u^2-n_u-4}{n_u(n_u-2)}}_{= 2 + \frac{2}{n_u} + \frac{1}{n_u-2}} \right) \\
            &= \frac{1}{|\mathring{V}_{\geq 4}(T)|+2} \cdot \sum\limits_{v \in \mathring{V}_{\geq 4}(T)} I'_v(T) + \frac{1}{|\mathring{V}_{\geq 4}(T)|+2} \left(2+ \frac{2}{n_u} + \frac{1}{n_u-2} \right) \\
            &= \frac{|\mathring{V}_{\geq 4}(T)|}{|\mathring{V}_{\geq 4}(T)|+2} \cdot \overline{I'}(T) + \frac{1}{|\mathring{V}_{\geq 4}(T)|+2} \left(2+ \frac{2}{n_u} + \frac{1}{n_u-2} \right).
        \end{align*}
        Now, assume that $\overline{I'}(T') \leq \overline{I'}(T)$. Then, 
        \begin{align*}
           &\qquad \overline{I'}(T') \leq \overline{I'}(T) \\
           &\Leftrightarrow \frac{|\mathring{V}_{\geq 4}(T)|}{|\mathring{V}_{\geq 4}(T)|+2} \cdot \overline{I'}(T) + \frac{1}{|\mathring{V}_{\geq 4}(T)|+2} \cdot \left(2+ \frac{2}{n_u} + \frac{1}{n_u-2} \right) \leq \overline{I'}(T) \\
           &\Leftrightarrow \frac{1}{|\mathring{V}_{\geq 4}(T)|+2} \cdot \left(2+ \frac{2}{n_u} + \frac{1}{n_u-2} \right) \leq \overline{I'}(T) \cdot \left(1 - \frac{|\mathring{V}_{\geq 4}(T)|}{|\mathring{V}_{\geq 4}(T)|+2} \right) \\
           &\Leftrightarrow \frac{1}{|\mathring{V}_{\geq 4}(T)|+2} \cdot \left(2+ \frac{2}{n_u} + \frac{1}{n_u-2} \right) \leq \overline{I'}(T) \cdot \frac{2}{|\mathring{V}_{\geq 4}(T)|+2} \\
           &\Leftrightarrow 2+ \frac{2}{n_u} + \frac{1}{n_u-2} \leq 2 \cdot \overline{I'}(T) \\
           &\Leftrightarrow \underbrace{1+ \frac{1}{n_u} + \frac{1}{2(n_u-2)}}_{> 1} \leq \overline{I'}(T).
        \end{align*}
        However, as $\overline{I'}(T) \leq 1$ by definition, this is clearly a contradiction. Thus, we must have $\overline{I'}(T') > \overline{I'}(T)$, thereby contradicting the maximality of $T$.
        \begin{figure}[htbp]
            \centering
            \includegraphics[scale=0.3]{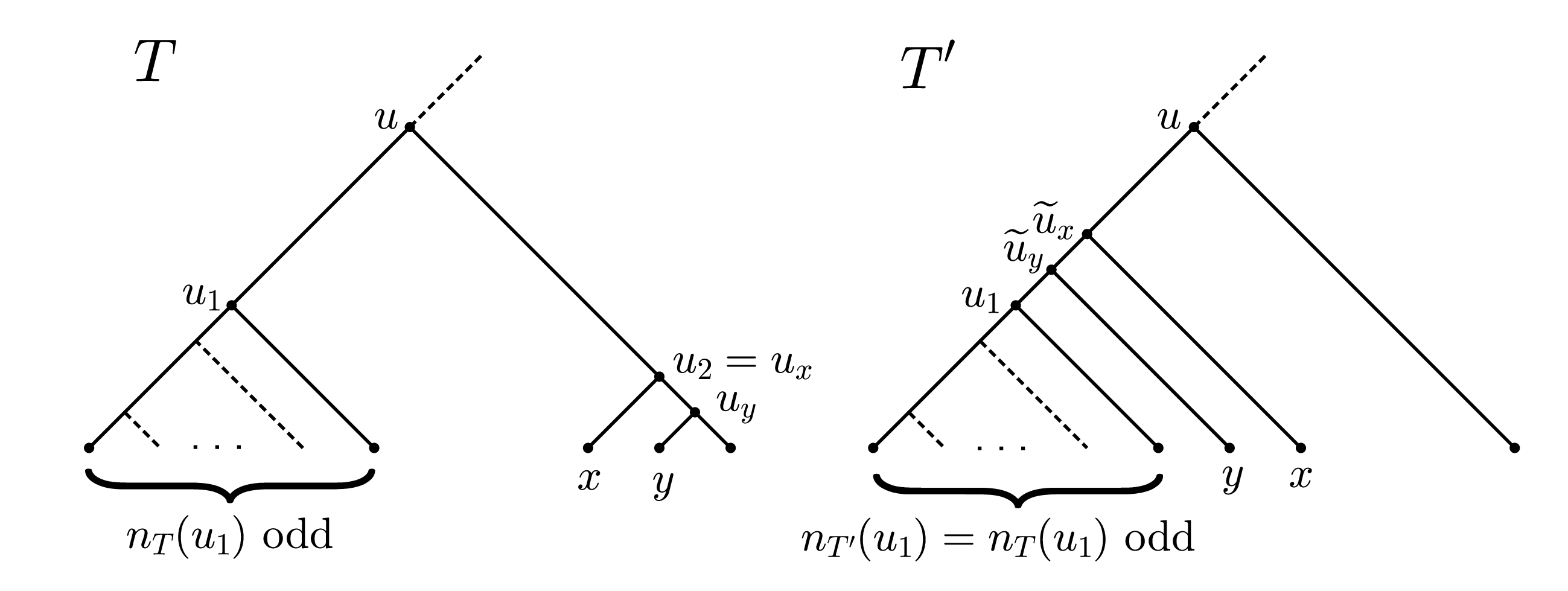}
            \caption{Trees $T$ and $T'$ as needed in Part~(ii-c) of the proof of Theorem~\ref{prop_meanI_max_b}. }
            \label{Fig_MeanIPrime_Part2c}
        \end{figure}
        \item If $n_T(u_2) = 2$, let $x$ denote one of the two leaves of $T_{u_2}$. We now delete the edge $(u_2,x)$, suppress $u_2$, subdivide the edge $(u,u_1)$ with a new degree-2 vertex $\widetilde{u}_2$, and add the edge $(\widetilde{u}_2,x)$ to obtain $T'$ (see Figure~\ref{Fig_MeanIPrime_Part2d}). We now note the following:
        \begin{itemize}
            \item There is one additional node contributing to $\overline{I'}(T')$, namely $\widetilde{u}_2$ because $n_T(u_1)\geq n_T(u_2)=2$ and $n_T(u_1)$ being odd imply $n_T(u_1)\geq 3$ and thus $n_{T'}(\widetilde{u}_2)\geq 4$. However, we have $I'_{\widetilde{u}_2}(T') > I_u(T)$ since:
            \begin{align*}
                I'_{\widetilde{u}_2}(T') 
                &= \frac{n_{T'}(\widetilde{u}_2)-1}{n_{T'}(\widetilde{u}_2)} \cdot I_{\widetilde{u}_2}(T') \quad \text{(since $n_{T'}(\widetilde{u}_2)$ is even)} \\
                &= \frac{(n_T(u)-1)-1}{(n_T(u)-1)} \cdot 1 \quad \text{(since $n_{T'}(\widetilde{u}_2) = n_{T'}(u)-1 = n_T(u)-1$)} \\
                &= \frac{n_T(u)-2}{n_T(u)-1}
            \end{align*}
            and
            \begin{align*}
                I'_u(T)
                &= I_u(T) \quad \text{(since $n_T(u)$ is odd)} \\
                &= \frac{n_{T}(u_1) - \left \lceil \frac{n_T(u)}{2} \right \rceil}{(n_T(u)-1) - \left \lceil \frac{n_T(u)}{2} \right \rceil} \\
                &= \frac{(n_T(u)-2) - \frac{n_T(u)+1}{2}}{n_T(u)-1-\frac{n_T(u)+1}{2}} \quad \text{(since $n_T(u_1)=n_T(u)-2$)} \\
                &= \frac{\frac{n_T(u)-5}{2}}{\frac{n_T(u)-3}{2}}
                 = \frac{n_T(u)-5}{n_T(u)-3}.
            \end{align*}
            Thus,
            \begin{align*}
                I'_{\widetilde{u}_2}(T')  -  I'_u(T) 
                &=  \frac{n_T(u)-2}{n_T(u)-1} - \frac{n_T(u)-5}{n_T(u)-3} \\
                &= \frac{(n_T(u)-2)(n_T(u)-3)-(n_T(u)-5)(n_T(u)-1)}{(n_T(u)-1)(n_T(u)-3)}\\
                &= \frac{(n_T(u)^2-5n_T(u)+6)-(n_T(u)^2-6n_T(u)+5)}{(n_T(u)-1)(n_T(u)-3)} \\
                &= \frac{n_T(u)+1}{(n_T(u)-1)(n_T(u)-3)} > 0 \quad \text{(since $n_T(u) \geq 4$)}.
            \end{align*}
        \item Moreover, $I'_u(T') = 1 > I'_u(T) = \frac{n_T(u)-5}{n_T(u)-3}$.
        \end{itemize}
        In summary, there is one additional node contributing towards $\Sigma I'(T')$ and $\overline{I'}(T')$ as compared to $\Sigma I'(T)$ and $\overline{I'}(T)$. However, its contribution is larger than the original contribution of node $u$, and node $u$ itself now contributes a maximal value of 1. This implies, that $\Sigma I'(T')>\Sigma I'(T)$ and $\overline{I'}(T') > \overline{I'}(T)$, thereby contradicting the maximality of $T$.
         \begin{figure}[htbp]
            \centering
            \includegraphics[scale=0.3]{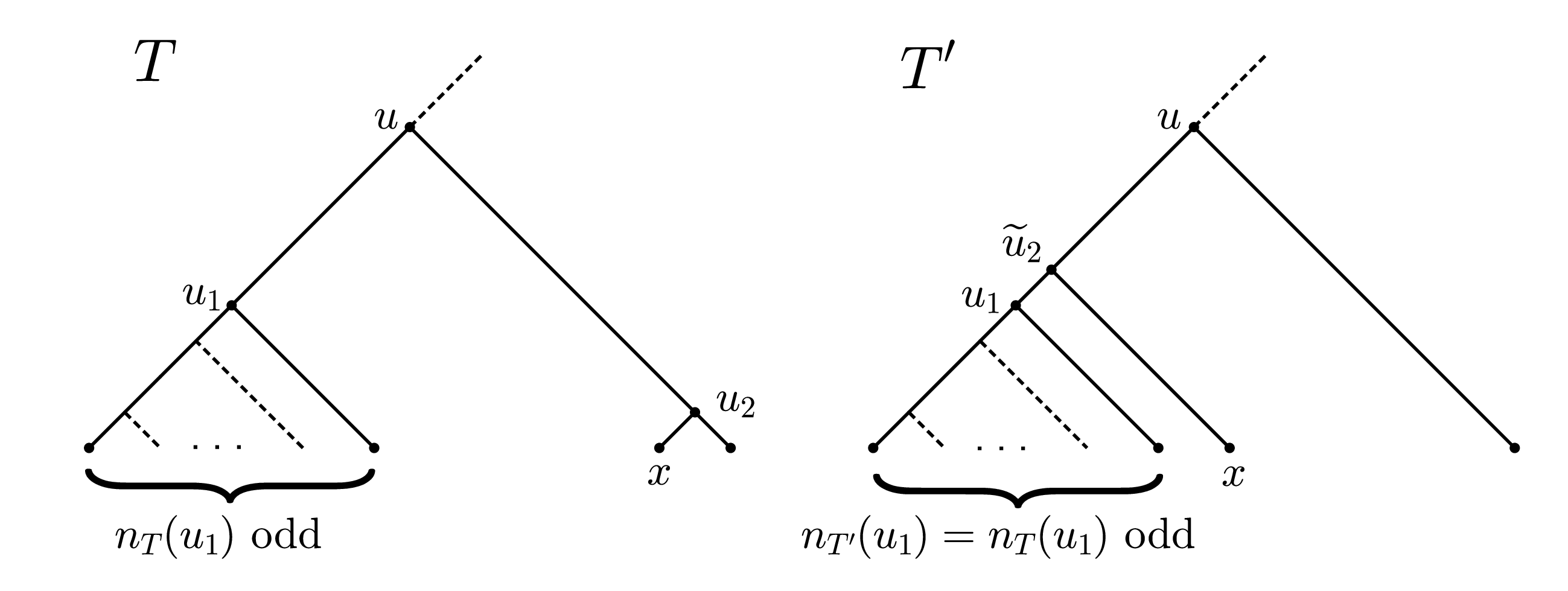}
            \caption{Trees $T$ and $T'$ as needed in Part~(ii-d) of the proof of Theorem~\ref{prop_meanI_max_b}. }
            \label{Fig_MeanIPrime_Part2d}
        \end{figure}
        \end{enumerate}
 \end{enumerate}
Thus, in all cases, we have $\Sigma I'(T')>\Sigma I'(T)$ and $\overline{I'}(T') > \overline{I'}(T)$ contradicting the maximality of $T$. Thus, the assumption that $T \neq \Tcat$ is a maximal tree was wrong. In particular, $\Tcat$ is the unique rooted binary tree $T \in \mathcal{BT}_{n \geq 4}^\ast$ maximizing $\Sigma I'(T)$ and $\overline{I'}(T)$. This completes the proof.
\end{proof} 

Now, we have a look at the minimal values of $\Sigma I'$ and $\overline{I'}$ on $\BTnstar$. 

\begin{theorem} \label{prop_meanI_min_b}
For every $n \in \mathbb{N}_{\geq 4}$ the minimal Total $I$ index $\Sigma I(T)$ over all $T \in \BTnstar$ is $\Sigma I(T)=0$ and this minimum is uniquely achieved by the maximally balanced tree $\Tmb$. In particular, for $n=2^h$ with $h \in \mathbb{N}_{\geq 2}$, $\Tfb = T^\mathit{mb}_{2^h}$ is the unique minimal tree. The same results hold for the correction method $I'$.\\
Also, for every $n \in \mathbb{N}_{\geq 4}$ the minimal Mean $I$ index $\overline{I}(T)$ over all $T \in \BTnstar$ is $\overline{I}(T)=0$ and this minimum is uniquely achieved by the maximally balanced tree $\Tmb$. In particular, for $n=2^h$ with $h \in \mathbb{N}_{\geq 2}$, $\Tfb = T^\mathit{mb}_{2^h}$ is the unique minimal tree. The same results hold for the correction method $I'$.
\end{theorem}
\begin{proof}
First, recall that by definition, $I_v \in [0,1]$ for any node $v \in \mathring{V}(T)$ with $n_v \geq 4$. This immediately implies that $\Sigma I(T)\geq 0$ and $\overline{I}(T) \geq 0$ for each $T \in \BTnstar$ with $n \geq 4$. 
Now, for each $v \in \mathring{V}(T)$ with $n_v \geq 4$ and children $v_1$ and $v_2$, we have $I_v=0$ if and only if $n_{v_1} = \left \lceil \frac{n_v}{2} \right \rceil$ and $n_{v_2} = \left \lfloor \frac{n_v}{2} \right \rfloor$. Moreover, $\Sigma I(T)=0$ and $\overline{I}(T)=0$ if and only if $I_v = 0$ for each $v \in \mathring{V}(T)$ with $n_v \geq 4$. In particular, $\Sigma I(T)=0$ and $\overline{I}(T) = 0$ if and only if $n_{v_1} = \left \lceil \frac{n_v}{2} \right \rceil$ and $n_{v_2} = \left \lfloor \frac{n_v}{2} \right \rfloor$ for each $v \in \mathring{V}(T)$ with $n_v \geq 4$ and children $v_1$ and $v_2$. However, this uniquely characterizes the maximally balanced tree $\Tmb$ (note that in the maximally balanced tree we have $n_{v_1} = \left \lceil \frac{n_v}{2}\right \rceil$ and $n_{v_2} = \left \lfloor \frac{n_v}{2} \right\rfloor$ for each $v \in \mathring{V}(T)$, but for $n_v \in \{2,3\}$ this holds in \emph{any} rooted binary tree). Thus, $\Tmb$ is the unique rooted binary tree in $\BTnstar$ with $n \geq 4$ minimizing the Total $I$ index $\Sigma I$ and the Mean $I$ index $\overline{I}$ and we have $\Sigma I(\Tmb)=\overline{I}(\Tmb)=0$. As the correction method $I'$ does not affect nodes with an $I_v$ value of 0, the same reasoning shows that $\Tmb$ is the unique rooted binary tree $T$ in $\BTnstar$ minimizing $\Sigma I'(T)$ and $\overline{I'}(T)$ and we have $\Sigma I'(\Tmb)=\overline{I'}(\Tmb) = 0$. This completes the proof. 
\end{proof}

\begin{remark} \label{remark_index_SumI_MeanI}
Note that Remark \ref{remark_I'_123} and Theorem \ref{prop_meanI_max_b} and Theorem \ref{prop_meanI_min_b} show that on $\BTnstar$ the caterpillar tree $\Tcat$ is the unique tree achieving the maximal value (for all $n\geq 1$) and the fully balanced tree $\Tfb$ is the unique tree achieving the minimal value (for $n=2^h$ with $h\geq 0$) of $\Sigma I$, $\Sigma I'$, $\overline{I}$ and $\overline{I'}$. This proves that all four measures are imbalance indices on $\BTnstar$ (opposed to $I_\rho$ and $I'_\rho$, which are neither balance nor imbalance indices according to our definitions).
\end{remark}

In addition to the binary case, we will now provide results on the maximal and minimal value in the arbitrary case. 
Note, however, that we only consider rooted trees $T \in \Tnstar$ with $\mathring{V}_{bin, \geq 4}(T) \neq \emptyset$ (as for all other trees, we trivially have $\Sigma I(T) = \Sigma I'(T) = \overline{I}(T) = \overline{I'}(T) = 0$). In particular, we only consider trees $T \in \mathcal{T}_{n \geq 4}^\ast$ that contain at least one binary node $v$ with $n_v \geq 4$.

\begin{theorem} \label{prop_I_max_a}
For every tree $T\in\mathcal{T}_{n\geq 4}^\ast$ with at least one binary node $v$ with $n_v \geq 4$, we have $\Sigma I(T)\leq n-3$ and $\Sigma I'(T)\leq\left \lfloor \frac{n-3}{2} \right \rfloor +  \sum\limits_{k=2}^{\left\lceil \frac{n-1}{2} \right\rceil}{\frac{2k-1}{2k}} < n-3$. These bounds are tight for all $n\geq 4$ and are achieved solely by the caterpillar tree $\Tcat$ and the tree which can be constructed from $\Tcat$ by contracting the inner edge leading to its only cherry. In particular, there are precisely two maximal trees for each $n\geq 4$. 
\end{theorem}

\begin{proof}
For the sake of a contradiction, assume that there is a tree $T$ that does not have one of the two described shapes but has maximal $\Sigma I$ index. We first construct a binary tree $T'$ from $T$ by resolving (if any) all multifurcations. To be more precise, if any exists let $v$ be a non-binary vertex with children $v_1,\ldots,v_k$ and without loss of generality let $n_{v_1}\geq n_{v_i}$ for $i=2,\ldots,k$. Now, delete the edges $(v,v_3),\ldots,(v,v_k)$, insert $k-2$ inner vertices $\widetilde{v}_3,\ldots,\widetilde{v}_k$ on the edge $(v,v_1)$ and insert the edges $(\widetilde{v}_3,v_3),\ldots,(\widetilde{v}_k,v_k)$. This procedure is repeated until all multifurcations are resolved. The resulting tree is then called $T'$ (with $T=T'$ if $T$ was binary). First, note that $T'$ is by construction binary. Second, note that any multifurcating vertex $v$ in $T$ does not contribute to $\Sigma I(T)$ (because $v\notin\mathring{V}_{bin,\geq 4}(T)$), and its replacement vertices $\widetilde{v}_3,\ldots,\widetilde{v}_k$ might or might not contribute a value $\geq 0$ to $\Sigma I(T')$. Thus, we already have $\Sigma I(T)\leq \Sigma I(T')$. If $T'\neq \Tcat$ we have $\Sigma I(T)\leq \Sigma I(T')<\Sigma I(\Tcat)$ (because of Theorem \ref{prop_meanI_max_b}), which contradicts the maximality of $T$. Thus, $T'=\Tcat$ must apply. This, in turn, means that all inner vertices of $T$ lie on the path from a leaf of maximal depth to the root. Let $v$ be the multifurcating vertex of $T$ with minimal depth. Since $T$ by assumption cannot be constructed from the caterpillar tree by contracting only the lowermost inner edge, we must have $n_T(v)=n_{T'}(v)\geq 4$. Recall that $v$ is not contributing to $\Sigma I(T)$, because it is not binary, but $v$ is contributing a value $>0$ to $\Sigma I(T')$, because in $T'$ it has (due to $n_{T'}(v)\geq 4$) the $I$ value $\frac{n_v-1-\left\lceil\frac{n_v}{2}\right\rceil}{n_v-1-\left\lceil\frac{n_v}{2}\right\rceil}>0$. Thus, in the case $T'=\Tcat$ we have $\Sigma I(T)<\Sigma I(T')$, which also contradicts the maximality of $T$. In total, $\Tcat$ and the tree which can be constructed from $\Tcat$ by contracting the inner edge leading to its only cherry are the only maximal trees in $\Tnstar$ concerning $\Sigma I$. The maximal value follows directly from Theorem \ref{prop_meanI_max_b}.

Since $I'_v>0$ if and only if $I_v>0$ the proof of the statements about the Total $I'$ index are analogous to the reasoning above where the maximal value follows again from Theorem \ref{prop_meanI_max_b}.
\end{proof}

\begin{remark}
In Theorem \ref{prop_I_max_a} it has been shown that for $n\geq 4$ the maximal trees on $\Tnstar$ concerning the Total $I$ index and the Total $I'$ index are the caterpillar tree $\Tcat$ and the tree that can be constructed from $\Tcat$ by contracting the lowermost inner edge. This proves that $\Sigma I$ and $\Sigma I'$ are neither balance nor imbalance indices on $\Tnstar$. Additionally, note that these two shapes are precisely the maximal trees on $\mathcal{T}_{n\geq 4}^\ast$ concerning the $\widehat{s}$-shape statistic (see Theorem \ref{thm_sstat_max_a}).
\end{remark}

\begin{proposition} \label{prop_meanI_max_a}
For every tree $T \in \mathcal{T}_{n \geq 4}^\ast$ with at least one binary node $v$ with $n_v \geq 4$, we have $\overline{I}(T) \leq 1$. This bound is tight for all $n \geq 4$. Moreover, any such tree $T \in \mathcal{T}_{n \geq 4}^\ast$ is a maximal tree if and only if each of its binary nodes $v \in \mathring{V}(T)$ with $n_v \geq 4$ has an $I_v$ value of one.\\
Similarly, for every tree $T\in\mathcal{T}_{n\geq 4}^\ast$ with at least one binary node $v$ with $n_v\geq 4$, we have $\overline{I'}(T)\leq 1$. This bound is tight for $n=5$ and all $n\geq 7$. Moreover, any such tree $T\in\Tnstar$ with $n=5$ or $n\geq 7$ is a maximal tree if and only if for each of its binary nodes $v\in\mathring{V}(T)$ with $n_v\geq 4$ it holds that its $I_v$ value is one and $n_v$ is odd.
\end{proposition}
\begin{proof}
First, recall that by definition, $I_v \in [0,1]$ for each binary node $v \in \mathring{V}(T)$ with $n_v \geq 4$. As the mean $I$ index of an arbitrary tree $T \in \mathcal{T}_{n \geq 4}^\ast$ with at least one binary node $v$ with $n_v \geq 4$ is defined as the mean of the $I_v$ values of its binary nodes $v \in \mathring{V}(T)$ with $n_v \geq 4$, it immediately follows that $\overline{I}(T) \leq 1$. For $n \geq 4$, we have $\overline{I}(\Tcat)=1$ (see Theorem~\ref{prop_meanI_max_b}) and thus the bound is tight. Finally, a tree $T \in \mathcal{T}_{n \geq 4}^\ast$ is a maximal tree (i.e. $\overline{I}(T)=1$) precisely if all its binary nodes $v \in \mathring{V}(T)$ with $n_v \geq 4$ have an $I_v$ value of one (as otherwise the mean of these values cannot be equal to one). 

Similarly, for the correction method $I'$, we have $I'_v \in [0,1]$ for each binary node $v \in \mathring{V}(T)$ with $n_v \geq 4$, and it immediately follows that $\overline{I'}(T) \leq 1$. To see that this bound is tight for all $n \geq 7$, consider a tree $T$ as depicted in Figure~\ref{Fig_MaxMeanI}. Here, the only node that contributes to $\overline{I'}(T)$ is node $v$ and we have $\overline{I'}(T) = I'_v = I_v= 1$, where the second to last equality follows from the fact that $n_v=5$ is odd. For $n=5$ simply consider the subtree $T_v$ in Figure \ref{Fig_MaxMeanI} to see that the bound is tight.
Finally, a tree $\in\Tnstar$ with $n=5$ or $n\geq 7$ is a maximal tree (i.e. $\overline{I'}(T)=1$) precisely if all its binary nodes $v \in \mathring{V}(T)$ with $n_v \geq 4$ have an $I'_v$ value of one (otherwise the mean of these values cannot be equal to one). However, this holds precisely if all binary nodes $v \in \mathring{V}(T)$ with $n_v \geq 4$ have an $I_v$ value of one and are such that $n_v$ is odd (as for a node with $I_v=1$ and $n_v$ even, we would have $I'_v = \frac{n_v-1}{n_v} I_v = \frac{n_v-1}{n_v} < 1$). This completes the proof. 
\end{proof}

\begin{figure}[htbp]
    \centering
    \includegraphics[scale=0.3]{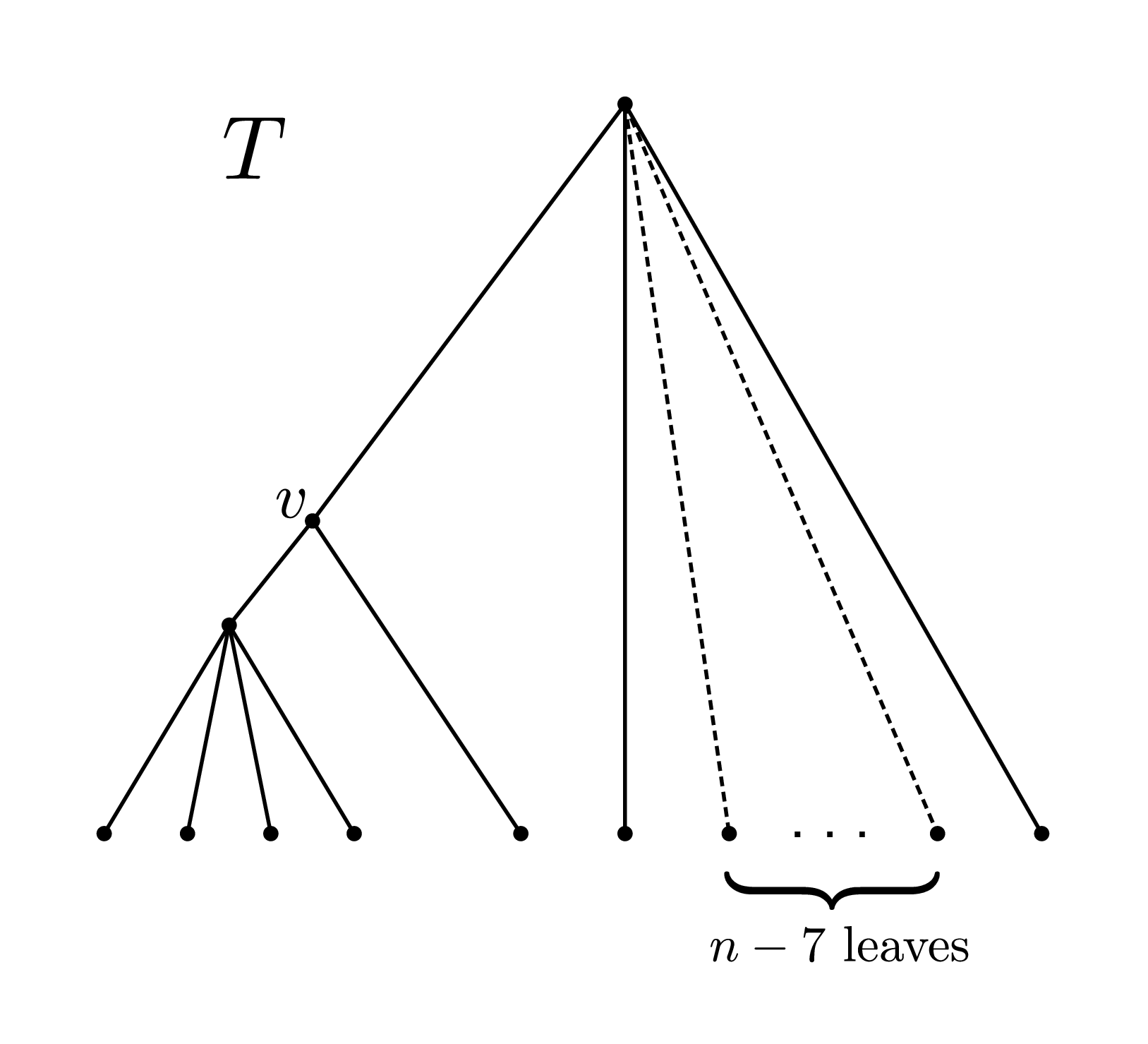}
    \caption{Rooted tree $T \in \mathcal{T}_{n \geq 7}^\ast$ with at least one binary node $v$ with $n_v \geq 4$ as needed in the proof of Proposition~\ref{prop_meanI_max_a}.}
    \label{Fig_MaxMeanI}
\end{figure}

\begin{proposition} \label{prop_meanI_min_a}
For every tree $T \in \mathcal{T}_{n \geq 4}^\ast$ with at least one binary node $v$ with $n_v \geq 4$, we have $\Sigma I(T)\geq 0$ and $\overline{I}(T) \geq 0$. These bounds are tight for all $n \geq 4$. Moreover, any such tree $T \in \mathcal{T}_{n \geq 4}^\ast$ is a minimal tree if and only if each of its binary nodes $v \in \mathring{V}(T)$ with $n_v \geq 4$ has an $I_v$ value of zero. The same results hold for the correction method $I'$.
\end{proposition}
\begin{proof}
First, recall that by definition, $I_v \in [0,1]$ for each binary node $v \in \mathring{V}(T)$ with $n_v \geq 4$. It thus immediately follows that $\Sigma I(T)\geq 0$ and $\overline{I}(T) \geq 0$ for any tree $T \in \mathcal{T}_{n \geq 4}^\ast$ with at least one binary node $v$ with $n_v \geq 4$. For $n \geq 4$, we have $\Sigma I(\Tmb)=\overline{I}(\Tmb)=0$ (see Theorem \ref{prop_meanI_min_b}) which shows that this bound is tight. Moreover, an arbitrary tree $T \in \mathcal{T}_{n \geq 4}^\ast$ with at least one binary node $v$ with $n_v \geq 4$ is a minimal tree (i.e. $\Sigma I(T)=\overline{I}(T)=0)$ if and only if all its binary nodes $v \in \mathring{V}(T)$ with $n_v \geq 4$ have an $I_v$ value of zero (as otherwise the sum and mean of these values cannot be zero). For the correction method $I'$ this follows analogously by noting that for each $v \in \mathring{V}(T)$ with $n_v \geq 4$, we have $I'_v = 0$ if and only if $I_v=0$. 
\end{proof} 

\begin{remark}
Note that while Proposition~\ref{prop_meanI_max_a} and Proposition~\ref{prop_meanI_min_a} provide a characterization of all maximal, respectively minimal, trees for the mean $I$ and mean $I'$ index, their exact numbers are -- to our knowledge -- not yet known.
\end{remark}

Last in this section, we will have a look at the properties of $I_v$ and $I_v'$ under the Yule model. The following results are based on \citet{Farris1976} and \citet{Slowinski1990} and have already been stated by \citet{Purvis2002} but without (complete and detailed) proofs.

\begin{lemma} \label{lem_Iv_YuleEV}
Let $T_n \in \BTn$ be a phylogenetic tree with $n \geq 4$ leaves sampled under the Yule model and let $v \in \mathring{V}(T_n)$ be an arbitrary vertex of $T_n$ with $n_v \geq 4$. Then, we have $E_Y(I_v) = \frac{1}{2}$ if $n_v$ is odd and $E_Y(I_v)=\frac{n_v/2}{(n_v-1)}>\frac{1}{2}$ monotonically decreasing with $\lim\limits_{n_v \rightarrow \infty}E_Y(I_v)=\frac{1}{2}$ if $n_v$ is even.
Using the correction method $I'$ we have $E_Y(I_v')=\frac{1}{2}$ independent of $n_v$.
\end{lemma}
\begin{proof}
Let $T_n \in \BTn$ be a phylogenetic tree with $n \geq 4$ leaves sampled under the Yule model. Then, the $I_v$ values for any $v \in \mathring{V}(T_n)$ with $n_v \geq 4$ are uniformly distributed on $\{0,\frac{1}{(n_v-1)- \left \lceil \frac{n_v}{2} \right \rceil},\frac{2}{(n_v-1)- \left \lceil \frac{n_v}{2} \right \rceil},\ldots,1\}$ if $n_v$ is odd. If $n_v$ is even, $I_v=0$ has probability $\frac{1}{n_v-1}$, whereas all other possible values $>0$ have probability $\frac{2}{n_v-1}$ \citep{Farris1976,Slowinski1990}. 
Thus, if $n_v$ is odd the expected value of the $I_v$ value of any such node $v$ is
\begin{equation*}
\begin{split}
    E_Y(I_v)& =\sum\limits_{i=1}^{n_v-1- \left\lceil \frac{n_v}{2} \right \rceil}{\frac{2}{n_v-1} \cdot \frac{i}{n_v-1- \left \lceil \frac{n_v}{2} \right \rceil}}= \frac{2}{n_v-1} \cdot  \frac{2}{n_v-3} \cdot \sum\limits_{i=1}^{\frac{n_v-3}{2}}{i}\\
    &=  \frac{2}{n_v-1} \cdot  \frac{2}{n_v-3} \cdot \left( \frac{1}{2} \cdot \frac{n_v-3}{2} \cdot \frac{n_v-1}{2} \right) = \frac{1}{2}
\end{split}
\end{equation*}
and if $n_v$ is even we have
\begin{equation*}
\begin{split}
    E_Y(I_v)& =\sum\limits_{i=1}^{n_v-1- \frac{n_v}{2}}{\frac{2}{n_v-1} \cdot \frac{i}{n_v-1- \frac{n_v}{2} }}= \frac{2}{n_v-1} \cdot  \frac{2}{n_v-2} \cdot \sum\limits_{i=1}^{\frac{n_v-2}{2}}{i}\\
    &=  \frac{2}{n_v-1} \cdot  \frac{2}{n_v-2} \cdot \left( \frac{1}{2} \cdot \frac{n_v-2}{2} \cdot \frac{n_v}{2} \right) = \frac{n_v/2}{(n_v-1)}>\frac{1}{2}.
\end{split}
\end{equation*}
If we use the $I'$ correction method, we have $E_Y(I_v')=E_Y(I_v)$ for $n_v$ odd and if $n_v$ is even we have $E_Y(I_v')=E_Y(\frac{n_v-1}{n_v}\cdot I_v)=\frac{n_v-1}{n_v} \cdot \frac{n_v/2}{(n_v-1)} = \frac{1}{2}$ since $I_v'=\frac{n_v-1}{n_v}\cdot I_v$.

Furthermore, the sequence $(a_n)_{n>3}$ with $a_{n}=\frac{n/2}{(n-1)}$ is monotonically decreasing because $\frac{a_{n+1}}{a_{n}} = \frac{(n+1)/2}{(n+1-1)} \cdot \frac{(n-1)}{n/2}=\frac{n^2-1^2}{n^2}<1$ and  has a lower bound of $\frac{1}{2}$.
Thus, the sequence converges and we have $\lim\limits_{n \rightarrow \infty}a_n=\lim\limits_{n \rightarrow \infty}a_{n+1}=\lim\limits_{n \rightarrow \infty}\frac{(n+1)/2}{n}=\frac{1}{2}\cdot \left(1 + \lim\limits_{n \rightarrow \infty}\frac{1}{n}\right)=\frac{1}{2}$. This completes the proof.
\end{proof}

The following results are based on Lemma \ref{lem_Iv_YuleEV} in this manuscript.

\begin{proposition} \label{prop_Iv_YuleEV}
Let $T_n \in \BTn$ be a phylogenetic tree with $n \geq 4$ leaves sampled under the Yule model. Then, for the $I_\rho$ value of $T_n$, we have $E_Y(I_{\rho}(T_n)) = \frac{1}{2}$ if $n$ is odd and $E_Y(I_\rho(T_n))=\frac{n/2}{n-1}>\frac{1}{2}$ if $n$ is even.
For the Mean $I$ index $\overline{I}(T_n)$, we have $\frac{1}{2} \leq E_Y(\overline{I}(T_n)) \leq \frac{n_{v,min}/2}{(n_{v,min}-1)} \leq \frac{2}{3}$ with $n_{v,min}$ being the smallest even subtree size $> 3$ in $T_n$. Using the correction method $I'$, we have $E_Y(I_\rho'(T_n))=E_Y(\overline{I'}(T_n))=\frac{1}{2}$. Finally, for the Total $I'$ index, we have $E_Y(\Sigma I'(T_n))=\frac{n}{4}-\frac{1}{2}$.
\end{proposition}
\begin{proof}
As $n_\rho=n \geq 4$, we can use Lemma~\ref{lem_Iv_YuleEV} to immediately conclude that $E_Y(I_\rho(T_n)) = E_Y(I_p) =  \frac{1}{2}$ if $n$ is odd and $E_Y(I_\rho(T_n)) = E_Y(I_p) =  \frac{n/2}{n-1}>\frac{1}{2}$ if $n$ is even. Similarly, $E_Y(I'_\rho(T_n)) = E_Y(I'_\rho)= \frac{1}{2}$.

For the Mean $I$ index note that in a rooted binary tree with $n$ leaves we have $m<n-1$ inner nodes $v_1,\ldots,v_m$ with $n_{v_i} \geq 4$ for each $i \in \{1, \ldots,m\}$, and thus using Lemma \ref{lem_Iv_YuleEV}, we can construct an upper and lower bound for $E_Y(\overline{I}(T_n))$:
$$ \frac{1}{2} \leq \min_{i \in \{1,\ldots,m\}}{E_Y(I_{v_i})} \leq E_Y(\overline{I}(T_n))= E_Y \left(\frac{1}{m}\cdot \sum\limits_{i=1}^{m}{I_{v_i}} \right)
= \frac{1}{m}\cdot \sum\limits_{i=1}^{m}{E_Y(I_{v_i})} \leq \max\limits_{i \in \{1,\ldots,m\}}{E_Y(I_{v_i})}.$$
We can further concretize the upper bound $\max\limits_{i \in \{1,\ldots,m\}}{E_Y(I_{v_i})}$ using the fact that $E_Y(I_v)=\frac{n_v/2}{(n_v-1)}>\frac{1}{2}$ for $n_v$ even is monotonically decreasing for $n_v \rightarrow \infty$ (Lemma~\ref{lem_Iv_YuleEV}). Therefore, a node $v_i$ with the smallest number $n_{v_i}$ of descending leaves yields the maximal value of $E_Y(I_{v_i})$. The smallest possible even subtree size $>3$ in a binary tree is $4$. Thus, we have $E_Y(\overline{I}(T_n)) \leq \frac{n_{v,min}/2}{(n_{v,min}-1)} \leq \frac{2}{3}$ with $n_{v,min}$ being the smallest even pending subtree size in $T_n$. 
Using the results of Lemma~\ref{lem_Iv_YuleEV} for the correction method $I'$, we have $E_Y(\overline{I'}(T_n))=\frac{1}{2}$ because the mean is unbiased. Moreover, by linearity of the expectation and as $E_Y(I'_{v}(T_n)) = \frac{1}{2}$ (see Lemma~\ref{lem_Iv_YuleEV}), for the Total $I'$ index, we have
\begin{align*}
    E_Y(\Sigma I'(T_n)) &= \frac{1}{2} \cdot E_Y \left[ \left \lvert \mathring{V}_{bin, \geq 4}(T_n) \right \rvert \right],
\end{align*}
where $E_Y \left[ \left \lvert \mathring{V}_{bin, \geq 4}(T_n) \right \rvert \right]$ denotes the expected number of nodes $v$ in $T_n$ with $n_v \geq 4$. Using the fact that the expected number of cherries in $T_n$ equals $\frac{n}{3}$ \citep{McKenzie2000} and the expected number of pitchforks (i.e. pending caterpillar trees on 3 leaves) equals $\frac{n}{6}$ for $n\geq 4$ \citep[Proof of Proposition 3]{Wu2015}\footnote{Note that \citep[Proposition 3]{Wu2015} states this fact for $n \geq 7$, but the proof shows that it already holds for $n \geq 4$.}, we get
\begin{align*}
    E_Y \left[ \left \lvert \mathring{V}_{bin, \geq 4}(T_n) \right \rvert \right] &= \underbrace{n-1}_{| \mathring{V}(T_n)|} - \underbrace{\frac{n}{3}}_{\text{number of inner nodes } v \text{ with } n_v = 2} - \underbrace{\frac{n}{6}}_{\text{number of inner nodes } v \text{ with } n_v = 3} = \frac{n}{2}-1.
\end{align*}
In particular, $ E_Y(\Sigma I'(T_n)) = \frac{1}{2} \left(\frac{n}{2}-1 \right) = \frac{n}{4} - \frac{1}{2}$ as claimed.
\end{proof}

\begin{proposition} \label{prop_Iv_YuleVar}
Let $T_n \in \BTn$ be a phylogenetic tree with $n \geq 4$ leaves sampled under the Yule model and let $v \in \mathring{V}(T_n)$ be an arbitrary node of $T_n$ with $n_v \geq 4$. Then, we have
$$V_Y(I_v)=\begin{cases} \frac{n_v^2-6n_v+17}{12(n_v-1)(n_v-3)} \xrightarrow{\ n_v \to \infty \ } \frac{1}{12}  & \textit{ if }n_v\textit{ is odd }\\
\frac{n_v^4 - 6 n_v^3 + 12 n_v^2 - 4 n_v}{12 (n_v-1)^3 (n_v-2)} \xrightarrow{\ n_v \to \infty \ } \frac{1}{12} & \textit{ if }n_v\textit{ is even }
\end{cases}$$
and 
$$V_Y(I_v')=\begin{cases} V_Y(I_v)=\frac{n_v^2-6n_v+17}{12(n_v-1)(n_v-3)} \xrightarrow{\ n_v \to \infty \ } \frac{1}{12}  & \textit{ if }n_v\textit{ is odd }\\
\frac{n_v^3 - 6 n_v^2 + 12 n_v - 4}{12 (n_v-2) (n_v-1) n_v }
\xrightarrow{\ n_v \to \infty \ } \frac{1}{12}
 & \textit{ if }n_v\textit{ is even. }
\end{cases}$$
In particular, $V_Y(I_\rho(T_n)) = V_Y(I_\rho)$ and $V_Y(I'_\rho(T_n)) = V_Y(I'_\rho)$ are obtained by substituting $n$ for $n_v$ in the expressions above.
\end{proposition}

\begin{proof}
Let $T_n \in \BTn$ be a phylogenetic tree with $n \geq 4$ leaves sampled under the Yule model.
Then, for any $v \in \mathring{V}(T_n)$ with $n_v \geq 4$, the $I_v$ values are uniformly distributed on $\{0,\frac{1}{(n_v-1)- \left \lceil \frac{n_v}{2} \right \rceil},\frac{2}{(n_v-1)- \left \lceil \frac{n_v}{2} \right \rceil},\ldots,1\}$ if $n_v$ is odd. If $n_v$ is even, $I_v=0$ has probability $\frac{1}{n_v-1}$, whereas all other possible values $>0$ have probability $\frac{2}{n_v-1}$ \citep{Farris1976,Slowinski1990}. 
By Lemma~\ref{lem_Iv_YuleEV}, the expected value of $I_v$ under the Yule model is $E_Y(I_v)=\frac{1}{2}$ if $n_v$ is odd and $E_Y(I_v)=\frac{n_v/2}{(n_v-1)}>\frac{1}{2}$ if $n_v$ is even. Thus, distinguishing the parity of $n_v \geq 4$, we can calculate the variance of $I_v$ as follows. For $n_v$ odd, we have:
\begin{equation*}
\begin{split}
    V_Y(I_v)& = \sum\limits_{i=1}^{n_v-1- \left \lceil \frac{n_v}{2} \right \rceil}{\frac{2}{n_v-1} \cdot \left( \frac{i}{n_v-1- \left \lceil \frac{n_v}{2} \right \rceil} -\frac{1}{2} \right)^2} = \frac{2}{n_v-1} \cdot \sum\limits_{i=1}^{\frac{n_v-3}{2}}{\left( \frac{2i}{n_v-3} -\frac{1}{2} \right)^2}\\
    &= \frac{2}{n_v-1} \cdot \left(\frac{4}{(n_v-3)^2} \cdot \sum\limits_{i=1}^{\frac{n_v-3}{2}}{i^2} - \frac{2}{n_v-3} \cdot \sum\limits_{i=1}^{\frac{n_v-3}{2}}{i} + \frac{1}{4} \cdot \sum\limits_{i=1}^{\frac{n_v-3}{2}}{1}\right)  \\
    &= \frac{2}{n_v-1} \cdot \left(\frac{4}{(n_v-3)^2} \cdot \left( \frac{1}{6} \cdot \frac{n_v-3}{2} \cdot \frac{n_v-1}{2} \cdot (n_v-2)\right) - \frac{2}{n_v-3} \cdot \left( \frac{1}{2} \cdot \frac{n_v-3}{2} \cdot \frac{n_v-1}{2}\right) + \frac{n_v-3}{8} \right)\\
    &=\frac{1}{3} \cdot \frac{n_v-2}{n_v-3} + \frac{1}{4} \cdot \frac{n_v-3}{n_v-1} -\frac{1}{2} 
    =- \frac{1}{2(n_v-1)} + \frac{n_v-2}{3(n_v-3)} +\frac{1}{12} \\
    &= \frac{n_v^2-6n_v+17}{12(n_v-1)(n_v-3)}= \frac{n_v^2-6n_v+17}{12n_v^2-48n_v+36} \xrightarrow{\ n_v \to \infty \ } \frac{1}{12}.
\end{split}
\end{equation*}
For $n_v$ even, we have: 
\begin{equation*}
\begin{split}
    V_Y(I_v)& = \sum\limits_{i=1}^{n_v-1- \frac{n_v}{2}}{\frac{2}{n_v-1} \cdot \left( \frac{i}{n_v-1- \frac{n_v}{2}} -\frac{n_v/2}{(n_v-1)} \right)^2} = \frac{2}{n_v-1} \cdot \sum\limits_{i=1}^{\frac{n_v-2}{2}}{\left( \frac{2i}{n_v-2} -\frac{n_v}{2(n_v-1)} \right)^2}\\
    &= \frac{2}{n_v-1} \cdot \left(\frac{4}{(n_v-2)^2} \cdot \sum\limits_{i=1}^{\frac{n_v-2}{2}}{i^2} - \frac{2n_v}{(n_v-1)(n_v-2)} \cdot \sum\limits_{i=1}^{\frac{n_v-2}{2}}{i} + \frac{n_v^2}{4(n_v-1)^2} \cdot \sum\limits_{i=1}^{\frac{n_v-2}{2}}{1}\right)  \\
    &= \frac{2}{n_v-1} \cdot \left(\frac{4}{(n_v-2)^2} \cdot \left( \frac{1}{6} \cdot \frac{n_v-2}{2} \cdot \frac{n_v}{2} \cdot (n_v-1)\right) \right.\\
    &\left. \quad - \frac{2n_v}{(n_v-1)(n_v-2)} \cdot \left( \frac{1}{2} \cdot \frac{n_v-2}{2} \cdot \frac{n_v}{2}\right) + \frac{n_v^2(n_v-2)}{8(n_v-1)^2} \right)\\
    &=\frac{1}{3} \cdot \frac{n_v}{n_v-2} - \frac{1}{2} \cdot \frac{n_v}{(n_v-1)^2} +\frac{1}{4}\frac{n_v^2(n_v-2)}{(n_v-1)^3} \\
    &=\frac{n_v (n_v^3 - 6 n_v^2 + 12 n_v - 4)}{12(n_v-1)^3(n_v-2)}
    =\frac{n_v^4- 6n_v^3 + 12n_v^2 - 4n_v}{12n_v^4 - 60n_v^3 + 108n_v^2 - 84n_v + 24} \xrightarrow{\ n_v \to \infty \ } \frac{1}{12}.
\end{split}
\end{equation*}
For the correction method $I'$, we have $E_Y(I_v')=\frac{1}{2}$ for all $v \in \mathring{V}(T_n)$ with $n_v\geq 4$ (see Lemma \ref{lem_Iv_YuleEV}) and the variance does not change if $n_v$ is odd, i.e. $V_Y(I_v')=V_Y(I_v)$. However, if $n_v$ is even, the variance is not equal to $V_Y(I_v)$. Instead we have
\begin{align*}
    V_Y(I_v') &= V_Y \left(\frac{n_v-1}{n_v} \cdot I_v \right) 
    = \frac{(n_v-1)^2}{n_v^2} \cdot V_Y(I_v) 
    = \frac{(n_v-1)^2}{n_v^2} \cdot \frac{n_v^4 - 6n_v^3 + 12 n_v^2 - 4 n_v}{12 (n_v-1)^3 (n_v-2)} \\
    &= \frac{n_v^4 - 6n_v^3 + 12 n_v^2 - 4 n_v}{12 (n_v-2) (n_v-1) n_v^2} 
    = \frac{n_v^3 - 6n_v^2 + 12 n_v - 4}{12 (n_v-2) (n_v-1) n_v} \xrightarrow{\ n_v \to \infty \ } \frac{1}{12}.
\end{align*}
\end{proof}

\subsubsection{Rooted quartet index}

Recall that the rooted quartet index \citep{Coronado2019} $rQI(T)$ of a tree $T\in\Tnstar$ is defined as the sum of the $rQI$-values of its rooted quartets, i.e. \[ rQI(T) \coloneqq \sum\limits_{Q \in \mathcal{Q}(T)} rQI(Q) = \sum\limits_{i=1}^4 |\{Q \in \mathcal{Q}(T): Q \text{ has shape } Q_i^\ast \}| \cdot q_i, \] where $q_0 = 0$ and $0 < q_1 < q_2 < q_3 < q_4$. If only binary trees $T\in\BTnstar$ are considered, this is the same as $q_3$ times the number of rooted quartets of shape $Q_3^*$, i.e. the fully balanced quartets, so \citet{Coronado2019} suggested the following alternative rooted quartet index for binary trees: \[ rQIB(T) \coloneqq \frac{1}{q_3}\cdot rQI(T)=|\{Q\in\mathcal{Q}(T):Q\text{ has shape }Q_3^\ast\}|. \] 

While the rooted quartet index for arbitrary trees is \emph{not} local (see Proposition \ref{locality_rQuartet_a}), we show in Proposition \ref{locality_rQuartet_b} that the rooted quartet index for binary trees is local. 

\begin{proposition} \label{locality_rQuartet_a}
The rooted quartet index for arbitrary trees is not local.
\end{proposition}
\begin{proof}
Consider the two trees $T$ and $T''$ in Figure \ref{fig_locality} on page \pageref{fig_locality}, which differ only in their subtrees rooted at $v$. Note that in both $T$ and $T''$ the vertex $v$ has exactly 5 descendant leaves. Now, for $T$ and $T''$ we have $rQI(T)-rQI(T'')=(107\cdot q_0+103\cdot q_3)-(52\cdot q_0+50\cdot q_2+103\cdot q_3+5\cdot q_4)=55\cdot q_0-50\cdot q_2-5\cdot q_4$. For $T_v$ and $T_v''$ we have $rQI(T_v)-rQI(T_v'')=5\cdot q_0-5\cdot q_4$. For the sake of a contradiction, assume that $rQI(T)-rQI(T'')=rQI(T_v)-rQI(T_v'')$. Then, we would have \[ 55\cdot q_0-50\cdot q_2-5\cdot q_4 = 5\cdot q_0-5\cdot q_4 \quad \Leftrightarrow \quad 50\cdot q_0 = 50\cdot q_2 \quad \Leftrightarrow \quad q_0 = q_2, \] which is a contradiction to $q_0<q_1<q_2$. Thus, the rooted quartet index for arbitrary trees is not local. Note that this property is due to the fact that changing the subtree $T_v$ might change the shape of a quartet on $u,v,w,x$ with $u,v,w\in V_L(T_v)$ and $x\in V_L(T)\setminus V_L(T_v)$.
\end{proof}

\begin{proposition} \label{locality_rQuartet_b}
The rooted quartet index for binary trees is local.
\end{proposition}
\begin{proof}
Let $T'$ be the binary tree that we obtain from $T\in\BTnstar$ by exchanging a subtree $T_v$ of $T$ with a binary subtree $T_v'$ on the same number of leaves. Now, recall that $rQIB(T)$ and $rQIB(T')$ count the number of fully balanced quartets induced by $T$ and $T'$, respectively. Obviously, all quartets induced by leaves in $V_L(T) \setminus V_L(T_v) = V_L(T') \setminus V_L(T_v')$ (i.e. all quartets that do not contain leaves of $T_v$ and $T_v'$) are of identical shape in $T$ and $T'$. It thus suffices to consider all induced quartets of $T$ and $T'$ that contain at least one leaf in $V_L(T_v) = V_L(T'_v)$. Thus, let $Y \subseteq V_L(T)$ with $|Y|=4$ be such that $Y \cap V_L(T_v) \neq \emptyset$ and consider $T_{|Y}$ and $T'_{|Y}$. We now distinguish two cases:
\begin{enumerate}
    \item $|Y \cap V_L(T_v)| \in \{1,2,3\}$, i.e. $Y$ contains one, two or three leaves of $V_L(T_v) = V_L(T_v')$. Then, $T_{|Y}$ and $T'_{|Y}$ always have the same shape. This is due to the fact the the subtrees of $T_{|Y}$ and $T'_{|Y}$ induced by the leaves in $Y \cap V_L(T_v)$ have the same shape (as there is only one binary tree with one, two or three leaves). In particular, if $T_{|Y}$ is fully balanced (and thus counts towards $rQIB(T)$), then $T'_{|Y}$ is also fully balanced (and thus counts towards $rQIB(T')$) and vice versa.
    \item $|Y \cap V_L(T_v)|=4$, i.e. all leaves in $Y$ are from $V_L(T_v)=V_L(T'_v)$. Then, $T_{|Y}$ and $T'_{|Y}$ do not necessarily have the same shape (as there are two distinct rooted binary trees on four leaves). However, if $|Y \cap V_L(T_v)|=4$, then $T_{|Y} = {T_v}_{|Y}$ and $T'_{|Y} = {T'_v}_{|Y}$. In particular, $T_{|Y}$ and $T'_{|Y}$ have the same shape if and only if ${T_v}_{|Y}$ and ${T'_v}_{|Y}$ have the same shape. 
\end{enumerate}
In total, this implies that only if $Y \subseteq V_L(T_v) = V_L(T_v')$, the induced quartets $T_{|Y}$ and $T'_{|Y}$ might be of different shape (in which case ${T_v}_{|Y}$ and ${T'_v}_{|Y}$ are of different shape as well). This in turn implies that $rQIB(T)-rQIB(T') = rQIB(T_v)-rQIB(T_v')$, which completes the proof.
\end{proof}

\subsubsection{\texorpdfstring{$\widehat{s}$}{s}-shape statistic}

In this section, we will turn our attention to the $\widehat{s}$-shape statistic. Recall that the $\widehat{s}$-shape statistic \citep{blum2006c} $\widehat{s}(T)$ of a tree $T\in\Tnstar$ is defined as \[ \widehat{s}(T)\coloneqq\sum\limits_{v\in\mathring{V}(T)} \log(n_v-1), \] where the logarithm base can be chosen arbitrarily.

We remark that a tree shape statistic related to the $\widehat{s}$-shape statistic has been studied in the literature for so-called \emph{binary search trees}, where similar results to the ones given in this manuscript were obtained \citep{Fill1996}. However, here we phrase our results in terms of arbitrary rooted trees and not in terms of binary search trees. We begin by considering the computation time, recursiveness and locality. 

\begin{proposition} \label{runtime_sShape}
For every tree $T\in\Tnstar$, the $\widehat{s}$-shape statistic $\widehat{s}(T)$ can be computed in time $O(n)$.
\end{proposition}
\begin{proof}
A vector containing the values $n_u$ for each $u\in V(T)$ can be computed in time $O(n)$ by traversing the tree in post order, setting $n_u=1$ if $u$ is a leaf and calculating $n_u=n_{u_1}+\ldots+n_{u_k}$ otherwise (where $u_1, \ldots, u_k$ denote the children of $u$). Then, the $\widehat{s}$-shape statistic can be computed from this vector in time $O(n)$ since the cardinality of $\mathring{V}(T)$ is at most $n-1$.
\end{proof}

\begin{proposition} \label{recursiveness_sShape}
The $\widehat{s}$-shape statistic is a recursive tree shape statistic. We have $\widehat{s}(T)=0$ for $T\in\mathcal{T}_1^\ast$, and for every tree $T\in\Tnstar$ with $n\geq 2$ and standard decomposition $T=(T_1,\ldots,T_k)$ we have \[ \widehat{s}(T)=\sum\limits_{i=1}^k \widehat{s}(T_i)+\log\left(-1+\sum\limits_{i=1}^k n_i\right). \]
\end{proposition}
\begin{proof}
Using $n_u\coloneqq n_T(u)=n_{T_i}(u)$ for all $u\in V(T_i)$, the $\widehat{s}$-shape statistic fulfills the recursion \[ \widehat{s}(T) = \sum\limits_{u\in\mathring{V}(T)} \log(n_u-1) = \sum\limits_{i=1}^k \left(\sum\limits_{u\in\mathring{V}(T_i)} \log(n_u-1) \right)+\log(n_{\rho}-1) = \sum\limits_{i=1}^k \widehat{s}(T_i)+\log\left(-1+\sum\limits_{i=1}^k n_i\right). \] Thus, it can be expressed as a recursive tree shape statistic of length $x=2$ with the recursions (where $\widehat{s}_i$ is the simplified notation of $\widehat{s}(T_i)$ and $n_i$ denotes the leaf number of $T_i$)
\begin{itemize}
    \item $\widehat{s}$-shape statistic: $\lambda_1=0$ and $r_1(T_1,\ldots,T_k)=\widehat{s}_1+\ldots+\widehat{s}_k+\log(n_1+\ldots+n_k-1)$
    \item leaf number: $\lambda_2=1$ and $r_2(T_1,\ldots,T_k)=n_1+\ldots+n_k$
\end{itemize}
It can easily be seen that $\lambda\in\mathbb{R}^2$ and $r_i: \underbrace{\mathbb{R}^2\times\ldots\times\mathbb{R}^2}_{k\text{ times}} \rightarrow \mathbb{R}$, and that all $r_i$ are independent of the order of subtrees. This completes the proof.
\end{proof}

\begin{proposition} \label{locality_sShape}
The $\widehat{s}$-shape statistic is local.
\end{proposition}
\begin{proof}
Let $T'$ be the tree that we obtain from $T\in\Tnstar$ by exchanging a subtree $T_v$ of $T$ with a subtree $T_v'$ on the same number of leaves. Note that $\mathring{V}(T)\setminus\mathring{V}(T_v)=\mathring{V}(T')\setminus\mathring{V}(T_v')$ and $n_T(w)=n_{T'}(w)$ if $w\in \mathring{V}(T)\setminus\mathring{V}(T_v)$, because changing the shape of $T_v$ does not change the number of descendant leaves of $w\in \mathring{V}(T)\setminus\mathring{V}(T_v)$ as $T_v$ and $T_v'$ have the same number of leaves. Also note that $n_T(w)=n_{T_v}(w)$ if $w\in\mathring{V}(T_v)$ and $n_{T'}(w)=n_{T_v'}(w)$ if $w\in\mathring{V}(T'_v)$. Hence, we have 
\begin{equation*}
\begin{split}
    \widehat{s}(T) - \widehat{s}(T') &= \sum\limits_{w\in\mathring{V}(T_v)} \log(n_T(w)-1) + \sum\limits_{w\in\mathring{V}(T)\setminus\mathring{V}(T_v)} \log(n_T(w)-1) - \sum\limits_{w\in\mathring{V}(T'_v)} \log(n_{T'}(w)-1)\\
    &\qquad - \sum\limits_{w\in\mathring{V}(T')\setminus\mathring{V}(T'_v)} \log(n_{T'}(w)-1)\\
    &= \sum\limits_{w\in\mathring{V}(T_v)} \log(n_{T_v}(w)-1) - \sum\limits_{w\in\mathring{V}(T'_v)} \log(n_{T_v'}(w)-1) = \widehat{s}(T_v) - \widehat{s}(T_v').
\end{split}
\end{equation*}
Thus, the $\widehat{s}$-shape statistic is local.
\end{proof}

Now, we will have a look at the maximal value of the $\widehat{s}$-shape statistic.

\begin{theorem}\label{thm_sstat_cat}
The caterpillar tree $T_n^\mathit{cat}$ is the unique rooted binary tree on $n$ leaves maximizing $\widehat{s}$. Moreover, we have $\widehat{s}(T_n^\mathit{cat})=\log((n-1)!)$.
\end{theorem}

Before we can prove this theorem, we need the following proposition as well as one more technical lemma. 

\begin{proposition} \label{prop_s_standarddecomp}
Let $T$ be a rooted binary tree with $n\geq 2$ leaves and with standard decomposition $T=(T_1,\ldots,T_k)$. Now, if $\widehat{s}(T)$ is minimal (maximal) then $\widehat{s}(T_1),\ldots,\widehat{s}(T_k)$ are also minimal (maximal).  
\end{proposition}
\begin{proof}
First, note that $\widehat{s}(T)=\widehat{s}(T_1)+\ldots+\widehat{s}(T_k)+\log(n-1)$ (see Proposition \ref{recursiveness_sShape}). Now assume that at least one of the maximal pending subtrees, say $T_1$, was not minimal (maximal) amongst all trees with the same leaf number. Then we could find a tree $\widetilde{T}_1$ with $n_1$ leaves and with $\widehat{s}(\widetilde{T}_1)<\widehat{s}(T_1)$ (or the other way round in case of maximality). We could then construct a tree $T'$ with $n$ leaves and with standard decomposition $T'=(\widetilde{T}_1,T_2,\ldots,T_k)$, and we would have \[ \widehat{s}(T')=\widehat{s}(\widetilde{T}_1)+\widehat{s}(T_2)+\ldots+\widehat{s}(T_k)+\log(n-1)<\widehat{s}(T_1)+\widehat{s}(T_2)+\ldots+\widehat{s}(T_k)+\log(n-1)=\widehat{s}(T) \] (or the other way around in case of maximality). This would clearly contradict the minimality (maximality) of $T$, which shows that the assumption was wrong. This completes the proof.
\end{proof}

\begin{lemma}\label{lem_binfrac}
Let $n \in \mathbb{N}_{\geq 3}$, let $n_2 \in \mathbb{N}$ such that $1 < n_2 \leq \frac{n}{2}$. Then, we have: \[ (n_2-1)! < \frac{(n-2)!}{(n-n_2-1)!}. \]
\end{lemma}
\begin{proof}
First note that after cancelling out the denominator, we get \[ \frac{(n-2)!}{(n-n_2-1)!}= (n-2)\cdot (n-3) \cdot \ldots \cdot (n-(n_2-1)-1). \]
Compare this to $(n_2-1)! = (n_2-1)\cdot (n_2-2)\cdot \ldots \cdot 1$. It can easily be seen that the number of factors in both products is identical, namely $n_2-1$. Moreover, note that the last factor of the second product is strictly smaller than the last factor of the first product, which equals $n-n_2 \geq \frac{n}{2}>1$, because $n_2 \leq \frac{n}{2}$. By the same argument, the second-to-last factor of the second product must be larger than the second-to-last factor of the first product (as both are just one more than the last factor) and so forth. This proves the assertion.
\end{proof}

Now we are finally in a position to prove Theorem \ref{thm_sstat_cat}.

\begin{proof}[Proof of Theorem \ref{thm_sstat_cat}]
We prove the last statement first. $\widehat{s}(T_n^\mathit{cat})=\sum\limits_{v \in \mathring{V}}\log(n_v-1)=\log\left(\prod\limits_{v \in \mathring{V}} (n_v-1) \right)= \log \left( \prod\limits_{i=1}^{n-1}i\right)=\log((n-1)!),$ where the product over all $i=1,\ldots,n-1$ stems from the fact that the caterpillar has one vertex with two descending leaves, one vertex with three descending leaves, and so forth, up to the root, which has $n$ descending leaves. So the factors $n_v-1$ run from 1 to $n-1$, accordingly.

Last, we show that the caterpillar is the only binary tree achieving the maximum value of $\widehat{s}$. We do this by induction on $n$. For $n\leq 3$ there is only one rooted binary tree (namely $\Tcat$), so there is nothing to show. For $n=4$, there are two rooted binary trees, namely $T_4^\mathit{cat}$ and $T_2^\mathit{fb}$. The latter has two vertices with two descending leaves and one vertex, namely the root, with four descending leaves, whereas the caterpillar has one node with two, one with three and one with four descending leaves each. Thus, $\widehat{s}(T_2^\mathit{fb})=\log(1\cdot 1\cdot 3)< \log(1\cdot 2\cdot 3)=\widehat{s}(T_4^\mathit{cat})$.

Now assume that the assertion holds for all rooted binary trees with up to $n-1$ leaves and consider a rooted binary tree $T=(T_1,T_2)$ with $n$ leaves. Then, by Proposition \ref{prop_s_standarddecomp}, we know that if $T$ maximizes $\widehat{s}$ then $T_1$ and $T_2$ maximize $\widehat{s}$ for $n_1$ and $n_2$, respectively. Without loss of generality, we assume $n_1 \geq n_2$, so in particular, $n_2 \leq \frac{n}{2}$. 

However, by the inductive hypothesis, $T_1$ and $T_2$ maximize $\widehat{s}$ precisely if $T_1$ and $T_2$ are caterpillars, and we have (as shown above) that $\widehat{s}(T_1) = \log((n_1-1)!)$ and $\widehat{s}(T_2)=\log((n_2-1)!)$.

We now want to show that $T$ is a caterpillar, too, so we need to show that $n_2=1$. Assume this is not the case, i.e. assume $n_2\geq 2$. Then, we have due to the recursiveness of $\widehat{s}$ (see Proposition \ref{recursiveness_sShape}):  \begin{align*}\widehat{s}(T)&= \log((n_1-1)!) + \log((n_2-1)!)+\log(n-1) = \log((n_1-1)!\cdot(n_2-1)!) +\log(n-1) \\ &= \log((\underbrace{n-n_2}_{=n_1}-1)!\cdot\underbrace{(n_2-1)!}_{\substack{<\frac{(n-2)!}{(n-n_2-1)!}\\ \mbox{\tiny by Lemma \ref{lem_binfrac}}}}) +\log(n-1)<
 \log\left((n-n_2-1)!\cdot \frac{(n-2)!}{(n-n_2-1)!}\right)+\log(n-1) \\&= \log((n-2)!)+\log(n-1)=\widehat{s}(T_{n-1}^\mathit{cat})+\log(n-1),
 \end{align*}
 where the latter equality stems from what we have shown in the beginning of this proof. So if $n_2 \geq 2$, we have $\widehat{s}(T) < \widehat{s}(T_{n-1}^{cat})+\log(n-1) =\widehat{s}\left( T_n^{cat}\right)$. This contradicts the maximality of $\widehat{s}(T)$. So the assumption was wrong and we can conclude that $n_2=1$. As $T_1$ is a caterpillar by the inductive hypothesis, this implies that $T$ is a caterpillar, too. This completes the proof.   \\
\end{proof}

\begin{remark} \label{rem_catnotunique} 
Note that the caterpillar tree is \emph{not} the unique tree in $\Tnstar$ maximizing $\widehat{s}$. Consider a tree $T$ that has the shape of a binary caterpillar except that the edge leading to the parent of the unique cherry has been contracted. Then, as the lowermost cherry in a binary caterpillar only contributes $\log(n_v-1)=\log(2-1)=\log(1)=0$ to $\widehat{s}$, $T$ has precisely the same $\widehat{s}$ value as $T_n^\mathit{cat}$. This shows that, as opposed to the binary case, the maximum is not unique in the arbitrary case. For this reason, the $\widehat{s}$-shape statistic does not fulfill our definition of an imbalance index (see Definition \ref{def_imbalance}) when arbitrary trees are considered, but only when it is restricted to $\BTnstar$, i.e. binary trees.
\end{remark}

In the following theorem we seek to characterize all arbitrary (i.e. not necessarily binary) maxima of $\widehat{s}$.

\begin{theorem} \label{thm_sstat_max_a}
Let $T\in\Tnstar$ be a tree with $n$ leaves and maximal $\widehat{s}$. Then, $\widehat{s}(T)=\log((n-1)!)$, and $T$ either equals $T_n^\mathit{cat}$ or it can be constructed by contracting the inner edge leading to the only cherry in $T_n^\mathit{cat}$.
\end{theorem}
\begin{proof}
By Remark \ref{rem_catnotunique}, it is clear that if a tree $T$ is a binary caterpillar or can be obtained from one by contracting the lowermost internal edge, then we have $\widehat{s}(T)=\log((n-1)!)$. 
Next we need to show that $\widehat{s}(T)=\log((n-1)!)$ is maximal even if $T$ is not binary. Assume this is not the case, i.e. assume there exists a tree $T$ with $\widehat{s}(T)>\log((n-1)!)$. By Theorem \ref{thm_sstat_cat}, $T$ cannot be binary, so $T$ contains at least one vertex $v$ with at least three children $w_1$, $w_2$ and $w_3$ (and possibly more). We now construct a tree $T'$ as follows: We delete $v$ and its incident edges and add two new vertices $v_1$ and $v_2$ as well as a new edge $e=(v_1,v_2)$. We then connect $w_1$ to $v_1$ and all other children of $v$ to $v_2$ by new edges, and if $v\neq\rho$ connect the parent of $v$ to $v_1$ by a new edge. This way, $n_v=n_{v_1}$, and $n_{v_2}\geq 2$ (as at least the two vertices $w_2$ and $w_3$ descend from $v_2$), which shows that $\widehat{s}(T')\geq \widehat{s}(T)$. Repeating this procedure until there is no more vertex with at least three children leads to a binary tree $T^*$, for which we have $\widehat{s}(T^*)\geq \widehat{s}(T)>\log((n-1)!)$ by assumption. This is a contradiction to Theorem \ref{thm_sstat_cat}, which is why also trees which are not binary cannot exceed this maximal $\widehat{s}$ value.

Last, we need to show that the maximal value of $\log((n-1)!)$ can only be achieved by trees of the described two shapes. Assume there is a tree $T$ such that $\widehat{s}(T)=\log((n-1)!)$ and such that $T$ does not have one of the two described shapes. Without loss of generality, we assume $T$ is minimal with this property, i.e. there is no tree with fewer leaves that has maximal $\widehat{s}$ value and is not a caterpillar or a caterpillar with the lowermost inner edge contracted. 

We now distinguish three cases. 
\begin{itemize}
    \item If the root $\rho$ of $T$ only has two children, we consider $T=(T_1,T_2)$. Using Proposition \ref{prop_s_standarddecomp} we can conclude that $T_1$ and $T_2$ also have maximal $\widehat{s}$ values, and by our choice of $T$ as a minimal example that does not have one of the described shapes, $T_1$ and $T_2$ both have one of the described shapes. Let $n_1$ and $n_2$ with $n_1\geq n_2$ denote their leaf numbers, respectively. It remains to show that $n_2=1$, because this implies that $T$ is also of one of the two described shapes. 
    
    With the previous considerations, we now have that $\widehat{s}(T_1)=\log((n_1-1)!)$ and $\widehat{s}(T_2)=\log((n_2-1)!)$ and $$\log((n-1)!)=\widehat{s}(T)=\log((n_1-1)!\cdot (n_2-1)! \cdot (n-1)).$$ 
    However, this holds if and only if
    \begin{align*}
        (n-1)!&=(n_1-1)!\cdot (n_2-1)!\cdot (n-1) \\
        (n-2)!&=(n_1-1)!\cdot (n_2-1)! \\
        (n-2)\cdot (n-3) \cdot  \ldots \cdot n_2 &= (n_1-1)!\\
         \underbrace{(n-2)\cdot (n-3) \cdot  \ldots \cdot (n-n_1)}_{\mbox{\tiny $(n_1-1)$ many terms, sorted in decreasing order}} &= \underbrace{(n_1-1)\cdot (n_1-2) \cdot\ldots \cdot 1}_{\mbox{\tiny $(n_1-1)$ many terms, sorted in decreasing order}}
    \end{align*}
    This implies that $n_1=n-1$ and thus $n_2=1$. Thus, as explained above, $T$ either is a binary caterpillar or a tree resulting from a binary caterpillar by contracting the lowermost inner edge. This is a contradiction to the choice of $T$.
    \item If the root $\rho$ of $T$ has at least four children $w_1$, $w_2$, $w_3$ and $w_4$ (and possibly more), we can modify $T$ to get a tree $T'$ as follows: We delete  the edges connecting $\rho$ to its children and insert two new vertices $\rho_1$ and $\rho_2$ as children of $\rho$. We connect $w_1$ and $w_2$ to $\rho_1$ and all other former children of $\rho$ to $\rho_2$. Now, the resulting tree $T'$ has $\widehat{s}(T')\geq\widehat{s}(T)$,  where equality applies if and only if $n=4$, i.e. if $w_1$, $w_2$, $w_3$ and $w_4$ are all leaves. This is because all subtrees descending from the children of $\rho$ have remained unchanged, $\rho$ itself still gives rise to a tree of size $n$, but now we have two additional vertices $\rho_1$ and $\rho_2$ which both have at least two descending leaves. However, the case $n=4$ is not possible as then we would have $T=T_4^{star}$ and $\widehat{s}(T)=\widehat{s}(T_4^{star})=\log(n-1)=\log(3)$, which is a contradiction to the assumption that $\widehat{s}(T)=\log((n-1)!)=\log(3!)$. So $T$ and $T'$ must have more than four leaves. But in this case, we get $\widehat{s}(T)<\widehat{s}(T')$, a contradiction to the maximality of $\widehat{s}(T)$. So in both cases, the assumption that such a tree $T$ exists and the root has degree at least four leads to a contradiction.
    \item Last, we consider the case where $\rho$ has precisely three children $w_1$, $w_2$ and $w_3$. If $n=3$ this implies that $T$ is the rooted star tree, which can be derived from a binary caterpillar by contracting the lowermost inner edge, so $T$ would have one of the two described shapes, a contradiction. So we must have $n>3$. Note that this implies that if we denote the number of leaves descending from the children of $\rho$ by  $n_{w_1}$, $n_{w_2}$ and $n_{w_3}$ such that, without loss of generality,  $n_{w_1} \geq n_{w_2} \geq n_{w_3}$, then $n_{w_1} +n_{w_2} >2$. 
    
    We now construct a tree $T'$ as follows: We delete the edges $(\rho,w_1)$ and $(\rho,w_2)$. We introduce a new vertex $v$ and new edges $(v,w_1)$, $(v,w_2)$ and $(\rho,v)$. Note that $n_v=n_{w_1} +n_{w_2} >2$. 
    However, now $\widehat{s}(T')>\widehat{s}(T)$, because by construction we have $$\widehat{s}(T')=\widehat{s}(T)+\underbrace{\log(n_v-1)}_{>0}>\widehat{s}(T),$$ as all subtree sizes of the subtrees of $T$ remain unchanged in $T'$, but additionally we have the new node $v$ contributing a positive value to $\widehat{s}(T')$ as $n_v>2$ and thus $n_v-1>1$. However, $\widehat{s}(T')>\widehat{s}(T)$ contradicts the maximality of $\widehat{s}(T)$.
    \end{itemize}
    
    Therefore, all three cases lead to a contradiction, which shows that such a tree $T$ cannot exist. Thus, all trees with maximal $\widehat{s}$ value must either be binary caterpillars, or they must be constructable from one by contracting its lowermost inner edge.  This completes the proof.
\end{proof}

\begin{remark}
In Theorem \ref{thm_sstat_max_a} it has been shown that the only trees maximizing $\widehat{s}$ on $\mathcal{T}_{n\geq 4}^\ast$ are the caterpillar tree $\Tcat$ and the tree that can be constructed from $\Tcat$ by contracting the lowermost inner edge. These are precisely the trees maximizing the Total $I$ and Total $I'$ indices (see Theorem \ref{prop_I_max_a}).
\end{remark}

Now we turn our attention to the minimum value of $\widehat{s}$. We start with the simpler arbitrary (i.e. not necessarily binary) case.

\begin{theorem}\label{thm_s_star} 
Let $T\in\Tnstar$ be a rooted tree minimizing $\widehat{s}$. Then, we have: $\widehat{s}(T)=\log(n-1)$. Moreover, $T$ is either the star tree $\Tstar$, or all its inner vertices other than the root must be parents of binary cherries, i.e. they can only have two descending leaves.
\end{theorem}
\begin{proof}
First consider the star tree $\Tstar$ and note that there is only one inner vertex, namely the root $\rho$, which is the ancestor of all $n$ leaves, so we have $n_{\rho}=n$. Using the definition of $\widehat{s}$, we get $\widehat{s}(\Tstar)=\sum\limits_{v \in \mathring{V}(\Tstar)}\log(n_v-1)=\log(n-1)$. 

Next, note that \emph{all} rooted trees with $n$ leaves have $n_{\rho}=n$, i.e. all trees have $\log(n-1)$ in their $\widehat{s}$ sum. So the value $\log(n-1)$ must indeed be minimal.

Next, note that a tree $T$ on $n$ leaves whose only inner vertices are the root $\rho$ and $m\geq 0$ parents of binary cherries achieves this minimum $\widehat{s}$ value, as we have $\widehat{s}(T)=\sum\limits_{v \in \mathring{V}(T)}\log(n_v-1)=\log(n-1)+m \cdot \underbrace{\log(2-1)}_{=0}=\log(n-1).$ So clearly, $T$ minimizes $\widehat{s}$.

Now let $T$ be a rooted tree with $n$ leaves that minimizes $\widehat{s}$ and which has an inner vertex $u$ with $n_u>2$. Note that we still have $n_{\rho}=n$, but we have at least one more summand in the sum of $\widehat{s}$, and as $n_u>2$, we have $\widehat{s}(T)=\sum\limits_{v \in \mathring{V}(T)}\log(n_v-1)\geq \log(n-1) + \log(n_u-1) >\log( n-1)+\log(2-1)=\log( n-1)$. This shows that $T$ cannot be minimal, which completes the proof.
\end{proof}

A direct consequence of Theorem \ref{thm_s_star} is the following corollary.

\begin{corollary}\label{cor_s_star}
For $n \in \{1,2\}$ there is precisely one rooted tree $T \in \Tnstar$ minimizing $\widehat{s}(T)$, whereas for each $n \in \mathbb{N}_{\geq 3}$, there are exactly $1 + \lfloor \frac{n}{2} \rfloor$ minimal trees.
\end{corollary}
\begin{proof}
For $n \in \{1,2\}$, the set $\Tnstar$ contains precisely one tree, and this tree trivially minimizes $\widehat{s}(T)$. Now, consider $n \geq 3$. By Theorem \ref{thm_s_star}, the star tree $\Tstar$ as well as each tree $T \in \Tnstar$ with the property that all inner vertices other than the root are parents of binary cherries minimize $\widehat{s}$. In particular, any other minimal tree can be obtained from the star tree by deleting one or more pairs of two leaves and their incident edges and re-attaching them as a pendant cherry. As a tree on $n$ leaves can have at most $\lfloor \frac{n}{2} \rfloor$ cherries, there are $\lfloor \frac{n}{2} \rfloor$ additional minimal trees (containing 1 up to $\lfloor \frac{n}{2} \rfloor$ cherries) next to the star tree (which contains 0 cherries). This completes the proof.
\end{proof}

Next, we will consider the binary case.

\begin{theorem} \label{thm_s-minPower2} 
Let $n=2^{h}$ for some $h\in \mathbb{N}_{\geq 0}$, and let $T$ be a rooted binary tree with $n$ leaves which minimizes $\widehat{s}$. Then, $T$ equals $\Tfb$. In other words, $\Tfb$ is the unique rooted binary tree minimizing $\widehat{s}$.  Moreover, we have $\widehat{s}\left( \Tfb \right)=\sum\limits_{i=0}^{h-1}2^i\cdot\log\left(2^{h-i}-1\right)$.
\end{theorem}

In order to prove Theorem \ref{thm_s-minPower2}, we need the following simple number theoretical result.

\begin{lemma} \label{lem_fracn}
Let $a,b \in \mathbb{N}_{>0}$ such that $a>b$. Then, we have $ \frac{a+b}{a} < \frac{a}{a-b}$.
\end{lemma}
\begin{proof} \begin{align*} \frac{a+b}{a} < \frac{a}{a-b} \Leftrightarrow (a+b)(a-b)<a^2 \Leftrightarrow a^2-b^2 < a^2 \Leftrightarrow b^2>0, \end{align*} which holds as $b>0$ by definition. This proves the assertion.
\end{proof}

We are now in a position to prove Theorem \ref{thm_s-minPower2}.

\begin{proof}[Proof of Theorem \ref{thm_s-minPower2}] 
Let $n=2^{h}$ and assume, seeking a contradiction, that $T$ is a tree on $n$ leaves minimizing $\widehat{s}$ but $T\neq T_{h}^{fb}$. Note that this implies that $h \geq 2$, because for each $h<2$, there is only one tree, which is by definition fully balanced. 

Now let $m<h$ denote the smallest value for which $T$ contains two fully balanced subtrees $T_m^\mathit{fb}$ with $2^m$ leaves but which do not belong to a fully balanced subtree $T_{m+1}^\mathit{fb}$ of size $2^{m+1}$. We now argue why such an $m$ must exist. First note that if $T$ has one leaf that does \emph{not} belong to a cherry, it must have another such leaf, too, because $n=2^h$ is even. So if $T$ has two leaves that do not belong to any cherry, we set $m=0$ and are done. Otherwise, we replace all cherries by leaves (thereby dividing the number of leaves by 2), increase $m$ by 1 and repeat this procedure until we have found our two required subtrees. In each step, we note that as $n=2^h$, dividing the number of leaves by 2 keeps the leaf number even, which is why the above argument holds in each round. Note that we can guarantee that $m<h-1$, because if $T$ had two subtrees $T_{h-1}^\mathit{fb}$, as $T$ has $n=2^h$ leaves, we would have $T=T_h^\mathit{fb}$, a contradiction to our assumption. 

So with $m$ chosen to be minimal with the property that $T$ contains two copies of $T_m^\mathit{fb}$ which do not belong to a subtree $T_{m+1}^\mathit{fb}$, we now call the root of the first one of these pending subtrees $x_1$ and the root of the second one $x_2$. 

Now there are two cases: Either the lowest common ancestor, say $w_0$, of $x_1$ and $x_2$ in $T$ is at the same time the parent of one of the vertices $x_1$ or $x_2$ or this is not the case. We consider these two cases separately.

\begin{enumerate}
    \item Let us first consider the case where $w_0$ is \emph{not} the parent of either $x_1$ or $x_2$. This implies that $T$ looks like the tree depicted in Figure \ref{figT} on page \pageref{figT}. In particular, in this case, it is crucial to note that both $x_1$ and $x_2$ must each have a sister subtree of size at least $2^{m+1}$ in $T$ ($T_A$ and $T_B$, respectively), because if one of them was smaller, $m$ would not be the minimal number found by the above algorithm that replaces fully balanced subtrees by leaves until it has two \enquote{unpaired} ones. 

    \begin{figure}\centering
    \includegraphics[scale=0.5]{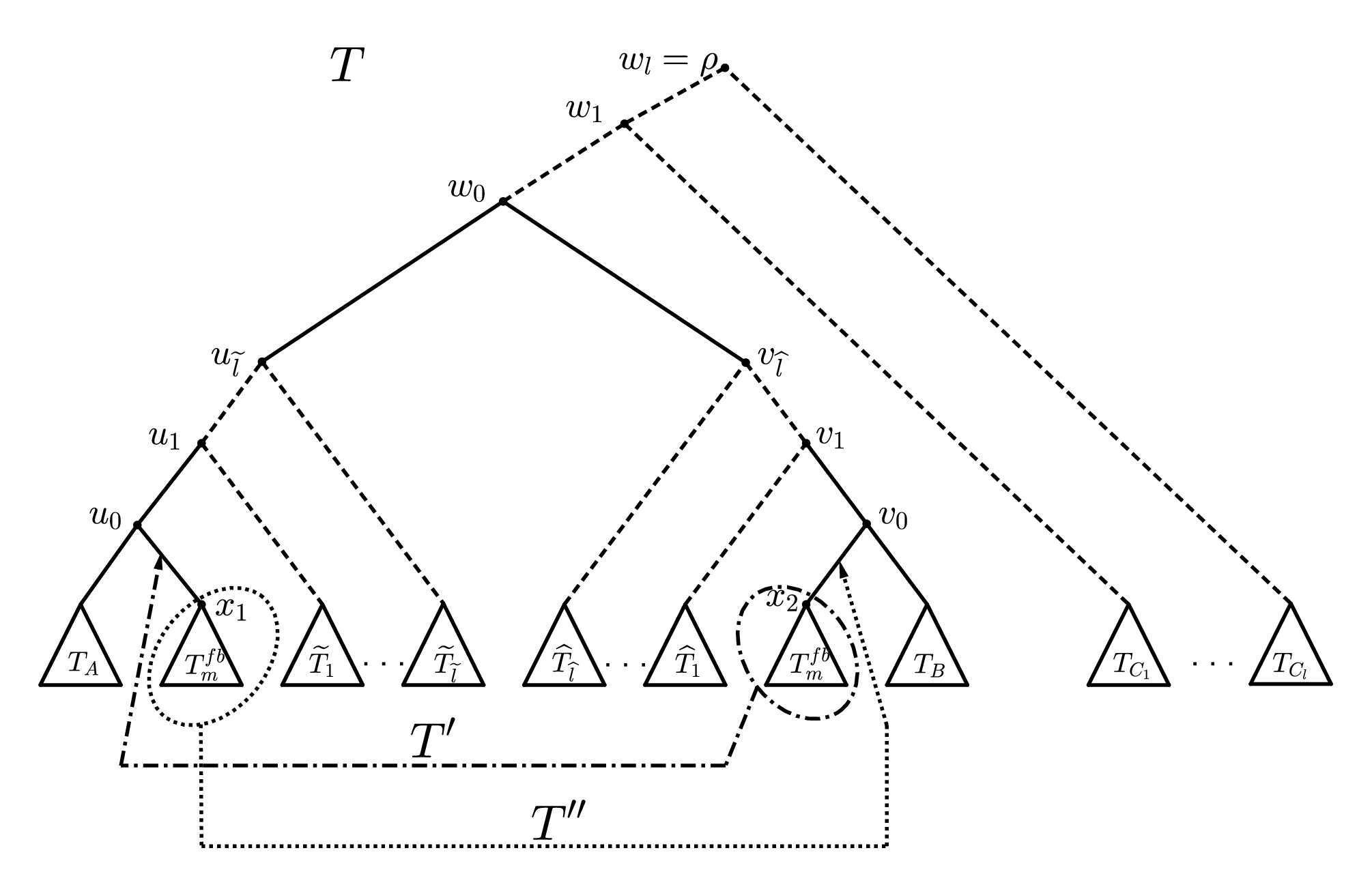}
    \caption{Trees $T$, $T'$, and $T''$ as needed in the first part of the proof of Theorem \ref{thm_s-minPower2}. For $T'$, $x_2$ and its corresponding subtree are moved to form a sister clade of $x_1$, and for $T''$, $x_1$ is moved to form a sister clade of $x_2$.}\label{figT}
    \end{figure}

    In the following, we denote the parent of $x_1$ by $u_0$, the parent of $x_2$ by $v_0$. There may or may not be vertices on the path from $u_0$ to $w_0$ other than $u_0$ and $w_0$ -- we denote their number by $\widetilde{l}$ (but $\widetilde{l}=0$ is possible) and label them from $u_1$ to $u_{\widetilde{l}}$. Each such vertex $u_i$ for $i\in \{1,\ldots,\widetilde{l}\}$, if it exists, gives rise to a subtree $\widetilde{T}_i$. Analogously, there may or may not be vertices on the path from $v_0$ to $w_0$ other than $v_0$ and $w_0$ -- we denote their number by $\widehat{l}$ (but $\widehat{l}=0$ is possible) and label them from $v_1$ to $v_{\widehat{l}}$. Each such vertex $v_i$ for $i\in \{1,\ldots,\widehat{l}\}$, if it exists, gives rise to a subtree $\widehat{T}_i$. Note that $w_0$ may or may not be the root of $T$. In case it is not, the nodes on the path from $w_0$ to the root may give rise to yet $l$ more subtrees $T_{C_1}$ to $T_{C_l}$. 

    Next, we construct trees $T'$ and  $T''$ from $T$ as follows (cf. Figure \ref{figT}):  For $T'$, we cut the edge leading to $x_2$, suppress vertex $v_0$ and re-attach the subtree of $x_2$ as a sister clade to $x_1$. This will create a new node $a$, the parent of the new subtree $T_{m+1}^{fb}$ formed by the two copies of $T_m^{fb}$ belonging to $x_1$ and $x_2$, respectively. Similarly, we construct $T''$ from $T$ as follows: We cut the edge leading to $x_1$, suppress vertex $u_0$ and re-attach the subtree of $x_1$ as a sister clade to $x_2$. This will create a new node $b$, the parent of the new subtree $T_{m+1}^\mathit{fb}$ formed by the two copies of $T_m^\mathit{fb}$ belonging to $x_1$ and $x_2$, respectively. 

    We now show that at least one of the trees $T'$ and $T''$ must have a lower $\widehat{s}$-value than $T$, contradicting the assumption. In order to do so, we analyze the differences between $T'$ and $T$ as well as $T''$ and $T$, respectively. Importantly, note that all subtrees descending from $x_1$ and $x_2$ in $T$ remain unchanged in $T'$ and $T''$, respectively. However, we observe the following differences:
    \begin{itemize} 
        \item In $T'$, the numbers $n_{v_i}$ of leaves descending from $v_i \in \{v_1,\ldots,v_{\widehat{l}}\}$ are reduced by $2^m$. 
        \item In $T'$, there is no vertex $v_0$, which is why the subtree of size $n_{v_0}$ is not there. 
        \item In $T'$, vertices $u_0,\ldots,u_{\widetilde{l}}$ each have $2^m$ descending leaves more, i.e. $n_{u_i}$ has increased by $2^m$ for each $i \in \{0,\ldots,\widetilde{l}\}$. 
        \item In $T'$, there is a new vertex $a$ which has $2^{m+1}$ descending leaves, i.e. $n_a=2^{m+1}$.
    \end{itemize}

    Analogously, note that $T$ and $T''$ differ in the following ways: 
    \begin{itemize} 
        \item In $T''$, the numbers $n_{u_i}$ of leaves descending from $u_i \in \{u_1,\ldots,u_{\widetilde{l}}\}$ are reduced by $2^m$. 
        \item In $T''$, there is no vertex $u_0$, which is why the subtree of size $n_{u_0}$ is not there. 
        \item In $T''$, vertices $v_0,\ldots,v_{\widehat{l}}$ each have $2^m$ more descending leaves, i.e. $n_{v_i}$ has increased by $2^m$ for each $i \in \{0,\ldots,\widehat{l}\}$. 
        \item In $T''$, there is a new vertex $b$ which has $2^{m+1}$ descending leaves, i.e. $n_b=2^{m+1}$.
    \end{itemize}

    \item Now consider the case in which the lowest common ancestor $w_0$ of $x_1$ and $x_2$ coincides with one of the parents of $x_1$ or $x_2$. Without loss of generality, $w_0$ coincides with $v_0$, the parent of $x_2$. This implies that $T$ looks like the tree depicted in Figure \ref{figTalt} on page \pageref{figTalt}. In particular, in this case, it is crucial to note that, by the same argument as in the first case, $x_1$  must have a sister subtree of size at least $2^{m+1}$ in $T$ (namely $T_A$), because otherwise we would have a contradiction to our choice of $m$.

    \begin{figure}
    \centering
    \includegraphics[scale=0.5]{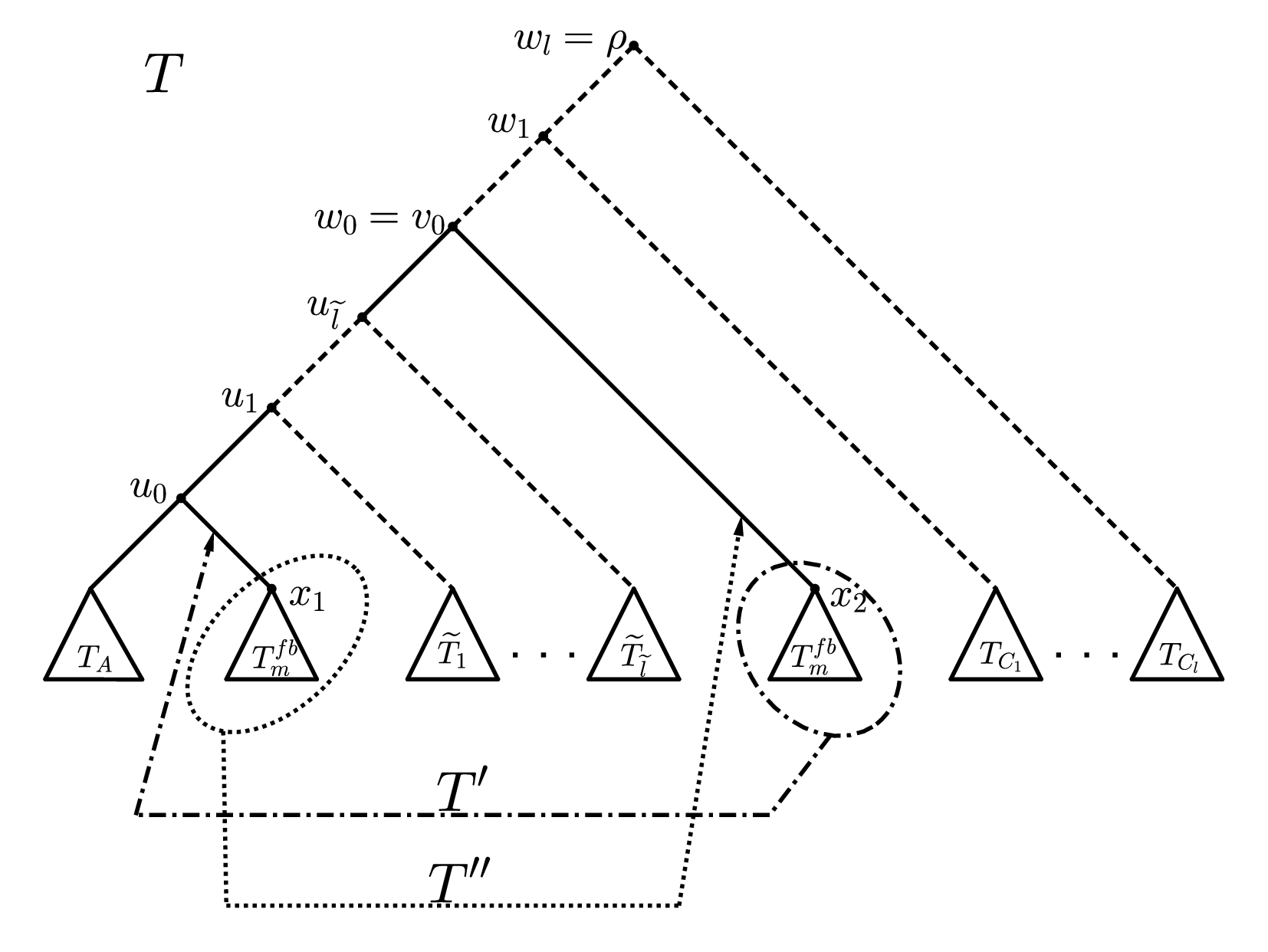}
    \caption{Trees $T$, $T'$, and $T''$ as needed in the second part of the proof of Theorem \ref{thm_s-minPower2}. For $T'$, $x_2$ and its corresponding subtree are moved to form a sister clade of $x_1$, and for $T''$, $x_1$ is moved to form a sister clade of $x_2$. }\label{figTalt}
    \end{figure}

    We again denote the parent of $x_1$ by $u_0$. There may or may not be vertices on the path from $u_0$ to $w_0$ other than $u_0$ and $w_0$ -- we denote their number by $\widetilde{l}$ (but $\widetilde{l}=0$ is possible) and label them from $u_1$ to $u_{\widetilde{l}}$. Each such vertex $u_i$ for $i\in \{1,\ldots,\widetilde{l}\}$, if it exists, gives rise to a subtree $\widetilde{T}_i$. Moreover, note that $w_0=v_0$ may coincide with the root of $T$ but does not necessarily have to. In case it does not, the nodes on the path from $w_0$ to the root may give rise to yet $l$ more subtrees $T_{C_1}$ to $T_{C_l}$. 

    We now construct trees $T'$ and $T''$ from $T$ (cf. Figure \ref{figTalt}): For $T'$, we cut the edge leading to $x_2$, suppress vertex $v_0$ and re-attach the subtree of $x_2$ as a sister clade to $x_1$. This will create a new node $a$, the parent of the new subtree $T_{m+1}^\mathit{fb}$ formed by the two copies of $T_m^\mathit{fb}$ belonging to $x_1$ and $x_2$, respectively. Similarly, we construct $T''$ from $T$ as follows: We cut the edge leading to $x_1$, suppress vertex $u_0$ and re-attach the subtree of $x_1$ as a sister clade to $x_2$. This will create a new node $b$, the parent of the new subtree $T_{m+1}^\mathit{fb}$ formed by the two copies of $T_m^\mathit{fb}$ belonging to $x_1$ and $x_2$, respectively. 

    We now analyze the differences between $T'$ and $T$ as well as $T''$ and $T$, respectively. Importantly, note that all subtrees descending from $x_1$ and $x_2$ in $T$ remain unchanged in $T'$ and $T''$, respectively. However, we observe the following differences:
    \begin{itemize} 
        \item In $T'$, there is no vertex $w_0=v_0$, which is why the subtree of size $n_{v_0}$ is not there. 
        \item In $T'$, vertices $u_0,\ldots,u_{\widetilde{l}}$ have $2^m$ more descending leaves each, i.e. $n_{u_i}$ has increased by $2^m$ for each $i \in \{0,\ldots,\widetilde{l}\}$. 
        \item In $T'$, there is a new vertex $a$ which has $2^{m+1}$ descending leaves, i.e. $n_a=2^{m+1}$.
    \end{itemize}

    Analogously, note that $T$ and $T''$ differ in the following ways: 
    \begin{itemize} 
        \item In $T''$, there is no vertex $u_0$, which is why the subtree of size $n_{u_0}$ is not there. 
        \item In $T''$, the numbers $n_{u_i}$ of leaves descending from $u_i \in \{u_1,\ldots,u_{\widetilde{l}}\}$ are reduced by $2^m$. 
        \item In $T''$, there is a new vertex $b$ which has $2^{m+1}$ descending leaves, i.e. $n_b=2^{m+1}$.
    \end{itemize}
\end{enumerate}

We now apply the definition of $\widehat{s}$ to both cases, i.e. to the case where $w_0$ is the parent of one of the vertices $x_1$ or $x_2$, as well as to the case where this does not hold. The following considerations hold for both cases, considering that in these equations, as the nodes $v_i$ for $i\in \{1,\ldots,\widehat{l}\}$ do not exist in Figure \ref{figTalt}, the respective products are in this case empty and thus equal to 1. Moreover, recall that we have $v_0=w_0$ in this case. 

For $\widehat{s}$, the above observations imply:
\begin{equation*}  
\widehat{s}(T')=\widehat{s}(T)+ \log \underbrace{\left(\frac{\left(\prod\limits_{i=0}^{\widetilde{l}} \left(n_{u_i}+2^m-1\right)\right)    \cdot  \left(\prod\limits_{i=1}^{\widehat{l}} \left(n_{v_i}-2^m-1\right)\right)  \cdot\left(\overbrace{2^{m+1}}^{=n_a}-1 \right) }{\left(\prod\limits_{i=0}^{\widetilde{l}} (n_{u_i}-1)\right)    \cdot  \left(\prod\limits_{i=1}^{\widehat{l}} (n_{v_i}-1)\right)   \cdot (n_{v_0}-1)}\right)}_{\eqqcolon t_1},
\end{equation*}
and
\begin{equation*} 
\widehat{s}(T'')=\widehat{s}(T)+ \log \underbrace{\left(\frac{\left(\prod\limits_{i=1}^{\widetilde{l}} \left(n_{u_i}-2^m-1\right)\right)    \cdot  \left(\prod\limits_{i=0}^{\widehat{l}}\left(n_{v_i}+2^m-1\right)\right)    \cdot \left(\overbrace{2^{m+1}}^{=n_b}-1 \right)  }{\left(\prod\limits_{i=1}^{\widetilde{l}} (n_{u_i}-1)\right)    \cdot  \left(\prod\limits_{i=0}^{\widehat{l}} (n_{v_i}-1)\right)   \cdot (n_{u_0}-1)  }\right)}_{\eqqcolon t_2}.
\end{equation*}

We now argue that at least one of the terms $t_1$ and $t_2 $ has to be strictly smaller than 1. Note that if $t_1<1$, it immediately follows that $\widehat{s}(T')<\widehat{s}(T)$, which would contradict the minimality of $\widehat{s}(T)$. Analogously, if $t_2<1$, we would have $\widehat{s}(T'')<\widehat{s}(T)$.  So if we manage to prove that at least one of the values $t_1$ and $t_2$ is strictly smaller than 1, this will complete the proof.

In the following, assume that we have $t_1 \geq 1$ and $t_2\geq 1$. We will show that this leads to a contradiction.

We first re-arrange $t_1$ a bit before we continue.

\begin{equation*}
t_1=\frac{\prod\limits_{i=0}^{\widetilde{l}} \left(n_{u_i} +2^m-1 \right) }{\prod\limits_{i=0}^{\widetilde{l}}(n_{u_i}-1)   } \cdot \frac{\prod\limits_{i=1}^{\widehat{l}} \left(n_{v_i}-2^m-1 \right)}{ \prod\limits_{i=1}^{\widehat{l}} (n_{v_i}-1)}\cdot \frac{2^{m+1}-1}{n_{v_0}-1}
\end{equation*}

Note that, as $t_1\geq 1$ by assumption and as the terms $\prod\limits_{i=1}^{\widehat{l}} \left(n_{v_i}-2^m-1 \right)$ and $\prod\limits_{i=1}^{\widehat{l}} (n_{v_i}-1)$ and $2^{m+1}-1$ and $n_{v_0}-1$ are all positive, this shows that we have:

\begin{equation*}
\frac{\prod\limits_{i=0}^{\widetilde{l}} \left(n_{u_i} +2^m-1  \right)   }{\prod\limits_{i=0}^{\widetilde{l}} (n_{u_i}-1 ) } \geq \frac{\prod\limits_{i=1}^{\widehat{l}} (n_{v_i}-1)}{\prod\limits_{i=1}^{\widehat{l}} \left( n_{v_i}-2^m-1\right)} \cdot \frac{n_{v_0}-1}{2^{m+1}-1}.
\end{equation*}

However, this is equivalent to:
\begin{equation*}
\prod\limits_{i=0}^{\widetilde{l}}\frac{ n_{u_i}  +2^m-1    }{ n_{u_i}-1  } \geq \prod\limits_{i=1}^{\widehat{l}}  \frac{n_{v_i}-1}{ n_{v_i}-2^m-1} \cdot \frac{n_{v_0}-1}{2^{m+1}-1}.
\end{equation*}

Using the fact that $n_{v_0}\geq 2^{m+1}+2^m$ (as $T_B$ has at least $2^{m+1}$ leaves in $T$ as explained above), we know that $n_{v_0}-2^m\geq 2^{m+1}$. Thus, we can use the fact that $\frac{n_{v_0}-1}{2^{m+1}-1}\geq \frac{n_{v_0}-1}{n_{v_0}-2^m-1}$ to conclude:  

\begin{equation} \label{eq_t1new}
\prod\limits_{i=0}^{\widetilde{l}}\frac{  n_{u_i}  +2^m-1   }{ n_{u_i}-1  } \geq \prod\limits_{i=0}^{\widehat{l}}  \frac{n_{v_i}-1}{ n_{v_i}-2^m-1}.
\end{equation}

Analyzing $t_2$ in a similar way, we derive:
\begin{equation} \label{eq_t2new}
\prod\limits_{i=0}^{\widehat{l}}\frac{ n_{v_i}+2^m-1      }{ n_{v_i}-1  } \geq \prod\limits_{i=0}^{\widetilde{l}}  \frac{n_{u_i}-1}{ n_{u_i}-2^m-1} .
\end{equation}

Using Lemma \ref{lem_fracn} with $a \coloneqq n_{v_i}-1$ and $b \coloneqq 2^m$ (where again $a>b$ is guaranteed as $n_{v_i} \geq 2^{m+1}+2^m$), we can see that $\frac{n_{v_i}-1}{n_{v_i}-2^m-1}>\frac{n_{v_i}+2^m-1}{n_{v_i}-1}$ for all $i\in\{1,\ldots,\widehat{l}\}$. However, this shows that the right-hand side of Equation \eqref{eq_t1new} is strictly larger than the left-hand side of Equation \eqref{eq_t2new}, which leads to:
\begin{equation*} 
\prod\limits_{i=0}^{\widetilde{l}}\frac{ n_{u_i}+2^m-1      }{ n_{u_i}-1  } \overset{\eqref{eq_t1new}}{\geq} \prod\limits_{i=0}^{\widehat{l}}  \frac{n_{v_i}-1}{ n_{v_i}-2^m-1}\overset{\mbox{\tiny Lem. \ref{lem_fracn}}}{>}\prod\limits_{i=0}^{\widehat{l}}\frac{ n_{v_i}     +2^m-1 }{ n_{v_i}-1  } \overset{\eqref{eq_t2new}}{\geq} \prod\limits_{i=0}^{\widetilde{l}}  \frac{n_{u_i}-1}{ n_{u_i}-2^m-1}.
\end{equation*}

In particular, we have:
\begin{equation} \label{eq_contr1new}
\prod\limits_{i=0}^{\widetilde{l}}\frac{ n_{u_i} +2^m-1     }{ n_{u_i}-1  }  > \prod\limits_{i=0}^{\widetilde{l}}  \frac{n_{u_i}-1}{ n_{u_i}-2^m-1}.
\end{equation}

However, again by Lemma \ref{lem_fracn} (using $a \coloneqq n_{u_i}-1$ and $b \coloneqq 2^m$, where again $a>b$ is guaranteed as each $u_i$ has at least the $2^m$ leaves induced by $x_1$ as well as the at least $2^{m+1}$ leaves of $T_A$), we can see that $\frac{n_{u_i}-1}{n_{u_i}-2^m-1}>\frac{n_{u_i}+2^m-1}{n_{u_i}-1}$ for all $i\in\{1,\ldots,\widetilde{l}\}$, which shows that
\begin{equation} \label{eq_contr2new}
\prod\limits_{i=0}^{\widetilde{l}}\frac{ n_{u_i}  +2^m-1    }{ n_{u_i}-1  }  < \prod\limits_{i=0}^{\widetilde{l}}  \frac{n_{u_i}-1}{ n_{u_i}-2^m-1}.
\end{equation}

The contradiction between Inequalities \eqref{eq_contr1new} and \eqref{eq_contr2new} completes the main part of the proof.

It remains to show the second assertion, namely that $\widehat{s}\left( \Tfb\right)=\sum\limits_{i=0}^{h-1}2^i\cdot\log\left(2^{h-i}-1\right)$. Recall that for all $i=0,\ldots, h$, $\Tfb$ has $2^i$ subtrees of size $2^{h-i}$ each (however, note that we do not have to consider the $2^h$ subtrees of size $2^{h-h}=2^0=1$, as the leaves do not contribute to $\widehat{s}$). This immediately shows that \[ \widehat{s}\left(\Tfb\right)= \sum\limits_{v \in \mathring{V}\left(\Tfb\right)} \log(n_v-1) \  =\sum\limits_{i=0}^{h-1}2^i \cdot \log\left(2^{h-i}-1\right). \] This completes the proof.
\end{proof}

\subsubsection{Symmetry nodes index}

Suitable only for binary trees, recall that the symmetry nodes index \citep{Kersting2021} $SNI(T)$ of a binary tree $T\in\BTnstar$ is the number of inner vertices that are \emph{not} symmetry vertices, i.e. $SNI(T)\coloneqq (n-1)-s(T)$.

In the following proposition we will show that the symmetry nodes index is not local, but recursive.

\begin{proposition} \label{locality_Sym}
The symmetry nodes index is not local.
\end{proposition}
\begin{proof}
Consider the two trees $T$ and $T'$ in Figure \ref{fig_locality} on page \pageref{fig_locality}, which only differ in their subtrees rooted at $v$. Note that in both $T$ and $T'$ the vertex $v$ has exactly 5 descendant leaves. Nevertheless, we have $SNI(T)-SNI(T')=6-4=2\neq 1=3-2=SNI(T_v)-SNI(T_v')$. Thus, the symmetry nodes index is not local. This property applies, because changing the subtree $T_v$ might change if a vertex $u\in anc(v)$ is a symmetry vertex or not.
\end{proof}

\begin{proposition} \label{recursiveness_Sym}
The symmetry nodes index is a binary recursive tree shape statistic. We have $SNI(T)=0$ for $T\in\mathcal{BT}_1^\ast$, and for every tree $T\in\BTnstar$ with $n\geq 2$ and standard decomposition $T=(T_1,T_2)$ we have 
\[ SNI(T) = SNI(T_1) + SNI(T_2) + \left(1-\mathcal{I}(CP(T_1)=CP(T_2))\right),\]
where $CP(T_i)$ is the Colijn-Plazotta rank of $T_i$ \citep{Colijn2018}.
\end{proposition}
\begin{proof}
Let $T=(T_1,T_2)$ be a tree in $\BTnstar$ with root $\rho$. The symmetry nodes index can be expressed as $SNI(T) = SNI(T_1) + SNI(T_1) + \delta_T$, where $\delta_T=1$ if $\rho$ is a symmetry vertex and 0 otherwise \cite[Lemma~3.2]{Kersting2021}. Using that the Colijn-Plazotta rank \citep{Colijn2018} provides unique ranks for all trees in $\BTnstar$ we have 
\begin{equation*}
\begin{split}
    SNI(T) &= SNI(T_1) + SNI(T_2) + \delta_T\\
    &= SNI(T_1) + SNI(T_2) + \left(1-\mathcal{I}(CP(T_1)=CP(T_2))\right).
\end{split}
\end{equation*}
The recursive expression of the Colijn-Plazotta rank $CP(T)$ can be obtained from \citet{Colijn2018}. Thus, the symmetry nodes index can be expressed as a binary recursive tree shape statistic of length $x=2$ with the recursions (where $SNI_i$ and $CP_i$ are simplified notations of $SNI(T_i)$ and $CP(T_i)$, respectively) 
\begin{itemize}
    \item Symmetry nodes index: $\lambda_1=0$ and $r_1(T_1,T_2)=SNI_1 + SNI_2 + \left(1-\mathcal{I}(CP_1=CP_2)\right)$
    \item Colijn-Plazotta rank: $\lambda_2=0$ and $r_2(T_1,T_2)=\frac{1}{2}\cdot \max\{CP_1,CP_2\}\cdot (\max\{CP_1,CP_2\}-1)$\\
    \phantom{Colijn-Plazotta rank: $\lambda_2=0$ and $r_2(T_1,T_2)=$} $+\min\{CP_1,CP_2\}+1$
\end{itemize}
It can easily be seen that $\lambda\in\mathbb{R}^2$ and $r_i: \mathbb{R}^2\times\mathbb{R}^2 \rightarrow \mathbb{R}$, and that all $r_i$ are independent of the order of subtrees. This completes the proof. 
\end{proof}

\begin{remark}
Note that instead of the Colijn-Plazzotta rank any other bijective map between the set of rooted binary trees and a set of real numbers that is itself a binary recursive tree shape statistic (for example the Furnas rank, Section \ref{factsheet_Furnas}) can be used. Depending on the chosen map, the number of recursions for $SNI$ (and thus the length of the binary recursive tree shape statistic) might vary.
\end{remark}

Next, we will have a look at the properties of the symmetry nodes index under the uniform model. We will first develop formulas for the expected value and variance of the number of symmetry nodes $s(T)$ and afterwards derive the corresponding formulas for the symmetry nodes index $SNI$.

\begin{proposition} \label{prop_sym_nodes_Unif} 
Let $T_n$ be a phylogenetic tree with $n$ leaves sampled under the uniform model, let $P_n(\xi)$ be the probability that $T_n$ has $\xi$ symmetry nodes, i.e. $s(T_n)=\xi$, and let $A_n(\xi)$ be the number of distinct binary trees $T\in\BTnstar$ with $s(T)=\xi$.\\
Then, we have $P_n(\xi) =\frac{A_n(\xi)\cdot n!}{(2n-3)!!\cdot 2^\xi}$ with $A_n(\xi) \coloneqq 0$ if $n \notin \mathbb{N}_{\geq 1}$, $\xi \notin \mathbb{N}_{\geq 0}$ or $\xi>n-wt(n)$,
$A_1(0)=1$, $A_2(1)=1$, and the following recursion for $n>2$ and $\xi\geq 1$, where we assume the sums to be zero if their index set is empty:
\begin{equation*}
    A_n(\xi) = A_{\frac{n}{2}}\left(\frac{\xi-1}{2}\right) + \sum\limits_{j=1}^{\left \lfloor \frac{\xi-1}{2} \right \rfloor }  A_{\frac{n}{2}}\left(j\right) \cdot A_{\frac{n}{2}}\left(\xi-j\right) + \binom{ A_{\frac{n}{2}}\left(\frac{\xi}{2}\right)}{2} + \sum\limits_{i=1}^{\lfloor\frac{n-1}{2}\rfloor} \quad \sum\limits_{j=0}^{\min\{i-wt(i),\xi\}} A_{i}(j)\cdot A_{n-i}(\xi-j).
\end{equation*}
\end{proposition}

\begin{remark} \label{remark_sym_nodes}
Note that $A_n(0)=0$ and $A_n(1)=1$ for all $n \in \mathbb{N}_{\geq 2}$ because the caterpillar $\Tcat$ is the unique tree in $\BTnstar$ with only one symmetry node (a single cherry), the minimal number of symmetry nodes for $n\geq 2$.
\end{remark}

\begin{proof}
First, note that the number of possible phylogenies $\mathcal{T}=(T,\phi) \in \BTn$ for any tree $T \in \BTnstar$ is $\frac{n!}{2^\xi}$ and thus only depends on the number of symmetry nodes $s(T)=\xi$ (cf.~\cite[Corollary~2.4.3]{Semple2003}). Since under the uniform model every phylogeny $\mathcal{T} \in \BTn$ has the same probability, namely $\frac{1}{(2n-3)!!}$ (because $|\BTn|=(2n-3)!!$ \citep[Corollary 2.2.4]{Semple2003}), and the number of phylogenies with $n$ leaves and $\xi$ symmetry nodes is $A_n(\xi) \cdot \frac{n!}{2^\xi}$, we have $P_n(\xi) =\frac{A_n(\xi)\cdot n!}{(2n-3)!!\cdot 2^\xi}$.

Now, we will prove the recursive formula for $A_n(\xi)$ by partitioning the set of binary trees with $n$ leaves and $\xi$ symmetry nodes, here denoted as $\BTnstar(\xi)$, and assessing the size of each subset individually. The initial values $A_1(0)=A_2(1)=1$ arise from the fact that the only trees in $\BTnstar$ for $n=1$ and $n=2$ have zero and one symmetry nodes, respectively. Also note that $A_n(\xi)=0$ if $\xi>n-wt(n)$ because of \citep[Theorem 3.5]{Kersting2021}.

So, now let $n\geq 3$, which implies $\xi\geq 1$ because of Remark \ref{remark_sym_nodes}. First, consider the subset containing all trees $T=(T_1,T_2) \in \BTnstar(\xi)$ whose root is a symmetry node, implying that both maximal pending subtrees are the same tree $T_1=T_2 \in \mathcal{BT}_{\frac{n}{2}}^\ast \left(\frac{\xi-1}{2}\right)$. The size of this subset is thus simply $A_{\frac{n}{2}}\left(\frac{\xi-1}{2}\right)$, which is zero if $\frac{n}{2}\notin\mathbb{N}_{\geq 1}$ or $\frac{\xi-1}{2} \notin \mathbb{N}_{\geq 0}$ in accordance with the definition.

Second, we count the trees $T=(T_1,T_2) \in \BTnstar(\xi)$ whose root is not a symmetry node, implying that the two maximal pending subtrees $T_1$ and $T_2$ cannot be equal. Here we distinguish two subcases:
\begin{itemize}
    \item The two maximal pending subtrees have equal size, i.e. $n_1=n_2=\frac{n}{2}$.\\
    For $s(T_1)\neq s(T_2)$, we can count all possible trees by summing over the number of symmetry nodes in the subtree with less symmetry nodes resulting in $\sum\limits_{j=1}^{\left \lfloor \frac{\xi-1}{2} \right \rfloor }  A_{\frac{n}{2}}\left(j\right) \cdot A_{\frac{n}{2}}\left(\xi-j\right)$. Note that $j\geq 1$, because $n\geq 3$ and $n_1=n_2$ imply that $n_1=n_2\geq 2$. In particular, $T_1$ and $T_2$ both have at least one cherry and thus at least one symmetry node.\\
    For $s(T_1)= s(T_2)=\frac{\xi}{2}$, provided that $\frac{\xi}{2} \in \mathbb{N}_{\geq 1}$, we have as many trees in $\BTnstar(\xi)$ as there are possibilities to pick two distinct trees from $\mathcal{BT}_{\frac{n}{2}}^{\ast}(\frac{\xi}{2})$, i.e. $\binom{ A_{\frac{n}{2}}\left(\frac{\xi}{2}\right)}{2}$.
    \item The two maximal pending subtrees have different sizes, i.e. $n_1\neq n_2$. Without loss of generality let $n_2<n_1$ implying $n_2\in\{1,\ldots,\lfloor\frac{n-1}{2}\rfloor\}$. Now, $T_2$ can have between 0 and $\min\{\xi,n_2-wt(n_2)\}$ symmetry nodes (because there are only $\xi$ in $T$ and a tree with $n_2$ leaves can have at most $n_2-wt(n_2)$ \citep[Theorem~3.5]{Kersting2021}). So, if $n_2$ and $\xi_2\coloneqq s(T_2)$ are fixed, there are $A_{n_2}(\xi_2)$ possibilities to choose $T_2$. Since $T$ must have a total of $\xi$ symmetry nodes and $n_1=n-n_2$, there are $A_{n-n_2}(\xi-\xi_2)$ possibilities to choose $T_1$. Note that each different choice of $T_1$ and $T_2$ leads to a different (unique) tree $T$ as $n_2<n_1$. Thus, the number of rooted binary trees in this case is \[\sum\limits_{n_2=1}^{\lfloor\frac{n-1}{2}\rfloor} \quad \sum\limits_{\xi_2=0}^{\min\{n_2-wt(n_2),\xi\}} A_{n_2}(\xi_2)\cdot A_{n-n_2}(\xi-\xi_2) = \sum\limits_{i=1}^{\lfloor\frac{n-1}{2}\rfloor} \quad \sum\limits_{j=0}^{\min\{i-wt(i),\xi\}} A_{i}(j)\cdot A_{n-i}(\xi-j).\]
\end{itemize}
Note that all considered cases are mutually exclusive, i.e. each tree is counted only once, and that there are no other ways to construct $T$. Hence, the total number of distinct binary trees with $n>2$ leaves and $\xi\geq 1$ symmetry nodes is 
\begin{equation*}
\begin{split}
    A_n(\xi) &= \underbrace{A_{\frac{n}{2}}\left(\frac{\xi-1}{2}\right)}_{\text{root is symmetry node}} + \underbrace{\sum\limits_{j=1}^{\left \lfloor \frac{\xi-1}{2} \right \rfloor }  A_{\frac{n}{2}}\left(j\right) \cdot A_{\frac{n}{2}}\left(\xi-j\right)}_{\substack{\text{root is no symmetry node, }\\ n_1=n_2, \quad s(T_1)\neq s(T_2)}} + \underbrace{\binom{ A_{\frac{n}{2}}\left(\frac{\xi}{2}\right)}{2}}_{\substack{\text{root is no symmetry node, }\\ n_1=n_2, \quad s(T_1)=s(T_2)}}\\
    &\qquad + \underbrace{\sum\limits_{i=1}^{\lfloor\frac{n-1}{2}\rfloor} \quad \sum\limits_{j=0}^{\min\{i-wt(i),\xi\}} A_{i}(j)\cdot A_{n-i}(\xi-j).}_{\text{root is no symmetry node, } n_1\neq n_2}
\end{split}
\end{equation*}
This completes the proof.
\end{proof}

From these observations we can directly conclude the statements in Corollary \ref{Cor_sym_uniform} about the expected value and the variance of the number of symmetry nodes and the symmetry nodes index of a phylogenetic tree randomly chosen under the uniform model.

\begin{corollary} \label{Cor_sym_uniform}
Let $T_n$ be a phylogenetic tree with $n$ leaves sampled under the uniform model. Then the expected value and the variance of $s(T_n)$ are 
\[ E_U(s(T_n))=\sum\limits_{\xi=1}^{n-wt(n)}{P_n(\xi)\cdot \xi}= \frac{n!}{(2n-3)!!}\cdot \sum\limits_{\xi=1}^{n-wt(n)}{A_n(\xi)\cdot \frac{\xi}{2^\xi}} \] and 
\[ V_U(s(T_n))=\sum\limits_{\xi=1}^{n-wt(n)}{P_n(\xi)\cdot (\xi-E_U(s(T_n)))^2} = \frac{n!}{(2n-3)!!}\cdot\sum\limits_{\xi=1}^{n-wt(n)}{\frac{A_n(\xi)}{2^\xi}\cdot (\xi-E_U(s(T_n)))^2}, \] where $A_n(\xi)$ and $P_n(\xi)$ can be computed with the formulas provided in Proposition \ref{prop_sym_nodes_Unif}.
As a consequence, for the symmetry nodes index $SNI(T_n)=n-1-s(T_n)$, we have $E_U(SNI(T_n))=E_U(n-1-s(T_n))=n-1-E_U(s(T_n))$ and $V_U(SNI(T_n))=V_U(s(T_n))$.
\end{corollary}

These formulas were implemented in $\mathsf{R}$ and were used to calculate the expected values $E_U(s(T_n))$ and variances $V_U(s(T_n))$ depending on $n$. The recursion of $A_n(\xi)$ takes a lot of computation time and the handling of such great numbers produced errors. These errors appeared for $n>50$, but were comparatively small; but for $n>140$ they were significant and noticeable. Thus, it could be an interesting question for future research if there is an exact formula or one which takes less computation time. Another method we used to tackle this problem is to approximate the exact values:\\
Since there seems to be a nearly perfectly linear correlation with $n$ (see Figure \ref{fig:e_valVarSym}), we approximated both the expected value and the variance with linear functions using the least squares method once for $n=10,\ldots,140$ (the first nine values were omitted to ignore the initial variability) and once for $n=60,\ldots,140$ (as using higher starting values yielded a better approximation for higher $n$). This resulted in the linear functions $E_U(s(T_n))\sim 0.27086 n+0.18545$ and $V_U(s(T_n))\sim 0.10491 n +0.02853$ for the approximation based on $n=10,\ldots,140$ as well as $E_U(s(T_n))\sim 0.27100 n+0.17208$ and $V_U(s(T_n))\sim 0.10494 n +0.02595$ for the approximation based on $n=60,\ldots,140$. The exact values as well as the approximated linear functions are shown in Figure \ref{fig:e_valVarSym}. Furthermore, we analyzed the differences between the exact and approximated values. For both the expected value and the variance these differences follow a similar pattern for both approximations (see Figure \ref{fig:residSym}), which definitely suggests that there is no exact linear connection between $E_U(s(T_n))$ or $V_U(s(T_n))$ and the number of leaves $n$. Nonetheless, if there is only a linear increase in the difference as the pattern might indicate, we can expect, for instance, errors $<1$ for the expected value for $n=1000$. A superlinear increase of the difference between the real and the  approximated expected values is not possible as the maximal number of symmetry nodes $n-wt(n)$ increases linearly (for higher $n$, $wt(n)$ becomes proportionally small and therefore, $n-wt(n)$ increases linearly with $n$). Thus, a linear approximation is sensible. We conjecture a similar behavior for the variance. For now, the linear functions from both approximations can be seen as sufficiently exact for most if not all application purposes that handle trees with not many more than $140$ leaves, but it seems plausible that the linear functions are also applicable for trees with significantly more leaves. For $n$ higher than $\approx 40$ the second approximation should be preferred over the first (see Figure \ref{fig:residSym}).

From these linear functions for the expected value and variance of $s(T_n)$ we can derive the corresponding approximated formulas for $n\geq 10$ for the symmetry nodes index using Corollary \ref{Cor_sym_uniform}: 
\begin{align*}
    E_U(SNI(T_n)) &= n-1-E_U(s(T_n)) \sim n-1- \left(0.27086 n+0.18545 \right) = 0.72914n - 1.18545 \\
    V_U(SNI(T_n)) &\sim  \ 0.10491 n +0.02853 \quad \text{based on $n=10,\ldots,140$}
\end{align*}
as well as
\begin{align*}
    E_U(SNI(T_n)) &= n-1-E_U(s(T_n)) \sim n-1- \left(0.271 n+0.17208 \right) = 0.729n - 1.17208\\
    V_U(SNI(T_n)) &\sim  \ 0.10494 n +0.02595  \quad \text{based on $n=60,\ldots,140$.}
\end{align*}

For smaller $n$, we have exactly $E_U(SNI(T_{1}))=0$, $E_U(SNI(T_{2}))=0$, $E_U(SNI(T_{3}))=1$, $E_U(SNI(T_4)) = \frac{3}{15} \cdot 0 + \frac{12}{15} \cdot 2 = \frac{8}{5}=1.6$, $E_U(SNI(T_5)) = \frac{15}{105} \cdot 1 + \frac{30}{105} \cdot 2 + \frac{60}{105} \cdot 3= \frac{17}{7}\approx2.43$ and similarly $E_U(SNI(T_6))\approx 3.10$, $E_U(SNI(T_7))\approx 3.88$, $E_U(SNI(T_8))\approx 4.59$ and $E_U(SNI(T_9))\approx 5.34$ as well as $V_U(SNI(T_{1,2,3}))=0$, $V_U(SNI(T_4))= \frac{3}{15} \cdot \left(\frac{8}{5}-0\right)^2 + \frac{12}{15} \cdot \left(\frac{8}{5}-2\right)^2 = \frac{80}{125}=0.64$ and $V_U(SNI(T_5))= \frac{15}{105} \cdot \left(\frac{17}{7}-1\right)^2 + \frac{30}{105} \cdot \left(\frac{17}{7}-2\right)^2 +\frac{60}{105} \cdot \left(\frac{17}{7}-3\right)^2=\frac{2730}{5145}\approx0.53$, $V_U(SNI(T_6))\approx0.75 $, $V_U(SNI(T_7))\approx 0.77 $, $V_U(SNI(T_8))\approx 0.90 $ and $V_U(SNI(T_9))\approx 0.98$. 

\begin{figure}[htbp]
	\centering
	\includegraphics[width=\textwidth]{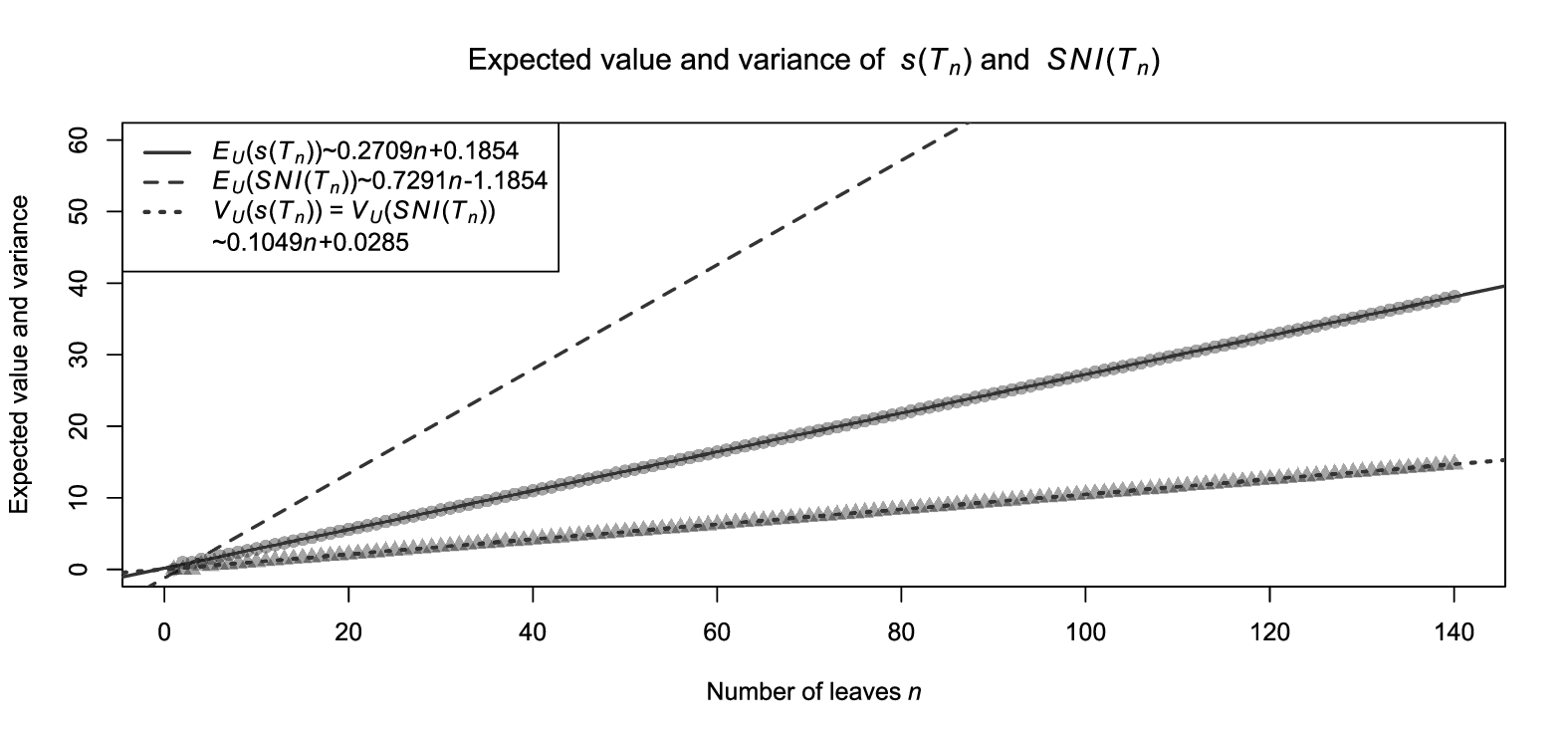}\vspace{-0.4cm}
	\caption{The expected values and variances of $s(T_n)$ and $SNI(T_n)$ under the uniform model depending on $n$. The data points represent the exact calculated data, $E_U(s(T_n))$ shown as circles and $V_U(s(T_n))=V_U(SNI(T_n))$ as triangles. The lines represent the respective linear functions approximated on the values for $n=10,...,140$. Please note that there are no visual differences for the second approximation based on $n=60,...,140$ since the differences of slope and intercept are comparatively small, which is why we decided to omit the second figure.}
	\label{fig:e_valVarSym}
\end{figure}
\begin{figure}[htbp]
	\centering
	\includegraphics[width=\textwidth]{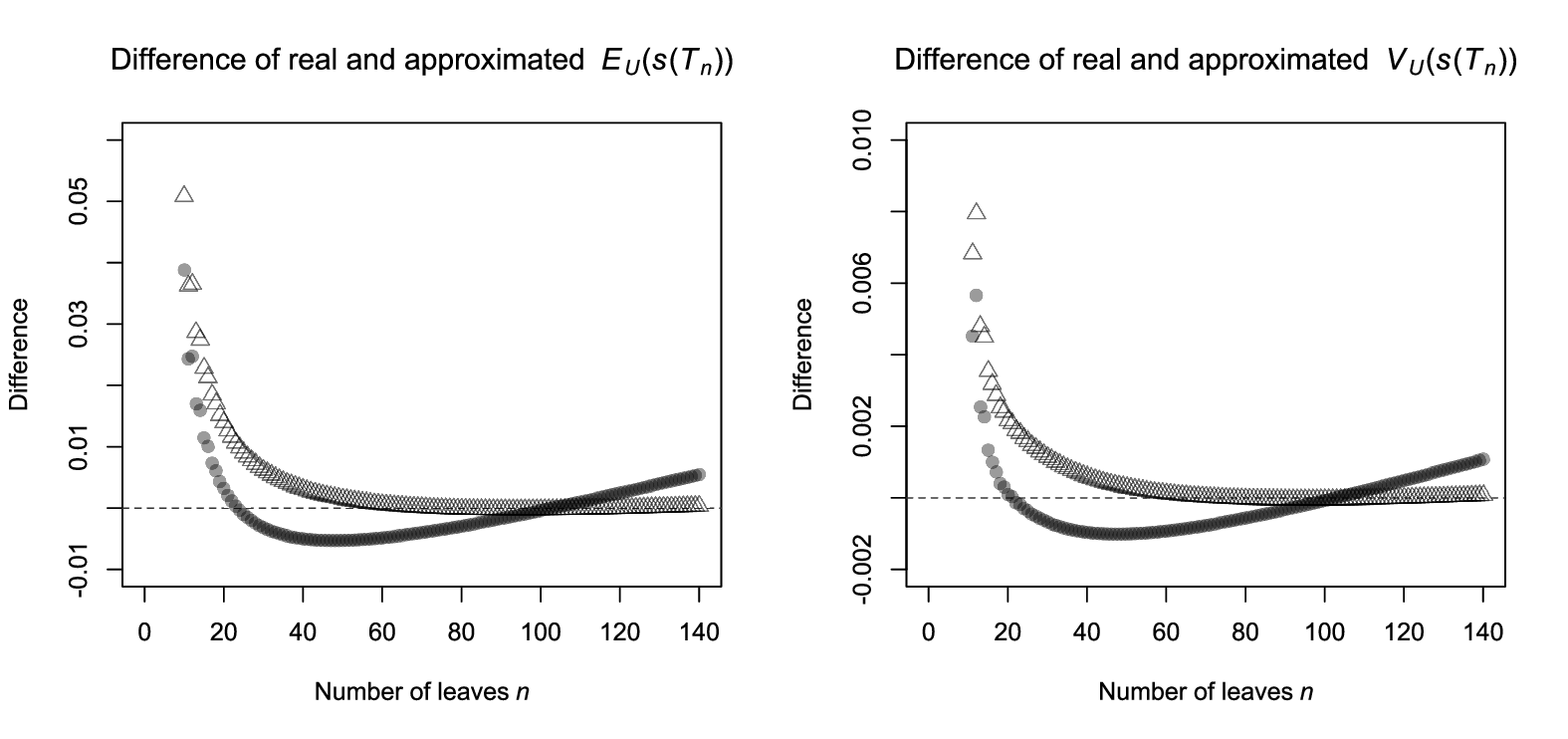}
	\caption{These plots show the difference between real and approximated expected values and variances of $s(T_n)$ with circles for the approximation based on $n=10,...,140$ (here $E_U(s(T_n))-(0.27086 n+0.18545)$ and $V_U(s(T_n))-(0.10491 n +0.02853)$) and with triangles for the approximation based on $n=60,...,140$ (here $E_U(s(T_n))- (0.271 n+0.17208)$ and $V_U(s(T_n))- (0.10494 n +0.02595)$).}
	\label{fig:residSym}
\end{figure}

\subsubsection{Variance of leaf depths}

Several of the presented indices are based on the leaf depths of the tree. The Sackin index is the sum of those depths, the average leaf depth -- like the name suggests -- is their average value and the variance of leaf depths is their variance. To be more precise, the variance of leaf depths \citep{Sackin1972, Coronado2020b} $\sigma_N^2(T)$ of a tree $T\in\Tnstar$ is defined as \[ \sigma_N^2(T) \coloneqq \frac{1}n \cdot \sum\limits_{x \in V_L(T)} \left( \delta_T(x) - \overline{N}(T) \right)^2 \] where $\overline{N}(T)$ denotes the average leaf depth of $T$.

As it is shown in the following propositions, the variance of leaf depths can be computed in time $O(n)$, it is a recursive tree shape statistic, and it is local.

\begin{proposition} \label{runtime_VLD}
For every tree $T\in\Tnstar$, the variance of leaf depths $\sigma_N^2(T)$ can be computed in time $O(n)$.
\end{proposition}
\begin{proof}
A vector containing the values $\delta_T(u)$ for each $u\in V(T)$ can be computed in time $O(n)$ by traversing the tree in pre order, setting $\delta_T(\rho)=0$ and calculating $\delta_T(u)=\delta_T(p_T(u))+1$ if $u\neq\rho$. Since the number of leaves is $n$, the average leaf depth $\overline{N}(T)$ can be computed from this vector in time $O(n)$. Then, the variance of leaf depths can be computed from this information in time $O(n)$ leading to a total computation time in $O(n)$.
\end{proof}

\begin{proposition} \label{recursiveness_VLD}
The variance of leaf depths is a recursive tree shape statistic. We have $\sigma_N^2(T)=0$ for $T\in\mathcal{T}_1^\ast$, and for every tree $T\in\Tnstar$ with $n\geq 2$ and standard decomposition $T=(T_1,\ldots,T_k)$ we have \[ \sigma_N^2(T) = \left(\sum\limits_{i=1}^k n_i\right)^{-1} \cdot \left( \sum\limits_{i=1}^k S^{(2)}(T_i) + 2\cdot\sum\limits_{i=1}^k S(T_i)\right) + 1 - \left(\sum\limits_{i=1}^k n_i\right)^{-2} \cdot \left( \sum\limits_{i=1}^k  S(T_i) + \sum\limits_{i=1}^k n_i \right)^2, \]
where $S^{(2)}(T)=\sum\limits_{x\in V_L(T)} \delta_T(x)^2$.
\end{proposition}
\begin{proof}
Let $T=(T_1,\ldots,T_k)$ be a tree in $\Tnstar$, and let $S^{(2)}(T)=\sum\limits_{x\in V_L(T)} \delta_T(x)^2$. Using $\sigma_N^2(T)=\frac{1}{n}\cdot S^{(2)}(T)-\frac{1}{n^2}\cdot S(T)^2$ (see \citet[Equation (3)]{Coronado2020b}) we have
\begin{equation*}
\begin{split}
    \sigma_N^2(T) &= \frac{1}{n}\cdot S^{(2)}(T) - \frac{1}{n^2}\cdot S(T)^2 \\
    &= \frac{1}{n}\cdot\left(\sum\limits_{i=1}^k S^{(2)}(T_i) + 2\cdot\sum\limits_{i=1}^k S(T_i) + n \right) - \frac{1}{n^2} \left( \sum\limits_{i=1}^k S(T_i) + \sum\limits_{i=1}^k n_i \right)^2\\
    &= \left(\sum\limits_{i=1}^k n_i\right)^{-1} \cdot \left( \sum\limits_{i=1}^k S^{(2)}(T_i) + 2\cdot\sum\limits_{i=1}^k S(T_i)\right) + 1 - \left(\sum\limits_{i=1}^k n_i\right)^{-2} \cdot \left( \sum\limits_{i=1}^k  S(T_i) + \sum\limits_{i=1}^k n_i \right)^2,
\end{split}
\end{equation*}
where the recursive expression for $S^{(2)}(T)$ is obtained from Lemma 11 in the supplementary material to \citet{Coronado2020b} and the recursive expression for $S(T)$ was established in Proposition \ref{recursiveness_Sackin} of the present manuscript. Thus, the variance of leaf depths can be expressed as a recursive tree shape statistic of length $x=4$ with the recursions (where $\sigma_i$ and $S^{(2)}_i$ and $S_i$ are simplified notations of $\sigma_N^2(T_i)$ and $S^{(2)}(T_i)$ and $S(T_i)$, respectively, and $n_i$ denotes the leaf number of $T_i$)
\begin{itemize}
    \item variance of leaf depths: $\lambda_1=0$ and $r_1(T_1,\ldots,T_k)=\frac{S^{(2)}_1+\ldots+S^{(2)}_k+2\cdot(S_1+\ldots+S_k)}{n_1+\ldots+n_k}+1-\frac{(S_1+\ldots+S_k+n_1+\ldots+n_k)^2}{(n_1+\ldots+n_k)^2}$
    \item $S^{(2)}$: $\lambda_2=0$ and $r_2(T_1,\ldots,T_k)=S^{(2)}_1+\ldots+S^{(2)}_k+2\cdot(S_1+\ldots+S_k)+n_1+\ldots+n_k$
    \item  Sackin index: $\lambda_3=0$ and $r_3(T_1,\ldots,T_k)=S_1+\ldots+S_k+n_1+\ldots+n_k$
    \item leaf number: $\lambda_4=1$ and $r_4(T_1,\ldots,T_k)=n_1+\ldots+n_k$
\end{itemize}
It can easily be seen that $\lambda\in\mathbb{R}^4$ and $r_i: \underbrace{\mathbb{R}^4\times\ldots\times\mathbb{R}^4}_{k\text{ times}} \rightarrow \mathbb{R}$, and that all $r_i$ are independent of the order of subtrees. This completes the proof.
\end{proof}

\begin{proposition} \label{locality_VLD}
The variance of leaf depths is not local.
\end{proposition}
\begin{proof}
Consider the two trees $T$ and $T'$ in Figure \ref{fig_locality}, which only differ in their subtrees rooted at $v$. Note that in both $T$ and $T'$ the vertex $v$ has exactly 5 descendant leaves. Nevertheless, we have $\sigma_N^2(T)-\sigma_N^2(T')=\frac{21}{25}-\frac{6}{25}=\frac{3}{5}\neq \frac{28}{25}=\frac{34}{25}-\frac{6}{25}=\sigma_N^2(T_v)-\sigma_N^2(T_v')$. Thus, the variance of leaf depths is not local. This is due to the different normalization factors $\frac{1}{n}$ for $T$ and $T'$ and $\frac{1}{n_v}$ for $T_v$ and $T_v'$, and due to the different average leaf depths $\overline{N}(T)$, $\overline{N}(T')$, $\overline{N}(T_v)$ and $\overline{N}(T_v')$.
\end{proof}

\subsection{Tree shape statistics that are balance indices}

In this section of the appendix, we will have a look at the maximal width, maximal difference in widths and maximal depth. In particular, we will show that all three of them fulfill our definition of a balance or imbalance index, results that -- to our knowledge -- have not yet been known or not yet been proven.

\subsubsection{Maximal width}

The maximal width or maximum width $mW(T)$ \citep{Colijn_phylogenetic_2014} of a binary tree $T\in\BTnstar$ with height $h(T)$ is defined as \[ mW(T) \coloneqq \max\limits_{i=0,\ldots,h(T)} w_T(i). \]

At first, we will have a look at the maximal value of $mW$ if $n$ is a power of two.

\begin{lemma} \label{lem_mW_max}
Let $T\in\BTnstar$ be a binary tree with $n=2^h$ leaves for some $h\in\mathbb{N}_{\geq 0}$. Then, $T=\Tfb$ if and only if $w_T(i)=2^i$ for all $i=0,\ldots,h$.
\end{lemma}
\begin{proof}
For $T=\Tfb$, we trivially have $w_T(i)=2^i$ for all $i=0,\ldots,h$.
Now, assume that $w_T(i)=2^i$ for all $i=0,\ldots,h$. Then, $T$ in particular contains $2^h$ vertices of depth $h$. As $n=2^h$ these vertices must all be leaves (if one or more of them were interior vertices, $T$ would have strictly more than $2^h$ leaves). In particular, $T$ contains $2^h$ leaves of depth $h$, which implies that $T = \Tfb$. This completes the proof.
\end{proof}

\begin{theorem} \label{prop_mW_max}
For every binary tree $T\in\BTnstar$ with $n=2^h$ and $h\in\mathbb{N}_{\geq 0}$, the maximal width fulfills $mW(T)\leq 2^h$. Moreover, for any given $n=2^h$ with $h\in\mathbb{N}_{\geq 0}$, there is exactly one tree $T\in\BTnstar$ reaching this upper bound, i.e. $mW(T)=2^h$, namely the fully balanced tree $\Tfb$. 
\end{theorem}
\begin{proof}
Trivially, we have $mW(\Tfb)=w_{\Tfb}(h)=2^h$ for every $h \in \mathbb{N}_{\geq 0}$. Now, let $T \in \mathcal{BT}^\ast_{2^h}$ be a rooted binary tree that maximizes $mW(T)$. As $T$ maximizes $mW$, we have $mW(T) \geq mW(\Tfb) = 2^h$. In particular, there exists a $j \in \{0, \ldots, h(T)\}$ such that $w_T(j) \geq 2^h$, i.e. $T$ contains at least $2^h$ vertices of depth $j$.
With the same reasoning as in the proof of Lemma \ref{lem_mW_max}, these vertices must all be leaves and there are precisely $2^h$ of them (i.e. the inequality is an equality). This implies that $T = \Tfb$, which completes the proof.
\end{proof}

\begin{theorem} \label{prop_mW_min}
For every binary tree $T\in\BTnstar$, the maximal width fulfills $mW(T)=1$ for $n=1$ and $mW(T)\geq 2$ for $n\geq 2$. Moreover, for any given $n\in\mathbb{N}_{\geq 1}$, there is exactly one tree $T\in\BTnstar$ reaching this lower bound, i.e. $mW(T)=1$ for $n=1$ and $mW(T)=2$ for $n\geq 2$, namely the caterpillar tree $\Tcat$.
\end{theorem}
\begin{proof}
Trivially, we have $mW(T_1^{cat})=1$ with $T_1^{cat}$ being the only tree in $\mathcal{BT}_1^\ast$, and $mW(\Tcat)=2$ for every $n \in \mathbb{N}_{\geq 2}$. Let $T \in \BTnstar$ with $n\geq 2$ have minimal $mW$ value, i.e.  $mW(T)\leq 2$. Assume $T \neq \Tcat$. Then there exists a vertex $v \in \mathring{V}(T)$ with two children $v_1,v_2 \in \mathring{V}(T)$ which again have two children each because $T$ is binary. However, it now follows that $mW(T)\geq w_T(\delta_T(v)+2)\geq 4$ which contradicts $T$ being minimal.
\end{proof}

\begin{remark} \label{remark_maxWidth}
    Since the fully balanced tree $\Tfb$ is the unique maximal tree when $n$ is a power of two (Theorem \ref{prop_mW_max}) and the caterpillar tree $\Tcat$ is the unique minimal tree on $\BTnstar$ (Theorem \ref{prop_mW_min}), we have shown that the maximal width indeed fulfills our definition of a balance index on $\BTnstar$.
\end{remark}

\subsubsection{Maximal difference in widths}

Next, we will have a look at the maximal difference in widths. Recall that the maximal difference in widths \citep{Colijn_phylogenetic_2014} of a binary tree $T\in\BTnstar$ with height $h(T)\geq 1$ is defined as
\[ delW(T) = \Delta W(T) \coloneqq \max\limits_{i=0,\ldots,h(T)-1}  |w_T(i+1)-w_T(i)|. \]

We will again start with the maximal value of $delW$ if $n$ is a power of two.

\begin{theorem} \label{prop_delW_max}
For every binary tree $T\in\BTnstar$ with $n=2^h$ and $h\in\mathbb{N}_{\geq 1}$, the maximal difference in widths fulfills $delW(T)\leq 2^{h-1}$. Moreover, for any given $n=2^h$ with $h\in\mathbb{N}_{\geq 1}$, there is exactly one tree $T\in\BTnstar$ reaching this upper bound, i.e. $delW(T)=2^{h-1}$, namely the fully balanced tree $\Tfb$.
\end{theorem}
\begin{proof} 
Using Lemma \ref{lem_mW_max}, we have for every $n=2^h$ with $h\in\mathbb{N}_{\geq 1}$ \[ delW(\Tfb) = \max\limits_{i=0,\ldots,h-1} |w_T(i+1)-w_T(i)| = \max\limits_{i=0,\ldots,h-1} 2^{i+1}-2^i = 2^{h-1}. \] Let $T \in\BTnstar$ with $n=2^h$ for some $h\in\mathbb{N}_{\geq 1}$ have maximal $delW$, i.e. $delW(T)\geq 2^{h-1}$. This implies that there exists a $j \in \{0,\ldots,h(T)-1\}$ such that $w_T(j+1)-w_T(j)\geq 2^{h-1}$ and thus $w_T(j+1)\geq 2^{h-1}+w_T(j)$. Since, additionally, all rooted binary trees fulfill $w_T(0)=1$ and $w_T(i+1)\leq 2\cdot w_T(i)$ for all $i\in\mathbb{N}_{\geq 0}$, we have $2\cdot w_T(j) \geq w_T(j+1) \geq w_T(j)+2^{h-1}$ which leads to $w_T(j) \geq 2^{h-1}$ and $w_T(j+1) \geq 2^h$. Now, using the same argument as in the proof of Theorem \ref{prop_mW_max} 
and Lemma \ref{lem_mW_max}, we have $T=\Tfb$.
\end{proof}

\begin{theorem} \label{prop_delW_min}
For every binary tree $T\in\BTnstar$ with $n\geq 2$, the maximal difference in widths fulfills $delW(T)\geq 1$. Moreover, for any given $n\in\mathbb{N}_{\geq 2}$, there is exactly one tree $T\in\BTnstar$ reaching this lower bound, i.e. $delW(T)=1$, namely the caterpillar tree $\Tcat$.
\end{theorem}
\begin{proof}
Trivially, we have $delW(\Tcat)=1$ for every $n \in \mathbb{N}_{\geq 2}$ because there is one vertex of depth zero (the root) and all other depths have precisely two vertices. Let $T \in \BTnstar$ with $n\geq 2$ have minimal $delW$, i.e. $delW(T)\leq 1$. This implies $w_T(i+1) \leq w_T(i)+1$ for all $i=0,\ldots,h(T)-1$.\\ 
Let $w_T(i)=2$ hold for any depth $0\leq i <h(T)$ then the only possible subsequent widths $w_T(i+1)$ in a binary tree are 0 (two leaves at depth $i$), 2 (one leaf and one inner node) and 4 (two inner nodes). However, 0 is not possible because there has to exist at least one node at depth $i+1\leq h(T)$ and 4 is not possible either as it contradicts $w_T(i+1) \leq w_T(i)+1$. Hence, $w_T(i+1)=2$ is the only possible subsequent width, i.e. at depth $i$ there has to be one leaf and one inner vertex. Since $w_T(0)=1$ and $w_T(1)=2$ because $n\geq 2$, we can conclude $w_T(i)=2$ for all $i=1,\ldots,h(T)$ which implies $T=\Tcat$.
\end{proof}

\begin{remark} \label{remark_maxdelW}
    Again, since we have shown that the fully balanced tree $\Tfb$ is the unique tree maximizing $delW$ when $n$ is a power of two (Theorem \ref{prop_delW_max}) and the caterpillar tree is the unique tree minimizing it on $\BTnstar$ (Theorem \ref{prop_delW_min}), the maximal difference in widths fulfills our definition of a balance index on $\BTnstar$ (see Definition \ref{def_balance}).
\end{remark}

\subsubsection{Maximal depth}

Lastly, we will show that the maximal depth is an imbalance index.

The maximal depth $mD(T)$ \citep{Colijn_phylogenetic_2014} of a binary tree $T\in\BTnstar$ is defined as \[ mD(T) \coloneqq \max\limits_{l\in V_L(T)} \delta_T(l) = \max\limits_{v\in V(T)} \delta_T(v) = h(T). \]

Since we have not been able to find proofs for the following statements, we will provide short proofs here.

\begin{theorem} \label{prop_mD_max}
For every binary tree $T\in\BTnstar$, the maximal depth fulfills $mD(T)\leq n-1$. Moreover, for any given $n\in\mathbb{N}_{\geq 1}$, there is exactly one tree in $\BTnstar$, namely the caterpillar tree $\Tcat$, reaching this upper bound, i.e. $mD(\Tcat)=n-1$. 
\end{theorem}
\begin{proof}
Note that each ancestor $w$ of a vertex $v$ with $w\neq v$ must fulfill $w\in\mathring{V}(T)$, and that $\delta_T(v)$ is exactly the number of such vertices $w$ in $T$. Together, we have $\delta_T(v)\leq |\mathring{V}(T)|=n-1$ for all $v\in V(T)$, which implies that $mD(T)=\max\limits_{v\in V(T)} \delta_T(v)\leq n-1$. Also note that $\Tcat$ is by definition the only tree in $\BTnstar$ containing at most one cherry. This means that it is also the only binary tree that actually contains a leaf that has all inner vertices as its ancestors and therefore $mD(\Tcat)=n-1$. 
\end{proof}

Now, we will have a look at the minimal value of $mD$ if $n$ is a power of two.

\begin{lemma} \label{lem_mD_min}
Let $T\in\BTnstar$ be a binary tree with height $h$. Then $T$ has at most $2^h$ leaves. Additionally, we have $T=\Tfb$ if and only if $n=2^h$.
\end{lemma}
\begin{proof} 
Let $T\in\BTnstar$ be a binary tree with height $h$. This means that we have $w_T(i)\leq 2^i$ for all $i=0,\ldots,h$ \cite[p.~400]{Knuth1}, as well as $w_T(i)=0$ for $i>h$. From
\begin{equation} \label{eq_mD_min}
    2n-1=|V(T)|=\sum\limits_{i=0}^{h}w_T(i) \leq \sum\limits_{i=0}^{h} 2^i = 2^{h+1}-1 = 2\cdot2^h-1
\end{equation}
we can directly conclude that $n\leq 2^h$. Moreover, Inequality \eqref{eq_mD_min} is an equality (implying $n=2^h$) if and only if $w_T(i)=2^i$ for all $i=0,\ldots,h$. Using Lemma \ref{lem_mW_max} this is the case if and only if $T=\Tfb$.
\end{proof}

\begin{theorem} \label{prop_mD_min}
For every binary tree $T\in\BTnstar$ with $n=2^h$ and $h\in\mathbb{N}_{\geq 0}$, the maximal depth fulfills $mD(T)\geq h$. Moreover, for any given $n=2^h$ with $h\in\mathbb{N}_{\geq 0}$, there is exactly one tree in $\BTnstar$, namely the fully balanced tree $\Tfb$, reaching this lower bound, i.e. $mD(\Tfb)=h$.
\end{theorem}
\begin{proof}
Lemma \ref{lem_mD_min} implies that any rooted binary tree $T$ with strictly more than $2^{h-1}$ leaves has height $h(T)>h-1$. Thus, any tree $T\in\BTnstar$ with $n=2^h$ fulfills $h(T)>h-1$ and thus $mD(T)=h(T)\geq h$. In Lemma \ref{lem_mD_min} it has also been shown that $\Tfb$ is the only tree with $n=2^h$ leaves and $h(T)=h$, and thus also the only tree with $n=2^h$ leaves and $mD(T)=h$.\\
\end{proof}

\begin{remark} \label{remark_maxdep}
    Since the caterpillar tree $\Tcat$ is the unique maximal tree on $\BTnstar$ (Theorem \ref{prop_mD_max}) and the fully balanced tree $\Tfb$ is the unique minimal tree when $n$ is a power of two (Theorem \ref{prop_mD_min}), the maximal depth is indeed an imbalance index on $\BTnstar$ according to Definition \ref{def_imbalance}.
\end{remark}

\subsection{Tree shape statistics that are not (im)balance indices} \label{Sec_nbi_tss_additional_proofs}

\subsubsection{Figures accompanying Table~\ref{Table_nbi_tss}} \label{Appendix_additionalfigures}
\begin{figure}[htbp]
    \centering
    \includegraphics[scale=0.15]{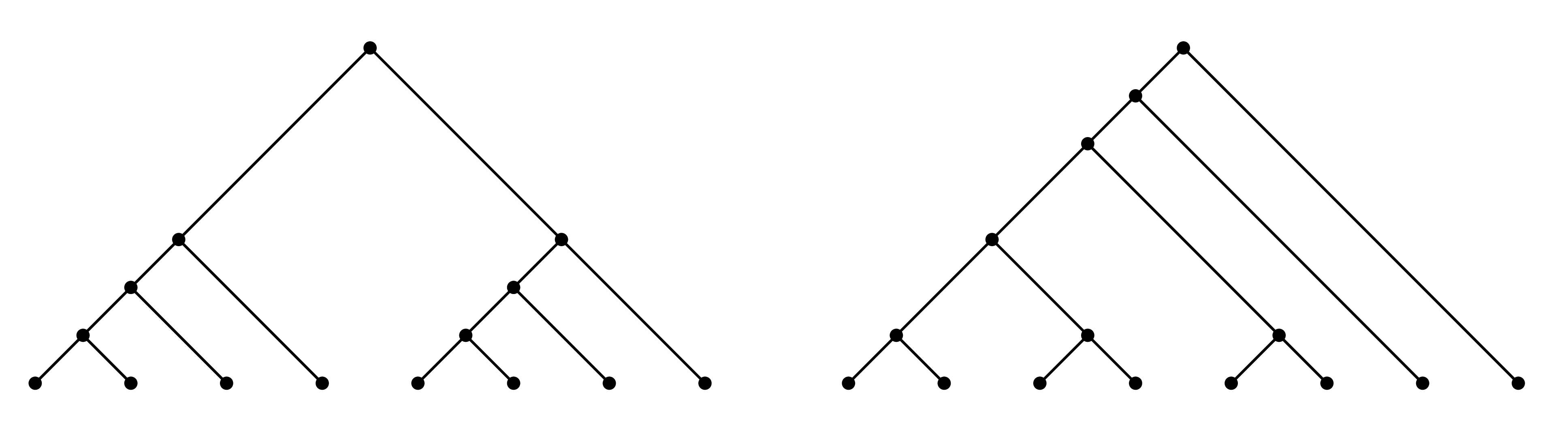}
    \caption{The unique rooted binary tree with 8 leaves and maximal APP index (left) and the unique rooted binary tree with 8 leaves and minimal APP index (right).}
    \label{Fig_APP_Extremal}
\end{figure}

\begin{figure}[htbp]
    \centering
    \includegraphics[scale=0.275]{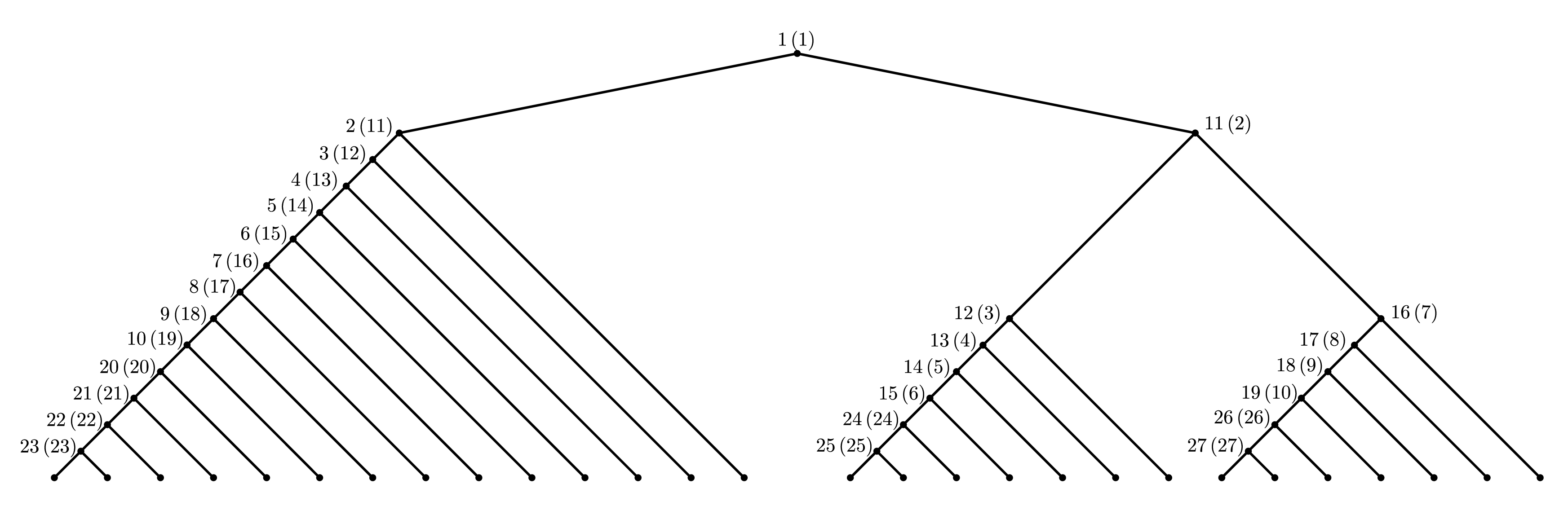}
    \caption{Tree $T$ has a different $\overline{I}_{10}'$ index, depending on the ranking used -- if we use the ranking not given in the brackets, nine of the ten oldest vertices are in the left (caterpillar) subtree and we get $\overline{I}_{10}'(T)=\frac{1}{10}\cdot\frac{2367}{280}\approx 0.8454$. If we take the second ranking (given in brackets), nine of the ten oldest vertices are in the right subtree, and we get $\overline{I}_{10}'(T)=\frac{1}{10}\cdot \frac{43}{6}\approx 0.7167$. Thus, the $\overline{I}_{10}'$ value does not solely depend on the tree shape, which is why we do not consider it an (im)balance index.}
    \label{fig_ex_Ivprime}
\end{figure}

\begin{figure}[htbp]
    \centering
    \includegraphics[scale=0.15]{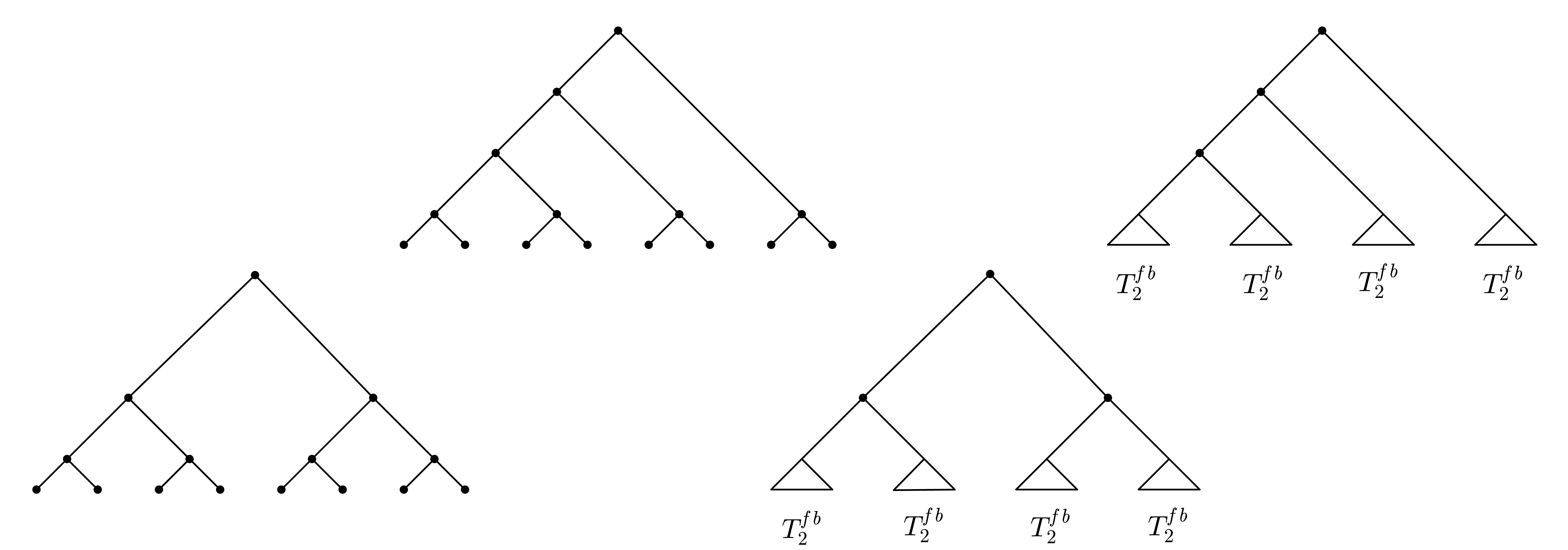}\\
    \caption{On the left there are two trees that have the maximal number of cherries for $n=8$ and show that $\Tfb$ is not the \emph{unique} maximal tree. The same counterexample applies to the modified cherry index as well as to the tree shape statistics ILnumber and ladder length. With a similar example on the right we can show for two further tree shape statistics that $\Tfb$ is not the \emph{unique} most balanced tree for all given $n=2^h$ with $h\in\mathbb{N}_{\geq 0}$: Both of these trees with 16 leaves each are maximal regarding the number of double cherries $T^\mathit{fb}_2$ with $dc(T)=4$ and both are minimal regarding the number of 4-caterpillars with $cat_4(T)=0$.}
    \label{Fig_cherry_max}
\end{figure}

\begin{figure}[htbp]
    \centering
    \begin{subfigure}[b]{0.25\textwidth}
    \includegraphics[scale=0.15]{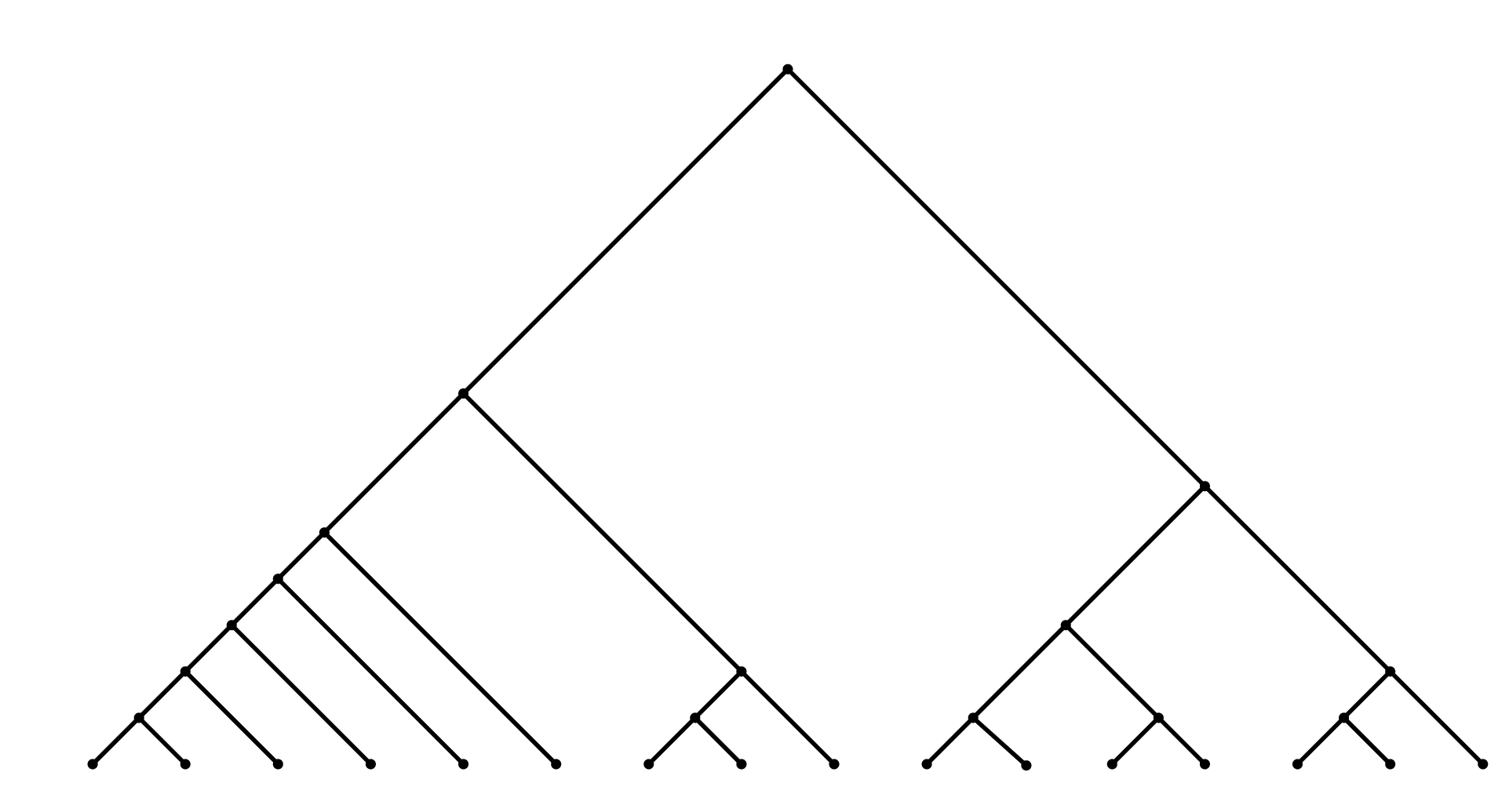}
    \caption{}
    \end{subfigure}
    \qquad
    \begin{subfigure}[b]{0.25\textwidth}
    \includegraphics[scale=0.15]{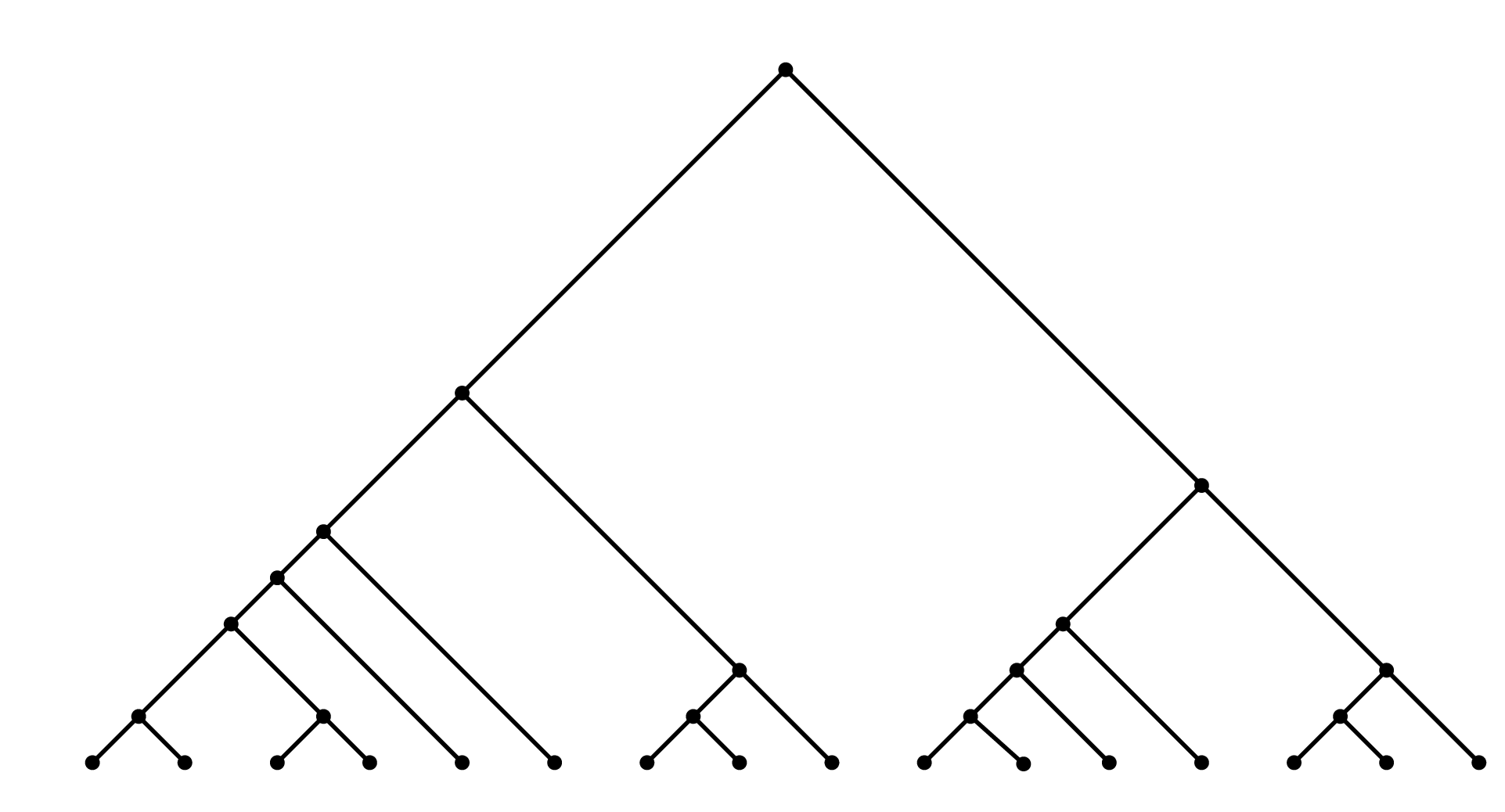}
    \caption{}
    \end{subfigure}
    \qquad
    \begin{subfigure}[b]{0.25\textwidth}
    \includegraphics[scale=0.15]{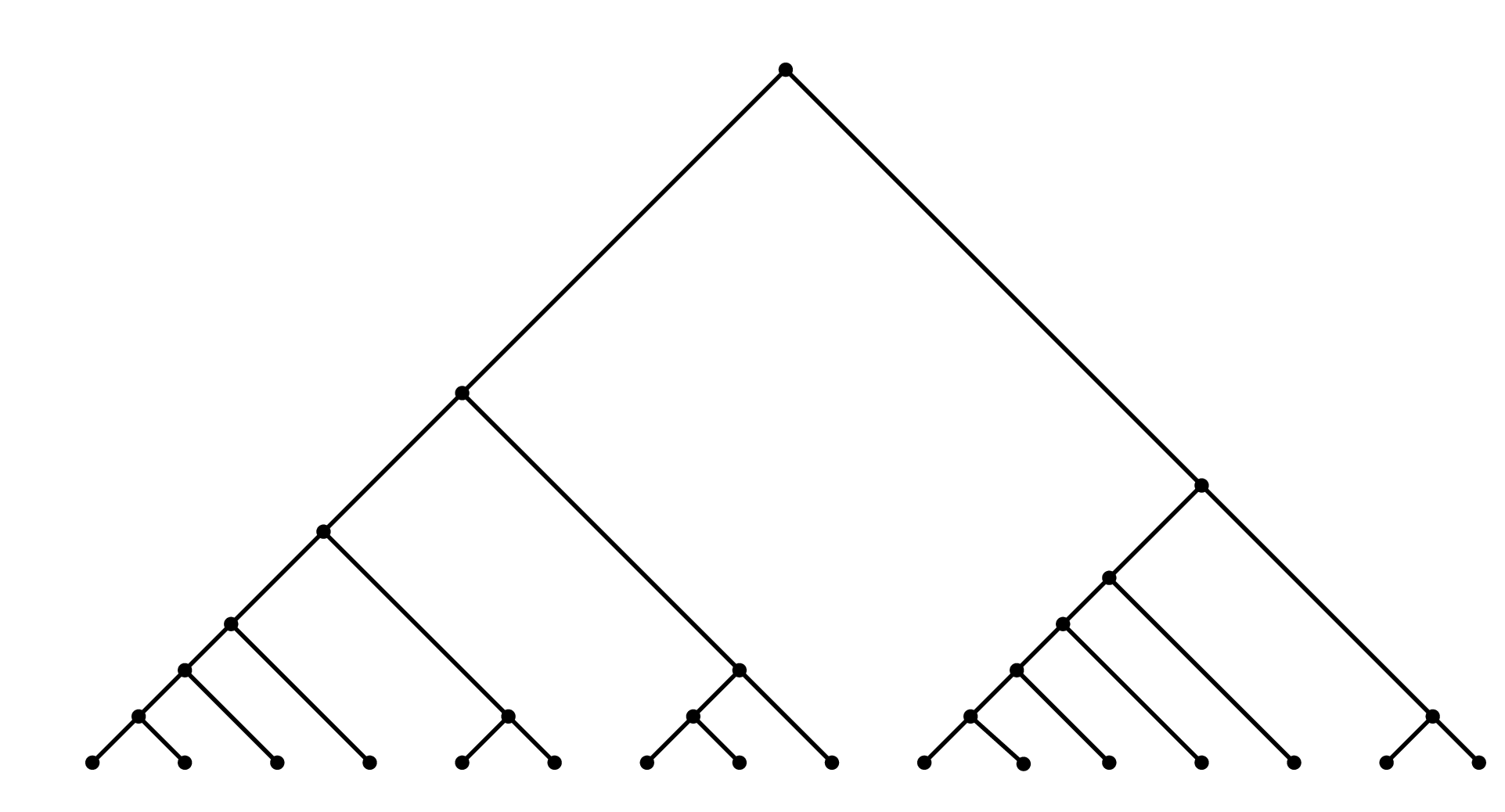}
    \caption{}
    \end{subfigure}
    \qquad
    \begin{subfigure}[b]{0.25\textwidth}
    \includegraphics[scale=0.15]{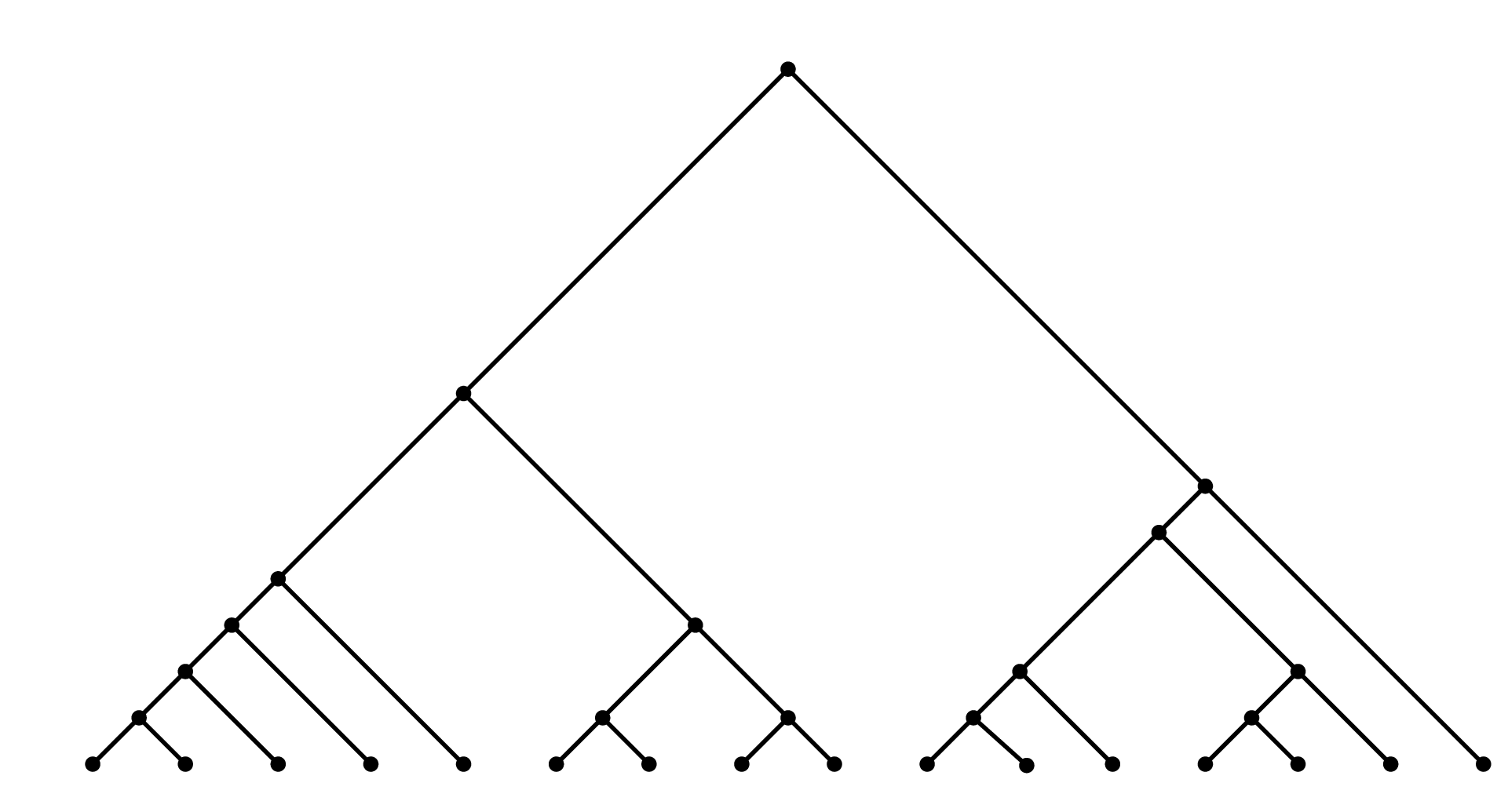}
    \caption{}
    \end{subfigure}
    \qquad
    \begin{subfigure}[b]{0.25\textwidth}
    \includegraphics[scale=0.15]{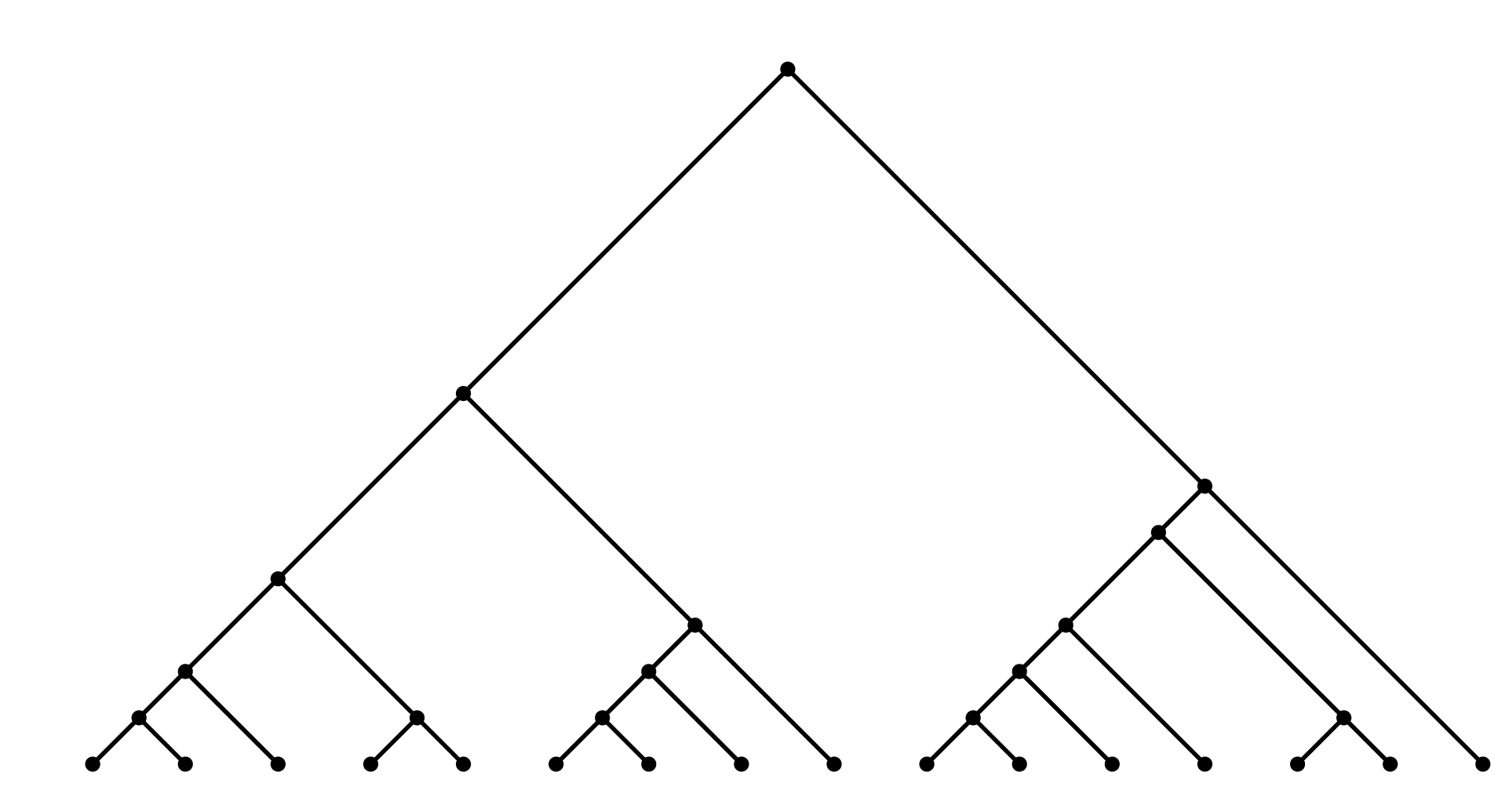}
    \caption{}
    \end{subfigure}
    \qquad
    \begin{subfigure}[b]{0.25\textwidth}
    \includegraphics[scale=0.15]{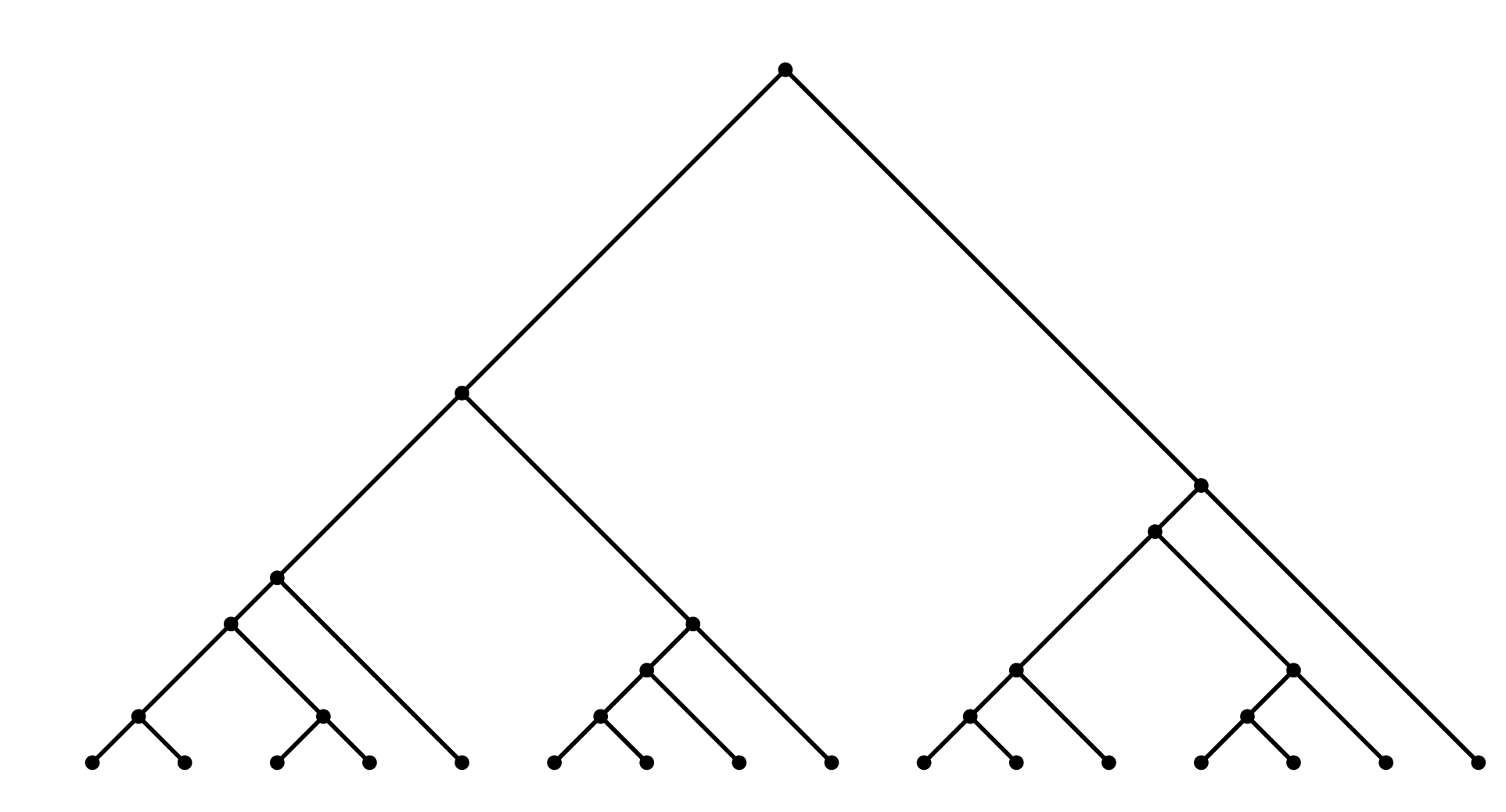}
    \caption{}
    \end{subfigure}
    \qquad
    \begin{subfigure}[b]{0.25\textwidth}
    \includegraphics[scale=0.15]{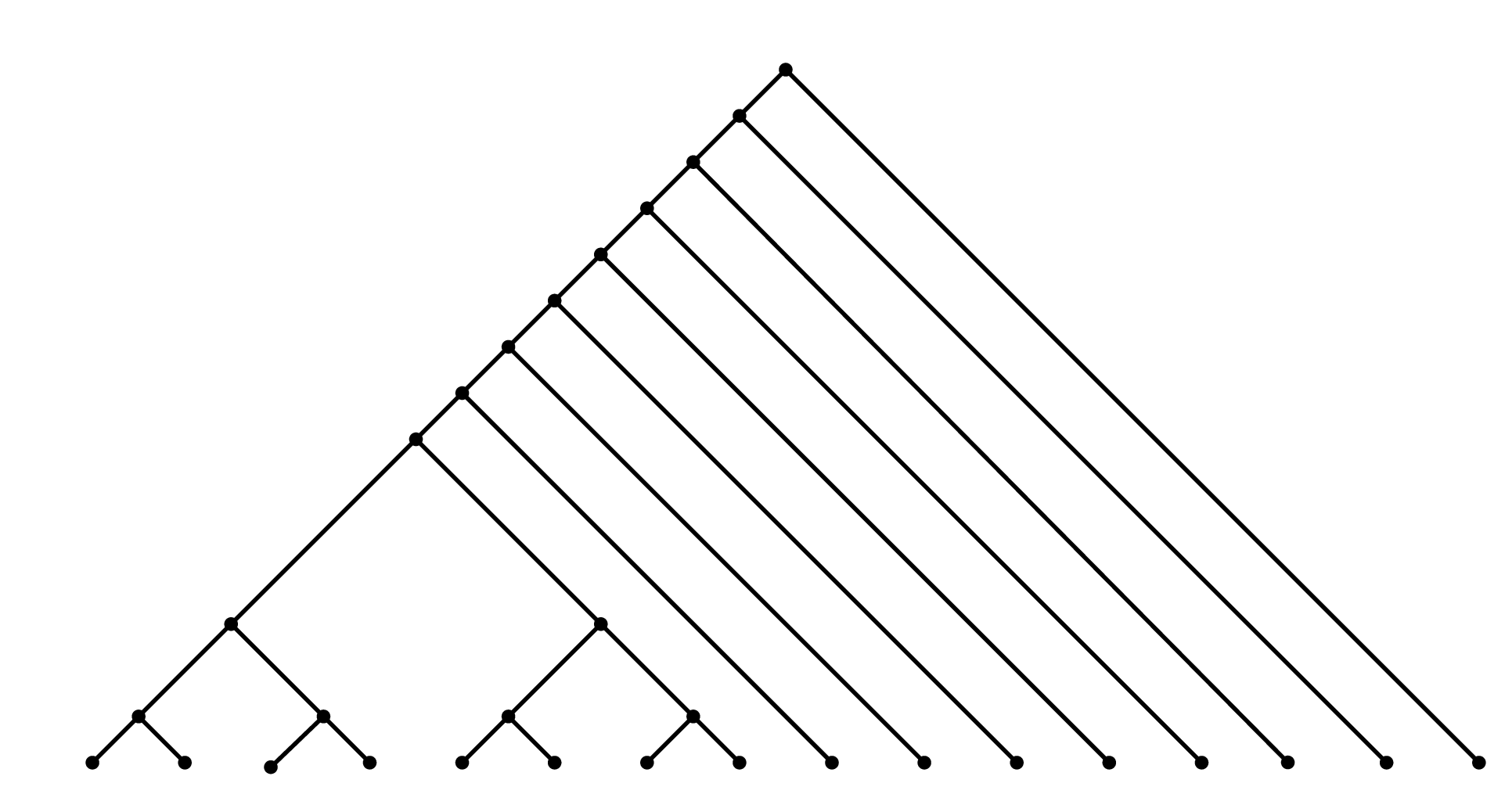}
    \caption{}
    \end{subfigure}
     \caption{The six minimal trees ((a) -- (f)) and the unique maximal tree ((g)) on 16 leaves as induced by $D$ (found with an exhaustive search using the computer algebra system Mathematica \cite{Mathematica}). This shows that neither $T_{16}^\mathit{cat}$ nor $T_4^\mathit{fb}$ are extremal, which is why we do not consider $D$ an (im)balance index.}
    \label{Fig_d1Trees}
\end{figure}
\clearpage

\subsubsection{Cherry index}

Although the cherry index appears quite frequently in the literature about tree balance, it is \emph{not} an (im)balance index according to our definition, because the fully balanced tree is not a unique extreme on $\BTnstar$ when $n$ is a power of two. However, since it is so popular, we do provide a fact sheet for it (see Section \ref{Sec_factsheets}). For the sake of completeness, we also provide some additional results in this section.

Recall that the cherry index of a tree $T\in\Tnstar$, denoted $ChI(T)$, is defined as the number of its cherries. We will now present some results on its computation time, its recursiveness and its locality.

\begin{proposition} \label{runtime_cherry}
For every tree $T\in\Tnstar$, the cherry index $ChI(T)$ can be computed in time $O(n)$.
\end{proposition}
\begin{proof}
The descendants of each node can be determined in  $O(n)$. Now, for every inner node $v$, we count how many of its descendants are leaves and denote their number by $ld(v)$. The computation time for this step is in $O(n)$ because there are only $2n-2$ descendants in total, and checking whether a node is a leaf takes constant time. In linear time we can then summarize $\binom{ld(v)}{2}$ over all those at  most $ n-1$ inner nodes.
\end{proof}

In 2007, Matsen showed that the cherry index is a binary recursive tree shape statistic \cite{Matsen2007}. The following proposition proves that it is also a recursive tree shape statistic when arbitrary trees are considered.

\begin{proposition} \label{recursiveness_cherry}
The cherry index is a recursive tree shape statistic. We have $ChI(T)=0$ for $T\in\mathcal{T}_1^\ast$, and for every tree $T\in\Tnstar$ with $n\geq 2$ and standard decomposition $T=(T_1,\ldots,T_k)$ we have \[ ChI(T)=\sum\limits_{i=1}^k ChI(T_i)+\binom{\sum\limits_{i=1}^k \mathcal{I}(ChI(T_i)=0)}{2}. \]
\end{proposition}
\begin{proof}
Recall that the cherry index is defined as the number of its cherries. Let $T\in\Tnstar$ be a tree with its standard decomposition $T=(T_1,\ldots,T_k)$ and let $u,v\in V_L(T)$ be two leaves in $T$. There are two cases to consider: 1) $u,v\in V_L(T_i)$ for some $i\in\{1,\ldots,k\}$. In this case, $u$ and $v$ form a cherry in $T$ if and only if they form a cherry in $T_i$. And 2) $u\in V_L(T_i)$ and $v\in V_L(T_j)$ for some $i,j\in\{1,\ldots,k\}$ with $i\neq j$. In this case, $u$ and $v$ form a cherry in $T$ if and only if $T_i$ and $T_j$ consist of only one leaf each, i.e. $n_i=n_j=1$. Since every tree with at least two leaves has at least one cherry, we have $n_i=n_j=1$ if and only if $ChI(T_i)=ChI(T_j)=0$. Also note that each choice of $i,j$ with $ChI(T_i)=ChI(T_j)=0$ induces a cherry. Taking case 1) and 2) together, we have
\begin{equation*}
    ChI(T)=\sum\limits_{i=1}^k ChI(T_i)+\binom{\sum\limits_{i=1}^k \mathcal{I}(ChI(T_i)=0)}{2}.
\end{equation*}
Thus, the cherry index can be expressed as a recursive tree shape statistic of length $x=1$ with the recursion (where $c_i$ is the simplified notation of $ChI(T_i)$)
\begin{itemize}
    \item cherry index: $\lambda_1=0$ and $r_1(T_1,\ldots,T_k)=c_1+\ldots+c_k+\binom{\mathcal{I}(c_1=0)+\ldots+\mathcal{I}(c_k=0)}{2}$.
\end{itemize}
It can easily be seen that $\lambda\in\mathbb{R}$ and $r_1:\underbrace{\mathbb{R}\times\ldots\times\mathbb{R}}_{k\text{ times}}\rightarrow\mathbb{R}$, and that $r_1$ is independent of the order of subtrees. This completes the proof.
\end{proof}

\begin{proposition} \label{locality_cherry}
The cherry index is local.
\end{proposition}
\begin{proof}
Recall that the cherry index is defined as the number of its cherries. When replacing a subtree $T_v$ in $T$ by a subtree $T_v'$ on the same number of leaves to obtain $T'$ there are two cases to consider:
\begin{enumerate}
    \item  $v$ is a leaf. Since $T_v$ and $T_v'$ must have the same number of leaves and there is only one tree shape with exactly one leaf, we have $T_v=T_v'$ and thus $T=T'$. This implies $ChI(T)-ChI(T')=0=ChI(T_v)-ChI(T_v')$.
    \item $v$ is not a leaf. This implies that any two leaves $u,w\in V_L(T)$ that form a cherry in $T$ or in $T'$ fulfill either $u,w\in V_L(T_v)$, respectively $V_L(T_v'),$ or $u,w\in V_L(T)\setminus V_L(T_v)=V_L(T) \setminus V_L(T_v')$. Also note that $p_T(u)=p_{T'}(u)$ if $u\in V_L(T)\setminus V_L(T_v)$ and that $p_T(u)=p_{T_v}(u)$ and $p_{T'}(u)=p_{T_v'}(u)$ if $u\in V_L(T_v)$, respectively $V_L(T_v')$, and $|V_L(T_v)|=|V_L(T_v')|\geq 2$, which is fulfilled as $v$ is not a leaf. Using these properties, we can split the cherry sets of $T$ and $T'$ in this case as follows:
\begin{equation*}
\begin{split}
    &ChI(T)-ChI(T')\\
    &=\quad |\{\{u,w\}:u,w\in V_L(T_v), p_T(u)=p_T(w)\}| + |\{\{u,w\}:u,w\in V_L(T)\setminus V_L(T_v), p_T(u)=p_T(w)\}|\\
    &\quad- |\{\{u,w\}:u,w\in V_L(T_v'), p_{T'}(u)=p_{T'}(w)\}| - |\{\{u,w\}:u,w\in V_L(T)\setminus V_L(T_v'), p_{T'}(u)=p_{T'}(w)\}|\\
    &=\quad |\{\{u,w\}:u,w\in V_L(T_v), p_{T_v}(u)=p_{T_v}(w)\}| + |\{\{u,w\}:u,w\in V_L(T)\setminus V_L(T_v), p_T(u)=p_T(w)\}|\\
    &\quad- |\{\{u,w\}:u,w\in V_L(T_v'), p_{T_v'}(u)=p_{T_v'}(w)\}| - |\{\{u,w\}:u,w\in V_L(T)\setminus V_L(T_v'), p_T(u)=p_T(w)\}|\\ 
    &=\quad |\{\{u,w\}:u,w\in V_L(T_v), p_{T_v}(u)=p_{T_v}(w)\}| - |\{\{u,w\}:u,w\in V_L(T_v'), p_{T_v'}(u)=p_{T_v'}(w)\}|\\
    &=\quad ChI(T_v)-ChI(T_v').
\end{split}
\end{equation*} 
\end{enumerate}
Since we have proven the locality criterion in both cases, the cherry index is local.
\end{proof}

Next, we will have a look at the maximal and minimal value of the cherry index.

\begin{theorem} \label{prop_cherry_max_a}
For every tree $T\in\Tnstar$, the cherry index fulfills $ChI(T)\leq\binom{n}{2}$. This bound is tight for all $n\in\mathbb{N}_{\geq 1}$. Also, for any given $n\in\mathbb{N}_{\geq 1}$, there is exactly one tree $T\in\Tnstar$ with maximal cherry index, i.e. $ChI(T)=\binom{n}{2}$, namely the rooted star tree $\Tstar$.
\end{theorem}
\begin{proof}
We have $ChI(\Tstar)=\binom{n}{2}$ because any pair of leaves forms a cherry. For $n=1,2$ the claim is true as $\Tstar$ is the only tree in this case. 
For $n\geq 3$, consider $T \in \Tnstar$ with $T \neq \Tstar$. Then, the root of $T$ has at least two direct descendants $a$ and $b$, one of which must be an inner vertex (without loss of generality let $a$ denote this vertex). Then, there must exist a leaf $u \in V_L(T_a)$ as well as a leaf $v \in V_L(T) \setminus V_L(T_a)$ because $b$ either is a leaf or has descending leaves. In particular, the pair $\{u,v\}$ cannot form a cherry, and thus $ChI(T)<ChI(\Tstar)$. This completes the proof.
\end{proof}

\begin{theorem} \label{prop_cherry_min_a}
For every tree $T\in\Tnstar$ with $n\geq 2$, the cherry index fulfills $ChI(T)\geq 1$. This bound is tight for all $n\in\mathbb{N}_{\geq 2}$. Also, for any given $n\in\mathbb{N}_{\geq 2}$, there is exactly one tree $T\in\Tnstar$ with minimal cherry index, i.e. $ChI(T)=1$, namely the caterpillar tree $\Tcat$.
\end{theorem}
\begin{proof}
Every rooted tree with at least two leaves has a cherry, i.e. $ChI(T)\geq 1$ for all $T\in\Tn^*$. The caterpillar tree is by definition the unique tree in $\BTnstar$ with precisely one cherry. Thus, it remains to show that there is no strictly  non-binary tree with exactly one cherry. For $n=2$ this is clear, as there is no strictly non-binary tree. 
For $n\geq 3$, let $T \in \Tnstar$ be a strictly non-binary tree. Then, $T$ has a vertex $v$ with at least three direct descendants that are either leaves ($l$) or inner vertices ($i$), i.e. we have the possibilities $\{l,l,l,\ldots\}$, $\{l,l,i,\ldots\}$, $\{l,i,i,\ldots\}$ or $\{i,i,i,\dots\}$ for the (mult)iset of descendants of $v$. In each case, we can find at least two cherries by using the fact that a pending subtree whose root is an inner vertex ($i$) has at least one cherry, or by forming a cherry from two single leaves ($l$). This completes the proof.
\end{proof}

\subsubsection{Clades of size \texorpdfstring{$x$}{x}}

In this subsection, we will show that the tree shape statistic called clades of size $x$ is not an (im)balance index for any choice of $x$ (according to our definition). To begin with, the tree shape statistic clades of size $x$ \citep{Rosenberg_mean_2006} of a binary tree $T\in\BTnstar$ with $x, n \in \mathbb{N}_{\geq 1}$, denoted $num_x(T)$, is defined as the number of pending subtrees in $T$ whose leaf number (the clade size) is precisely $x$. 

\begin{lemma} \label{lem_Numx_nbi}
For all $n\geq 5$ and for all $4\leq x < n$ there exists a rooted binary tree $T \in \BTnstar$ with $num_x(T)=0$, i.e. $T$ does not contain any pending subtree with precisely $x$ leaves.
\end{lemma}
\begin{proof}
This can be proven by induction on $n\geq 5$. For the base case $n=5$ and the only possible option $x=4$, we can use the tree $T=(T^\mathit{cat}_3,T^\mathit{cat}_2)$ which does not contain any pending subtree of size $x=4$ and therefore $num_4(T)=0$. Let the assertion hold up to some $n\geq 5$ and all $4\leq x < n$ and consider a rooted binary tree with $n+1$ leaves. In the case of $x=n$, we choose $T=(T^\mathit{cat}_{n-1},T^\mathit{cat}_2)$ and obtain $num_n(T)=0$. If $x \in \{4,\ldots,n-1\}$ we use the induction hypothesis to find a tree $T' \in\BTnstar$ with $num_x(T')=0$. Then, we choose $T=(T',T^{cat}_1)$ and obtain $num_x(T)=0$ because the subtrees in $T'$ remain unchanged, $x>1=|V_L(T^{cat}_1)|$ and $x<n+1=|V_L(T)|$.
\end{proof}

\begin{proposition} \label{prop_Numx_nbi}
For any choice of $x, n \in \mathbb{N}_{\geq 1}$, the tree shape statistic $num_x$ is not a balance index.
\end{proposition}
\begin{proof} 
For $x=1$, all trees $T\in\BTnstar$ have the same value $num_1(T)=n$, i.e. $\Tcat$ can neither have the unique minimal nor maximal value. For $x=2$ and $x=3$ the statistic $num_x$ matches the cherry index and the number of pitchforks, respectively, and both can be shown to not fulfill the balance index definition (counterexamples can be found in Section \ref{Sec_nbi_tss} of this manuscript). Similarly, $num_x$ is not a balance index for $x=n$, because any tree $T\in\BTnstar$ has $num_n(T)=1$, and for $x>n$, because any tree $T\in\BTnstar$ has $num_x(T)=0$.\\
Last but not least, we show that for all $x\geq 4$, the statistic $num_x$ is not a balance index on $\BTnstar$ with $n>x$, because $\Tcat$ is neither the unique tree with minimal nor maximal $num_x$ value: For all $4\leq x<n$, we have $num_x(\Tcat)=1$ because the caterpillar tree contains precisely one pending subtree of size $1,\dots,n-1$ and $n$, respectively. Using Lemma \ref{lem_Numx_nbi} we can therefore conclude that $\Tcat$ does not minimize $num_x$ on $\BTnstar$. Additionally, as $x\geq 4$ we have $we(x)\geq 2$ and therefore there exist $T_x$ and $T_x' \in \mathcal{BT}_x^\ast$ with $T_x \neq T_x'$. Choose an arbitrary tree $T_{n-x}\in \mathcal{BT}_{n-x}^\ast$, and consider the trees $T=(T_x,T_{n-x})$ and $T'=(T_x',T_{n-x})$. Then, we have $num_x(T)=num_x(T')\geq 1$ implying that $\Tcat$ is not the unique tree maximizing $num_x$ on $\BTnstar$, either. This completes the proof.
\end{proof}

\end{document}